\documentclass[a4paper,12pt,twoside]{book}

\usepackage[left=2.5cm,right=2.5cm,top=2.5cm,bottom=2.5cm]{geometry}
\usepackage[T1]{fontenc}
\usepackage[frenchb,english]{babel}
\usepackage{floatflt}
\usepackage{fancybox}
\usepackage{stmaryrd}
\usepackage{fancyhdr}
\usepackage{cite}
\usepackage{hyperref}
\usepackage{amsfonts,amssymb,amsmath,exscale,bbm}
\usepackage{footnpag} 
\usepackage[all,knot]{xy}  
\usepackage{color}
\usepackage{graphicx}
\usepackage{epsfig}
\usepackage{subfig}
\usepackage{enumerate}
\usepackage{stmaryrd}

\usepackage{amsmath}
\usepackage{amsfonts}
\usepackage{graphicx}
\usepackage{psfrag}
\usepackage{array}
\usepackage{bbm}
\usepackage{hyperref}
\usepackage{amsthm}
\theoremstyle{definition}
\newtheorem{definition}{Definition}
\newtheorem{theorem}{Theorem}

\newtheorem{proposition}{Proposition}
\newtheorem{corollary}{Corollary}
\newtheorem{conjecture}{Conjecture}
\newtheorem{hypothesis}{Working hypothesis}

\newcommand{\tr}{\mathrm{tr}}

\newcommand{\cA}{{\mathcal A}}
\newcommand{\cB}{{\mathcal B}}
\newcommand{\cC}{{\mathcal C}}
\newcommand{\cD}{{\mathcal D}}
\newcommand{\cE}{{\mathcal E}}
\newcommand{\cF}{{\mathcal F}}
\newcommand{\cG}{{\mathcal G}}
\newcommand{\cH}{{\mathcal H}}

\newcommand{\cL}{{\mathcal L}}
\newcommand{\cM}{{\mathcal M}}
\newcommand{\cN}{{\mathcal N}}
\newcommand{\cO}{{\mathcal O}}
\newcommand{\cP}{{\mathcal P}}

\newcommand{\cR}{{\mathcal R}}
\newcommand{\cS}{{\mathcal S}}
\newcommand{\cT}{{\mathcal T}}
\newcommand{\cV}{{\mathcal V}}

\newcommand{\cZ}{{\mathcal Z}}

\newcommand{\Inv}{{\mathrm{Inv}}}

\def\inv{{\mbox{\tiny -1}}}

\newcommand\beq{\begin{equation}}
\newcommand\eeq{\end{equation}}
\newcommand{\be}{\begin{equation}}
\newcommand{\ee}{\end{equation}}
\newcommand{\bes}{\begin{eqnarray}}
\newcommand{\ees}{\end{eqnarray}}
\newcommand{\bea}{\begin{eqnarray}}
\newcommand{\eea}{\end{eqnarray}}

\def\psihat{{\widehat \psi}}

\def\cc{{\cal C}}

\def\ot{{\,\otimes \,}}

\def\act{\rhd}

\def\vphi{{\varphi}}
\def\vphib{\overline{{\varphi}}}

\newcommand{\one}{\mbox{$1 \hspace{-1.0mm}  {\bf l}$}}

  \def\cc{{\cal C}}    
      \def\nn{{\nonumber}}

\def\tr{{\mathrm{tr}}}

\def\hpsi{{\widehat \psi}}

\def\act{{\, \triangleright\, }}

\newcommand{\su}{\mathfrak{su}}

\newcommand{\so}{\mathfrak{so}}
\newcommand{\SU}{\mathrm{SU}}

\newcommand{\U}{\mathrm{U}}

\newcommand{\SO}{\mathrm{SO}}

\def\extd{\mathrm {d}}
\newcommand{\e}{\epsilon}

\newcommand\acts\triangleright

\newcounter{letter} \newcounter{numeral} \newcounter{Numeral}

\newcommand\Tr{\mathrm{Tr}}
\def\vphi{\varphi}
\def\vphihat{\widehat{\varphi}}

\def\e{\mbox{e}}
\def\E{\mbox{E}}
\def\extd{\mathrm {d}}
\def\ve{\varepsilon}

\newcommand\maps{\colon}

\newtheorem{theo}{Theorem}
\newtheorem{lemma}[theo]{Lemma}

\begin{document}



\thispagestyle{empty}

\noindent{\sc Université Paris-Sud 11 \hfill  LPT Orsay}

\

\noindent{\sc MPI for Gravitational Physics \hfill Microscopic Quantum Structure}\\
\noindent{\sc (Albert Einstein Institute) \hfill \& Dynamics of Spacetime}

\vspace{2cm}
\begin{center}
\textbf{\Huge Tensorial methods and renormalization\\
in\\
Group Field Theories\\}
\vspace{1.5cm}
{\large Doctoral thesis in physics, presented by \\}
\vspace{5mm}
\textbf{\Large Sylvain Carrozza}

\vfill
{\large Defended on September 19$^{\rm{th}}$, 2013, in front of the jury\\
\vspace{10mm}
\begin{tabular}{rcl}
\hline
Pr. Renaud Parentani && Jury president \\
Pr. Bianca Dittrich && Referee \\
Dr. Razvan Gurau && Referee \\
Pr. Carlo Rovelli && Jury member \\
Pr. Daniele Oriti && Supervisor\\
Pr. Vincent Rivasseau && Supervisor \\
\hline
\end{tabular}
}
\vspace{1.5cm}

\includegraphics[height=2.2cm,keepaspectratio]{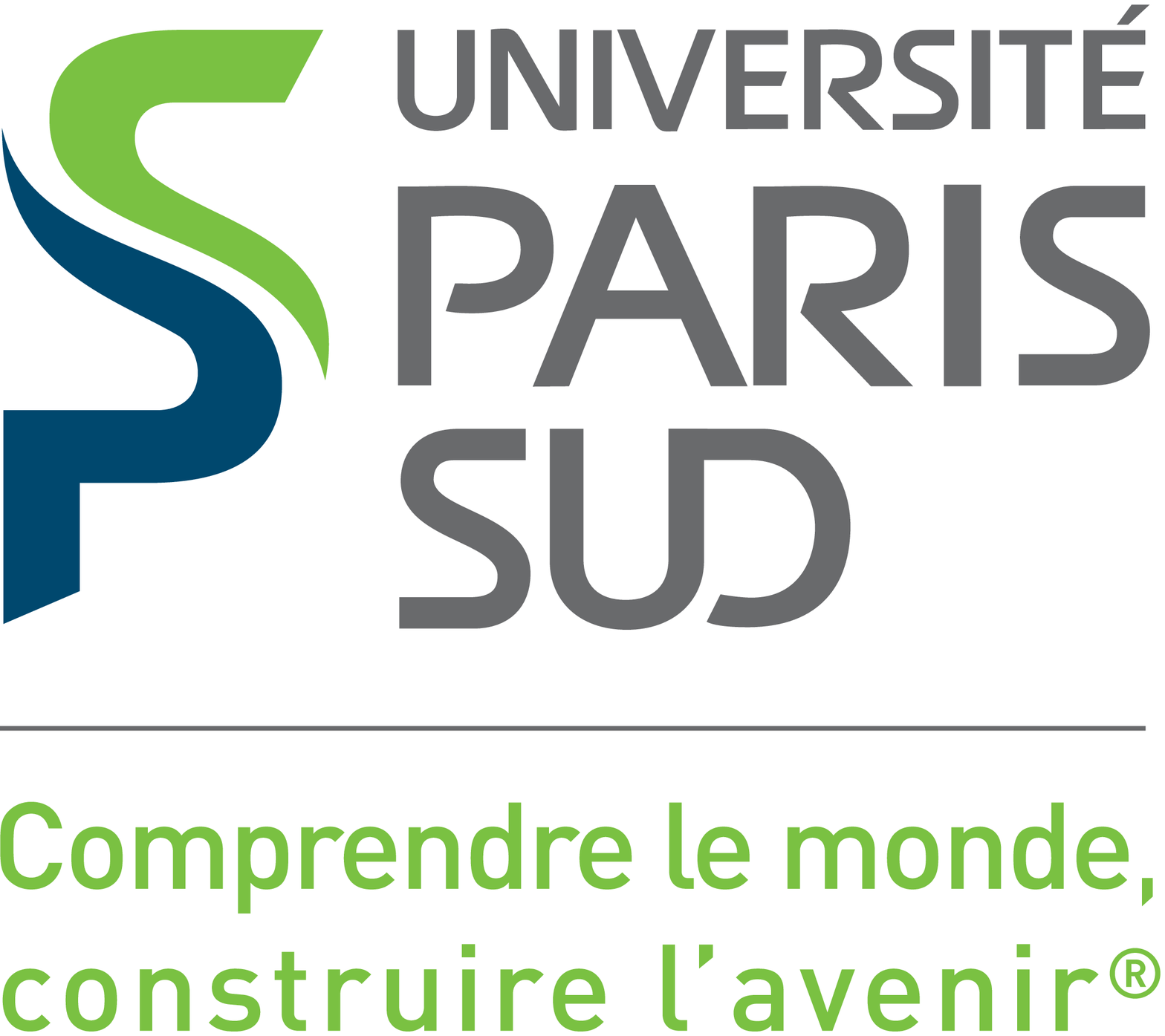}
\hspace{2cm}
\includegraphics[height=2.2cm,keepaspectratio]{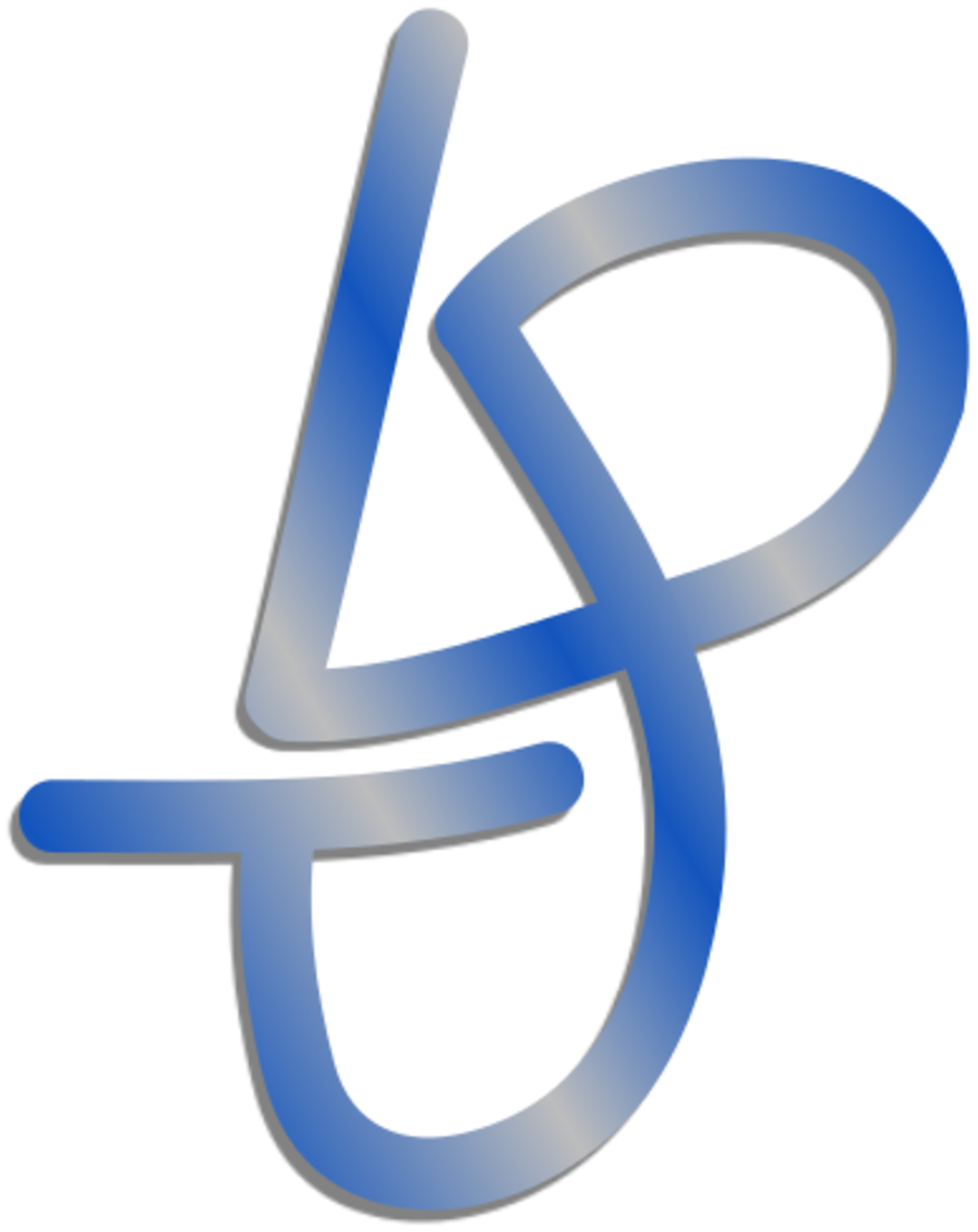}
\hspace{2cm}
\includegraphics[height=2.2cm,keepaspectratio]{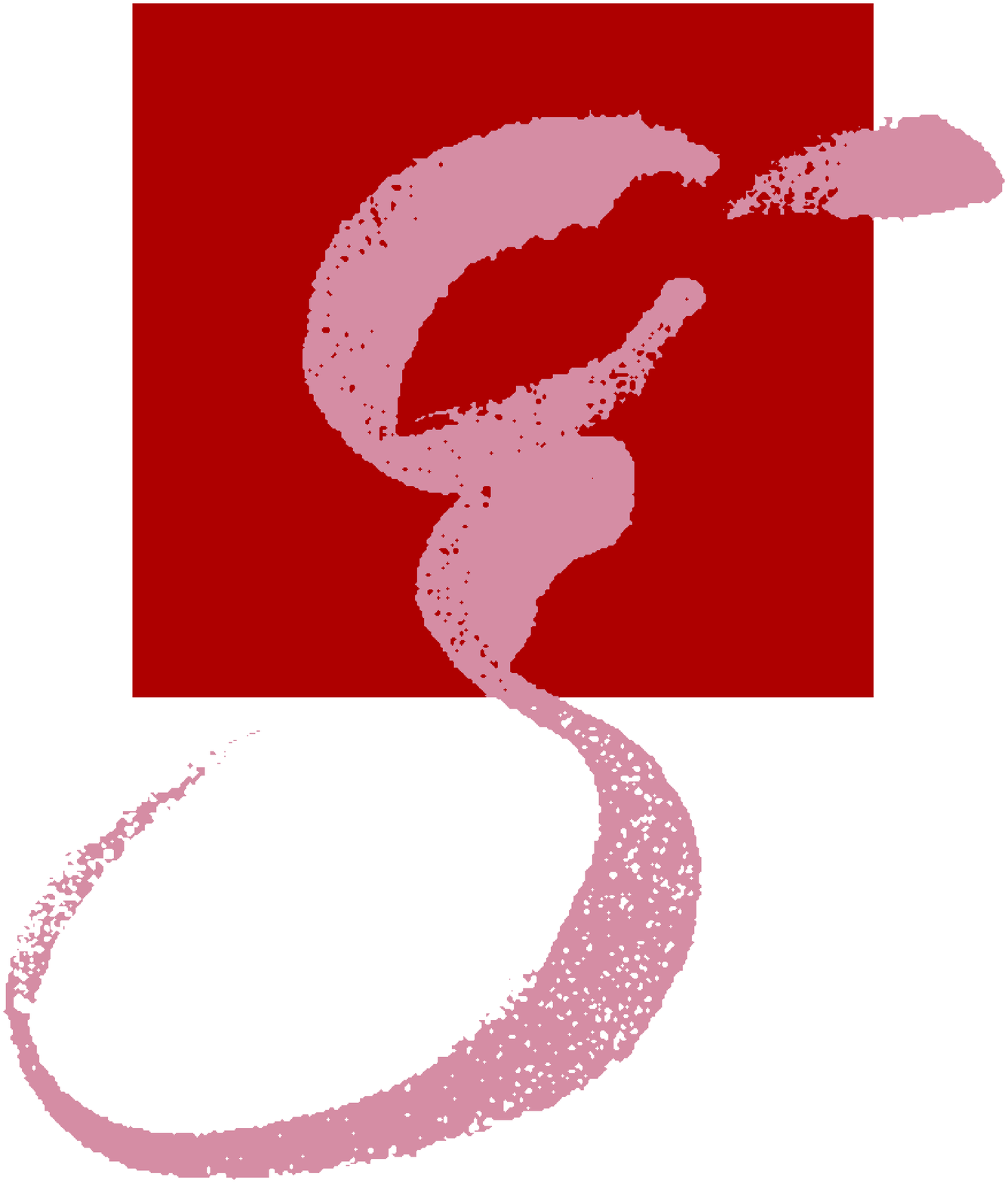}
\hspace{2cm}
\includegraphics[height=2.2cm,keepaspectratio]{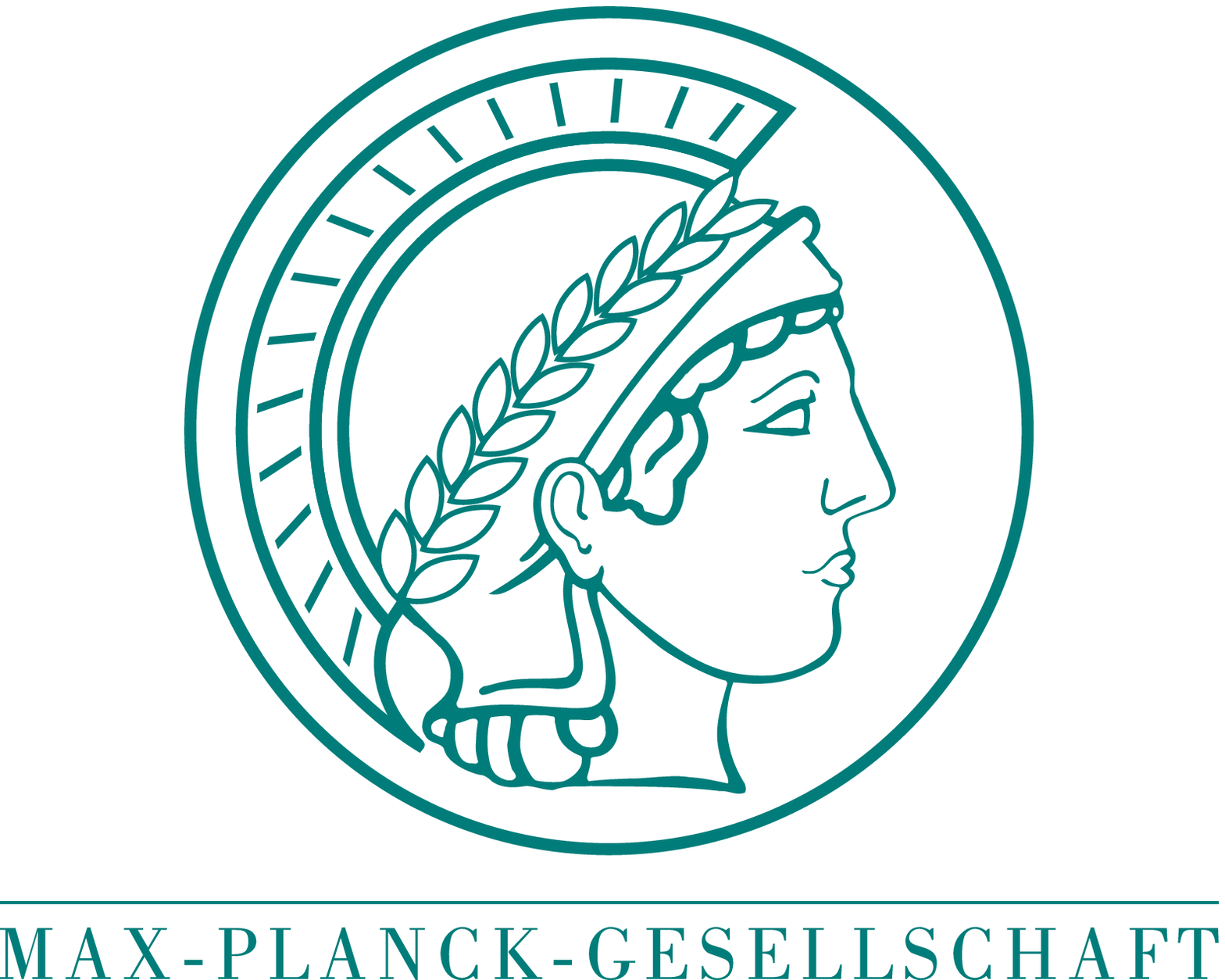}
\end{center}


\newpage
\thispagestyle{empty}
\cleardoublepage
\thispagestyle{plain}
\frontmatter
\setcounter{page}{1}
{\footnotesize{
\begin{center}
\textbf{Abstract:}\\
\end{center}
In this thesis, we study the structure of Group Field Theories (GFTs) from the point of view of renormalization theory. Such quantum field theories are found in approaches to quantum gravity related to Loop Quantum Gravity (LQG) on the one hand, and to matrix models and tensor models on the other hand. They model quantum space-time, in the sense that their Feynman amplitudes label triangulations, which can be understood as transition amplitudes between LQG spin network states. The question of renormalizability is crucial if one wants to establish interesting GFTs as well-defined (perturbative) quantum field theories, and in a second step connect them to known infrared gravitational physics. Relying on recently developed tensorial tools, this thesis explores the GFT formalism in two complementary directions. First, new results on the large cut-off expansion of the colored Boulatov-Ooguri models allow to explore further a non-perturbative regime in which infinitely many degrees of freedom contribute. The second set of results provide a new rigorous framework for the renormalization of so-called Tensorial GFTs (TGFTs) with gauge invariance condition. In particular, a non-trivial 3d TGFT with gauge group $\SU(2)$ is proven just-renormalizable at the perturbative level, hence opening the way to applications of the formalism to (3d Euclidean) quantum gravity.  

\vspace{0.3cm}

\vspace{1cm}

\begin{center}
\textbf{Résumé :}\\
\end{center}
Cette thèse présente une étude détaillée de la structure de théories appelées GFT ("Group Field Theory" en anglais), à travers le prisme de la renormalisation. Ce sont des théories des champs issues de divers travaux en gravité quantique, parmi lesquels la gravité quantique à boucles et les modèles de matrices ou de tenseurs. Elles sont interprétées comme des modèles d'espaces-temps quantiques, dans le sens où elles génèrent des amplitudes de Feynman indexées par des triangulations, qui interpolent les états spatiaux de la gravité quantique à boucles.
Afin d'établir ces modèles comme des théories des champs rigoureusement définies, puis de comprendre leurs conséquences dans l'infrarouge, il est primordial de comprendre leur renormalisation. C'est à cette tâche que cette thèse s'attèle, grâce à des méthodes tensorielles développées récemment, et dans deux directions complémentaires. Premièrement, de nouveaux résultats sur l'expansion asymptotique (en le cut-off) des modèles colorés de Boulatov-Ooguri sont démontrés, donnant accès à un régime non-perturbatif dans lequel une infinité de degrés de liberté contribue. Secondement, un formalisme général pour la renormalisation des GFTs dites tensorielles (TGFTs) et avec invariance de jauge est mis au point. Parmi ces théories, une TGFT en trois dimensions et basée sur le groupe de jauge $\SU(2)$ se révèle être juste renormalisable, ce qui ouvre la voie à l'application de ce formalisme à la gravité quantique.  
\vspace{0.3cm}

\noindent \textbf{Mots-clés:} gravité quantique, gravité quantique à boucles, mousse de spin, group field theory, modèles tensoriels, renormalisation, théorie de jauge sur réseau.
}}
\

\vfill
{\footnotesize{Thèse préparée au sein de l'Ecole Doctorale de Physique de la Région Parisienne (ED 107), dans le Laboratoire de Physique Théorique d'Orsay (UMR 8627),
Bât. 210, Université Paris-Sud 11, 91405 Orsay Cedex; 
et en cotutelle avec le Max Planck Institute for Gravitational Physics (Albert Einstein Institute),
Am Mühlenberg 1, 14476 Golm, Allemagne, dans le cadre de l'International Max Planck Research School (IMPRS).}}  


\newpage
\thispagestyle{plain}
\cleardoublepage
\thispagestyle{plain}

\noindent {\huge Acknowledgments}\\

First of all, I would like to thank my two supervisors, Daniele Oriti and Vincent Rivasseau. Obviously, the results exposed in this thesis could not be achieved without their constant implication, guidance and help. They introduced me to numerous physical concepts and mathematical tools, with pedagogy and patience. Remarkably, their teachings and advices were always complementary to each other, something I attribute to their open-mindedness and which I greatly benefited from. I particularly appreciated the trusting relationship we had from the beginning. It was thrilling, and to me the right balance between supervision and freedom. 

I feel honoured by the presence of Bianca Dittrich, Razvan Gurau, Renaud Parentani and Carlo Rovelli in the jury, who kindly accepted to examine my work. Many thanks to Bianca and Razvan especially, for their careful reading of this manuscript and their comments.

I would like to thank the people I met at the AEI and at the LPT, who contributed to making these three years very enjoyable. The Berlin quantum gravity group being almost uncountable, I will only mention the people I had the chance to directly collaborate with: Aristide Baratin, Francesco Caravelli, James Ryan, Matti Raasakka and Matteo Smerlak.

\

It is quite difficult to keep track of all the events which, one way or another, conspired to pushing me into physics and writing this thesis. It is easier to remember and thank the people who triggered these long forgotten events. 

First and foremost, my parents, who raised me with dedication and love, turning the ignorant toddler I once was into a curious young adult. Most of what I am today takes its roots at home, and has been profoundly influenced by my younger siblings: Manon, Julia, Pauline and Thomas. My family at large, going under the name of Carrozza, Dislaire, Fontès, Mécréant, Minden, Ravoux, or Ticchi, has always been very present and supportive, which I want to acknowledge here. 

The good old chaps, Sylvain Aubry, Vincent Bonnin and Florian Gaudin-Delrieu, deeply influenced my high school years, and hence the way I think today. Meeting them in different corners of Europe during the three years of this PhD was very precious and refreshing.

My friends from the ENS times played a major role in the recent years, both at the scientific and human levels. In this respect I would especially like to thank Antonin Coutant, Marc Geiller, and Baptiste Darbois-Texier: Antonin and Marc, for endless discussions about theoretical physics and quantum gravity, which undoubtedly shaped my thinking over the years; Baptiste for his truly unbelievable stories about real-world physics experiments; and the three of them for their generosity and friendship, in Paris, Berlin or elsewhere.  

\

Finally, I measure how lucky I am to have Tamara by my sides,
who always supported me with unconditional love. I found the necessary happiness and energy to achieve this PhD thesis in the dreamed life we had together in Berlin.

\newpage
\thispagestyle{plain}
\cleardoublepage
\thispagestyle{plain}
\vspace*{\stretch{1}}

\hfill\begin{minipage}{10cm}
{\footnotesize {\it
Wir sollen heiter Raum um Raum durchschreiten,

An keinem wie an einer Heimat hängen,

Der Weltgeist will nicht fesseln uns und engen,

Er will uns Stuf' um Stufe heben, weiten.}}
\vspace{0.2cm}
\newline {\footnotesize {\bf Hermann Hesse}, \textit{Stufen}, in \textit{Das Glasperlenspiel}, 1943.}
\end{minipage}

\vspace*{\stretch{1}}


\newpage
\thispagestyle{plain}
\tableofcontents

\selectlanguage{english}

\mainmatter
\chapter{Motivations and scope of the present work}



\hfill\begin{minipage}{10cm}
{\footnotesize {\it Nous sommes en 50 avant Jésus-Christ. Toute la Gaule est occupée par les Romains... Toute? Non! Un village peuplé d'irréductibles Gaulois résiste encore et toujours à l'envahisseur. Et la vie n'est pas facile pour les garnisons de légionnaires romains des camps retranchés de Babaorum, Aquarium, Laudanum et Petibonum\ldots}}
\vspace{0.2cm}
\newline {\footnotesize {\bf René Goscinny} and {\bf Albert Uderzo}, \it{Astérix le Gaulois}}
\end{minipage}

\section{Why a quantum theory of gravity cannot be dispensed with}


A consistent quantum theory of gravity is mainly called for by a conceptual clash between the two major achievements of physicists of the XX$^{\rm{th}}$ century. On the one hand, the realization by Einstein that classical space-time is a dynamical entity correctly described by General Relativity (GR), and on the other the advent of Quantum Mechanics (QM). The equivalence principle, upon which GR is built, leads to the interpretation of gravitational phenomena as pure geometric effects: the trajectories of test particles are geodesics in a curved four-dimensional manifold, space-time, whose geometric properties are encoded in a Lorentzian metric tensor, which is nothing but the gravitational field \cite{einstein_grossmann1913}. Importantly, the identification of the gravitational force to the metric properties of space-time entails the dynamical nature of the latter. Indeed, gravity being sourced by masses and energy, space-time cannot remain as a fixed arena into which physical processes happen, as was the case since Newton. With Einstein, space-time becomes a physical system \textit{per se}, whose precise structure is the result of a subtle interaction with the other physical systems it contains.
At the conceptual level, this is arguably the main message of GR, and the precise interplay between the curved geometry of space-time and matter fields is encoded into Einstein's equations \cite{einstein1915}. The second aspect of the physics revolution which took place in the early XX$^{\rm{th}}$ century 
revealed a wealth of new phenomena in the microscopic world, and the dissolution of most of the classical Newtonian picture at such scales: the disappearance of the notion of trajectory, unpredictable outcomes of experiments, statistical predictions highly dependent on the experimental setup \cite{bohr1983}... At the mathematical level, QM brings along an entirely new arsenal of technical tools: physical states are turned into vectors living in a Hilbert space, which replaces the phase space of classical physics, and observables become Hermitian operators acting on physical states. However, the conception of space-time on which QM relies remains deeply rooted in Newtonian physics: the Schrödinger equation is a partial differential equation with respect to fixed and physical space-time coordinates. For this reason, Special Relativity could be proven compatible with these new rules of the game, thanks to the Quantum Field Theory (QFT) formalism. The main difficulties in going from non-relativistic to relativistic quantum theory boiled down to the incorporation of the Lorentz symmetry, which also acts on time-like directions.  
Achieving the same reconciliation with the lessons of GR is (and has been proven to be) extraordinarily more difficult. The reason is that as soon as one contemplates the idea of making the geometry of space-time both dynamical and quantum, one looses in one stroke the fixed arena onto which the quantum foundations sit, and the Newtonian determinism which allows to unambiguously link space-time dynamics to its content. 
The randomness introduced by quantum measurements seems incompatible with the definition of a single global state for space-time and matter (e.g. a solution of a set of partial differential equations). 
And without a non-dynamical background, there is no unambiguous 'here' where quantum ensembles can be prepared, nor a 'there' where measurements can be performed and their statistical properties checked. 
In a word, by requiring background independence to conform to Einstein's ideas about gravity, one also suppresses the only remaining Newtonian shelter where quantum probabilities can safely be interpreted. This is probably the most puzzling aspect of modern physics, and calls for a resolution. 

\

But, one could ask, do we necessarily need to make gravity quantum? Cannot we live with the fact that matter is described by quantum fields propagating on a dynamical but classical geometry? A short answer would be to reject the dichotomous understanding of the world that would result from such a combination of \textit{a priori} contradictory ideas. On the other hand, one cannot deny that space-time is a very peculiar physical system, which one might argue, could very well keep a singular status as the only fundamentally classical entity. However, very nice and general arguments, put forward by Unruh in \cite{unruh1984}, make this position untenable (at least literally). Let us recapitulate the main ideas of this article here. In order to have the Einstein equations
\beq\label{einstein_eq}
G_{\mu \nu} = 8 \pi G T_{\mu \nu}
\eeq
as a classical limit of the matter sector, one possibility would be to interpret the right-hand side as a quantum average $\langle \hat{T}_{\mu \nu} \rangle$ of some quantum operator representing the energy-momentum tensor of matter fields. The problems with such a theory pointed out in \cite{unruh1984} are two-fold. First, quantum measurements would introduce discontinuities in the expectation value of $\hat{T}_{\mu \nu}$, and in turn spoil its conservation. Second, and as illustrated with a gravitational version of Schrödinger's cat gedanken experiment, such a coupling of gravity to a statistical average of matter states would introduce slow variations of the gravitational field caused by yet unobserved and undetermined matter states. Another idea explored by Unruh to make sense of (\ref{einstein_eq}) in such a way that the left-hand side is classical, and the right-hand side quantum, is through an eigenvalue equation of the type
\beq
8 \pi G \hat{T}_{\mu \nu} \vert \psi \rangle = G_{\mu \nu} \vert \psi \rangle \,.
\eeq
The main issue here is that the definition of the operator $\hat{T}_{\mu \nu}$ would have to depend non-linearly on the classical metric, and hence on the 'eigenvalue' $G_{\mu \nu}$. From the point of view of quantum theory, this of course does not make any sense. 

\

Now that some conceptual motivations for the search for a quantum description of the gravitational field have been recalled (and which are also the author's personal main motivations to work in this field), one should make a bit more precise what one means by 'a quantum theory of gravitation' or 'quantum gravity'. We will adopt the kind of ambitious though minimalistic position promoted in Loop Quantum Gravity (LQG) \cite{ashtekar_book , rovelli_book , thiemann_book}. Minimalistic because the question of the unification of all forces at high energies is recognized as not necessarily connected to quantum gravity, and therefore left unaddressed. But ambitious in the sense that one is not looking for a theory of quantum perturbations of the gravitational degrees of freedom around some background solution of GR, since this would be of little help as far as the conceptual issues aforementioned are concerned. Indeed, and as is for instance very well explained in \cite{perez_review2004, perez_review2012}, from the point of view of GR, there is no canonical way of splitting the metric of space-time into a background (for instance a Minkowski metric, but not necessarily so) plus fluctuations. Therefore giving a proper quantum description of the latter fluctuations, that is finding a renormalizable theory of gravitons on a given background, cannot fulfill the ultimate goal of reconciling GR with QM. On top of that, one would need to show that the specification of the background is a kind of gauge choice, which does not affect physical predictions. Therefore, one would like to insist on the fact that even if such a theory was renormalizable, the challenge of making Einstein's gravity fully quantum and dynamical would remain almost untouched.
This already suggests that introducing the background in the first place is unnecessary. Since it turns out that the quantum theory of perturbative quantum GR around a Minkowski background is not renormalizable \cite{goroff_sagnotti}, we can even go one step further: the presence of a background might not only be unnecessary but also problematic. 
The present thesis is in such a line of thought, which aims at taking the background independence of GR seriously, and use it as a guiding thread towards its quantum version \cite{smolin_background}.
In this perspective, we would call 'quantum theory of gravity' a quantum theory without any space-time background, which would reduce to GR in some (classical) limit. 

\

A second set of ideas which are often invoked to justify the need for a theory of quantum gravity concerns the presence of singularities in GR, and is therefore a bit more linked to phenomenology, be it through cosmology close to the Big Bang or the question of the fate of black holes at the end of Hawking's evaporation. It is indeed tempting to draw a parallel between the question of classical singularities in GR and some of the greatest successes of the quantum formalism, such as for example the explanation of the stability of atoms or the resolution of the UV divergence in the theory of black-body radiation. We do not want to elaborate on these questions, but only point out that even if very suggestive and fascinating proposals exist \cite{bojowald2001 , ashtekar2006 , bojowald2012 }, there is as far as we know no definitive argument claiming that the cumbersome genericity of singularities in GR has to be resolved in quantum gravity. This is for us a secondary motivation to venture into such a quest, though a very important one. While a quantum theory of gravity must by definition make QM and GR compatible, it only \textit{might} explain the nature of singularities in GR. Still, it would be of paramount relevance if this second point were indeed realized, since it would open the door to a handful of new phenomena and possible experimental signatures to look for. 

\

Another set of ideas we consider important but we do not plan to address further in this thesis are related to the non-renormalizability of perturbative quantum gravity. As a quantum field theory on Minkowski space-time, the quantum theory of gravitons based on GR can only be considered as an effective field theory \cite{burgess2003 , donoghue2012}, which breaks down at the Planck scale. Such a picture is therefore necessarily incomplete as a fundamental theory, as it was to be expected, but does not provide any clear clue about how it should be completed. At this point, two attitudes can be adopted. Either assume that one should first look for a renormalizable perturbative theory of quantum gravity, from which the background independent aspects will be addressed in a second stage; or, focus straight away on the background independent features which are so central to the very question of quantum gravity. Since we do not want to assume any \textit{a priori} connection between the UV completion of perturbative quantum general relativity and full-fledged quantum gravity, as is for instance investigated in the asymptotic safety program \cite{reuter2012 , dario2013}, the results of this thesis will be presented in a mindset in line with the second attitude. Of course, any successful fundamental quantum theory of gravity will have to provide a deeper understanding of the two-loops divergences of quantum GR, and certainly any program which would fail to do so could not be considered complete \cite{nicolai2005}. 

\

The purpose of the last two points was to justify to some extent the technical character of this PhD thesis, and its apparent disconnection with many of the modern fundamental theories which are experimentally verified. While it is perfectly legitimate to look for a reconciliation of QM and GR into the details of what we know about matter, space and time, we want to advocate here a hopefully complementary strategy, which aims at finding a general theoretical framework encompassing them both at a general and conceptual level. At this stage, we would for example be highly satisfied with a consistent definition of quantum geometry whose degrees of freedom and dynamics would reduce to that of vacuum GR in some limit; even if such a theory did not resolve classical singularities, nor it would provide us with a renormalizable theory of gravitons.

\section{Quantum gravity and quantization}








Now that we reinstated the necessity of finding a consistent quantum formulation of gravitational physics, we would like to make some comments about the different general strategies which are at our disposal to achieve such a goal. In particular, would a quantization of general relativity (or a modification thereof) provide the answer?

\

The most conservative strategy is the quantization program of classical GR pioneered by Bryce DeWitt \cite{dewitt2011pursuit}, either through Dirac's general canonical quantization procedure \cite{dirac_qm , dewitt_can} or with covariant methods \cite{dewitt_cov}. Modern incarnations of these early ideas can be found in canonical loop quantum gravity and its tentative covariant formulation through spin foam models \cite{rovelli_book , thiemann_book , perez_review2012}. While the Ashtekar formulation of GR \cite{ashtekar1986 , barbero1994} allowed dramatic progress with respect to DeWitt's formal definitions, based on the usual metric formulation of Einstein's theory, very challenging questions remain open as regards the dynamical aspects of the theory. In particular, many ambiguities appear in the definition of the so-called scalar constraint of canonical LQG, and therefore in the implementation of four-dimensional diffeomorphism invariance, which is arguably the core purpose of quantum gravity. There are therefore two key aspects of the canonical quantization program that we would like to keep in mind: first, the formulation of classical GR being used as a starting point (in metric or Ashtekar variables), or equivalently the choice of fundamental degrees of freedom (the metric tensor or a tetrad field), has a great influence on the quantization; and second, the subtleties associated to space-time diffeomorphism invariance have so far plagued such attempts with numerous ambiguities, which prevent the quantization procedure from being completed. The first point speaks in favor of loop variables in quantum gravity, while the second might indicate an intrinsic limitation of the canonical approach. 

\

A second, less conservative but more risky, type of quantization program consists in discarding GR as a classical starting point, and instead postulating radically new degrees of freedom. This is for example the case in string theory, where a classical theory of strings moving in some background space-time is the starting point of the quantization procedure. Such an approach is to some extent supported by the non-renormalizability of perturbative quantum GR, interpreted as a signal of the presence of new degrees of freedom at the Planck scale. Similar interpretations in similar situations already proved successful in the past, for instance with the four-fermion theory of Fermi, whose non-renormalizability was cured by the introduction of new gauge bosons, and gave rise to the renormalizable Weinberg-Salam theory. In the case of gravity, and because of the unease with the perturbative strategy mentioned before, we do not wish to give too much credit to such arguments. However, it is necessary to keep in mind that the degrees of freedom we have access to in the low-energy classical theory (GR) are not necessarily the ones to be quantized.

\

Finally, a third idea which is gaining increasing support in the recent years is to question the very idea of quantizing gravity, at least \textit{stricto sensu}. Rather, one should more generally look for a quantum theory, with possibly non-metric degrees of freedom, from which classical geometry and its dynamics would \textit{emerge}. Such a scenario has been hinted at from within GR itself, through the thermal properties of black holes and space-time in general. For instance in \cite{jacobson1995}, Jacobson suggested to interpret the Einstein equations as equations of states at thermal equilibrium. In this picture, space-time dynamics would only emerge in the thermodynamic limit of a more fundamental theory, with degrees of freedom yet to be discovered. This is even more radical that what is proposed in string theory, but also consistent with background independence in principle: there is no need to assume the existence of a (continuous) background space-time in this picture, and contrarily so, the finiteness of black hole entropy can be interpreted as suggestive of the existence of an underlying discrete structure. Such ideas have close links with condensed matter theory, which explains for example macroscopic properties of solids from the statistical properties of their quantum microscopic building blocks, and in particular with the theory of quantum fluids and Bose-Einstein condensates \cite{lorenzo2009 , barcelo2005}. Of course, the two outstanding issues are that no experiments to directly probe the Planck scale are available in the near future, and emergence has to be implemented in a fully background independent manner.  

\

After this detour, one can come back to the main motivations of this thesis, loop quantum gravity and spin foams, and remark that even there, the notion of emergence seems to have a role to play. Indeed, the key prediction of canonical loop quantum gravity is undoubtedly the discreteness of areas and volumes at the kinematical level \cite{rovelli_smolin1994}, and this already entails some kind of emergence of continuum space-time. In this picture, continuous space-time cannot be defined all the way down to the Planck scale, where the discrete nature of the spectra of geometric operators starts to be relevant. This presents a remarkable qualitative agreement with Jacobson's proposal, and in particular all the thermal aspects of black holes explored in LQG derive from this fundamental result \cite{pranzetti_sigma}. But there are other discrete features in LQG and spin foams, possibly related to emergence, which need to be addressed. Even if canonical LQG is a continuum theory, the Hilbert space it is based upon is constructed in an inductive way, from states (the spin-network functionals) labeled by discrete quantities (graphs with spin labels). We can say that each such state describes a continuous quantum geometry with a finite number of degrees of freedom, and that the infinite number of possible excitations associated to genuine continuous geometries is to be found in large superpositions of these elementary states, in states associated to infinitely large graphs, or both. In practice, only spin-network states on very small graphs can be investigated analytically, the limit of infinitely large graphs being out of reach, and their superpositions even more so. This indicates that in its current state, LQG can also be considered a theory of discrete geometries, despite the fact that it is primarily a quantization of GR. From this point of view, continuous classical space-time would only be recovered through a continuum limit. This is even more supported by the covariant spin foam perspective, where the discrete aspects of spin networks are enhanced rather than tamed. The discrete structure spin foam models are based upon, $2$-complexes, acquire a double interpretation, as Feynman graphs labeling the transitions between spin network states on the one hand, and as discretizations of space-time akin to lattice gauge theory on the other hand. Contrary to the canonical picture, this second interpretation cannot be avoided, at least in practice, since all the current spin foam models for four-dimensional gravity are constructed in a way to enforce a notion of (quantum) discrete geometry in a cellular complex dual to the foam. Therefore, in our opinion, at this stage of the development of the theory, it seems legitimate to view LQG and spin foam models as quantum theories of discrete gravity. And if so, addressing the question of their continuum limit is of primary importance. 

Moreover, we tend to see a connection between: a) the ambiguities appearing in the definition of the dynamics of canonical LQG, b) the fact that the relevance of a quantization of GR can be questioned in a strong way, and c) the problem of the continuum in the covariant version of loop quantum gravity. Altogether, these three points can be taken as a motivation for a strategy where quantization and emergence both have to play their role. It is indeed possible, and probably desirable, that some of the fine details of the dynamics of spin networks are irrelevant to the large scale effects one would like to predict and study. In the best case scenario, the different versions of the scalar constraint of LQG would fall in a same universality class as far as the recovery of continuous space-time and its dynamics is concerned. This would translate, in the covariant picture, as a set of spin foam models with small variations in the way discrete geometry is encoded, but having a same continuum limit. The crucial question to address in this perspective is that of the existence, and in a second stage the universality of such a limit, in the sense of determining exactly which aspects (if any) of the dynamics of spin networks are key to the emergence of space-time as we know it. The fact that these same spin networks were initially thought of as quantum states of \textit{continuous} geometries should not prevent us from exploring other avenues, in which the continuum only emerge in the presence of a very large number of \textit{discrete} building blocks. 

\

This PhD thesis has been prepared with the scenario just hinted at in mind, but we should warn the reader that it is in no way conclusive in this respect. Moreover, we think and we hope that the technical results and tools which are accounted for in this manuscript are general enough to be useful to researchers in the field who do not share such point of views. The reason is that, in order to study universality in quantum gravity, and ultimately find the right balance between strict quantization procedures and emergence, one first needs to develop a theory of renormalization in this background independent setting, which precisely allows to consistently erase information and degrees of freedom. This thesis is a contribution to this last point, in the Group Field Theory (GFT) formulation of spin foam models. 

\section{On scales and renormalization with or without background}





The very idea of extending the theory of renormalization to quantum gravity may look odd at first sight. The absence of any background seems indeed to preclude the existence of any physical scale with respect to which the renormalization group flow should be defined. A few remarks are therefore in order, about the different notions of scales which are available in quantum field theories and general relativity, and the general assumption we will make throughout this thesis in order to extend such notions to background independent theories. 

\

Let us start with relativistic quantum field theories, which support the standard model of particle physics, as well as perturbative quantum gravity around a Minkowski background. The key ingredient entering the definition of these theories is the flat background metric, which provides a notion of locality and global Poincaré invariance. The latter allows in particular to classify all possible interactions once a field content (with its own set of internal symmetries) has been agreed on \cite{weinberg1}. More interesting, this same Poincaré invariance, combined with locality and the idea of renormalization \cite{wilson_nobel , vincent_book , salmhofer}, imposes further restrictions on the number of independent couplings one should work with. When the theory is (perturbatively) non-renormalizable, it is consistent only if an infinite set of interactions is taken into account, and therefore loses any predictive power (at least at some scale). When it is on the contrary renormalizable, one can work with a finite set of interactions, though arbitrarily large in the case of a super-renormalizable theory. For fundamental interactions, the most interesting case is that of a just-renormalizable theory, such as QED or QCD, for which a finite set of interactions is uniquely specified by the renormalizability criterion. In all of these theories, what is meant by 'scale' is of course an energy scale, in the sense of special relativity. However, renormalization and quantum field theory are general enough to accommodate various notions of scales, as for example non-relativistic energy, and apply to a large variety of phenomena for which Poincaré invariance is completely irrelevant. A wealth of examples of this kind can be found in condensed matter physics, and in the study of phase transitions. The common feature of all these models is that they describe regimes in which a huge number of (classical or quantum) degrees of freedom are present, and where their contributions can be efficiently organized according to some order parameter, the 'scale'. As we know well from thermodynamics and statistical mechanics, it is in this case desirable to simplify the problem by assuming instead an infinite set of degrees of freedom, and adopt a coarse-grained description in which degrees of freedom are collectively analyzed. Quantum field theory and renormalization are precisely a general set of techniques allowing to efficiently organize such analyzes. Therefore, what makes renormalizable quantum field theories so useful in fundamental physics is not Poincaré invariance in itself, but the fact that it implies the existence of an infinite reservoir of degrees of freedom in the deep UV. 

\

We now turn to general relativity. The absence of Poincaré symmetry, or any analogous notion of space-time global symmetries prevents the existence of a general notion of energy. Except for special solutions of Einstein's equations, there is no way to assign an unambiguous notion of localized energy to the modes of the gravitational field\footnote{We can for instance quote Straumann \cite{straumann}: \begin{quote} This has been disturbing to many people, but one simply has to get used to this fact. There is no "energy-momentum tensor for the gravitational field". \end{quote}}. The two situations in which special relativistic notions of energy-momentum do generalize are in the presence of a global Killing symmetry, or for asymptotically flat space-times. In the first case, it is possible to translate the fact that the energy-momentum tensor $T^{\mu \nu}$ is divergence free into both local and integral conservation equations for an energy-momentum vector $P^\mu \equiv T^{\mu \nu} K_\nu$, where $K_\nu$ is the Killing field. In the second case, only a partial generalization is available, in the form of integral conservation equations for energy and momentum at spatial infinity. One therefore already loses the possibility of localizing energy and momentum in this second situation, since they are only defined for extended regions with boundaries in the approximately flat asymptotic region. In any case, both generalizations rely on global properties of specific solutions to Einstein's equations which cannot be available in a background independent formulation of quantum gravity. We therefore have to conclude that, since energy scales associated to the gravitational field are at best solution-dependent, and in general not even defined in GR, a renormalization group analysis of background independent quantum gravity cannot be based on space-time related notions of scales.

\

This last point was to be expected on quite general grounds. From the point of view of quantization à la Feynman for example, all the solutions to Einstein's equations (and in principle even more general 'off-shell' geometries) are on the same footing, as they need to be summed over in a path-integral (modulo boundary conditions). We cannot expect to be able to organize such a path-integral according to scales defined internally to each of these geometries. But even if one takes the emergent point of view seriously, GR suggests that the order parameter with respect to which a renormalization group analysis should be launched cannot depend on a space-time notion of energy. This point of view should be taken more and more seriously as we move towards an increasingly background independent notion of emergence, in the sense of looking for a unique mechanism which would be responsible for the emergence of a large class of solutions of GR, if not all of them. In particular, as soon as such a class is not restricted to space-times with global Killing symmetries or with asymptotically flat spatial infinities, there seems to be no room for the usual notion of energy in a renormalization analysis of quantum gravity. 

\

However, it should already be understood at this stage that the absence of any background space-time in quantum gravity, and therefore of any natural physical scales, does not prevent us from using the quantum field theory and renormalization formalisms. As was already mentioned, the notion of scale prevailing in renormalization theory is more the number of degrees of freedom available in a region of the parameter space, rather than a proper notion of energy. Likewise, if quantum fields do need a fixed background structure to live in, this needs not be interpreted as space-time. As we will see, this is precisely how GFTs are constructed, as quantum field theories defined on (internal) symmetry groups rather than space-time manifolds. More generally, the working assumption of this thesis will be that a notion of scale and renormalization group flow can be defined \textit{before}\footnote{Obviously, this 'before' does not refer to time, but rather to the abstract notion of scale which is assumed to take over when no space-time structure is available anymore.} space-time notions become available, and studied with quantum field theory techniques, as for example advocated in \cite{vincent_tt1 , vincent_tt2}. The only background notions one is allowed to use in such a program must also be present in the background of GR. The dimension of space-time, the local Lorentz symmetry, and the diffeomorphism groups are among them, but they do not support any obvious notion of scale. Rather, we will postulate that the 'number of degrees of freedom' continues to be a relevant order parameter in the models we will consider, that is in the absence of space-time. This rather abstract scale will come with canonical definitions of UV and IR sectors. They should by no means be understood as their space-time related counter-parts, and be naively related to respectively small and large distance regimes. Instead, the UV sector will simply be the corner of parameter space responsible for divergences, or equivalently where 'most' of the degrees of freedom sit. A natural renormalization group flow will be defined, which will allow to average out the contributions of the degrees of freedom, from higher to lower scales. The only strong conceptual assumption we will make in this respect is that such an abstract definition of renormalization is physical and can be used to describe the emergence of space-time structures. However, at this general level of discussion, we would like to convey the idea that such a strong assumption is in a sense also minimal. Indeed, if one wants to be able to speak of emergence of space and time, one also needs at least one new parameter which is neither time nor space. We simply call this order parameter 'scale', and identify it with one of the central features of quantum field theory: the renormalization group. It is in our view the most direct route towards new physics in the absence of space and time, as quantum gravity seems to require. 

\section{Purpose and plan of the thesis}
%
%
%
%
%

We are well aware of the fact that the previous motivations cannot be taken for granted. They can be contested in various ways, and also lack a great deal of precision. The reader should see them as a guiding thread towards making full sense of the emergence of space-time from background independent physics, rather than definitive statements embraced by the author. 
From now on, we will refrain from venturing into more conceptual discussions, and mostly leave the specific examples worked out in this thesis speak for themselves, hoping that they will do so in favor of the general ideas outlined before.

\

The rest of the thesis is organized as follows. In chapter \ref{GFT}, we will start by recalling the two main ways of understanding the construction of GFT models. One takes its root in the quantization program for quantum gravity, in the form of loop quantum gravity and spin foam models. In this line of thoughts, GFTs are generating functionals for spin foam amplitudes, in the same way as quantum field theories are generating functionals for Feynman amplitudes. In this sense, they complete the definition of spin foam models by assigning canonical weights to the different foams contributing to a same transition between boundary states (spin networks). Moreover, a quantum field theory formalism is expected to provide easier access to non-perturbative regimes, and hence to the continuum. For example, classical equations of motion can be used as a way to change vacuum \cite{edr}, or to study condensed phases of the theory \cite{gfc}. Of course, this specific completion of the definition of spin foam models relies on a certain confidence in the quantum field theory formalism. Alternative but hopefully complementary approaches exist, such as coarse-graining methods imported from condensed matter physics and quantum information theory \cite{bianca_cyl, bianca_review , bahr2012}. Though, if one decides to stick to quantum field theory weights, it seems natural to also bring renormalization in. From this point of view, perturbative renormalizability of GFTs is a self-consistency check, and is rather independent from the continuum limit and the emergence question. A second set of motivations is given by discrete approaches to quantum gravity. We will first focus on the successful example of matrix models, which allowed to define random two-dimensional surfaces, and henceforth achieve a quantization of two-dimensional quantum gravity. We will then outline natural extensions of matrix models, known as tensor models, which will be crucial to the rest of the thesis, particularly in their modern versions. In this perspective, GFTs appear as enriched tensor models, which allow to define finer notions of discrete quantum geometries, and the question of emergence of the continuum is put to the forefront, merely by construction. We will then comment on the relations between these two historical paths, and advocate for a middle path, in which the hard questions to be faced are renormalization and universality of the continuum limit.  

In Chapter \ref{color_tensor}, we will move to the more recent aspects of GFTs and tensor models, following the introduction of colored models by Gurau in 2009. This will be the occasion to summarize the main results and tools which have been developing fast since then. This will especially include combinatorial and topological properties of colored graphs, which shape all the models we will discuss in the later chapters, and will be at the core of all the original results of this thesis. The important notion of tensor invariance will also be introduced and motivated in this chapter.

Finally, all the concepts introduced in the first chapters will be used in Chapters \ref{largeN} and \ref{renormalization} to present the original results of this PhD thesis. They come in two types, and as we will try to illustrate, are quite complementary. The first set of results concerns the so-called $1/N$ expansion of topological GFT models. It applies to GFTs with cut-off (given by the parameter $N$), in which a particular scaling of the coupling constant allows to reach an asymptotic many-particle regime at large $N$. We will in particular focus on asymptotic bounds on the amplitudes which allow to classify their contributions according to the topology of the underlying cell-complex. This will be argued to be a rough preliminary version of the second set of results, which concerns full-fledge renormalization. Tensorial Group Field Theories (TGFTs), which are refined versions of the cut-off models with new non-trivial propagators, will be introduced. They have a built-in notion of scale, which generates a well-defined renormalization group flow, and will give rise to somewhat dynamical versions of the $1/N$ expansions. We will give several examples, culminating in a detailed study of a just-renormalizable TGFT based on the group $\SU(2)$ in three dimensions, and incorporating the 'closure constraint' of spin foam models. As will be argued at this point, this TGFT can be considered a field theory realization of the original Boulatov model for three-dimensional quantum gravity \cite{boulatov}. 

The results accounted for in Chapters \ref{largeN} to \ref{chap:su2} are mainly based on four publications, the first two in collaboration with Daniele Oriti \cite{vertex, edge}, and the two others with Daniele Oriti and Vincent Rivasseau \cite{u1 , su2}. At the end of Chapter \ref{chap:su2}, we also include unpublished work \cite{beta_su2} about the flow equations of the TGFT introduced in \cite{su2}. Finally, a fifth paper \cite{boulatov_phase}, in collaboration with Aristide Baratin, Daniele Oriti, James Ryan and Matteo Smerlak, will be briefly evoked in the conclusions.

\chapter{Two paths to Group Field Theories}\label{GFT}

\hfill\begin{minipage}{10cm}
{\footnotesize {\it
Men's memories are uncertain and the past that was differs little from the past that was not.}}
\vspace{0.2cm}
\newline {\footnotesize {\bf Cormac McCarthy}, \textit{Blood Meridian}.}
\end{minipage}

\vspace{0.4cm}

In this chapter, we recall two historical and conceptual paths leading to the GFT formalism. The first follows the traditional quantization route, from canonical LQG, to spin foam models, and finally to GFTs, seen as generating functionals for the latter. The second approach relies from the start on the idea of discretization, and takes its root in the successes of matrix models. We will then argue in favor of taking these two strategies equally seriously to investigate the properties of GFTs, and the physical picture emerging from them.

	\section{Group Field Theories and quantum General Relativity}

		\subsection{Loop Quantum Gravity}
		
		Loop Quantum Gravity is a tentative approach to the canonical quantization of general relativity. It is conceptually very much in line with the old Wheeler-DeWitt (WdW) theory, at the notable exception that many previously formal expressions could be turned into rigorous equations thanks to LQG techniques. This breakthrough triggered many developments, from cosmology to black hole physics, all relying on the well-understood kinematical properties of LQG. Unfortunately, to this date the canonical quantization program could not be completed, a proper definition of the dynamics being still challenging, despite some recent progress \cite{Laddha:2011mk, casey_weak}.

\

As in the WdW theory, one starts from a canonical formulation of general relativity, which requires hyperbolicity of the space-time manifold $\cM$. This allows to introduce a global foliation, in terms of a time function $t \in \mathbb{R}$ indexing three-dimensional spacelike hypersurfaces $\Sigma_t$. One then recasts Einstein's equation in the form of a Hamiltonian system, dictating the evolution of spatial degrees of freedoms in this parameter $t$. In the WdW theory, the Hamiltonian formulation is the ADM one \cite{adm}, for which the configuration space is the set of all spatial metrics $g_{ab}$ ($1\leq a , b \leq 3$) on $\Sigma_t$. This is where LQG departs from the WdW theory, with the Ashtekar-Barbero formulation of classical GR \cite{ashtekar1986, barbero1994} as a starting point, hence shifting the emphasis to forms and connection variables.

\
 
Let us first introduce the four-dimensional formalism. Instead of using the space-time metric $g_{\mu \nu}$ ($0\leq \mu , \nu \leq 3$) as configuration variable, one instead introduces a tetrad field
\beq
e^{I}(x) = e^{I}_{\mu} (x) \extd x^\mu\,, 
\eeq 
which is a quadruple of $1$-forms on $\cM$, indexed by a Minkowski index $I \in \llbracket 0, 3\rrbracket$. The metric is only a derived quantity, which simply writes as
\beq
g_{\mu \nu} (x) = e^{I}_{\mu} (x) e^{J}_{\nu} (x) \eta_{IJ}
\eeq  
where $\eta$ is the Minkowski metric (with signature $[-, + , + , +]$). This redundant rewriting of the metric introduces a new local symmetry, corresponding to the action of $\SO( 1 , 3 )$ on the Minkowski index labeling the tetrad:
\beq\label{loc_so13}
e^{I} (x) \mapsto {\Lambda^{I}}_{J}(x)  e^{J} (x) \,, \qquad \Lambda (x) \in \SO( 1 , 3 )\,.
\eeq
Einstein's equations can then be recovered from a first order variational principle, depending on an additional spin connection $\omega$, which is a $1$-form with values in the Lie algebra $\so(1,3)$. The particular action lying at the foundations of current LQG and spin foams is the Holst action (for vacuum gravity without cosmological constant):
\beq\label{holst}
S[ e , \omega ] = \frac{1}{\kappa} \int_{\cM} \Tr \left[ \left( \star e \wedge e + \frac{1}{\gamma} e \wedge e \right) \wedge F(\omega) \right]\,,
\eeq 
where $\kappa$ is a constant containing Newton's $G$, $\gamma$ is the Barbero-Immirzi parameter, and $F(\omega) = \extd \omega + \omega \wedge \omega$ is the curvature of the spin connection. Varying $S$ with respect to $e$ and $\omega$ independently provides two sets of equations: the first forces $\omega$ to be the unique torsion-free connection compatible with $e$, and the second gives Einstein's equations. The term proportional to the Barbero-Immirzi parameter is topological, and therefore plays no role as far as classical equations are concerned. It is however a key ingredient of the loop quantization, and as we will see is a cornerstone of the current interpretation of quantum geometry in LQG. 

\ 

Assuming hyperbolicity of $\cM$, it is possible to perform a canonical analysis of (\ref{holst}) leading to the Ashtekar-Barbero formulation of GR. Details can for example be found in \cite{perez_review2012 , Geiller:2012dd}, and at the end of the day a $3+1$ splitting is obtained for both the tetrad and the spin connection in terms of canonical pairs $(E^a_i , A^j_b)$:
\beq
E^{a}_i \equiv \det(e) e^{a}_i \,, \qquad A^i_a \equiv \omega^i_a + \gamma \omega^{0 i}_a = \frac{1}{2} \epsilon^{i}_{jk} \omega^{jk}_a + \gamma \omega^{0 i}_a \,, \qquad \{E^{a}_i (x) ,  A^j_b (y) \} = \kappa \gamma \delta^{a}_{b} \delta^{i}_{j} \delta(x,y) \,, 
\eeq 
where now $x$ and $y$ refer to coordinates on the spatial hypersurface $\Sigma$. $A$ is an $\SU(2)$ connection encoding parallel transport on $\Sigma$, and the densitized triad $E$ can be used to define geometric quantities such as areas and volumes embedded in $\Sigma$. The most interesting aspect of such a formulation is that the dynamics of GR is encoded in three sets of constraints, bearing some similarities with Yang-Mills theory on a background space-time:
\beq
\cG_i = \cD_a E^{a}_i \,, \qquad \cC_a = E^{b}_k F^{k}_{ba} \,, \qquad \cS = \epsilon^{ijk} E^{a}_i E^{b}_j F_{ab \, k} + 2\frac{\gamma^2 -1}{\gamma^2}  E^{a}_{[ i} E^{b}_{j ]} ( A_{a}^{i} - \omega_{a}^{i}) ( A_{b}^{j} - \omega_{b}^{j} ) \,.
\eeq 
These constraints are first class, which means that the Poisson bracket of two constraints is itself a linear combination of constraints, and therefore weakly vanishes. In addition to defining a submanifold of admissible states in phase space, they therefore also generate gauge transformations, which encode all the symmetries of the triad formulation of GR. More precisely, the Gauss constraint $\cG_i$ generates $\SU(2)$ local transformations induced by the $\SO(1,3)$ symmetry (\ref{loc_so13}), while the vector constraint generates spatial diffeomorphism on $\Sigma$. The third and so-called scalar constraint $\cS$ is responsible for time reparametrization invariance, therefore extending the gauge symmetries to space-time diffeomorphisms.

\

In its canonical formulation, GR is a fully constrained system, which can be quantized following Dirac's program \cite{dirac_qm, henneaux_teitelboim}. The first step consists in finding a representation of the basic phase space variables as operators acting on a kinematical Hilbert space $\cH_{kin}$, and such that Poisson brackets are turned into commutation relations in the standard way. Then, one has to promote the constraints themselves to operators on $\cH_{kin}$. Third, one should perform the equivalent of finding the constraint surface in the classical theory, which means finding the states annihilated by the constraint operators. This step is crucial in the sense that it is only at this stage that the Hilbert space of physical sates $\cH_{phys}$ is defined. Finally, one has to look for a complete set of physical observables, that is a maximal set of operators commuting with the constraints. The main achievements of LQG regarding this quantization program concerns the first three steps. Contrary to the WdW theory, a kinematical Hilbert space could explicitly be constructed. Moreover, a representation of the constraint algebra on this Hilbert space could be specified, without any anomaly, therefore completing the second step. Ambiguities remain, but they mainly concern the scalar constraint. The Gauss and vector constraints, which form a closed subalgebra, could moreover be solved explicitly. The resulting Hilbert space of gauge and diffeomorphism invariant states has also been proved unique under certain additional assumptions \cite{lost}. This provides a solid arena to study the dynamics of LQG further, as encoded by the scalar constraint. However, its quantum definition remains ambiguous, and a complete set of solutions out of reach, hence precluding the completion of the canonical program. It should also be noted that, even if a satisfactory space of solutions to the scalar constraint were to be found, completing Dirac's program would remain incredibly difficult: even at the classical level, finding a complete (and manageable) set of Dirac observables for GR remains an outstanding task, and possibly forever so. 
   
\

The structures of the kinematical Hilbert space of LQG, and of the solutions to the Gauss and vector constraints, are one of the main inputs for the covariant approach, and therefore deserve to be described in some details here. A distinctive feature of the loop approach is the choice of basic variables to be quantized, which are holonomies of the connection rather than the connection itself. To any path $\ell$ in $\Sigma$ is associated the unique holonomy
\beq
h_\ell ( A ) = \cP \exp\left( - \int_\ell A \right) \in \SU(2)\,, 
\eeq
where $\cP$ stands for path ordering. This $\SU(2)$ element encodes parallel transport of the triad along $\ell$, and collectively the set of all the holonomy functionals fully capture the information contained in the connection. There is an important subtlety however, lying in the fact that smooth connections are replaced by a generalized space of connections before the quantization. It can be given an inductive definition in terms of holonomy functionals associated to closed embedded graphs \cite{Ashtekar:1994wa}, which play a central role in the definition of the kinematical Hilbert space. As far as we know, alternative constructions might be proposed at this level, relying on slightly different generalizations of smooth connections. This concerns for example regularity properties of the graphs (usually assumed to be piecewise analytic), as was already explored in \cite{Zapata:2004xq}. But more importantly so, it seems to us that the relevance of the combinatorial properties of the graphs supporting LQG states might have been underestimated. As we will see in the GFT context, combinatorial restrictions can dramatically improve analytic control over the theory, without affecting its conceptual aspects in any clear way. It seems to us that this freedom, so far overlooked in the canonical approach, would deserve to be investigated and taken advantage of. In any case, the standard kinematical Hilbert space of LQG is spanned by so-called \textit{cylindrical} functionals $\Psi_{\cG, f}$, labeled by a closed oriented graph $\cG$ with $L$ links, and a function $f: \SU(2)^{L} \rightarrow \mathbb{C}$. $\Psi_{\cG, f} (A)$ only depends on the holonomies $h_\ell (A)$ along the links of $\cG$, through the following formula:
\beq
\Psi_{\cG, f} (A) = f( h_{\ell_1} (A) , \ldots , h_{\ell_L} (A) )\,.
\eeq  
$f$ is moreover assumed to be square-integrable with respect to the Haar measure on $\SU(2)^{L}$, which allows to define a scalar product between cylindrical functionals associated to a same graph:
\beq
\langle \Psi_{\cG, f_1} \vert \Psi_{\cG, f_2} \rangle \equiv \int [ \extd h_\ell ]^{L} \overline{f_1 ( h_{\ell_1} , \ldots , h_{\ell_L}) } f_2 ( h_{\ell_1} , \ldots , h_{\ell_L}) \,.
\eeq
This scalar product naturally extends to arbitrary couples of cylindrical functionals, via embeddings of two distinct graphs into a common bigger one. The kinematical Hilbert space $\cH_{kin}$ is the completion with respect to this scalar product of the vector space generated by the cylindrical functionals. It can be shown to be nothing but the space $L^2 ( \overline{\cA}, \extd \mu_{AL})$ of square-integrable functions on the space of generalized connections $\overline{\cA}$, with respect to the Ashtekar-Lewandowski measure $\extd \mu_{AL}$ \cite{Ashtekar:1994mh}. The (matrix elements of the) holonomies are turned into operators, acting by multiplication on the kinematical states. As for the momenta $E$, there are also smeared out before quantization. For any two-dimensional surface $S$ in $\Sigma$, and any $\su(2)$-valued function $\alpha$, one defines the \textit{flux}:
\beq
E ( S , \alpha) = \int_S \alpha^i E_i \,.
\eeq
The action of an elementary flux on a kinematical state can be found by formally turning the densitized triad into a derivative operator with respect to the connection. This construction provides an anomaly-free representation of the classical \textit{holonomy-flux algebra}, that is the Poisson algebra formed by the smeared variables just introduced, which is moreover unique under certain assumptions regarding the diffeomorphism symmetry \cite{lost}. Solving the Gauss constraint is relatively straightforward once this representation has been constructed, since its (quantum) action on cylindrical functionals can be simply inferred from its (classical) action on holonomies. A finite (generalized) gauge transformation is labeled by one $\SU(2)$ element $g(x)$ at each point $x \in \Sigma$, acting on holonomies and cylindrical functionals as:
\bes
h_\ell (A) &\mapsto& (g \act h_\ell) (A) \equiv g_{s(\ell)} h_\ell (A) g_{t(\ell)}^\inv \,, \\ 
\Psi_{\cG, f} (A) = f(  h_{\ell} (A)  ) &\mapsto& (g \act \Psi_{\cG, f}) (A) \equiv f(  (g \act h_{\ell}) (A) )  \,.
\ees  
Thanks to group-averaging techniques, any cylindrical functional can be projected down to a gauge invariant one. Combined with the Peter-Weyl theorem, this allows to construct an explicit orthonormal basis for the gauge invariant Hilbert space. It is the subspace of $\cH_{kin}$ spanned by the \textit{spin network functionals}, which are special cylindrical functionals $\Psi_{\cG, \{ j_\ell \} , \{ \iota_n \}}$. The half-integers $\{ j_\ell \}$ label representations of $\SU(2)$ associated to the links $\{ \ell \}$ of $\cG$, while $\iota_n$ is an $\SU(2)$ intertwiner attached to the node $n$, compatible with the different representations meeting at $n$. $\Psi_{\cG, \{ j_\ell \} , \{ \iota_n \}}$ is the contraction of representations matrices $D^{j_\ell}$ with the intertwiners, according to the pattern of the decorated graph $(\cG, \{ j_\ell \} , \{ \iota_n \})$, which can schematically be denoted:
\beq
\Psi_{\cG, \{ j_\ell \} , \{ \iota_n \}} (A) = \prod_n \iota_n \cdot \prod_{\ell} D^{j_\ell} ( h_\ell (A) )\,.
\eeq
Solving the vector constraint is more intricate, because its solutions have to be found in the topological dual of an appropriate dense subspace of $\cH_{kin}$. This aspect of the canonical theory is largely unrelated to the present thesis, therefore it is sufficient to recall the main conceptual idea, and skip all the technical details. When acted upon a spin network, a diffeomorphism simply moves and deforms the graph on which it is supported. One therefore needs to repackage spin network functionals into equivalent classes of decorated graphs which can be deformed into one another. All the information about the embeddings of the graphs is washed out, except for their knot structures, which provide additional quantum numbers characterizing the diffeomorphism invariant classes. As we will see, a different route has been developed in the covariant theory as regards the diffeomorphism symmetry. The embedding information is altogether dispensed with, and the spin networks states are instead labeled by abstract graphs. The topological structure is then recovered from the graph itself, and likewise smooth diffeomorphisms (be them spatial or four-dimensional) are only expected to be (approximately) recovered in a yet to be defined continuum limit.  

\

The essential outcome of this partial completion of Dirac's quantization program is the notion of \textit{quantum geometry}, which provides a physical interpretation for spin network states. From the quantum fluxes, it is possible to construct geometric operators quantizing the area of a $2$-dimensional surface, or the volume of a $3$-dimensional region. Because the spin network functionals are eigenstates of such operators, it was possible to determine their spectra, which turned out to be discrete. More precisely, it was shown that the quanta of area are carried by the links of a spin network: each link with spin $j_\ell$, puncturing a surface $S$, contributes with a term $\pm 8 \pi \gamma \ell_P^2$ to the (oriented) area of $S$, where $\ell_P$ is the Planck length. Similarly, the nodes of a spin network carry infinitesimal volumes, in a rather complicated but still discrete fashion. These results are at the chore of the applications of LQG techniques to cosmology and black hole physics, and also enter crucially into the semi-heuristic constructions on which spin foam models are based. The fact that the Immirzi parameter $\gamma$ parametrizes these spectra is the reason why we claimed earlier it crucially enters the geometric interpretation of spin networks states.


		
%
%
%
%
%
		\subsection{Spin Foams}

Spin Foam Models (SFMs) are a covariant approach to the quantization of GR, initially introduced to circumvent the difficulties with the scalar constraint encountered in canonical LQG. Like any known quantum theory, LQG is expected to have a second equivalent representation, in the form of a path-integral à la Feynman. SFMs provide tentative formulations of this covariant theory, with the idea that LQG techniques can again make old formal theories well-defined, and that the dynamics of quantum GR might be more amenable in a four-dimensional setting. In the gravitational context, Feynman's formulation of quantum mechanics suggests to define transition amplitudes between three-dimensional boundary states of the gravitational field by integrating over histories (i.e. space time geometries interpolating between such boundary states), with a weight given by an action for GR. Schematically, call $h_1$ and $h_2$ two boundary $3$-geometries, and denote by $g$ a space time manifold having $h_1$ and $h_2$ as boundaries. Then one would like to define the transition amplitude between $h_1$ and $h_2$ by:
\beq\label{abstract_transition}
\langle h_1 \vert h_2 \rangle_{phys} = \int \cD g \, \exp\left( \rm{i} S_{GR} (g) \right)\,,
\eeq
where $\cD g$ is a suitable probability measure on the space of interpolating $4$-geometries, and $S_{GR}$ is an action for GR. This goal is mathematically very challenging, essentially because of the dubious existence of probability measures on spaces of continuous geometries. Indeed, if we were to use coordinates in order to define such quantities, in addition to having to handle the usual difficulties associated to measures on infinite dimensional spaces, one would also have to face the even harder question of diffeomorphism invariance. Even the left-hand side of (\ref{abstract_transition}) is problematic, as it assumes the availability of a well controlled space of $3$-geometries. Finally, even at this very abstract level, what exactly should be summed over is not really clear (metric degrees of freedom only or topologies as well?). 
		
		\
		
In order to make (\ref{abstract_transition}) more concrete, one can follow the ideas behind LQG, but this time in a four-dimensional framework. This can be done in several ways, which all seem to point in a same direction. They can be classified according to the amount of inspiration and results which can be directly traced back to the canonical theory. The first, historical route, has been to reproduce Feynman's heuristic construction from canonical LQG. While this allows to deduce the general form of spin foam amplitudes, no definite model can be inferred, since the dynamics of canonical LQG is itself not well understood. Rather, one hopes in reverse to be able to pin-point the right model on the covariant side (right-hand side of equation (\ref{abstract_transition})), and deduce the definition of the physical Hilbert space (scalar product on the left-hand side of equation (\ref{abstract_transition})). In this approach, one takes as many features of the canonical theory as possible for granted, such as the boundary $SU(2)$ spin network states, and relies on covariant quantization techniques for the dynamics only. In the second approach, one starts the quantization from scratch in the covariant setting, but like in canonical LQG one hopes that shifting the emphasis from metric to tetrads and connection variables is profitable. Finally, the third strategy consists in taking seriously the type of degrees of freedom and discrete features of canonical LQG, but being critical towards strict quantization procedures. The latter will be discussed in details in the next section. Here, we only present and comment on the advantages and shortcomings of the first two strategies.

\

The dynamics of canonical LQG is encoded into the vector and scalar constraints, which together must ensure space-time diffeomorphism invariance. The vector constraint can be imposed at the kinematical level, and with the additional combinatorial abstraction previously mentioned, provides an intermediate Hilbert space spanned by (combinatorial) spin networks. The physical Hilbert space should be deduced by projecting down to states annihilated by the quantum scalar constraint $\widehat{\cS}$. The trick lies in the fact that formally
\beq
\widehat{\cS} \vert s \rangle = 0 \; \Leftrightarrow \; \forall t \in \mathbb{R}\,,\; \exp(i t \widehat{\cS}) \vert s \rangle  = \vert s \rangle \,,
\eeq
and therefore the projector on physical states can be given the formal definition:
\beq
P \equiv \int \extd t \exp(i t \widehat{\cS})\,.
\eeq 
The parameter $t$ is not a time variable, but an abstract parameter of gauge transformations generated by $\widehat{\cS}$. However, formally one can reproduce Feynman's original derivation of the path-integral, by cutting the integral on $t$ into infinitesimal pieces and inserting resolutions of the identity in terms of the spin network basis. This heuristic argument \cite{Reisenberger:1996pu, rovelli_book} leads to the spin foam general ansatz replacing (\ref{abstract_transition}):
\beq\label{abstract_sf}
\langle s_1 \vert s_2 \rangle_{phys} = \sum_{\cF : s_1 \to s_2} A_\cF\,,
\eeq  
where $s_1 = (\cG_1 , \{ j_{\ell_1} \} , \{ \iota_{v_1} \} )$ and $s_2 = (\cG_2 , \{ j_{\ell_2} \} , \{ \iota_{v_2} \} )$ are two spin network states and $\cF$ is a spin foam interpolating between them. In addition to $s_1$ and $s_2$, $\cF$ is labeled by a triplet $(\cC , \{ j_f \} , \{ \iota_e \} )$, where $\cC$ is a $2$-complex with boundary $\cG_1 \cup \cG_2$, $\{ j_f \}$ are $\SU(2)$ representations associated to its faces $\{ f \}$, and $\{ \iota_e \}$ are intertwiners associated to its edges $\{ e \}$. Compatibility between $\cF$ and its boundary requires a face $f$ touching a boundary link $\ell$ to be labeled by $j_f = j_\ell$, and an edge $e$ touching a boundary vertex $v$ to be associated to an intertwiner $\iota_e = \iota_v$. 
The main advantages of (\ref{abstract_sf}) over (\ref{abstract_transition}) is that boundary states have well-defined background-independent labels, and the sum-over-path is combinatorial in nature. Apart from that, the conceptual setting is identical: boundary states represent quantum spatial geometries, while the foams are one dimensional higher analogues which can be given the interpretation of quantum space-times. The main shortcoming is that no clear-cut derivation of this ansatz from canonical LQG is available, and therefore viewing SFMs as the covariant realization of the same thing might be misleading, at least in a literal sense. In any case, this general heuristic argument does not provide much clue as to how the amplitudes $\cA_\cF$ should be defined, nor as to which exact combinatorial structures should be summed over. We discard the question of the combinatorics for the moment, since it will be addressed at length in the core chapters of this thesis. At this stage we only point out that in most of the literature on spin foams, the $2$-complexes are assumed to be dual to simplicial decompositions of some topological $4$-manifold. Taking seriously the fact that the scalar constraint of LQG (as defined by Thiemann \cite{Thiemann:1996ay}) only acts at the nodes of LQG states, one can further assume that the amplitudes $\cA_\cF$ can be factorized over elementary contributions $\cA_v$, only depending on the group-theoretic labels related to the vertex $v$. It is therefore on the definition of the so-called \textit{vertex amplitude} that most of the efforts have been concentrated, and we refer to \cite{Bianchi:2012nk} for details or additional references. We would like to see this assumption as a 'locality principle', akin to the fact that $S_{GR}$ is best understood as an integral over the $4$-manifolds entering the formal definition (\ref{abstract_transition}). In this respect, it is worth-mentioning one particular derivation of spin foam dynamics, outlined in the review \cite{Bianchi:2012nk}, which is tightly related to the canonical theory: LQG provides the degrees of freedom and their quantum geometric interpretation, the dynamics of the Engle-Peirera-Rovelli-Livine (EPRL) \cite{eprl} model is then argued for on the basis of Lorentz covariance and some locality principle. This approach is conceptually similar to the original parts of this thesis, and it seems to us could benefit from the new 'locality principle' we will introduce. For more insights into the relation between canonical LQG and SFMs, we refer to \cite{Alexandrov:2011ab}.

\ 

The other approach to SFMs, more independent from canonical LQG, originates from the Ponzano-Regge model for quantum gravity, which is a spin foam model for Euclidean quantum gravity in three dimensions. The basic idea is that, because the classical theory is topological, it can be discretized before quantization without loss of information. Applying a path-integral quantization to the discretized classical theory, which takes the form of an $\SU(2)$ $BF$ theory, provides a SFM with the same structure as inferred from the heuristic reasoning based on LQG. This model has then been generalized, in several ways, to four-dimensional models with Euclidean or Lorentzian signatures (see for instance \cite{Baez:1997zt, bc, eprl, fk, Dupuis:2011fz, Dupuis:2011wy , bo_bc , bo_holst}). The classical starting point is not the Holst action itself, but rather the Holst-Plebanski one, which recasts gravity as an $\SO(1 , 3)$ $BF$ theory with additional constraints. The Holst-Plebanski action is 
\beq\label{plebanski}
S[ B , \omega , \lambda ] = \frac{1}{\kappa} \int_{\cM} \left[ \left( \star B^{IJ} + \frac{1}{\gamma} B^{IJ}  \right) \wedge F_{IJ} (\omega) + \lambda_{IJKL} B^{IJ} \wedge B^{KL} \right]\,,
\eeq 
where $\lambda^{IJKL}$ is a Lagrange multiplier symmetric under the exchange of pairs $(IJ)$ and $(KL)$, and such that $\varepsilon_{IJKL} \lambda^{IJKL} = 0$. $B$ is an $\so(1,3)$-valued $2$-form, the bivector, and the $B \wedge B$ term ensures that on shell:
\beq
B = \pm \star (e \wedge e)\,, \; \rm{or} \; B = \pm e \wedge e\,,
\eeq 
where $e$ is a triad. Solving for the equations of motion the variation of $\lambda$ provides therefore gives back the Holst action (in one sector). The main interest of the Holst-Plebanski action principle is that $BF$ theory can be rigorously quantized by spin foam techniques. This consists in two steps: 1) define a discretization of the $B$-field and the curvature on some adequate cell complex, which thanks to the topological nature of $BF$, captures all its dynamical features; 2) quantize the discrete theory through path-integral methods. In three dimensions gravity and $BF$ theory coincide, hence the initial interest in such a strategy. 
In four dimensions, the idea is to start from the exact quantization of $BF$, and use a discrete version of the Plebanski constraints to reintroduce the tetrad degrees of freedom \textit{at the quantum level}. Note that in such approaches, one works with the full $\SO(1 , 3)$ (or $\SO(4)$) group rather than $\SU(2)$, and therefore slightly generalizes the formalism. This strategy of quantizing first, and only then constraining the degrees of freedom, is at the same time one of the main appeals and the most cumbersome issue of the spin foam quantization. While this allowed to construct interesting quantum gravity models, some of them well-connected to the canonical theory, it is certainly problematic from the conceptual point of view, and to a large extent explains the variety of models one can construct: ill-defined recipes necessary introduce a handful of ambiguities. We refer again to \cite{Alexandrov:2011ab} for a review of quantization and discretization ambiguities in SFMs.  
		
		\
		
We finish this discussion of SFMs by a brief introduction to the Ponzano-Regge model. Classically, three-dimensional Euclidean gravity can be formulated through the first order action:
\beq
S_{3d} [ e , \omega ] = \int_{\cM} \Tr \left( e \wedge F(\omega) \right)\,,
\eeq  		
where the triad $e$ is an $\su(2)$-valued $1$-form (the $B$ field), $\omega$ is an $\su(2)$-connection with curvature $F(\omega)$, and $\Tr$ stands for the trace in $\su(2)$. As in the four-dimensional context, variation with respect to $\omega$ imposes the torsion-free equation, and the variation of $e$ provides Einstein's equation: $F(\omega) = 0$; hence space-time is flat. The only degrees of freedom of vacuum $3d$ Euclidean gravity are global, which is the sense in which the theory is said to be topological, and this flatness condition should therefore be preserved at the quantum level. Indeed, the triad can formally be integrated in the (ill-defined) continuous partition function  
\beq\label{pr_formal}
\cZ_{3d} = \int \cD \omega \int \cD e \exp\left( \rm{i} S_{3d} [ e , \omega ] \right) = \int \cD \omega \, \delta( F(\omega) )\,,
\eeq
suggesting that a rigorous definition should in a way measure the 'volume' of the set of flat connections on $\cM$. This expectation can be made more rigorous in the discrete: let us introduce a cellular decomposition $\Delta$ of $\cM$, and its dual $2$-complex $\Delta^{*}$. For definiteness, one can for example assume that $\Delta$ is a three-dimensional simplicial complex: elementary cells are tetrahedra, glued along their boundary triangles. $\Delta^{*}$ is a set of nodes $n$, lines $\ell$ and faces $f$ such that: inside each tetrahedron of $\Delta$ one finds a unique node, two nodes are connected by a line whenever the two dual tetrahedra share a triangle, and faces are collections of lines closing around the edges of $\Delta$. Similarly to lattice gauge theory, one then discretizes the connection by extracting its holonomy along each line of $\Delta^{*}$, noted $h_\ell \in \SU(2)$. As for the triad, it can be integrated along the edges of $\Delta$, and provides a Lie algebra element $X_e = X_f$ for each edge $e \in \Delta$ dual to the face $f \in \Delta^{*}$. This allow to discretize the action in the following way:
\beq
S_{\Delta} (X , h) = \sum_{f \in \Delta^{*}} \Tr \left( X_f H_f \right)\,,
\eeq
where
\beq
H_f = \underset{\ell \in f}{\overrightarrow{\prod}} h_\ell
\eeq		
is the oriented product of the holonomies around the face $f$. The formal partition function (\ref{pr_formal}) is made concrete by summing holonomies with the Haar measure $\extd h_\ell$ over $\SU(2)$, and triads with the Lebesgue measure $\extd X_f$ on $\su(2) \simeq \mathbb{R}^3$:
\beq\label{pr_holo}
Z_{PR} (\Delta) = \int [ \extd X_f ] \int [ \extd h_\ell ] \exp \left( {\rm{i}} \sum_f {\rm{tr}} (X_f H_f) \right) = \int [ \extd h_\ell ] \prod_{f \in \Delta^{*} } \delta \left( H_f \right)\,.
\eeq		
This formulation of the Ponzano-Regge partition function, in lattice gauge theoretic language making flatness of the geometry apparent, will be primary in the GFT context. Two other formulations are however possible. One could instead integrate the connection degrees of freedom, and express $Z_{PR}$ as an integral over Lie algebra variables \cite{Baratin:2010wi}. The subtlety is that these are non-commutative variables, but we will see how it can be done and put to good use in later chapters. The other rewriting, which we describe now, makes the connection to spin foam models explicit. One simply relies on the Peter-Weyl theorem, and expands the $\delta$-functions in representations:
\beq
\delta = \sum_{j \in \frac{\mathbb{N}}{2}} ( 2 j + 1 ) \chi_j \,,
\eeq  
where $\chi_j$ is the character of the $j^{\rm{th}}$ representation of $\SU(2)$. Each character can be decomposed in products of Wigner matrices with individual line holonomies as variables. Each variable $h_\ell$ will appear exactly three times, one for each edge of the dual triangle of $\ell$. The $h_\ell$ holonomies can therefore be integrated explicitly, yielding one $3$-valent intertwiner per line, which is nothing but a $3 j$ symbol. Finally, these can be contracted together, four by four, following the tetrahedral pattern associated to each vertex of $\Delta^{*}$. The partition function then takes the original form proposed by Ponzano and Regge in their seminal work \cite{ponzano_regge}:
\beq
Z_{PR} (\Delta) = \sum_{\{j_f\}} \prod_{f \in \Delta^* } (- 1)^{2 j_f } (2 j_f + 1 ) \prod_{v \in \Delta^* } \{ 6 j \} (v) \,,
\eeq
where the sum runs over all possible spin attributions to the faces of $\Delta^*$, and $\{ 6 j \} (v)$ denotes the evaluation of the $6j$ symbol on the six spin labels of the faces of $\Delta^*$ running through $v$. This is the first SFM ever proposed, and we see that in the simplicial setting, the vertex amplitude is essentially captured by the $6 j$ symbol.
		
		\
		
		
		
		
		
		
		
		\subsection{Summing over Spin Foams}
		
		From the point of view of the heuristic path towards spin foams from canonical LQG (formula (\ref{abstract_sf})), it is clear that when computing transition amplitudes, one should not only sum over the group-theoretic data, but also on the combinatorial structure of the foam. The first question to come to mind is then: with which measure? Since the precise form of the amplitudes is dictated by a quantization of the discretized theory, where the foam is chosen once and for all at the classical level, this is a difficult question to address. On the other hand, if we do not give too much credit to the heuristic construction à la Feynman, and focus instead on the Ponzano-Regge and related spin foam models for three-dimensional gravity, it is possible to argue for a different strategy. Indeed, an interesting property of the Ponzano-Regge model is that it is formally triangulation independent: that is, performing any possible local move which do not change the topology (Pachner moves) on the triangulation only affects the amplitude through a generically divergent overall factor. The unpleasant formal character of this invariance can even be cured at the price of trading the $\SU(2)$ group for its quantum-deformed version. Is obtained in this way the so-called Turaev-Viro model, which is usually interpreted as 3d quantum gravity with a non-vanishing positive cosmological constant \cite{turaev_viro}. Since 3d quantum gravity turns out to be equivalent in this sense to the definition of a topological invariant for $3$-manifolds, it is tempting to assume that 4d quantum gravity will be likewise equivalent to the definition of a diffeomorphism invariant for $4$-manifolds. Again, in the realm of triangulated manifolds, diffeomorphism invariance can be understood as triangulation invariance, and more precisely invariance under local Pachner moves. This key property is what is generally understood as entailing the definition of a well-behaved \textit{state-sum model}, applying tools from topological quantum field theory to quantum gravity \cite{Barrett:1995mg}. 
		
		\
		
	In this thesis we tend to favor the heuristic argument leading to the idea of SFMs over the very formal idea of a quantum invariant. The outstanding difficulties paving the way to a satisfactory diffeomorphism invariant in 4d, and the rigidity of such a strategy are discouraging to us. As explained in the introduction, we would rather favor an approach in which the notions of scale and renormalization have a role to play, thus allowing for more flexible models, where approximate rather than exact invariance matters. This route is currently explored in details by Bianca Dittrich and collaborators \cite{Bahr:2011uj, Dittrich:2011zh, Dittrich:2011vz, bianca_review , bianca_cyl}. Rather than looking for an exact state-sum model, they are instead developing coarse-graining and renormalization methods allowing to consistently improve the dynamics, and hopefully reach a diffeomorphism-invariant phase in some \textit{infinite refinement limit} of the foams. In addition to providing a constructive method towards diffeomorphism invariance, this framework aims at developing the necessary tools to perform approximate calculations. From the point of view of applications to realistic physical situations, this seems to us a good alternative to the more abstract incarnations of state-sum models. 
	
		\
		
More prosaically, most of the SFMs currently under investigation not only fail to implement topological or diffeomorphism invariance, but even fail to verify the axioms which are at the basis of the state-sum approach \cite{Baez:1997zt}. For instance, the 'projector' defining the intertwiner space of the Lorentzian EPRL model does actually not square to itself. As a result, spin foam amplitudes do not compose well: if $\cF_{12}$ is a spin foam mapping a spin network state $s_1$ to $s_2$, and $\cF_{23}$ maps $s_2$ to $s_3$, one has in general
\beq
A_{\cF_{23} \circ \cF_{12}} \neq A_{\cF_{23}} A_{\cF_{12}}\,,
\eeq
contrary to the formalism advocated in \cite{Baez:1997zt}. Even in the Ponzano-Regge model, such requirements cannot be implemented in full generalities. Because of the formal of nature equation (\ref{pr_holo}), plagued with divergences, one needs to introduce a regulator cutting-off large spin labels. Such a cut-off spoils both the triangulation invariance and the morphism properties of the Ponzano-Regge
amplitudes. It is true that in this case one could resort to the well-behaved Turaev-Viro model. However, from a more physical point of view it is not clear why the cosmological constant should be necessary to the very definition of spin foam models. These pathologies seem on the contrary to suggest that the strict notion of state-sum model, with rigid composition rules, is too narrow to accommodate interesting and physically sound proposals such as the EPRL model. 

\

Let us summarize shortcomings of SFMs such as the Ponzano-Regge model or the EPRL one. First, they do not provide any canonical measure on the space of foams to be summed over, and even this space is not clearly constrained. Second, the presence of divergences (at least when the cosmological constant vanishes) calls for regularization and renormalization procedures, which as in usual quantum field theories are expected to map different combinatorial structures to one another (via coarse-graining or renormalization steps). Again, the quantization procedure leading to specific models does not seem to provide any hint as to how this should be done. Third, and this is certainly the main drawback, the first issue together with the absence of triangulation independence precludes any complete definition of the transition amplitudes (\ref{abstract_sf}), even at the formal level. We would like to add a fourth trouble, which we will address in more details later one: the question of the \textit{continuum limit} of SFMs. If spin networks boundary states are thought of as encoding elementary excitations of the gravitational field, or equivalently atoms of space, one is naturally inclined to address the question of the dynamics of \textit{very many} such states, collectively representing macroscopic and approximately smooth geometries. Indeed, if LQG and SFMs can ever be related to continuous GR in full generality, and hence established as proper quantum theories of gravity, it seems unlikely to be at the level of spin networks with a few links and nodes. For if such states were to describe macroscopic spaces, these would be highly discrete ones (i.e. sampled with very few points), hardly comparable to the smooth structure we experience in everyday life and physics. It is true that they might be appropriate to describe specific physical situations where the number of relevant degrees of freedom are effectively small, such as for instance in cosmology, but they cannot themselves explain the emergence of classical space-time. On the contrary, handling combinatorially large and complicated spin networks can only be made possible if new approximation tools are developed. Individually, and even more altogether, these four challenges seem to point towards essentially two alternatives, both relying on renormalization. The first, inspired by lattice gauge theories, is to find a consistent way of \textit{refining} SFMs, as already explained. The second is to instead look for a way of consistently \textit{summing} spin foam amplitudes. 
		
		

		\subsection{Towards well-defined quantum field theories of Spin Networks}
		
The application of Group Field Theory to spin foam models originates from seminal work by De Pietri, Freidel, Krasnov, Reisenberger and Rovelli \cite{dPFKR, GFT_rovelli_reisenberg}. The central idea is to construct a generating functional for spin foam amplitudes, which allows to sum them with quantum field theory techniques. Let us give a concrete illustration with the first GFT ever proposed, the Boulatov model \cite{boulatov}, which generates Ponzano-Regge amplitudes. Following the general QFT procedure, we want to encode the boundary states of the model into functionals of one or several fields, and for definiteness let us say a single scalar field $\vphi$. Let us moreover aim at the lattice gauge formulation of the Ponzano-Regge amplitudes given in equation (\ref{pr_holo}). It is then natural to assume the field $\vphi$ to have support on several copies of $\SU(2)$. In its simplicial version, the boundary states of the Ponzano-Regge model are labeled by closed graphs with three-valent vertices, whose analogue in the field theory formalism are convolutions of the fields $\vphi$. It is thus necessary to work with three copies of $\SU(2)$, and hence natural to assume $\vphi \in L^2 (\SU(2)^3 )$ (with respect to the Haar measure). Incidentally, one immediately recognizes how to recover the spin network boundary states of the Ponzano-Regge model: through the harmonic expansion of $\vphi$ in terms of Fourier modes labeled by triplets of spins. The partition function of the looked for field theory will have the generic structure:
\beq\label{gft_structure}
\cZ_{GFT} = \int \extd \mu_C (\vphi) \exp( - S_{int} (\vphi))\,,
\eeq 
where $\extd \mu_C$ is a Gaussian measure defining the notion of propagation of boundary data, and $S_{int}$ is the interaction part of an action, encoding the non-Gaussian part of the dynamics. Note that, for reasons which will become clear shortly, the kinetic part of the action is encoded in the Gaussian measure $\extd \mu_C$, together with the (ill-defined) Lebesgue measure on $L^2 (\SU(2)^3 )$. Let us focus on the combinatorics first. A boundary field is associated to a node of a spin network, while its variables label the three links connected to this node. Therefore in the bulk the GFT fields must label spinfoam edges, and a field variable be associated to a couple $w = (e , f)$ where $f$ is a face incident on the edge $e$, called a wedge (see Figure \ref{tetra_wedges}). We want moreover to recover the $4$-valent interaction vertices of the Ponzano-Regge model when constructing the perturbative theory, through an expansion of the exponential term in (\ref{gft_structure}). Therefore $S_{int}$ must consist in a single term, which precisely convolutes four GFT fields according to a tetrahedral pattern. Since in any QFT we are free to encode the non-combinatorial aspects of the dynamics in the propagator rather than the interaction, we can, without loss of generality fix:
\beq\label{Sint_3d1}
S_{int} (\vphi) = \lambda \int [ \extd g_i ]^{6} \, \vphi ( g_1 , g_2 , g_3 ) \vphi ( g_3 , g_5 , g_4 ) \vphi ( g_5 , g_2 , g_6 ) \vphi ( g_4 , g_6 , g_1 )\,,
\eeq
where $\lambda$ is a new coupling constant. This ensures that, whatever the precise form of the propagator $C$, the perturbative expansion of (\ref{gft_structure}) will be labeled by arbitrary $2$-complexes verifying the condition that at each vertex meet exactly $4$ edges and $6$ faces, with a 'tetrahedral' pattern. See Figure \ref{tetra_dual1}.
\begin{figure}[h]
  \centering
  \subfloat[$\omega = 1$]{\label{int_boulatov}\includegraphics[scale=0.6]{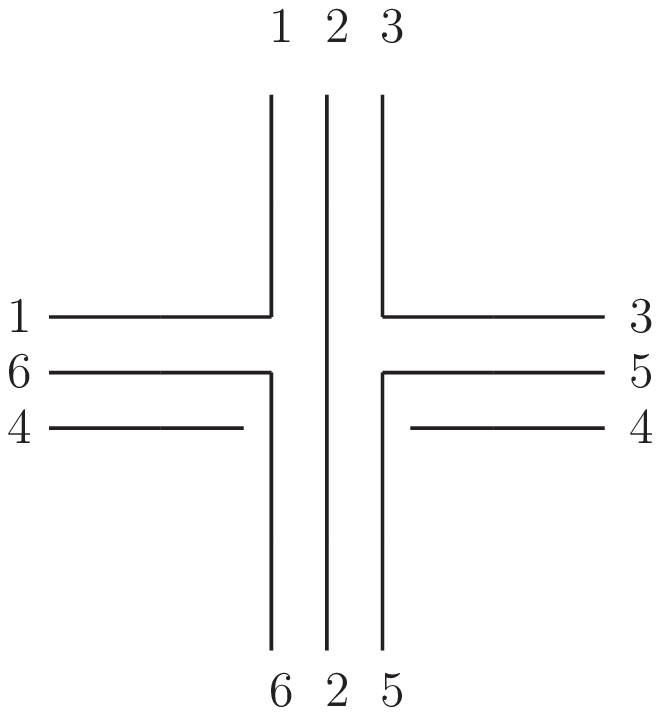}}   	    
  \subfloat[$\omega = 0$]{\label{tetra_wedges}\includegraphics[scale=0.6]{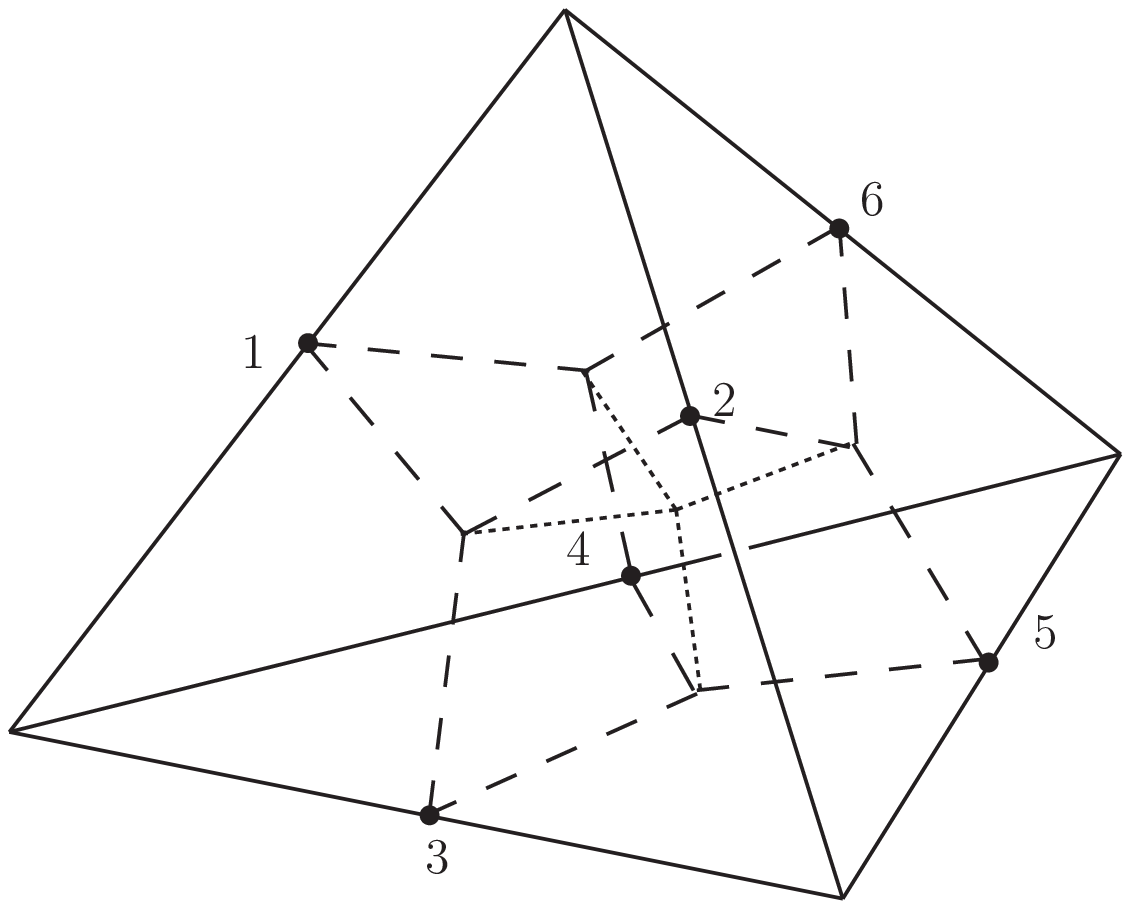}}
  \caption{(\ref{int_boulatov}) Interaction vertex of the Boulatov model: each half-line, dual to a triangle, is composed of three strands, themselves dual to the strands of the triangle; (\ref{tetra_wedges}) elementary tetrahedron, dual to the vertex, with its six wedges represented in dashed lines.}\label{tetra_dual1}
\end{figure}
We now need to find the correct covariance, so that the amplitudes contain the relevant $\delta$-functions ensuring triviality of the holonomies around any face. Recall that the covariance defines a kernel $C( g_1 , g_2 , g_3 ; g_1' , g_2' , g_3')$ (the propagator), through:
\beq
\int \extd \mu_C (\vphi) \, \vphi ( g_1 , g_2 , g_3) \vphi (g_3' , g_2' , g_1') = C( g_1 , g_2 , g_3 ; g_1' , g_2' , g_3')
\eeq
The first data which should be introduced in this kernel are the holonomies $h_\ell$ associated to the lines of the foam. They cannot be the field variables $g_i$ or $g_i'$ themselves, which as we pointed out rather label the wedges of the foam. Therefore $C$ should be expressed as an integral over a single $\SU(2)$ parameter. It should also contain three $\delta$ functions depending on $h$, such that when constructing the amplitudes and integrating out the $g$'s one recovers formula (\ref{pr_holo}). It is then not difficult to convince oneself and check that the solution is:
\beq\label{cov_dege}
C( g_1 , g_2 , g_3 ; g_1' , g_2' , g_3') = \int \extd h \, \delta( g_1 h g_1'^{\inv}) \delta( g_2 h g_2'^{\inv}) \delta( g_3 h g_3'^{\inv})\,.
\eeq
In this case, one can indeed show that:
\beq
\cZ_{GFT} = \sum_{\Delta*} \frac{(- \lambda)^{V(\Delta*)}}{s(\Delta^*)} Z_{PR} (\Delta)\,,
\eeq
where $\Delta$ runs over all the $2$-complexes (without boundaries) verifying the aforementioned combinatorial constraints, $V(\Delta^*)$ is the number of vertices in $\Delta^*$, and $s( \Delta^* )$ is a symmetry factor. It is now clear why we did not want to write the kinetic term of this field theory in an exponentiated form: because of the integration over the holonomy $h$, $C$ does not define an invertible operator, and therefore this representation is strictly speaking not available. A different but equivalent way of understanding this point is sometimes adopted (see for instance the presentation adopted in the classical references \cite{Perez:2003, rovelli_book, freidel_gft, daniele_rev2006}). Instead of integrating over the whole Hilbert space $L^2 (\SU(2)^3)$ with respect to a degenerate Gaussian measure, one can restrict beforehand the space of fields to the \textit{gauge invariant} ones:
\beq\label{gauge1}
\forall h \in \SU(2) \, , \qquad \vphi( g_1 , g_2 , g_3) = \vphi( h g_1 , h g_2 , h g_3)\,, 
\eeq
and use an invertible propagator (the trivial one with no $h$ variable). The constraint (\ref{gauge1}) is nothing but the Gauss one in the GFT setting. It implies that boundary states are really gauge invariant (abstract) spin networks. Using an unconstrained space of fields with the degenerate covariance (\ref{cov_dege}) is morally equivalent, in the sense that only the gauge invariant part of the field is propagated in this framework. The reasons why we prefer to adopt the degenerate covariance formalism are: a) it is mathematically rigorous; b) it allows to use simple notions of locality for the interactions \cite{Krajewski:2010yq}; c) it seems convenient to encode as much as the dynamics as possible in the propagator, this way combinatorial issues and the problem of how gravitational constraints should be imposed get clearly separated. In particular, this thesis explores ways of defining a locality principle, and covariances implementing the Gauss constraint, such that a well-defined theory of renormalization becomes available. 
\

The correspondence between GFTs and SFMs we have just illustrated is very general, as the example of the Boulatov model hopefully makes it clear. GFTs can from this perspective be characterized as quantum field theories for spin networks, generalizing the correspondence between relativistic QFTs and scattering states one finds at the root of particle physics. The only constraints are on the valencies of the spin networks nodes on the boundary, and the structures of the spin foam vertices in the bulk. The former are determined by the field content, while the latter depend on the choice of interactions. More precisely, if we want to have $n$-valent spin networks nodes in the boundary, we need to bring in one new GFT field with $n$ variables; and to any type of vertex in the spin foam amplitudes must be associated a certain convolution of fields in $S_{int}$. Therefore the formalism is in principle general enough to accommodate any finite numbers of nodes in the boundary, and any finite number of vertices in the bulk. Likewise, the type of boundary data and dynamics can be specified at will by changing the group and the covariance. In particular, any constraint arising from the spin foam quantization, such as simplicity constraints, can be included in the covariance. From this point of view, the Gauss constraint is a defining feature of GFTs, since it is responsible for the presence of holonomy degrees of freedom. This explains why we will concentrate on this aspect in the remainder of this thesis.

\

The appeals of the GFT formalism from the point of view of SFMs and LQG are numerous. First and foremost, the QFT formalism allows to fix canonical weights for the sum over foams. While a full justification of these weights from first principles is not available, they have the merit of being well-defined, therefore completing the definition of SFMs. They can moreover be partially justified: with an appropriate definition of the GFT coupling constant, the combinatorial weight of a complex is the order of its automorphism group, which can be argued to be a discrete, purely combinatorial analogue of the diffeomorphism group. Second, the divergences of SFMs are now understood as divergences of the amplitudes of a field theory. Such features have to be expected in any QFT, and tools to control them are well understood both conceptually and mathematically, thanks to renormalization theory. Through the embedding of SFMs into GFTs, the question of the value of the cosmological constant and that of the occurrence of divergences become clearly separated, as it seems to us they should be. A third feature which can be seen as an advantage too, is that in GFT the topology of space-time is itself dynamical. For instance in the version of the Boulatov model introduced above, the foams contributing to the partition function are dual to arbitrary gluings of tetrahedra, which include in particular all types of triangulated topological manifolds, but also more pathological structures.  Therefore GFT can potentially explain the local topology of macroscopic space-time, in addition to its metric properties. This feature is sometimes called third quantization, and we refer to \cite{ds_3} for a more detailed exposition in the context of GFTs. Finally, as quantum field theories on local symmetry groups rather than space-time, GFTs make possible to incorporate all the tools which are so crucial to quantum field theories in a background independent context. Especially, perturbative methods in the QFT sense are available, without having to resort to perturbation in the space-time sense. This is very different from the usual perturbative approach to quantum gravity, which is so to speak doubly perturbative: perturbative QFT methods are used to analyze quantum perturbations over a background metric. This aspect of GFTs has already been taken advantage of, in the limited context of 'graviton propagator' calculations, where the existence and physical relevance of the perturbative expansion in the coupling constants of GFTs is assumed from the start \cite{lqg_propa1, lqg_propa2, lqg_propa3}. 
\

In order to determine whether this list of merits is truly realized in GFTs, a lot of work is needed, both at the conceptual and mathematical level. The key result to achieve in this respect is to establish GFTs as well-defined perturbative quantum field theories. To this effect, one needs to generalize renormalization theory to this new context. It is only equipped with such a new toolbox that we will be able to determine whether a specific GFT model can be taken seriously as a field theory or not, and in which sense. It seems that the key physical questions are to be settled down only then. The most pressing one is to determine in which sector of a given theory (if any) the classical effective dynamics of GR lies. In particular, can we relate the first few orders of the perturbative expansion to continuum macroscopic physics? Or, on the contrary, does it emerge from the interaction of very many degrees of freedom, and can as a result only be captured by non-perturbative techniques, or by perturbation in a different phase of the theory? The somehow mysterious interpretation of the coupling constants of a GFT from the gravity perspective is certainly related. Whatever the answers to these questions, for which in our view renormalization methods applied to phase transitions could play a discriminant role, it must be admitted that if GFTs are used to complete the definition of SFMs, the problem of their renormalizability have to be faced head on and resolved.  
\

This concludes our summary of motivations for studying GFTs, from the point of view of loop quantum gravity and spin foam models. As is certainly clear to the reader, several points of the reasoning are semi-heuristic and formal. Moreover, several key combinatorial aspects of GFTs seem poorly motivated by the quantization procedures, and it seems to us essentially independent of them. It is therefore of paramount importance to take some distance with the quantization, and present the more combinatorial motivations behind GFTs, which we do in the next section. 		
		
		
		
		
	\section{Group Field Theories and random discrete geometries}
	
		\subsection{Matrix models and random surfaces}
	
Since GFTs are in a sense higher dimensional incarnations of matrix models, we start with a brief presentation of basic aspects of the latter \cite{david1985planar , Ginsparg:1991bi , DiFrancesco:1993nw , Ginsparg:1993is }. Matrix models are statistical models for matrix-like degrees of freedom, in the sense that 'locality' is based on a matrix rather than point-wise product. For instance, let $M$ be an hermitian matrix of size $N \times N$. We can construct an action $S(M)$ for $M$ by requiring it to be invariant under conjugation of $M$. This plays the role of 'locality' principle in the same way as Lorentz and gauge invariance are at the root of local quantum field theories, i.e. by providing a set of allowed interactions. It is then possible to show that $S(M)$ has to be a sum of products of invariants of the form: $\tr M^k$, with $k \in \mathbb{N}^*$. If we further restrict to \textit{connected} invariants, then $S(M)$ is simply a sum of such terms. The simplest non-trivial connected invariant action retains the first two terms only:
\beq
S(M) = \frac{1}{2} \tr M^2 - \lambda \tr M^3\,,
\eeq  
where $\lambda$ is a coupling constant. The partition function of the matrix model is then defined by:
\beq\label{matrix_3}
\exp{\cZ} = \int \extd M \, \exp \left(- S(M) \right) = \int \extd M \, \exp \left( - \frac{1}{2} \tr M^2 + \lambda \tr M^3\right)\,,
\eeq
where $\extd M$ is  the invariant measure on the set of $N \times N$ hermitian matrices. This theory can be expanded in perturbations in $\lambda$ and shown to generate Feynman amplitudes labeled by \textit{ribbon graphs}. The propagator, deduced from the quadratic Gaussian term in $S(M)$, can be pictured as a line with two strands, each strand carrying one index of the matrix $M$. More precisely, the covariance of this theory identifies indices as follows:
\beq
C_{ij;kl} = \int [\extd M] M_{ij} M_{kl} \exp \left( - \frac{1}{2} \tr M^2 \right) = \delta_{j , k} \delta_{l , i} \,.
\eeq
The interaction part of the action introduces a single $3$-valent vertex, which identifies $6$ strands pairwise. The propagator and interaction vertex are represented in Figure \ref{int_propa_matrix}. 
\begin{figure}[h]
\begin{center}
\includegraphics[scale=0.7]{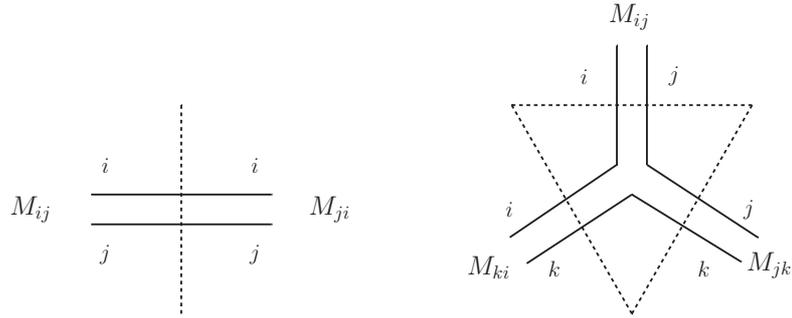}
\caption{Propagator and vertex interaction of the matrix model (\ref{matrix_3}). In dashed lines: the dual edge of the propagator and the dual triangle of the vertex.}
\label{int_propa_matrix}
\end{center}
\end{figure}
The free energy $\cZ$ is then indexed by closed and connected ribbon graphs $\cG$, an example of an open ribbon graph being given in Figure \ref{ribbon}. 
\begin{figure}[h]
\begin{center}
\includegraphics[scale=0.7]{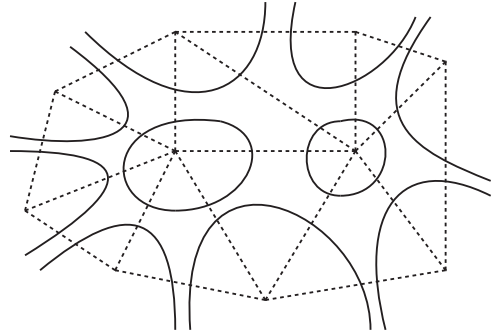}
\caption{An open ribbon graph, with two closed faces, and its dual triangulation.}
\label{ribbon}
\end{center}
\end{figure}
The ribbon structure of the graphs allows to define the notion of face: a face of $\cG$ is a set of strands forming a loop. We note $n(\cG)$ the number of ($3$-valent) vertices of $\cG$, and $F(\cG)$ its number of faces. It is then not difficult to see that: 
\bes\label{partition_matrix}
\cZ &=& \sum_{\cG \; {\rm{connected}}} \frac{1}{s(\cG)} \lambda^{n(\cG)} \cA_\cG \,, \\
\cA_\cG &=& N^{F(\cG)}\,,
\ees   
where $s(\cG)$ is a symmetry factor. The fact that the amplitude $\cA_\cG$ is weighted by the number of faces is easy to understand: since in $\cG$, the indices of the strands are identified by $\delta$-functions across propagators and vertices, one can trivially sum all of them but one per face; we are therefore left with one free index per face, which sums to $N$. The main interest of the occurrence of such ribbon graphs in the perturbative expansion of the matrix models, is that their duals are triangulated surfaces. To see this, it suffices to associate a transverse line to each propagator, and a triangle to each vertex, as represented in Figures \ref{int_propa_matrix} and \ref{ribbon}. In this dual picture, the role of the propagator is to identify pairwise the edges of the $n(\cG)$ triangles generated by $\cG$, yielding a closed triangulated surface. Note that, very importantly, the vertices of the triangulation are dual to the faces of $\cG$, so that a complete and unambiguous correspondence is established between a ribbon graph and its dual triangulation. This interpretation of matrix models as statistical models for discretized surfaces is very general. The statistical ensemble one is spanning depends on which matrix ensemble and which interactions are being used. For instance, in the model we just presented, the hermiticity of $M$ restricts the sum to orientable triangulations, while both orientable and non-orientable triangulations would be included had we used symmetric matrices instead. Also, one could work with quadrangulations upon replacing the third order interaction by a fourth order one, or include arbitrarily wild kinds of discretizations thanks to all the other matrix invariants.  

\

The first technical aspect which makes the link between matrix models and random triangulations so interesting is the existence of a $1/N$ expansion, bringing powerful analytical control over the formal sum (\ref{partition_matrix}). One can use the arbitrary size of the matrix $N$ to unravel universal properties of the statistical model, in the limit of infinitely large matrices. In practice, one looks for a rescaling of the coupling constant
\beq
\lambda \mapsto \frac{\lambda}{N^{\alpha}}\,,
\eeq 
such that, when $N \to + \infty$: a) the most divergent configurations have a uniform divergence degree at each order in $\lambda$; b) these configurations are infinitely many. In such a situation, $N$ can be used as a new perturbative parameter, in such a way that the leading order term in $N$ captures infinitely many orders in $\lambda$. This is in this sense that a $1/N$ expansion captures non-perturbative effects. In our context, large orders in the coupling constants correspond to discrete surfaces with many building blocks, hence the $1/N$ expansion is particularly relevant to the question of the continuum limit of such discrete models. Let us consider again the model (\ref{matrix_3}). The unique value of $\alpha$ verifying the two previous conditions is $1/2$, and in this case the amplitude of a graph $\cG$ becomes:
\beq
\cA_\cG = N^{F(\cG) - \frac{1}{2} n(\cG)}\,.  
\eeq
Let us call $V(\Delta_\cG)$, $E(\Delta_\cG)$ and $T(\Delta_\cG)$ the numbers of vertices, edges, and triangles in the triangulation $\Delta_\cG$ dual to $\cG$. Each edge of $\Delta_\cG$ being shared by exactly two triangles, we have $2 E(\Delta_\cG) = 3 T(\Delta_\cG)$ and therefore:
\beq
F(\cG) - \frac{1}{2} n(\cG) = V(\Delta_\cG) - \frac{1}{2} T(\Delta_\cG) = V(\Delta_\cG) - E(\Delta_\cG) + T(\Delta_\cG) = \chi(\Delta_\cG) = 2 - 2 g(\Delta_\cG)\,, 
\eeq
where $\chi(\Delta_\cG)$ is the Euler characteristic of $\Delta_\cG$, and $g(\Delta_\cG)$ its genus. Recall that closed two-dimensional topological manifolds are fully characterized by two invariants, the Euler characteristic and the orientability, so that in the orientable case we are considering the genus fully determines the topology of the triangulation. Therefore the $1/N$ expansion is actually a topological expansion:
\bes
\cZ &=& \sum_{\cG \; {\rm{connected}}} \frac{1}{s(\cG)} \lambda^{n(\cG)} N^{2 - 2 g( \Delta_\cG )} \\
&=& \sum_{g \in \mathbb{N}^* } N^{2 - 2 g} \cZ_g (\lambda)\,,
\ees
where 
\beq
\cZ_g (\lambda) \equiv \sum_{\cG \; {\rm{connected}}\vert g(\Delta_\cG) = g} \frac{1}{s(\cG)} \lambda^{n(\cG)}
\eeq
sums all the triangulations with genus $g$. As a result, in the limit of infinitely large matrices, spherical topologies ($g=0$) dominate; and the more holes a non-spherical surface has, the more it is suppressed.
  		
Assuming that one can first take the large $N$ limit to restrict oneself to the spherical sector, one can then look for a continuum limit in this sector. Because $\cZ_0 (\lambda)$ contains all the spherical triangulations, with arbitrary number of triangles, it is meaningful to ask whether it is dominated by large triangulations in some regime of $\lambda$. Such large order terms will start dominating the behavior of $\cZ_0$ close to its convergence radius. It can be shown that $\cZ_0$ has a non-zero and finite convergence radius, providing a critical value $\lambda_c$ for the coupling constant, and such that:
\beq
\cZ_0 (\lambda) \underset{\lambda \to \lambda_c}{\sim} \vert \lambda - \lambda_c \vert^{2 - \gamma}\,.
\eeq
$\gamma$ is the critical exponent, and is equal to $-1/2$ in this model, so that the free energy associated to the spherical sector diverges close to the critical point. At least at a formal level, this supports the idea that a continuum phase of the theory is reached by tuning $\lambda$ to its critical value. We can illustrate this in the following way. Suppose that each triangle of the model is assumed equilateral, and attributed some elementary area $a$. Then the mean area of a triangulation in the statistical ensemble is given by 
\beq
\langle A \rangle = a \langle n(\cG) \rangle = a \frac{\extd}{\extd \lambda} \ln( \cZ_0 (\lambda) ) \underset{\lambda \to \lambda_c}{\sim} \frac{a}{\vert \lambda - \lambda_c \vert }\,.
\eeq 
So that if we send the fiducial area to $0$ when approaching criticality, in such a way that the mean value of the area is kept fixed, one obtains a statistical model for infinitely refined spherical triangulations with a given area. This has certainly the flavor of the continuum! 
	
	One can go further and construct statistical ensembles containing all possible (oriented) manifolds. This is achieved in the so-called \textit{double scaling} limit, which consists in taking the large $N$ and critical limits in a correlated manner. It is possible at all thanks to the fact that $\cZ_g$ can be shown to have the same critical point as $\cZ_0$, for any $g \in \mathbb{N}$. More precisely, one has 
\beq
\cZ_g (\lambda) \underset{\lambda \to \lambda_c}{\sim} \vert \lambda - \lambda_c \vert^{\frac{(2 - \gamma)(2 - 2 g)}{2}}\,,
\eeq
which suggests to take the joint limit $N \to + \infty$ and $\lambda \to \lambda_c$, keeping the ratio
\beq
\kappa^{-1} \equiv N \vert \lambda - \lambda_c \vert^{\frac{2 - \gamma}{2}}
\eeq
fixed. In such a limit, the damping effect due to the large $N$ limit and the enhancement of higher genera close to the critical point compensate, in such a way that all topologies contribute to the free energy:
\beq
\cZ \sim \sum_{g  \in \mathbb{N}} \kappa^{2 g - 2 } f_g \,,
\eeq
where $f_g$ are some computable constants. 

The way these results are usually proven is through a rewriting of the free energy in terms of $\U(N)$ invariant quantities, namely the eigenvalues of $M$. Such a tool is not (yet?) available in higher dimensions, so we do not want to provide any details along this line here. The interested reader may refer to \cite{david1985planar , Ginsparg:1991bi , DiFrancesco:1993nw , Ginsparg:1993is }. It should also be noted that these results are only a first step, and that the precise correspondence between the discrete theory and its continuum limit is studied by different means, for instance thanks to Schwinger-Dyson equations for the so-called \textit{loop observables}, which in the large $N$ limit reproduce the Witt algebra \cite{Ginsparg:1991bi , DiFrancesco:1993nw , Ginsparg:1993is }, and can be related to the Wheeler-DeWitt equation for 2d quantum gravity \cite{staudacher_moore_seiberg}. Finally, mathematically rigorous characterizations of the continuum limit of matrix models have been obtained thanks to random planar maps, see \cite{legall} and references therein.

\

As far as quantum gravity is concerned, what makes matrix models interesting is their connection to quantum gravity in two dimensions. We only recall the correspondence at the discrete level, which can be shown to extend to the continuum phase thanks to the Schwinger-Dyson equations. The action for gravity in two dimensions with cosmological constants is 
\beq\label{gr_2d}
S_{2d} = \frac{1}{G} \int_{S} \extd x^2 \, \sqrt{- g} \left( - R(g) + \Lambda \right) = - \frac{4 \pi}{G} \chi(S) + \frac{\Lambda}{G} A_S \,, 
\eeq
where the second equality is a consequence of the Gauss-Bonnet theorem. The Einstein term being topological in two dimensions, only the Euler characteristic and the area $A_S$ can be dynamical in the vacuum. Intuitively, this suggests that a fine enough discretization of the theory is enough to encode the dynamics of these two global degrees of freedom. We can for instance introduce an equilateral triangulation $\Delta$ of $S$, each individual triangle having a same area $a$. The Einstein term of (\ref{gr_2d}) can then be replaced by $\chi(\Delta)$, and the area approximated by $a T(\Delta)$. This provides an action principle for a discrete theory, from which we can deduce a path-integral quantization. Remarkably, the resulting partition function is:
\beq
\cZ_{2d} = \sum_{\Delta} \exp\left( \frac{4 \pi}{G} \chi(\Delta) - \frac{\Lambda a}{G} T(\Delta)\right)\,,
\eeq 
and can therefore be matched to the free energy (\ref{partition_matrix}) of the matrix model (with the large $N$ rescaling). 
The sum over $\Delta$ can include arbitrary topologies, so that one really obtains a third quantization of $2d$ gravity. The correspondence imposes the following identifications:
\beq
\lambda \leftrightarrow \exp\left( - \frac{\Lambda a}{G} \right) \;, \qquad N \leftrightarrow \exp\left( \frac{4 \pi}{G} \right)\,.  
\eeq
Therefore the large $N$ limit of the matrix model corresponds to the weak coupling regime of $2d$ gravity. Thanks to the double scaling, one can moreover extend the sum over topologies to the continuum phase.

\

Let us summarize the situation. Matrix models are at first sight algebraic models, with no particular relation to gravity or geometry in general. However the relation to discrete geometry is readily seen at the perturbative level, since Feynman amplitudes are labeled by discretized surfaces. Still at the discrete level, the amplitudes can be related to discrete gravity path integrals, the two coupling constants of gravity being encoded in the coupling constant of the matrix model and the size parameter $N$. Finally, the matrix model formulation allows to reach a continuum phase when sending $N$ to infinity, dominated by spherical topologies, and which can also be extended to all types of manifolds. In a sense, what matrix models achieve is a definition of a measure on continuous geometries thanks to discrete methods, which reminds the relationship between Riemann sums and the Riemann integral. 

		
		
		
		
				
		
		\subsection{Higher dimensional generalizations}
		
The success of matrix models motivated extensions to higher dimensions already in the nineties, called \textit{tensor models} \cite{Ambjorn_tensors, gross, Sasakura:1990fs}. Viewing a matrix as a $2$-tensor, it is natural to introduce $d$-tensors in $d$ dimensions. The action $S(T)$ of a tensor is inspired by its matrix counter-part: the kinetic term convolutes indices of two tensors pairwise, and is interpreted as the identification of two $(d-1)$-simplices; and the interaction is built in such a way as to represent elementary $d$-cells. It is also natural to work with a single interaction, dual to a $d$-simplex, as it is the interaction with the smallest number of fields which can be given a $d$-dimensional interpretation, and also any discretization of a manifold can be subdivided in such a way that all cells are simplices. For example, in three dimensions, one can define:
\beq\label{tensor_3d}
S(T) = \frac{1}{2} \sum_{i_1 , i_2 , i_3} T_{i_1 i_2 i_3} T_{i_3 i_2 i_1} - \lambda \sum_{i_1 , \ldots , i_6} T_{i_1 i_2 i_3} T_{i_3 i_5 i_4} T_{i_5 i_2 i_6} T_{i_4 i_6 i_1}\,,
\eeq
where all the indices run from $1$ to $N$. As we will see later on, the type of tensors one is working with has important implications. This part of the generalization from matrix models is also not obvious, and the strategy which was adopted in the early nineties was to symmetrize the tensor indices. Since one does not want to give them any physical meaning, it seems at first sight a reasonable assumption, but it is not the only way to fulfill the requirement, and is responsible for key difficulties of this approach. One major challenge is to control the overwhelmingly complicated sum over triangulations generated in perturbative expansion. Contrary to matrix models, not all of them are discretization of topological manifolds: quite differently, mild singular contributions such as pseudo-manifolds are included, but also highly degenerate triangulations with extended singularities (see for instance \cite{Gurau:2010nd}, in the context of GFTs). This comes from the fact that in higher than two dimensions, prescribing simple local gluing rules for $d$-simplices along their $(d-1)$-subsimplices is not restrictive enough to eliminate these pathological structures. In particular, and as we will explore in more details in the next section, the data encoded in these simple combinatorial models is not rich enough to capture the structure of the simplices of dimension strictly less than $(d-2)$. For instance, in the model defined by the action (\ref{tensor_3d}), no data is associated to the vertices of the triangulation, therefore the topological structure around the vertices of the triangulations is essentially arbitrary \cite{Gurau:2010nd}. 

With hindsight, this lack of combinatorial structure is what prevented all the achievements of matrix models from being reproduced in higher dimensions. Without the necessary analytical control over the perturbative series, no $1/N$ expansion could be formulated for these early versions of tensor models, and therefore none of the other appeals of matrix models could be investigated either. The only part of the story which remained true was the interpretation of the amplitudes as discrete gravity path integrals, at least for triangulations of manifolds. Interpreting the elementary $d$-simplices as equilateral, a discrete metric can be assigned to each configuration, in the same way as in two dimensions. The amplitudes themselves can then be matched to the exponential of a discrete version of the Einstein-Hilbert action. Volume terms coming from the cosmological constant are again related to the coupling constant $\lambda$, and the curvature is captured by deficit angles around $(d-2)$-simplices. Such considerations gave birth to the Dynamical Triangulations program, and later on their Causal versions, where ensembles of discrete space-times are generated and summed over numerically rather than by analytical means. We refer to the lecture notes \cite{Ambjorn:2012jv} for details and references.   
		
		This situation changed dramatically thanks to the pioneering work of Gurau, which upon slightly restricting the combinatorial structures of tensor models, could define a $1/N$ expansion. A wealth of results could be gathered after this breakthrough, giving new support in favor of analytical studies of random triangulations.  
		
		
		
		
		\subsection{Bringing discrete geometry in}
		
		We can now explain how GFT comes about in this discrete approach to quantum gravity. Tensor models and dynamical triangulations are the minimalistic backbone of discrete gravity path integrals, in the sense that metric degrees of freedom are encoded in a purely combinatorial way. While this was fine enough in the case of two dimensional gravity, where only topology matters, the strategy can be questioned when it comes to three or four dimensions. In general relativity enter several important background structures. A topological manifold of the appropriate dimension is one such structure, and is arguably exhaustive in two dimensions. However, this topological manifold has to be endowed with a differential structure, which is not uniquely specified by the topology in four dimensions. Most importantly, local Lorentz invariance can also be argued to be a primary ingredient of any quantum theory of gravity. It is experimentally tested with an overwhelming precision, and seems to be rather hard to make emergent from a more fundamental theory. In this respect, getting the continuum of GR out of a fundamentally discrete model seems ambitious enough, so that it seems reasonable to save ourselves the burden of explaining Lorentz invariance as well. 
		
		In GFTs, a notion of local Lorentz invariance (or Euclidean invariance in Euclidean models) is assumed from scratch. To this effect, the purely combinatorial indices of tensor models become instead elements of some subgroup of the Lorentz or rotation group. For example in a 3d Euclidean context, one goes from the interaction part of (\ref{tensor_3d}) to (\ref{Sint_3d1}) by turning the indices into $\SU(2)$ group elements. It is in this combinatorial sense that a GFT field can be considered a tensor. However, that is not all: after this new type of data has been introduce, one needs to provide them with a discrete geometric meaning. We simply follow the matrix/tensor models reasoning: if the GFT field $\vphi$ is assumed to represent an elementary building block of geometry, then the geometric data should refer to this building block. Let us again use the Boulatov model as an example, in which case $\vphi( g_1 , g_2 , g_3 )$ is to be interpreted as a flat triangle, and the variables $g_i$ label its edges. It is the role of the constraint (\ref{gauge1}) to introduce an $\SU(2)$ flat discrete connection at the level of the amplitudes, encoded in the elementary line holonomies $h_\ell$. The natural interpretation of the variable $g_i$ is as the holonomy from a reference point inside the triangle, to the center of the edge $i$. Thanks to the flatness assumption, this holonomy is independent of the path one chooses to compute $g_i$. The meaning of the constraint (\ref{gauge1}) is also clear: it simply encodes the freedom in the choice of reference point. From the discrete geometric perspective, the Boulatov model can therefore naturally be called a \textit{second quantization of a flat triangle}: the GFT field $\vphi$ is the wave-function of a quantized flat triangle, and the path integral provides an interacting theory for such quantum geometric degrees of freedom. This point of view was already present in the early stages of SFMs \cite{Barbieri1997,BaezBarrett_tetra}, and has guided the development of this research field ever since. It was more recently advocated in \cite{Baratin:2010wi , daniele_rev2011}, providing a new look at the construction of four-dimensional models \cite{bo_bc , bo_holst}.
		
		\
		
		The natural question which can arise at this point is: which type of geometric data one should introduce? Or more specifically, why should we work with holonomy variables rather than simply edge vectors, as is done for example in Regge calculus \cite{lr_loll}? As it turns out, there is a general correspondence between these two alternatives, which gives us the opportunity to introduce the Lie algebra formalism for GFTs, initially introduced in \cite{Baratin:2010wi}. 
		To avoid unnecessary complications, we illustrate this dual representation on the Boulatov model restricted to the $\SO(3)$ group. The technical tool allowing to use (non-commuting) Lie algebra variables $x_i \in \su(2) \sim \mathbb{R}^3$, is the group Fourier transform \cite{PR3,majidfreidel,karim}. In our case, it maps $L^{2} (\SO(3)^3)$ to a space $L^{2}_{\star} (\so(3)^3)$, endowed with a non-commutative $\star$-product. The Fourier transform of $\vphi \in L^{2} (\SO(3)^3)$ is defined as: 
\beq \label{fourier}
\vphihat (x_1, x_2, x_3) :=\int  [\extd g_i]^3\, \vphi (g_1, g_2, g_3)  \,\e_{g_1}(x_1) \e_{g_2}(x_2) \e_{g_3}(x_3) ,
\eeq
where $\e_g \maps \su(2) \!\sim\! \mathbb{R}^3 \to \U(1)$ are non-commutative plane-waves, and functions on $\SO(3)$ are now identified with functions on $\SU(2)$ invariant under $g \to - g$. The definition of the plane-waves involves a choice of coordinates on the group. Following \cite{Baratin:2010wi}, we adopt: 
\beq
\forall g \in \SU(2) \,, \qquad \e_g \maps \, x \mapsto e^{\rm{i} \Tr(x |g|)}
\eeq 
where for $g \in \SU(2)$ we denote $|g| \equiv \rm{sign}(\Tr \,g) g$, and $\Tr$ is the trace in the fundamental representation of $\SU(2)$. Note that other choices are possible, and some may be more convenient than others \cite{cdr_duflo}. The Lie algebra variables can be given a simple metric interpretation, as vectors associated to the edges of the triangles \cite{Baratin:2010wi}. One can therefore start from the Lie algebra representation, and provide an independent construction of the theory, following a similar procedure as in group space. The same combinatorial structure of the action can be assumed, entailing the same simplicial interpretation, except that the pointwise product for functions on $\SU(2)$ is replaced by the non-commutative and non-local $\star$-product. The latter is induced by the group structure of $\SU(2)$, as dual to the convolution product for functions on the group. Defined first on plane-waves: 
\beq  \label{star}
(\e_{g} \star \e_{g'})(x) \!:=\! \e_{gg'}(x)\,,
\eeq
it is then extended to the image of the non-commutative Fourier transform, i.e. $L^{2}_{\star} (\so(3)^3)$, by linearity. This formalism becomes particularly interesting when it comes to the geometric constraints. Indeed, if $x_1, x_2, x_3$ have to be interpreted as the edge vectors of a flat triangle, they should close, as for i
\beq
x_1 + x_2 + x_3 = 0 \,.
\eeq
This condition needs to be imposed at the operator level on the GFT field $\vphihat$. A possible version of this projector is constructed out of a non-commutative notion of $\delta$-function. Defining
\beq
\delta_{x}(y) \equiv \int \extd h \, e_{g^{-1}}(x) e_{g}(y) \,,
\eeq
it is easy to verify that
\beq
\int \extd y \, (\delta_{x} \star f)(y) = \int \extd y \, (f \star \delta_{x})(y) = f(x)\,,
\eeq 
and therefore $\delta_{x}$ plays the role of Dirac distribution at point $x$ in $L^2_\star (\so(3))$. We can therefore impose the following closure constraint on the GFT field:
\beq\label{gauge_metric}
\vphihat = \widehat{C} \star \vphihat \, ,
\eeq
with:
\beq
\widehat{C}(x_1, x_2, x_3) \equiv \delta_{0}(x_1 + x_2 + x_3)\,.
\eeq
It is then easy to check that transforming back (\ref{gauge_metric}) to $L^{2} (\SO(3)^3)$ gives back the gauge invariance condition (\ref{gauge1}). This confirms that the Boulatov model can be understood as a second quantization of a flat triangle, and we refer to \cite{Baratin:2010wi} for more details.

In four dimensions, the same correspondence between group and Lie algebra representation has been put to profit \cite{bo_bc, bo_holst}. There, the GFT field represents a quantum tetrahedron, and Lie algebra elements correspond to bivectors associated to its boundary triangles. In addition to the closure constraint (again equivalent to the Gauss constraint in group space), additional geometricity conditions have to be imposed to guarantee that the bivectors are built from edge vectors of a geometric tetrahedron. These additional constraints are nothing by the simplicity constraints, and non-commutative $\delta$-functions can again be used to implement them.

		
		
		
		
	
	\section{A research direction}
	
	Let us summarize the two possible takes on GFTs we have been presenting so far. We first focused on quantization procedures, either in canonical or covariant form, applied to Einstein's field equations. The key steps entering this line of thoughts were: a) adopt connection and flux rather than metric variables at the classical level; b) thanks to canonical LQG techniques, show that spin network functionals are good kinematical states for quantum GR; c) resort to a semi-heuristic covariant formulation, spin foams, to identify the dynamics of these states; d) introduce GFTs as generating functionals for spin foam amplitudes. Already from this point of view, we could see non-trivial combinatorial assumptions entering GFT models, as well as unorthodox quantization rules motivating spin foam models for four-dimensional gravity. Actually, discrete geometric considerations rather than strict quantization procedures are arguably at the root of most of the modern spin foam models. As a result, the combinatorics should play an important role, but seems essentially unconstrained by these procedures: one usually works with simplicial complexes because it seems like a natural starting point. However, there is no strong case for it, even less from the GFT point of view: as in any quantum field theory, it would be preferable to have a large set of interactions at our disposal, determined by a symmetry principle. The second set of works we focused on, matrix and tensor models, put on the contrary most of their emphasis on the combinatorics. GFTs in this perspective are simply enriched tensor models, where tensor indices carry geometrical information. The conceptual gap between 2d and 4d gravity supports the idea that such additional data might be needed in 4d, contrary to 2d where purely combinatorial models are able to capture the very limited metric aspects of gravity. 
	
\

As for example put forward in \cite{daniele_rev2011}, we would like to use these two sets of incomplete motivations for GFTs as	a way to find new research directions. If we allow ourself to make a simple synthesis between the two, we could say that GFTs are quantum field theories of discrete geometries, in which boundary states are LQG like, namely a subset of spin network functionals (for the appropriate group). The key questions to ask can then be split into two classes. The first concerns model building, in the sense of finding the appropriate notion of quantum discrete geometry, both in the boundary and in the bulk. The second concerns more generic features of the formalism, such as the combinatorics, symmetry principles, regularization and renormalization. It should be noted that these two sets of questions are rather independent and should therefore better be explored in parallel. In this thesis, we are only concerned with the generic aspects of GFTs, and therefore we will mostly work in situations where discrete geometric aspects are either irrelevant, or well-controlled and unambiguous. More precisely, we would first like to understand to which extent GFTs can be defined rigorously as perturbative quantum field theories. And second, one would like to develop tools to explore the large triangulations/foams regime of GFTs. Whatever the position one adopts as regards discreteness in GFTs and SFMs, such regimes exist and therefore deserve to be studied. This is especially true in the current situation, in which the connection between the most refined models and continuum GR remain rather elusive. Whether this connection is to show up at the perturbative level already seems rather unlikely in the discrete gravity interpretation of these models, and if any such connection exists one would rather expect an emergence scenario \cite{daniele_hydro}. In such a mindset, the large triangulations regime is the physically relevant one. These two open problems will be addressed by two complementary means: large $N$ methods similar to the $1/N$ expansion of matrix models; and renormalizability studies of enriched GFT models, called Tensorial Group Field Theories (TGFTs). The latter will be established as the first well-behaved perturbative quantum field theories related to spin foam models, which is in our opinion the strongest result of this thesis. Key to these two series of results are new combinatorial tools which recently revived tensor models. In the next section, we therefore give some motivations for introducing them, and a detailed account of the combinatorial backbone provided by these new improved tensor models.  

%
%
%
%
		
%
%

\chapter{Colors and tensor invariance}\label{color_tensor}

\hfill\begin{minipage}{10cm}
{\footnotesize {\it
In any case, one does not have the right today to maintain that the foundation must consist in a \emph{field theory} in the sense of Maxwell. The other possibility, however, leads in my opinion to a 
renunciation of the time-space continuum and to a purely algebraic physics.
}}
\vspace{0.2cm}
\newline {\footnotesize {\bf Albert Einstein}, letter to Paul Langevin, October 1935.\footnote{translated and cited
by John Stachel in \cite{stachel1986einstein}.}}
\end{minipage}

\vspace{0.4cm}


Colors were first introduced by Gurau as a way to restrict the class of complexes generated by GFT models \cite{Gurau:2009tw}. In addition to giving a handle on their topological properties 
\cite{Gurau:2010nd , jimmy}, it was soon shown that colored models support a $1/N$ expansion \cite{RazvanN, RazvanVincentN, Gurau:2011xq}. Aside from these purely combinatorial and topological 
considerations, a careful study of the shift symmetry of BF theory and its imprint on the Boulatov model \cite{diffeos} gave more support to the physical relevance of the colored models. We will 
briefly recall these two sets of initial motivations in the first section of this chapter.

The discovery of the $1/N$ expansion then triggered a lot of activity, mainly focused on colored tensor models \cite{Gurau:2011xp, critical, dr_dually, v_revisiting,v_new, jr_branched, wjd_double}. 
These can be seen as a revival of the purely combinatorial approach to quantum gravity, but also as a necessary step towards the understanding of more involved GFTs. Interestingly, just like matrix 
models, colored tensor models have already found applications outside quantum gravity \cite{vrv_ising, v_dimers, vh_dimers, pspin, vf_packed}. This is because they provide universal probabilistic 
tools \cite{universality}, and as such can be expected to be relevant to a wide range of statistical physics problems. We introduce the main properties of colored tensor models in the second section 
of this chapter.

Finally, a key evolution we would like to comment on in the third section is the introduction of tensor invariance \cite{r_vir, universality, uncoloring}. It was indeed soon realized that the color 
restriction can be understood as a $\U(N)^{\otimes d}$ symmetry, where $d$ is the rank of the tensors. This opened a whole new set of applications, and in particular allowed to launch a 
renormalization program for Tensorial Group Field Theories (TGFTs), that is for GFTs based on tensorial interactions \cite{tensor_4d , vincent_tt2}.  

	\section{Colored Group Field Theories}
		
		\subsection{Combinatorial and topological motivations}
		
		
In order to understand the difficulties with GFTs as introduced so far, let us consider again a generic 3d model, defined by the partition function (\ref{gft_structure}) and the interaction 
(\ref{Sint_3d1}). We do not need to make any assumption about the covariance $C$ at this stage, since we are mainly concerned with the labels of the Feynman amplitudes, not their structure. These 
labels are $3$-stranded graphs, built out from the elementary vertex shown in Figure \ref{tetra_dual1}. A (closed) $3$-stranded graph can be equivalently thought of as a $2$-complex, with faces 
identified by the closed circuits of strands. 

The key objection put forward in \cite{Gurau:2010nd} is that, while the lines and faces of the $2$-complex unambiguously identify the triangles and edges of the dual triangulation, there is no 
combinatorial data associated to its vertices. As a result, the tetrahedra can be glued in arbitrary wild fashions along there common vertices, which leads to the presence of extended topological 
singularities. That is to say that the topological space naturally associated to such a pathological triangulation
contains points with neighborhoods not homeomorphic to a $3$-ball, and that these singular points can form subspaces of dimension higher than $1$. Hence, these triangulations cannot be interpreted as 
topological manifolds. Even worst, they are not homeomorphic to \textit{pseudomanifolds}, which are milder versions of singular topological spaces (see again \cite{Gurau:2010nd}). Examples of such 
configurations can easily be constructed by gluing vertices in such a way that a given face runs several times through the same line. It is therefore natural to look for combinatorial prescriptions which 
prevent the generation of such faces, called \textit{tadfaces}.

These unwelcomed properties of ordinary GFTs were cured by the introduction of new data in the action: the colors. These are purely combinatorial labels which identify the $(d-1)$-simplices constituting 
a given $d$-simplex. One therefore needs $(d+1)$-colors in a rank-$d$ GFT, labeling $(d+1)$ independent fields. It is moreover convenient to work with complex fields, which will ensure nice orientability 
properties. 
For instance, the colored rank-$3$ GFT is defined in terms of fields $\{ \vphi_\ell, \ell =1 ,
\ldots , 4 \}$, 
each distributed according to a (complex) Gaussian measure $\extd \mu_C$, and the interaction convolutes the four fields together:
\bes
\cZ_{col} &=& \int [\prod_{\ell =1}^{4} \extd \mu_C (\vphi_\ell , \vphib_\ell )]\, \exp( - S_{col} (\vphi_\ell))\,, \\
S_{col} (\vphi_\ell) &=& \lambda \int  \vphi_1 \vphi_2 \vphi_3 \vphi_4 + \overline{\lambda} \int \vphib_1 \vphib_2 \vphib_3 \vphib_4 \,, 
\ees
where in the second line we kept the convolution pattern implicit\footnote{It is made explicit in Figure \ref{int_edge_labels}. It is the same as in equation (\ref{Sint_3d1}), up to a reordering of the variables.}. The only 
effect of the coloring is to reduce the combinatorial complexity of the Feynman graphs generated by the GFT: they must be \textit{bipartite} and \textit{colorable}. The first condition comes from the complex
nature of the field: there are two types of vertices, associated to the interaction of $\vphi$ or $\vphib$ fields respectively; and propagator lines can only connect two vertices of different types. The second
condition comes from the colors: each line is now equipped with a color label, and only four lines with distinct colors can interact. But, and this is essential, 
the amplitude of a colorable and bipartite stranded graph in the colored GFT is identical to its amplitude in the non-colored theory. 

The main advantage of the colored model is that it generates pseudomanifolds only \cite{Gurau:2010nd}. In the 3d context we are focusing on, this means that at most pointlike topological singularities can occur, 
located at the vertices of the triangulations. This topological information is fully encoded in the colored structure of the stranded graphs, which can equivalently be represented by \textit{colored graphs}.
\begin{definition}
Let $n \geq 3$. A $n$-colored graph is a bipartite regular graph of valency $n$, whose edges are colored by labels $\ell \in \{1 , \ldots, n\}$, 
and such that at each vertex meet $n$ edges with distinct colors. 
\end{definition}
These graphs have been extensively used in the mathematical literature, as an efficient combinatorial way of representing topological manifolds \cite{FerriGagliardi, Vince_gene}. This goes under the name of \textit{crystallization} theory.
These techniques have been partly rediscovered and partly directly imported into the GFT context, though with a different nomenclature, and one notable difference: crystallizations are a particular subclass of $n$-colored graphs
which represent manifolds, while in GFT we cannot avoid dealing with the singular pseudomanifolds as well. The general formalism is discussed at length in \cite{Gurau:2011xp}, and we should content ourselves 
with the main ideas here, which we again introduce in three dimensions. To begin with, we will represent the vertex interaction associated to $\vphi$ (resp. $\vphib$) fields by a white (resp. black) dot, and a propagator between
fields of color $\ell$ with a line of color $\ell$. The first interesting point to notice is that the stranded substructure of a Feynman graph can be fully inferred from its colored representation. Indeed
a strand can be canonically labeled by a couple $(\ell \ell')$ of colors, where $\ell$ and $\ell'$ are the colors of the two types of lines in which the strand runs. See Figure \ref{int_edge_labels}. Therefore the faces, which 
are connected circuits of strands, are equivalently labeled by connected cycles of edges of color $\ell$ and $\ell'$. And due to the color prescription, tadfaces cannot occur in colored GFTs. This solves
a first pathology of ordinary GFTs. 
\begin{figure}[h]
\begin{center}
\includegraphics[scale=0.6]{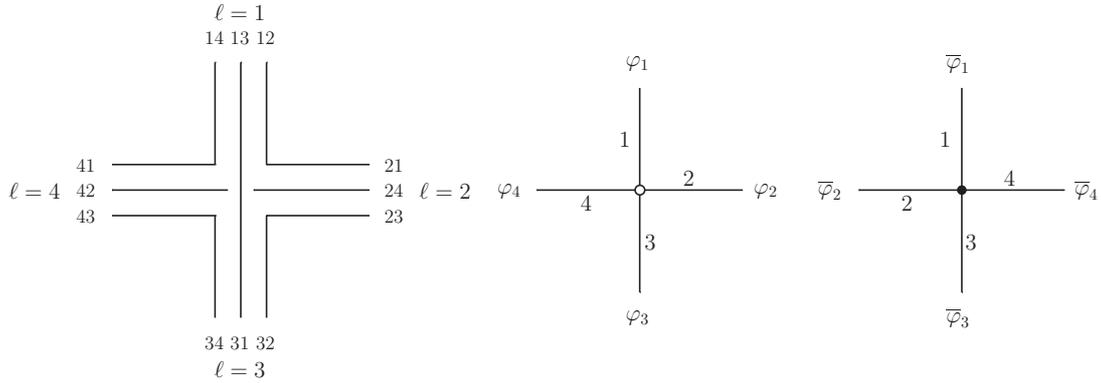}
\caption{Interaction(s) in a stranded representation with labels (left), and as a node of a colored graph (right).}
\label{int_edge_labels}
\end{center}
\end{figure}
The second key aspect is that the colors give direct access to the vertices of the triangulation. First, they can be attributed a color label in the following way: in a tetrahedron, the vertex of color $\ell$ is the one
opposite to the triangle of color $\ell$. Consider a $4$-colored graph $\cG$. We define the $3$-\textit{bubbles} of color $\ell$ as the connected components of edges with colors in $\{1, \ldots , 4\} \setminus \{ \ell \}$. See Figure \ref{cgraph1} for an example.
\begin{figure}[h]
\begin{center}
\includegraphics[scale=0.6]{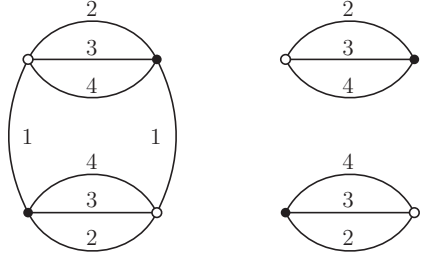}
\caption{A $4$-colored graphs of with $4$ nodes (left), and its two $3$-bubbles of color $1$.}
\label{cgraph1}
\end{center}
\end{figure}
In the same way as the $4$-colored graph $\cG$ represents a simplicial complex of dimension $3$, a $3$-bubble $b$ is a $3$-colored graph and represents a triangulated two dimensional surface. This is nothing
but the the boundary of a 3d cell dual to a particular vertex $v$. And hence the topology of the neighborhood of $v$ can be directly read out from that of the surface encoded by $b$: if it is a sphere (i.e. of genus $0$),
$v$ is regular, otherwise (i.e. of genus $\geq 1$) it is topologically singular. We will use this idea at length in Chapter \ref{largeN}.

The notion of bubble is very general. In arbitrary dimension $d \geq 2$, a $(d+1)$-colored graph $\cG$ comes with nested sets of $k$-bubbles, for $2 \leq k \leq d$, simply defined as the maximally connected
subgraphs made of edges with $k$ fixed colors. They define particular embedded surfaces of dimension $k-1$, dual to the $(d - k)$-simplices of the complex encoded in $\cG$, which in turn provide a
full homology for the topological spaces generated by colored GFTs \cite{Gurau:2009tw , Gurau:2011xp}.

		\subsection{Motivation from discrete diffeomorphisms}
		
		
		

A second argument in favor of colors one can find in the literature relies on the symmetries of discrete gravity models. However, it mainly applies to 3d for now, and therefore is not
as generic as the combinatorial and topological motivations. 

It goes as follows. Because of the strict duality between spin foam models and discrete gravity path-integrals, as for example uncovered in \cite{Baratin:2010wi}, it is tempting to implement discrete gravity symmetries
directly at the level of the GFT. Of particular interest is the diffeomorphism group: if one wants to reach a continuous limit in which the full symmetry of GR is restored, identifying
a residual action of the diffeomorphism group at a discrete level, and subsequently turning it into a symmetry (or at least an approximate one) is a natural strategy. This point of view has long been advocated by Bianca Dittrich and 
collaborators \cite{bianca_review}. Thanks to the metric variables introduced in \cite{Baratin:2010wi}, this could be partially explored at the GFT level by the authors of \cite{diffeos}, who focused on the Boulatov and Ooguri models.
In particular, they identified a translation symmetry in the Boulatov model, a quantum-deformed version of vertex translation invariance interpreted as discrete diffeomorphism invariance. It is then natural to try to
use this symmetry as a defining feature of the Boulatov model, acting directly on the fields and leaving the action invariant. But this could only be achieved in the colored version of the model,
the reason being essentially that in order to translate a given vertex in a tetrahedron, one needs to act on specific edge variables of three of its triangles. Therefore the triangles need to be identified by additional 
labels in the action: the colors. Moreover, they argued that this symmetry would systematically be broken at the quantum level in a non-colored model.

In order to make this argument stronger, one would need to investigate it further in a four-dimensional context. In the Ooguri colored model, a similar translation symmetry could be uncovered \cite{diffeos}, this time associated
to the edges of the 4d complexes. It however has no geometric interpretation, and is broken by the imposition of simplicity constraints. Hence it cannot be used in favor of colored models for four-dimensional quantum gravity.       

We will come back to these symmetries in Chapter \ref{largeN}, to motivate new formulations of the Boulatov and Ooguri models which are well adapted to the computation of scaling bounds.  
		
	\section{Colored tensor models}
	
	Let us now introduce tensor models in some details. They rely on the same structures as colored GFTs, but exact calculations are made a lot easier. Hence, this provides a simplified context in which to introduce 
the main tools we will use again in the next sections.
		
		\subsection{Models and amplitudes}






Colored tensor models turn out to have very specific properties from dimension three onwards, we will therefore exclude the matrix case and assume $d \geq 3$ from now on. The degrees of freedom are $d +1$ colored rank-$d$
tensors $\vphi_{i_1, \ldots , i_d}^{\ell}$, and their conjugates $\vphib_{i_1, \ldots , i_d}^{\ell}$. $\ell \in \{1, \ldots , d+1\}$ is again the color index, and the indices run from $1$ to $N$. From the GFT
perspective, these tensors can equivalently be considered as truncated version of the Fourier transforms of fields defined on $d$ copies of a compact Lie group. For instance, the Fourier dual of $\U(1)$ is
$\mathbb{Z}$, and can be truncated to a set of cardinal $N$. But because of the choice of covariances made in tensor models, they are essentially insensible to the nature of the group. The kernels of the propagators are 
given by simple Kronecker $\delta$'s with respect to the indices and the colors. For any $\ell$, we define the Gaussian measure $\extd \mu_{C^{\ell}}$ by
\beq
\int \extd \mu_{C^\ell} ( \vphi^{\ell} , \vphib^{\ell} ) \vphi^{\ell}_{i_1 , \ldots , i_d }  \vphib^{\ell}_{j_1 , \ldots , j_d} = \prod_{k = 1}^{d} \delta_{i_k , j_k} \,.
\eeq
The interaction is given by a simplicial term and its conjugate. We introduce a notation with two indices\footnote{Additions are understood modulo $d+1$ in the color set, and $i_{\ell \ell'}$ is identified with $i_{\ell' \ell}$.}
\beq
\forall \ell \in \{1 , \ldots, d + 1 \} \,, \qquad {\boldsymbol{i_\ell}} \equiv (i_{\ell \ell-1} , \ldots , i_{\ell 1} ,  i_{\ell d +1 } , \ldots , i_{\ell \ell + 1})\,,
\eeq
which makes its definition compact:
\beq
S( \vphi^\ell , \vphib^\ell ) = \frac{\lambda}{N^{d(d-1)/4}} \sum_{\{ i_{\ell \ell'} , \ell < \ell' \} } \, \prod_{\ell = 1}^{d+1} \vphi^{\ell}_{{\boldsymbol{i_\ell}}} \qquad + \qquad \mathrm{c.c}\,.
\eeq
For simplicity of the discussion, we already introduced the unique rescaling of the coupling constant $\lambda$ which makes the large $N$ expansion possible. The partition function of this tensor model expands as a power
series in $(\lambda \overline{\lambda})$:
\beq
\cZ = \int [\extd \mu_{C^\ell} ( \vphi^{\ell}, \vphib^\ell )] \, \exp\left( S( \vphi^\ell , \vphib^\ell ) \right) = \sum_\cG \frac{(\lambda \overline{\lambda})^{\cN(\cG)/2}}{s(\cG)} \cA_\cG\,,
\eeq
where the sum runs over $(d+1)$-colored graphs, $\cN (\cG)$ is the number of nodes of a graph $\cG$, and $\cA_\cG$ its amplitude. The latter is easily computed. Each node brings a $N^{-d(d-1)/4}$ factor, and identifies indices two by two, according
to the face structure of the graph. Following a given face across propagators and vertices, we are able to trivially sum all the Kronecker $\delta$'s, until we have just one index left in this face. This gives a free sum and hence a factor $N$
per face contributing to the amplitude. Therefore, we simply have:
\beq\label{ampl_tensors}
\cA_\cG = N^{\vert \cF(\cG) \vert - \cN(\cG) \frac{d(d-1)}{4}}\,,
\eeq 
where $\cF(\cG)$ is the set of faces of the graph $\cG$.
		
		\subsection{Degree and existence of the large N expansion}
		



The $1/N$ expansion of tensor models relies on a rewriting of the amplitudes (\ref{ampl_tensors}) in terms of appropriate combinatorial objects. Of particular relevance is the notion of \textit{jacket}, first introduced in 
\cite{lin}, and better understood in \cite{RazvanN, RazvanVincentN, jimmy}. 
\begin{definition}
Let $\cG$ be a $(d+1)$-colored graph. A \textit{jacket} $J$ of $\cG$ is a $2$-subcomplex of the ($2$-complex represented by) $\cG$, labeled by a cycle $\sigma = (\ell_1 \ldots \ell_{d+1})$ in the color set (or equivalently by $\sigma^{\inv}$).  
$J$ consists of the same nodes and edges as $\cG$, while its faces are the faces of colors $(\ell_1 \ell_{2}),\, (\ell_2 \ell_{3}), \ldots, \, (\ell_d \ell_{d+1}),\, (\ell_{d+1} \ell_{1})$. 
\end{definition}
The jackets label particular ribbon subdiagrams in the stranded diagrams encoded by the colored graphs, which are dual to $2$-dimensional discretized surfaces embedded in the dual simplicial complexes \cite{jimmy}. They are also orientable, thanks
to the complex structure of the tensor models, and therefore their topology is fully captured by a single positive integer: the genus. It is moreover possible to count the faces of a graph $\cG$ in terms of the combinatorial data associated
to all its jackets, and prove a simple relation with their genera (see \cite{Gurau:2011xp} and references therein):
\begin{proposition}
 The amplitude of a $(d+1)$-colored graph is equal to:
\beq
\cA_\cG = N^{d - \frac{2}{(d-1)!} \omega(\cG)}\,,
\eeq
where $\omega(\cG)$, called \textit{degree} of $\cG$, is the sum of the genera of its jackets:
\beq
\omega(\cG) \equiv \sum_{J} g_J\,. 
\eeq
\end{proposition}

This very important result shows the existence of the $1/N$ expansion: given its expression, $\omega$ is a positive integer and therefore the divergence degree at large $N$ is uniformly bounded by $d$. At each order of the $1/N$ expansion,
one needs to sum all the graphs with a given degree, in the same way as one needed to sum all triangulations with a given genus in matrix models. Notice however that, while $\omega(\cG)$ encodes some topological information about 
the graph $\cG$, it is not a topological invariant (see again \cite{Gurau:2011xp} for a detailed discussion of the topological properties of colored tensor models).
		
		\subsection{The world of melons}
		
		
		

After the discovery of the $1/N$ expansion, the research efforts mainly focused on the leading order sector. The only configurations which survive when $N$ is sent to infinity are the $\omega = 0$ graphs. They were given a simple recursive characterization 
in \cite{critical}, in terms of so-called \textit{melons}. An \textit{elementary melon} is a connected subgraph with two nodes, $d$ internal lines, and $2$ external legs, as represented on the left-side of Figure \ref{melo_contr1}. 
\begin{figure}[h]
\begin{center}
\includegraphics[scale=0.6]{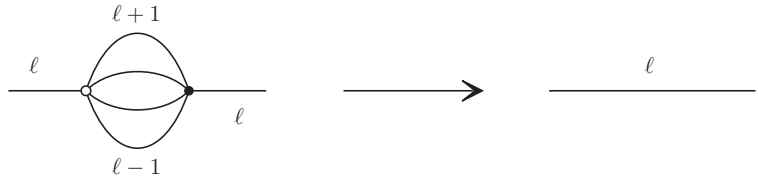}
\caption{An elementary melopole (left) and its contraction (right).}
\label{melo_contr1}
\end{center}
\end{figure}
In \cite{critical}, it
is first demonstrated that an $\omega =0$ graph necessarily contains an elementary melon; and in a second step, that contracting an elementary melon as pictured in Figure \ref{melo_contr1} does not affect the degree of the graph. The second property is
easy to understand from formula (\ref{ampl_tensors}), since $2$ nodes and $d(d-1)$ faces are suppressed in the contraction process. As for the first, it relies on specific combinatorial relations, which we do not want to elaborate on here.
Therefore, any connected leading-order graph can be reduced to the unique graph with two nodes, called the \textit{supermelon}, in a contraction process which involves only deletions of elementary melons. Reciprocally, any such graph can be obtained
from the supermelon and successive dressings of its lines with elementary melons. These graphs are called \textit{melonic}, an example of which is provided in Figure \ref{ex_melon1}.   
\begin{figure}[h]
\begin{center}
\includegraphics[scale=0.7]{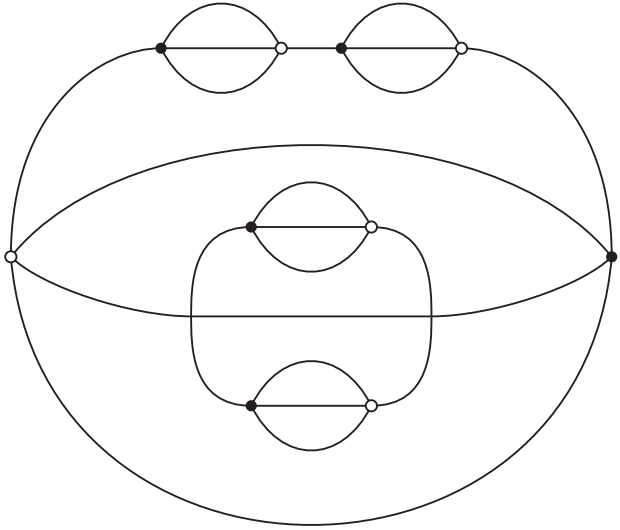}
\caption{A melonic graph in $d = 3$.}
\label{ex_melon1}
\end{center}
\end{figure}

\
The free-energy of the melonic sector was also studied in details in \cite{critical}, and was shown to have a finite radius of analyticity in $g \equiv \lambda \overline{\lambda}$. Close to the critical value $g_c$, graphs with very many melons dominate, which signals the transition to a continuum phase. The melonic free-energy $F$ behaves like 
\beq
F(g) \underset{g \to g_c }{\sim} K \left( \frac{g_c - g}{g_c} \right)^{2 - \gamma_{\rm{melons}}}
\eeq
in the vicinity of the critical point, with a critical exponent $\gamma_{\rm{melons}} = 1/2$ which turns out to be independent of the dimension $d$. The nested structure of melonic graphs suggests that this phase is a "branched polymer" one, describing a very crumpled metric. This last point has been rigorously confirmed in the recent work \cite{jr_branched}: the continuous emergent space at the critical point has Hausdorff dimension $2$ and spectral dimension $4/3$.

\ 
From the quantum gravity perspective, this result can be interpreted in two ways. Either a well-behaved emergent spherical space has to be looked for in a double scaling or a generalization thereof, in a similar fashion as higher genera are incorporated to matrix models; or it can alternatively be understood as one other sign suggesting that tensor models need to be supplemented with additional pre-geometric data. Ongoing efforts aim at exploring this alternative, with first encouraging results towards the realization of a double scaling in tensor models \cite{wjd_double}, and preliminary steps towards the analysis of a melonic phase transition in the Boulatov-Oguri GFTs \cite{boulatov_phase}.   
		
	\section{Tensor invariance}
				
		\subsection{From colored simplices to tensor invariant interactions}
		
A new set of developments of tensor models have been triggered by the introduction of tensor invariant interactions. It was first realized in \cite{r_vir} that, in a rank-$d$ colored tensor model, $d$ fields can be directly integrated, yielding an effective theory for the last field. The new effective interactions are labeled by $d$-bubbles (in the color of the last field), as easily understood from the point of view of the initial perturbative expansion: instead of summing over the initial interactions and propagators, one first integrates the propagators in three of the colors, yielding new effective bubble vertices to which the propagators in the last color can be hooked. These interactions are infinitely many, and the coupling constants in front are functions of the initial $\lambda \overline{\lambda}$. What is remarkable with these simple tensor models is that their kernels are also products of Kronecker $\delta$'s, which identify indices according to which face they belong. If one furthermore assumes the new coupling constants to be independent, one obtains a whole class of tensor models with a single field, but an infinite set of freely adjustable interactions.

\
Let us assume the single effective field to have color $0$, and the bubble vertices to be constructed out of color $1$ to $d$. The only faces which matters are the color-$(0\ell)$ ones, and are identified thanks to the colored structure of the bubbles. With the notations introduced before, in a given bubble interaction $b$, the $(0\ell)$ index of a $\vphi^0$ field is identified to the $(0 \ell)$ index of a $\vphib^0$, and both are in the same position in these tensors. We can therefore interpret the color $\ell$ as labeling the position of an index in the tensors $\vphi^0$ and $\vphib^0$. And the color conservation of the faces in the initial colored models translates into the requirement that: in the bubble $b$, an index in the $\ell^{\rm{th}}$ position of a tensor $\vphi^0$ must be contracted with an index of a $\vphib^0$ in the same position. 

\

These considerations allowed the authors of \cite{uncoloring} to develop an axiomatic formulation of random tensor models, based on this single-field framework. Consider a complex rank-$d$ tensor $T_{i_1 \ldots i_d}$, and assumes that it transforms as a tensor products of fundamental representations of $\U(N)$. That is to say that matrices $U^{(1)}\, \, \ldots , \, U^{(d)} \in \U(N)$ act on $T$ and its conjugate as:
\bes
T_{i_1 \ldots i_d} &\rightarrow& \sum_{j_1 ,\, \ldots ,\, j_d} U^{(1)}_{i_1 j_1} \, \ldots U^{(d)}_{i_d j_d} T_{j_1 \ldots j_d}\,,\\
\overline{T}_{i_1 \ldots i_d} &\rightarrow& \sum_{j_1 ,\, \ldots ,\, j_d} \overline{U}^{(1)}_{i_1 j_1} \, \ldots \overline{U}^{(d)}_{i_d j_d} T_{j_1 \ldots j_d}\,.
\ees
This large symmetry group, $\U(N)^{\otimes d}$, acts independently on each of the indices. The action can then be required to be \textit{tensor invariant}, i.e. invariant under $\U(N)^{\otimes d}$. It turns out \cite{universality} that tensor invariants are generated by monomials, which contract the indices of $p$ tensors $T$ and $p$ tensors $\overline{T}$, in such a way that the $\ell^{\rm{th}}$ index of a $T$ is always contracted with the $\ell^{\rm{th}}$ index of a $\overline{T}$. They are labeled by $d$-colored graphs, with white (resp. black) dots representing $T$ (resp. $\overline{T}$) tensors, and a line of color $\ell$ picturing the contraction of two indices in the $\ell^{\rm{th}}$ position. They are the analogues of the trace invariants of matrix models. If we moreover restrict to \textit{connected} such invariants, we find that the action is a sum of bubble interactions, i.e. a sum of monomials labeled by connected $d$-colored graphs.
A general probability theory for such independent and identically distributed (i.i.d.) tensors has been constructed at the perturbative level \cite{universality}. It relies on a generalization of the $1/N$ expansion for colored tensor models, which provides the tools to prove universality theorem akin to the central limit theorem \cite{universality}. Another avenue towards the same results is provided by Schwinger-Dyson equations \cite{v_revisiting}, which generalize the loop equations of matrix models \cite{r_vir, r_sdgene}. They even allowed to uncover the existence of slightly modified $1/N$ expansions, which retain more than melonic graphs at leading-order \cite{v_new}. Finally, recent non-perturbative studies \cite{r_constructive} of the same models opened a new era, in which constructive field theory methods \cite{rz_loop1, rz_loop2} are expected to deepen our understanding of the large $N$ limit of i.i.d. tensor models.
			
		\subsection{Generalization to GFTs}
			

The concept of tensor invariance unlocked a new understanding of tensor models, which allows to view the colors as a consequence of a symmetry principle, rather than labels of extra fields introduced by hand. With insight, the relative failure of analytical methods in the early incarnations of tensor models \cite{Ambjorn_tensors, gross, Sasakura:1990fs} is tied to an unfortunate symmetrization of the tensor indices. The colored substrate making the topology and combinatorics tractable is only recovered with unsymmetrized tensors.   
The wealth of results gathered in this new framework establish tensor models as a very active field of research, in rapid development. From the point of view of spin foam models and group field theories, it is at the same time a competitor and an important source of inspiration. The more we will learn about tensor models, the more we will be able to understand the quantum field theory properties of group field theories. 

\

In particular, the notion of tensor invariance seems to be of paramount importance, as will be illustrated in Chapters \ref{renormalization}, \ref{chap:u1} and \ref{chap:su2}. It has been first applied to fields defined on group manifolds in \cite{tensor_4d}, which considers a 4d GFT based on the $\U(1)$ group. The main innovation of this article is to replace the sharp cut-off in Fourier space, which would immediately yield a tensor model, with a non-trivial propagator based on the Laplace operator on $\U(1)$. This introduces a non-i.i.d $1/(\sum_{k = 1}^{4} i_k^2)$ weight into the propagator, hence defining a non-trivial field theory. In this context, tensor invariance provides a notion of locality, well-adapted to a renormalizability analysis. The field theory of \cite{tensor_4d} is the first renormalizable instance in a large class of GFTs based on bubble interactions, known as Tensorial Group Field Theories (TGFTs). They have been generalized to various groups and propagators \cite{josephsamary, joseph_etera, u1, fabien_dine, su2, joseph_d2}, slowly bridging the gap with more involved 4d quantum gravity models. It is in our view a very important line of research to pursue, to which a large part of this thesis is dedicated.

\chapter{Large N expansion in topological Group Field Theories}\label{largeN}

In this chapter, we present results about the large $N$ expansion of topological GFTs, initially obtained in \cite{vertex} and \cite{edge}. We first focus on the Boulatov model, and then generalize to its four-dimensional counter-part, the Ooguri model, which can be considered the backbone of any of the current quantum gravity models. Our main purpose will be to illustrate how the geometric data encoded in the group or algebra variables of the fields can be exploited to understand the scalings of the amplitudes with respect to the large $N$ parameter. We will in particular detail how both of these models can be reformulated in variables adapted to the symmetry uncovered in \cite{diffeos}. While we think these formulations can be useful in various situations, we will only present direct applications to the large $N$ expansion. The two main results in this respect will be: scaling bounds in terms of the bubbles of the Feynman diagrams, which prove that topological singularities do not contribute to the first few orders of the large $N$ expansion; scaling bounds in terms of the jackets which slightly generalize the jacket bounds of \cite{RazvanN, RazvanVincentN}, and thereby provide an independent way of proving the existence of the $1/N$ expansion of these models.  

	\section{Colored Boulatov model}
	
		\subsection{Vertex variables}

\subsubsection{Colored Boulatov model and shift symmetry}

We begin with a complete definition of the colored Boulatov model, and the construction of its vertex representation, as introduced in \cite{diffeos} and further developed in \cite{vertex}. 
For clarity of the presentation, we will restrict again to functions on $\SO(3) \sim \SU(2) / \mathbb{Z}_{2}$, identified as functions $f$ on $\SU(2)$ such that $f(g) = f(-g)$ for all $g \in
\SU(2)$. This will allow us to use the group Fourier transform of \cite{majidfreidel}, which is bijective in the case of $\SO(3)$, but not in the case of $\SU(2)$. Note however that we could work with the full $\SU(2)$ group, at the price of using the generalized Fourier transform introduced in \cite{karim}. As this would introduce heavy notations without changing any of the scaling bounds we can deduce from the vertex formulation, we refrain
from doing so and work within the simplified framework. Recall that the group Fourier transform maps functions $f \in L^{2}(\SO(3))$ to functions $\hat{f}$ defined on the Lie algebra $\so(3) \sim \mathbb{R}^{3}$. 
The $\star$-product endows this space of Lie algebra functions with a non-commutative structure, reflecting the group structure of $\SU(2)$. It is indeed the Fourier dual of the convolution of functions in $L^{2}(\SO(3))$, in the sense that:
\beq
\forall f_1\,, f_2 \in L^{2}(\SO(3))\,, \qquad  \left( \widehat{f}_{1} \star \widehat{f}_{2} \right)(x)\,=\, \int \extd g \left[ \int \extd h f_1(g h^{-1})\,f_2(h)\right] e_g(x)\,.
\eeq
This implies in particular that an integral of the point-wise product of two functions on the group manifold is equal to the integral over $\mathbb{R}^3$ of the $\star$-product of their Fourier transforms:
\beq
\int \extd g f_{1}(g^{-1}) f_{2}(g) = \int \extd x \left( \widehat{f}_{1} \star \widehat{f}_{2} \right)(x)\,.
\eeq
\  

It is convenient at this stage to work 'on-shell', i.e. in a space of fields verifying the Gauss constraint. We therefore introduce four complex fields $\vphi_\ell \in L^2 ( \SU(2)^3 )$, labeled by a color index $\ell \in \{ 1, \cdots , 4\}$, verifying the gauge
invariance condition:
\beq\label{gauge_v}
\forall h \in \SU(2),  \qquad \vphi_\ell (hg_1, hg_2, hg_3)  \, = \, \vphi_\ell (g_1, g_2, g_3),
\eeq 
as well as the $g_i \leftrightarrow g_i^\inv$ invariance. 
The colored Boulatov model \cite{cboulatov} can then be defined by the action:
\bes \label{action_boulatov}
S[\vphi , \vphib ]&=& S_{kin} [\vphi , \vphib] + S_{int}[\vphi , \vphib] ,\\
S_{kin} [\vphi, \vphib] &=& \frac{1}{2} \int [\extd g_i]^3 \sum_{\ell=1}^4   \, \vphi_\ell(g_1, g_2, g_3) \overline{\vphi_{\ell}}(g_1, g_2, g_3) ,\\
S_{int}[\vphi, \vphib] &=& \lambda \int [\extd g_{i} ]^6 \, \vphi_1(g_1, g_2, g_3) \vphi_2(g_3, g_5, g_4) \vphi_3(g_5, g_2, g_6) \vphi_4(g_4, g_6, g_1)
\nn \\
&+& \overline{\lambda} \int [\extd g_{i} ]^6 \, \overline{\vphi_1}(g_1, g_2, g_3) \overline{\vphi_2}(g_3, g_5, g_4) \overline{\vphi_3}(g_5, g_2, g_6) \overline{\vphi_4}(g_4, g_6, g_1)\,,
\ees
and the partition function
\beq
\cZ = \int \extd \mu_{inv} (\vphi_\ell , \vphib_\ell ) \exp\left( - S [\vphi, \vphib] \right) \,,
\eeq
where $\mu_{inv}$ is the (formal) Lebesgue measure on the space of gauge invariant fields. Alternatively, one can work with the Fourier transformed fields $\vphihat_\ell \in L^2_{\star} (\mathbb{R}^3)$, and a non-commutative path integral weighted by an action with the same structure as (\ref{action_boulatov}), except that the $\star$-product is used in place of the pointwise product. The geometrical interpretation of the fields and of the interactions is recalled in Figure \ref{edge_rep}. By convention, we will refer to the interaction between $\vphi_\ell$ fields (resp. $\vphib_\ell$ fields) as the clockwise (resp. anti-clockwise) interaction, this nomenclature referring to the graphical representation we adopt to distinguish these two interactions. 

\begin{figure}[h]
  \centering
  \subfloat[Field]{\label{triangle_edge}\includegraphics[scale=0.5]{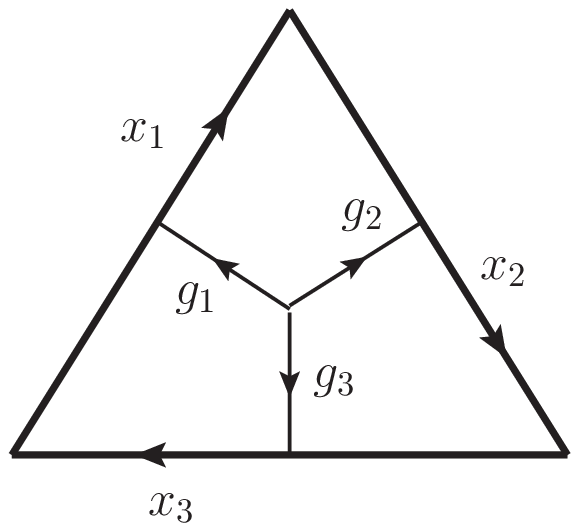}}            
  \subfloat[Interaction vertex (clockwise)]{\label{interaction_edge}\includegraphics[scale=0.5]{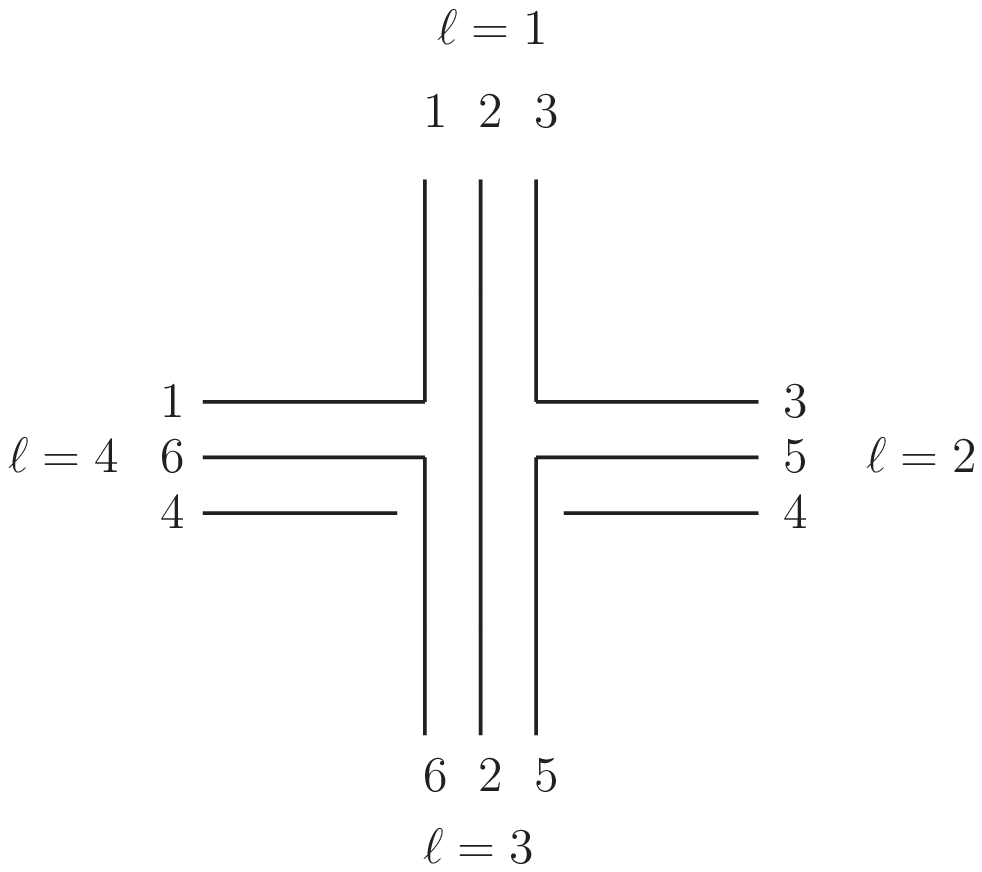}}
  \subfloat[Geometrical interpretation]{\label{tetrahedron_edge}\includegraphics[scale=0.4]{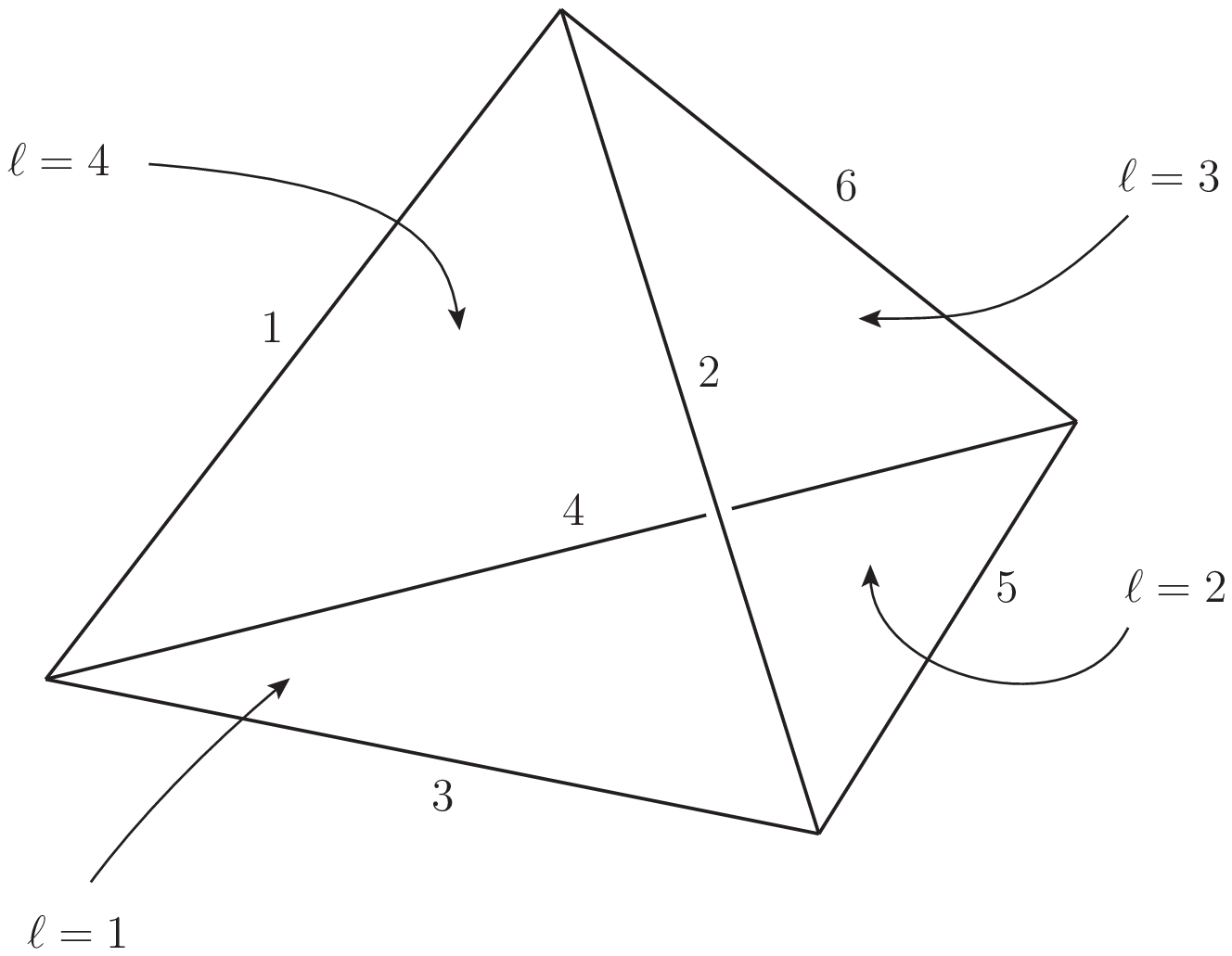}}
  \caption{Graphical representation of a field, and the interaction vertex in usual edge variables.}
  \label{edge_rep}
\end{figure}

\

In \cite{diffeos}, this model was shown to respect (quantum) symmetries, given by actions of the Drinfel'd double $\cD \SO(3)\!=\!\cc(\SO(3))\rtimes \mathbb{C}\SO(3))$ on the ('on-shell') fields. We will focus on the translational part of these actions, interpreted as (discrete) diffeomorphisms \cite{laurentdiffeo,biancadiffeos,biancabenny}.
They have four generators $\{\cT^{\ell'} , \ell'=1\cdots 4\}$, each $\cT^{\ell'}$ acting non-trivially on fields of color $\ell \neq \ell'$. For instance, $\cT^{3}$ acts on $\vphi_1$ as:
\beq
\cT^{3}_{\varepsilon} \act \vphi_1(g_1, g_2, g_3) \equiv (\e_{g_1^\inv} \star \e_{g_3})(\varepsilon) \, \vphi_1(g_1, g_2, g_3) \, = \e_{g_1^\inv g_3}(\varepsilon) \, \vphi_1(g_1, g_2, g_3).
\eeq

This can be interpreted as translations of the edges $1$ and $3$, respectively by $\varepsilon$ and $-\varepsilon$, with a deformation given by the $\star$-product. This is clearer in metric variables, where the previous equation can be (schematically) written as:
\beq
\cT^{3}_{\varepsilon} \act \vphihat_1(x_1, x_2, x_3) = \bigstar_{\varepsilon} \, \vphihat_1(x_1-\varepsilon, x_2, x_3+\varepsilon)\,. 
\eeq
As a result, the action of $\cT^{3}$ on the field of color $1$ can geometrically be interpreted as a deformed translation of one of its vertices, as represented in Figure \ref{trans_vertex}. Furthermore, we can assign colors to the vertices of the tetrahedron defining the interaction term, with the convention that $v_\ell$ should be the vertex opposed to the triangle of color $\ell$ (as introduced in the previous chapter). This induces a color label for the vertices of the different triangles. In this picture, the action of $\cT^{3}$ on $\vphihat_1$ corresponds to a translation of the vertex of color $3$ in the triangle of color $1$.

\begin{figure}[h]
\begin{center}
\includegraphics[scale=0.5]{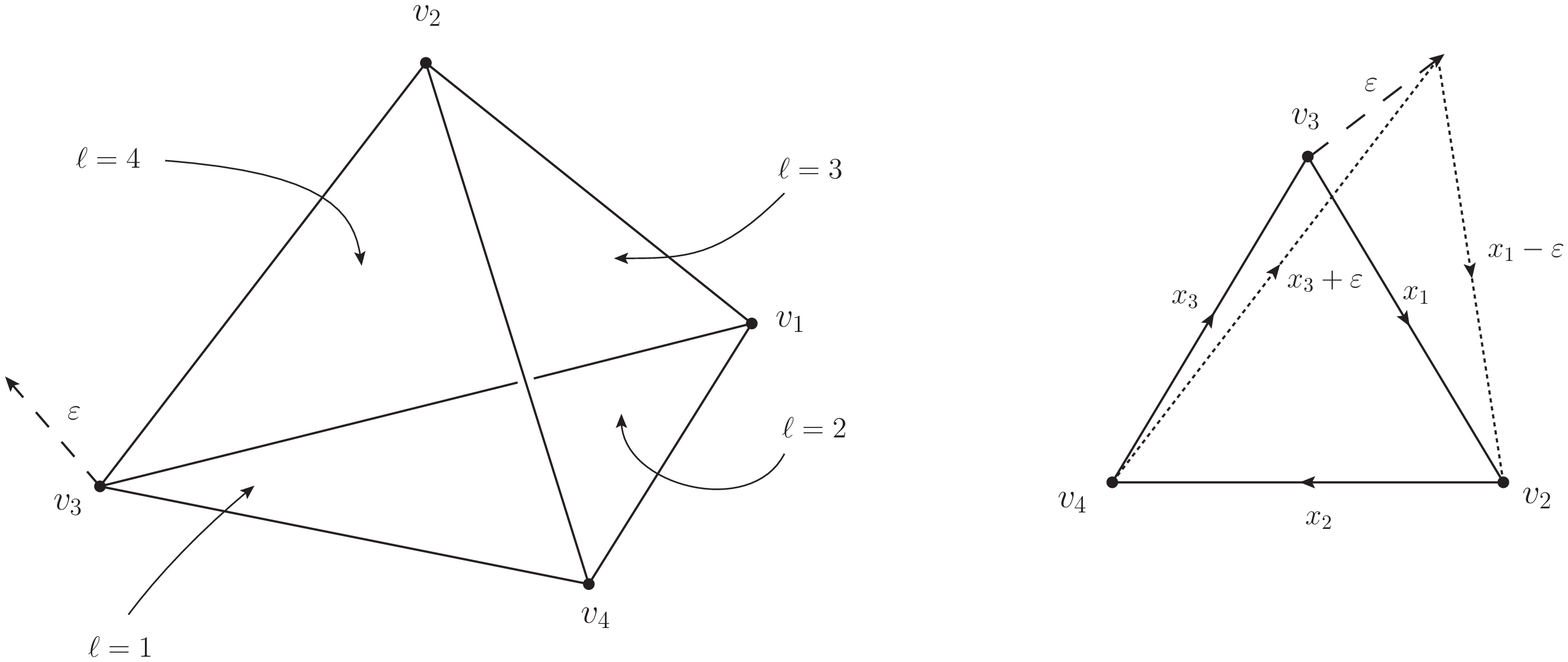} 
\caption{Action of $\cT^{3}_\epsilon$ on the interaction term, and resulting transformation of $\vphihat_1$.} \label{trans_vertex}
\end{center}
\end{figure}

This geometrical interpretation generalizes to any generator and any field: $\cT^{\ell'}_\varepsilon$ translates the vertex of color $\ell'$ in $\vphi_\ell$ (if any) by a quantity $\varepsilon$. With our conventions, the symmetries are therefore given by the following equations:
\bes \label{VertexTranslation1}
\cT^{1}_{\varepsilon} \act \vphi_1(g_1, g_2, g_3) &:=&  \vphi_1(g_1, g_2, g_3) \nn \\
\cT^{1}_{\varepsilon} \act \vphi_2(g_3, g_5, g_4) &:=& \e_{g_4^\inv g_5}(\varepsilon) \, \vphi_2(g_3, g_5, g_4) \nn \\
\cT^{1}_{\varepsilon} \act \vphi_3(g_5, g_2, g_6) &:=& \e_{g_5^\inv g_6}(\varepsilon) \, \vphi_3(g_5, g_2, g_6) \nn \\
\cT^{1}_{\varepsilon} \act \vphi_4(g_4, g_6, g_1) &:=& \e_{g_6^\inv g_4}(\varepsilon) \, \vphi_4(g_4, g_6, g_1) \nn 
\ees
\bes
\cT^{2}_{\varepsilon} \act \vphi_1(g_1, g_2, g_3) &:=& \e_{g_2^\inv g_1}(\varepsilon) \, \vphi_1(g_1, g_2, g_3) \nn \\
\cT^{2}_{\varepsilon} \act \vphi_2(g_3, g_5, g_4) &:=&  \vphi_2(g_3, g_5, g_4) \nn \\
\cT^{2}_{\varepsilon} \act \vphi_3(g_5, g_2, g_6) &:=& \e_{g_6^\inv g_2}(\varepsilon) \, \vphi_3(g_5, g_2, g_6) \nn \\
\cT^{2}_{\varepsilon} \act \vphi_4(g_4, g_6, g_1) &:=& \e_{g_1^\inv g_6}(\varepsilon) \, \vphi_4(g_4, g_6, g_1) \nn 
\ees
\bes
\cT^{3}_{\varepsilon} \act \vphi_1(g_1, g_2, g_3) &:=& \e_{g_1^\inv g_3}(\varepsilon) \, \vphi_1(g_1, g_2, g_3) \nn \\
\cT^{3}_{\varepsilon} \act \vphi_2(g_3, g_5, g_4) &:=& \e_{g_3^\inv g_4}(\varepsilon) \, \vphi_2(g_3, g_5, g_4) \nn \\
\cT^{3}_{\varepsilon} \act \vphi_3(g_5, g_2, g_6) &:=&  \vphi_3(g_5, g_2, g_6) \nn \\
\cT^{3}_{\varepsilon} \act \vphi_4(g_4, g_6, g_1) &:=& \e_{g_4^\inv g_1}(\varepsilon) \, \vphi_4(g_4, g_6, g_1) \nn 
\ees
\bes
\cT^{4}_{\varepsilon} \act \vphi_1(g_1, g_2, g_3) &:=& \e_{g_3^\inv g_2}(\varepsilon) \, \vphi_1(g_1, g_2, g_3) \nn \\
\cT^{4}_{\varepsilon} \act \vphi_2(g_3, g_5, g_4) &:=& \e_{g_5^\inv g_3}(\varepsilon) \, \vphi_2(g_3, g_5, g_4) \nn \\
\cT^{4}_{\varepsilon} \act \vphi_3(g_5, g_2, g_6) &:=& \e_{g_2^\inv g_5}(\varepsilon) \, \vphi_3(g_5, g_2, g_6) \nn \\
\cT^{4}_{\varepsilon} \act \vphi_4(g_4, g_6, g_1) &:=& \vphi_4(g_4, g_6, g_1). \nn 
\ees
Note that the Hopf algebra deformations (i.e. the $\star$-products) are defined such that the plane-waves generating the translations are always of the form $\e_{g_i^\inv g_j}(\varepsilon)$. This feature has also a geometrical meaning: it guarantees that the transformed fields stay invariant under diagonal left action of $\SO(3)$, that is the triangles remain closed after translation of one of their vertices.

To be complete, we would need to specify how these translations act on products of fields. This step, which depends on the $\cD \SO(3)$ co-product, again amounts to a choice of $\star$-product orderings of the plane waves resulting from the actions on individual fields. One result of \cite{diffeos} is that it is possible to define them in such a way that the action, and in particular its interaction term, are left invariant. We postpone this task to the next section, where the use of vertex variables will make the definitions more geometrically transparent.

The interpretation of these symmetries is very nice. As just mentioned, they are interpreted as translations of the vertices of the triangulation, which at the level of simplicial gravity are the discrete counterparts of the diffeomorphisms \cite{biancadiffeos}. At the discrete gauge field theory level, that is in group space, they impose triviality of the holonomy around a loop encircling a vertex of the boundary triangulation (this is apparent in the group representation of the GFT interaction vertex), which is the content of the diffeomorphism constraints of 3d gravity. Finally, in the spin foam formulation they generate the recurrence relations satisfied by $6j$-symbols \cite{BarrettCraneWdW,eterasimonevalentin}, which again encode the diffeomorphism invariance of the theory in algebraic language. We refer to \cite{diffeos} for a detailed discussion of these aspects.

\begin{figure}[h]
\begin{center}
\includegraphics[scale=0.5]{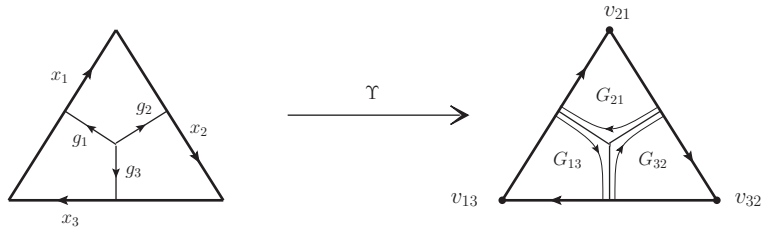} 
\caption{Map from edge to vertex variables.} \label{change_variables_v}
\end{center}
\end{figure}

\subsubsection{Vertex Lie algebra variables}

Following the interpretation of the symmetries as vertex translations, it is tempting to change variables so that the new fields directly depend on the generators of these translations. The idea is simply to write each edge Lie algebra variable as a difference between the positions of its two endpoints (with respect to some arbitrary reference point). Each triangle is now to be described by the three positions of its vertices, as represented in Figure \ref{change_variables_v}. However, due to the non-commutative nature of this symmetry, one would actually like to define a twisted version of this change of variables, schematically:
\beq
\forall \ell \in \left\lbrace 1,..,4\right\rbrace ,\, \psi_\ell(v_{21} , v_{13} , v_{31}) \equiv \bigstar_{v_{21}} \bigstar_{v_{13}} \bigstar_{v_{32}} \vphihat_\ell(v_{21} - v_{13} , v_{32} - v_{21} , v_{13} - v_{32})\,,
\eeq
where the $\bigstar$'s stand for implicit internal $\star$-products. 
To this effect one can introduce the map $\Upsilon$, which to any left-invariant field $\phi \in L^{2}(\SO(3)^{3})$, associates a function of three $\so(3)$ elements, defined as:
\beq
\Upsilon[\phi](v_{21}, v_{13}, v_{32}) \equiv \int \extd g_1 \extd g_2 \extd g_3 \phi(g_1, g_2, g_3) \e_{g_2^\inv g_1}(v_{21}) \e_{g_1^\inv g_3}(v_{13}) \e_{g_3^\inv g_2}(v_{32})\,.
\eeq
The variable $v_{ij}$ is interpreted as a position variable for the vertex shared by the edges $i$ and $j$. The intrinsic geometry of a triangle in $\mathbb{R}^3$ is fully characterized, up to rotations, by specifying three
edge vectors constrained to close. If we are given the three positions of its vertices instead, there is a redundancy: a simultaneous translation of all the vertices simply corresponds to a change of the origin of the coordinates in $\mathbb{R}^3$ and does not affect the geometry.
This redundancy should also manifests itself in the properties of the function $\Upsilon[\phi]$. 
However, because the translation symmetry is encoded in a quantum group acting on products of representations, in order to define it properly we need to specify an ordering of the arguments of the fields. We therefore
interpret $\Upsilon[\phi]$ as a tensor product of representations: 
\beq
\Upsilon[\phi] = \int \extd g_1 \extd g_2 \extd g_3 \, \phi(g_1, g_2, g_3) \, \e_{g_2^\inv g_1} \ot \e_{g_1^\inv g_3} \ot \e_{g_3^\inv g_2}\,.
\eeq
We can then introduce the global translation operator $\cT_{\ve}$, defined on tensor products of plane-waves by:
\bes
\cT_{\ve} \act \left( \e_{g_1}(x_1) \ot \cdots \ot \e_{g_N}(x_N ) \right) &\equiv& \bigstar_{\ve} \left( \e_{g_1}(x_1 + \ve) \ot \cdots \ot \e_{g_N}(x_N + \ve) \right) \nn \\
&\equiv& \e_{g_1 \cdots g_N}(\ve)\,
\left( \e_{g_1}(x_1) \ot \cdots \ot \e_{g_N}(x_N) \right)\,.
\ees
We immediately verify that:
\beq
\cT_{\ve} \act \Upsilon[\phi](v_{21}, v_{13}, v_{32}) = \Upsilon[\phi](v_{21}, v_{13}, v_{32})\,.
\eeq
Therefore $\Upsilon$ maps the space of left-invariant functions in $L^{2}(\SO(3)^{3})$ to a set of functions of three $\so(3)$ elements invariant under
$\cT_{\ve}$. One can moreover prove it to be injective, and hence bijective between $L^{2}(\SO(3)^{3})$ and its image. We will not detail this point here, but we will do it for the slightly more involved Ooguri model, in Proposition \ref{bij_4d}. This property ensures that
the theory can fully be formulated with vertex variables, provided by the new fields:
\beq
\forall \ell \in \left\lbrace 1,..,4\right\rbrace ,\qquad \psihat_\ell = \Upsilon[\vphi_\ell]\,.
\eeq
The key point of this transformation lies in the following: the initial gauge invariance condition, that is diagonal left-translation invariance in group space, is traded for a global (twisted) translation invariance in Lie algebra vertex variables.

\

As already proven in \cite{diffeos}, the original Boulatov action can then be re-written in terms of the new fields $\psihat_\ell$. With the conventions of this thesis, we have:
{\footnotesize
\bes
S_{kin}[\psihat , \overline{\psihat}] &=& \sum_\ell \int [\extd^3 v_i]^2\, \psihat_\ell(v_1, v_2, v_3)  \star {\overline{\psihat}}_\ell(v_1, v_2, v_3) \,, \\
S_{int}[\psihat , \overline{\psihat}] &=& \lambda \int  [\extd^3 v_i]^3\, \psihat_1(- v_2, v_3, - v_4) \star \psihat_2(- v_4, v_3, v_1) \star \psihat_3(- v_4, v_1, - v_2) \star \psihat_4(v_1, v_3, - v_2) \label{int_vertex} \\
&+& \overline{\lambda} \int  [\extd^3 v_i]^3\,{\overline{\psihat}}_1(v_2, - v_3, v_4) \star {\overline{\psihat}}_2(v_4,- v_3,- v_1) \star {\overline{\psihat}}_3(v_4,- v_1, v_2) \star {\overline{\psihat}}_4(- v_1,- v_3, v_2) \, .\nn
\ees}
We notice that in all the integrals we have one free variable which can be fixed to any value without changing the value of the action; this amounts to a choice of origin from which measuring the position of the vertices. This is also reflected in the four translation symmetries not being independent, one of them being automatically verified when the others are imposed; in other words, the model knows about the intrinsic geometry of the triangles and of the tetrahedra they form, and correctly does not depend on their embedding in $\mathbb{R}^3$. 

We remark also that, in the interaction, each vertex variable appears in three different fields, so that we have a $\star$-product of three terms for each $v_\ell$. The extra signs encode orderings of the $\star$-products, which can again be interpreted as defining the orientations of the triangles. Consider for example the triangle of color $1$. From Figure \ref{trans_vertex}, its orientation is given by the cyclic ordering $(x_1, x_2, x_3)$ of its edge variables, which induces a natural cyclic ordering of its vertices: $(v_2, v_3, v_4)$ (notice that by convention, we actually choose the reverse ordering). This induces in turn an ordering of the triangles attached to the vertex $v_1$: $(\ell = 2, \ell =3, \ell = 4)$ (see again the left part of Figure \ref{trans_vertex}). This is the (cyclic) order in which, in the clockwise interaction term, the $\star_{v_1}$-product of fields having $v_1$ in their arguments (that is $\psihat_2$, $\psihat_3$ and $\psihat_4$) has to be computed. In the anticlockwise interaction term, this has to be reversed. That is why the variable $v_1$ appears with a positive sign in the first interaction term, and a minus sign in the second. This discussion generalizes to any color, so that in the end signs in front of variables $v_\ell$ are fully determined by the ordering of variables in the field $\hpsi_\ell$ of the same color.

\subsubsection{Shift symmetry in vertex variables}

Now we have a vertex representation of the classical theory, it is interesting to discuss further the translation symmetries. As expected, we have simpler formulas in this representation. 

Let us first discuss the action of translations on individual fields. The transformations read:
\bes \label{VertexTranslation2}
\cT^{1}_{\varepsilon} \act \hpsi_1(v_2, v_3, v_4) &=&  \hpsi_1(v_2, v_3, v_4) \nn \\
\cT^{1}_{\varepsilon} \act \hpsi_2(v_4, v_3, v_1) &=&  \hpsi_2(v_4, v_3, v_1 + \varepsilon) \nn \\
\cT^{1}_{\varepsilon} \act \hpsi_3(v_4, v_1, v_2) &=&  \hpsi_3(v_4, v_1 + \varepsilon, v_2) \nn \\
\cT^{1}_{\varepsilon} \act \hpsi_4(v_1, v_3, v_2) &=&  \hpsi_4(v_1 + \varepsilon, v_3, v_2) \nn 
\ees
\bes
\cT^{2}_{\varepsilon} \act \hpsi_1(v_2, v_3, v_4) &=&  \hpsi_1(v_2 + \varepsilon, v_3, v_4) \nn \\
\cT^{2}_{\varepsilon} \act \hpsi_2(v_4, v_3, v_1) &=&  \hpsi_2(v_4, v_3, v_1) \nn \\
\cT^{2}_{\varepsilon} \act \hpsi_3(v_4, v_1, v_2) &=&  \hpsi_3(v_4, v_1, v_2 + \varepsilon) \nn \\
\cT^{2}_{\varepsilon} \act \hpsi_4(v_1, v_3, v_2) &=&  \hpsi_4(v_1, v_3, v_2 + \varepsilon) \nn 
\ees
\bes
\cT^{3}_{\varepsilon} \act \hpsi_1(v_2, v_3, v_4) &=&  \hpsi_1(v_2, v_3 + \varepsilon, v_4) \nn \\
\cT^{3}_{\varepsilon} \act \hpsi_2(v_4, v_3, v_1) &=&  \hpsi_2(v_4, v_3 + \varepsilon, v_1) \nn \\
\cT^{3}_{\varepsilon} \act \hpsi_3(v_4, v_1, v_2) &=&  \hpsi_3(v_4, v_1, v_2) \nn \\
\cT^{3}_{\varepsilon} \act \hpsi_4(v_1, v_3, v_2) &=&  \hpsi_4(v_1, v_3 + \varepsilon, v_2) \nn 
\ees
\bes
\cT^{4}_{\varepsilon} \act \hpsi_1(v_2, v_3, v_4) &=&  \hpsi_1(v_2, v_3, v_4 + \varepsilon) \nn \\
\cT^{4}_{\varepsilon} \act \hpsi_2(v_4, v_3, v_1) &=&  \hpsi_2(v_4 + \varepsilon, v_3, v_1) \nn \\
\cT^{4}_{\varepsilon} \act \hpsi_3(v_4, v_1, v_2) &=&  \hpsi_3(v_4 + \varepsilon, v_1, v_2) \nn \\
\cT^{4}_{\varepsilon} \act \hpsi_4(v_1, v_3, v_2) &=&  \hpsi_4(v_1, v_3, v_2) \nn 
\ees
Thus each field $\hpsi_\ell$ can be interpreted as living in the representation space of (the translation part of) three copies of the deformed 3d Poincar\'e group $\cD \SO(3)$. This makes the interpretation of these transformations as vertex translations more explicit, and clarifies the very definition of the GFT.

The deformation of the translations manifests itself when acting on products of fields. This is a question we left open in the previous sections, exactly because it is more easily understood in the vertex formulation. To define the action of the translations on a product of fields, we need to interpret it as a tensor product. There is no canonical choice: for example the integrand $\psi_1^{234} \psi_2^{431} \psi_3^{412} \psi_4^{132}$ in the interaction term (\ref{int_vertex}) can be interpreted as the evaluation of $\psi_1^{234} \ot \psi_2^{431} \ot \psi_3^{412} \ot \psi_4^{132}$, but also of $\psi_2^{431} \ot \psi_1^{234} \ot \psi_3^{412} \ot \psi_4^{132}$, and generally of any permutation of the representation spaces. The Hopf algebra deformation of the translations required to make the interaction invariant will then depend on this additional convention. 
For definiteness let us interpret the term $\psi_1^{234} \psi_2^{431} \psi_3^{412} \psi_4^{132}$ as the evaluation of $\psi_1^{234} \ot \psi_2^{431} \ot \psi_3^{412} \ot \psi_4^{132}$. The Hopf algebra structure of the symmetries then has to be consistent with orderings of $\star$-products (i.e. signs) in equation (\ref{int_vertex}). This requires to distinguish colors $\{1,3\}$ from $\{2,4\}$, since the corresponding variables have opposite signs in (\ref{int_vertex}). 
All this suggests the following definition of translations, on products of fields, which we give in group variables. If $\{\phi_i, \, i =1, \cdots, N\}$ are living in the representation space of $\cT^{\ell}$, then:
\bes
\cT^{\ell}_{\varepsilon} \act (\phi(g_1) \ot \cdots \ot \phi(g_N)) &\equiv& \e_{g_1 \cdots g_N} (\varepsilon) (\phi(g_1) \ot \cdots \ot \phi(g_N))\,, \qquad \rm{if} \; \ell \in \{1, 3 \} \\
\cT^{\ell}_{\varepsilon} \act (\phi(g_1) \ot \cdots \ot \phi(g_N)) &\equiv& \e_{g_N \cdots g_1} (\varepsilon) (\phi(g_1) \ot \cdots \ot \phi(g_N))\,, \qquad \rm{if} \; \ell \in \{2, 4 \}\,.
\ees 
With this definition, and the tensor product interpretation of the interaction term we gave, the action is indeed invariant under translations. For instance, in metric variables, the integrand of the interaction part of the action is simply translated with respect to its variable of color $\ell$ under the transformation $\cT^{\ell}$. As a result, and because it is defined by integrals over the whole space $\su(2)$, the invariance follows.

\subsubsection{Vertex group variables}

 As far as the quantum theory is concerned, and in particular for practical calculations, it is convenient to Fourier transform back the vertex formulation to group space. The dual group variables are $G_{ij} \equiv g_{i}^{\inv}
g_j$, Fourier duals of the $v_{ij}$. Due to the translation invariance of the Lie algebra fields, in group space the configuration fields are distributions $\widetilde{\psi}_\ell$ of the form:
\beq
\widetilde{\psi}_\ell (G_1, G_2, G_3) = \delta(G_1 G_2 G_3) \psi_\ell(G_1, G_2, G_3)\,, 
\eeq
where $\psi_\ell$ are regular functions. It is interesting at this point to notice the double duality between algebra/group and edge/vertex variables: one has to impose closure constraints in algebra edge variables and group vertex variables; the same translate into translation invariance in group edge variables and algebra vertex variables. The precise forms of these different constraints are summarized in Table \ref{constraints}.  

\begin{table}[h]
\centering
{\footnotesize
\begin{tabular}{| c | c | c |}
    \hline
   & Edge variables & Vertex variables  \\ \hline
   & & \\
Group & $\vphi_\ell ( g_1 , g_2 , g_3 ) = \vphi_\ell ( h g_1 , h g_2 , h g_3 )$ & $\widetilde{\psi}_\ell (G_1, G_2, G_3) = \delta(G_1 G_2 G_3) \psi_\ell(G_1, G_2, G_3)$  \\ \hline
& & \\
Algebra & $\vphihat_\ell (x_1 , x_2 , x_3 ) = \delta_0 ( x_1 + x_2 + x_3) \star \vphihat_\ell (x_1 , x_2 , x_3 )$ & $ \psihat_\ell (v_1 , v_2 , v_3 ) = \bigstar_{\varepsilon} \psihat_\ell (v_1 + \varepsilon, v_2 + \varepsilon, v_3 + \varepsilon)$ \\ \hline
  \end{tabular}}
\caption{Fields and constraints in the different representations of the Boulatov model.}
\label{constraints}
\end{table}

 In terms of the newly defined fields, the action takes the form:
\bes
S[\psi , \overline{\psi}] &=&  \frac{1}{2} \sum_{\ell} \int [\extd G^\ell_i] \psi_\ell(G^\ell_1, G^\ell_2, G^\ell_3) \delta(G^\ell_1 G^\ell_2 G^\ell_3) \overline{\psi}_\ell(G^\ell_1, G^\ell_2, G^\ell_3) \nn \\
&+& \lambda \int [\prod_{l \neq l'} \extd G^l_{l'}] \cV(G^l_{l'}) \psi_1^{234} \psi_2^{431} \psi_3^{412} \psi_4^{132} \label{action_vertex} \\
&+& \overline{\lambda} \int [\prod_{\ell \neq \ell'} \extd G^l_{l'}] \cV(G^l_{l'}) {\overline{\psi}}_1^{\lower0.1in \hbox{\footnotesize234}} {\overline{\psi}}_2^{\lower0.1in \hbox{\footnotesize431}}
{\overline{\psi}}_3^{\lower0.1in \hbox{\footnotesize412}} {\overline{\psi}}_4^{\lower0.1in \hbox{\footnotesize 132}}\nn ,
\ees
with the notation convention $\psi_{\ell}^{ijk} \equiv \psi_{\ell}(G^{\ell}_i, G^{\ell}_j, G^{\ell}_k)$, and a vertex function defined by:
\bes\label{kerint_v}
\cV(G^l_{l'}) &=& \delta(G^{1}_{2} G^{1}_{3} G^{1}_{4})\delta(G^{2}_{4} G^{2}_{3} G^{2}_{1})\delta(G^{3}_{4} G^{3}_{1} G^{3}_{2})\delta(G^{4}_{1} G^{4}_{3} G^{4}_{2}) \nn \\
		&& \delta(G^{4}_{2} G^{3}_{2} G^{1}_{2})\delta(G^{4}_{3} G^{1}_{3} G^{2}_{3}) \delta(G^{1}_{4} G^{3}_{4} G^{2}_{4}). 
\ees 
The notations in the vertex function are as follows: upper indices label triangles or equivalently fields, and lower indices are associated to vertices, with the convention that a color $\ell$ labels
the vertex opposite to the triangle of the same color $\ell$ (see Figure \ref{color_conventions_v}). 

Notice the appearance of a kernel of the kinetic term (i.e. the distributional part of the fields $\widetilde{\psi}_\ell$). This is the Fourier dual of the translation invariance of the fields we
described in the Lie algebra
representation. It is interpreted as a consistency constraint on the three group elements associated to a quantized triangle: their ordered product needs by construction to be trivial. 

The vertex
function consists (first line) in the same distributional factors of the fields $\widetilde{\psi}_{\ell}$, imposing the mentioned consistency conditions, while the second lines are flatness conditions
associated to paths around the vertices of the tetrahedra, hence guaranteeing flatness of the connection in the boundary of the tetrahedron (see Figure \ref{tetrahedron_v}). Note that only three of
these flatness constraints appear in the interaction kernel, while a tetrahedron has four vertices. This is
obviously because only three of these flatness constraints are independent. We can therefore choose any triplet of these four constraints\footnote{The fourth one, associated to the vertex of color 1
being
simply $\delta(G^{2}_{1} G^{3}_{1} G^{4}_{1})$.} to express the same distribution, implementing all four constraints. 

\begin{figure}[h]
  \centering
  \subfloat[Coloring of vertices]{\label{color_conventions_v}\includegraphics[scale=0.5]{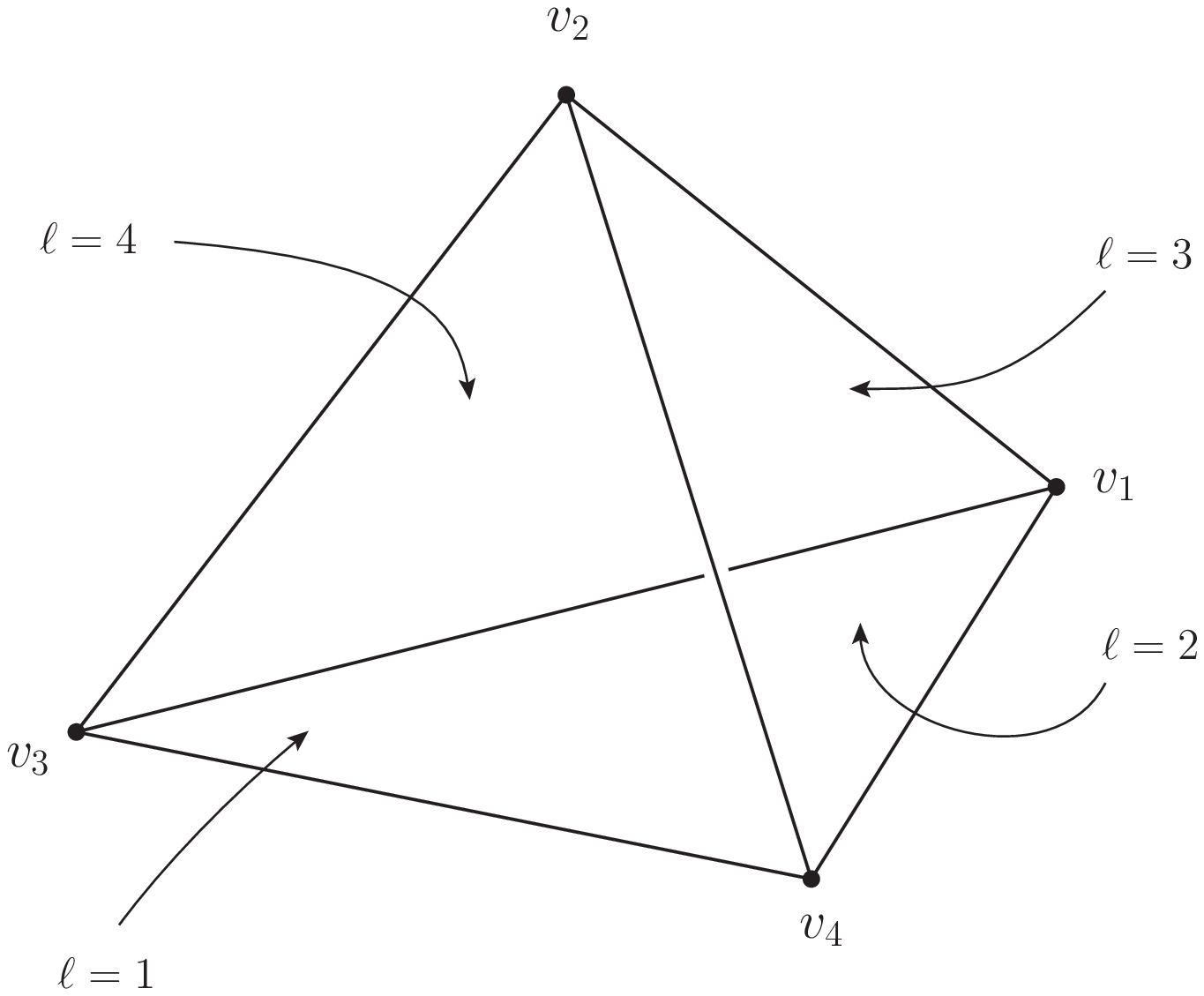}}                
  \subfloat[Holonomy variables]
{\label{variables_v}\includegraphics[scale=0.5]{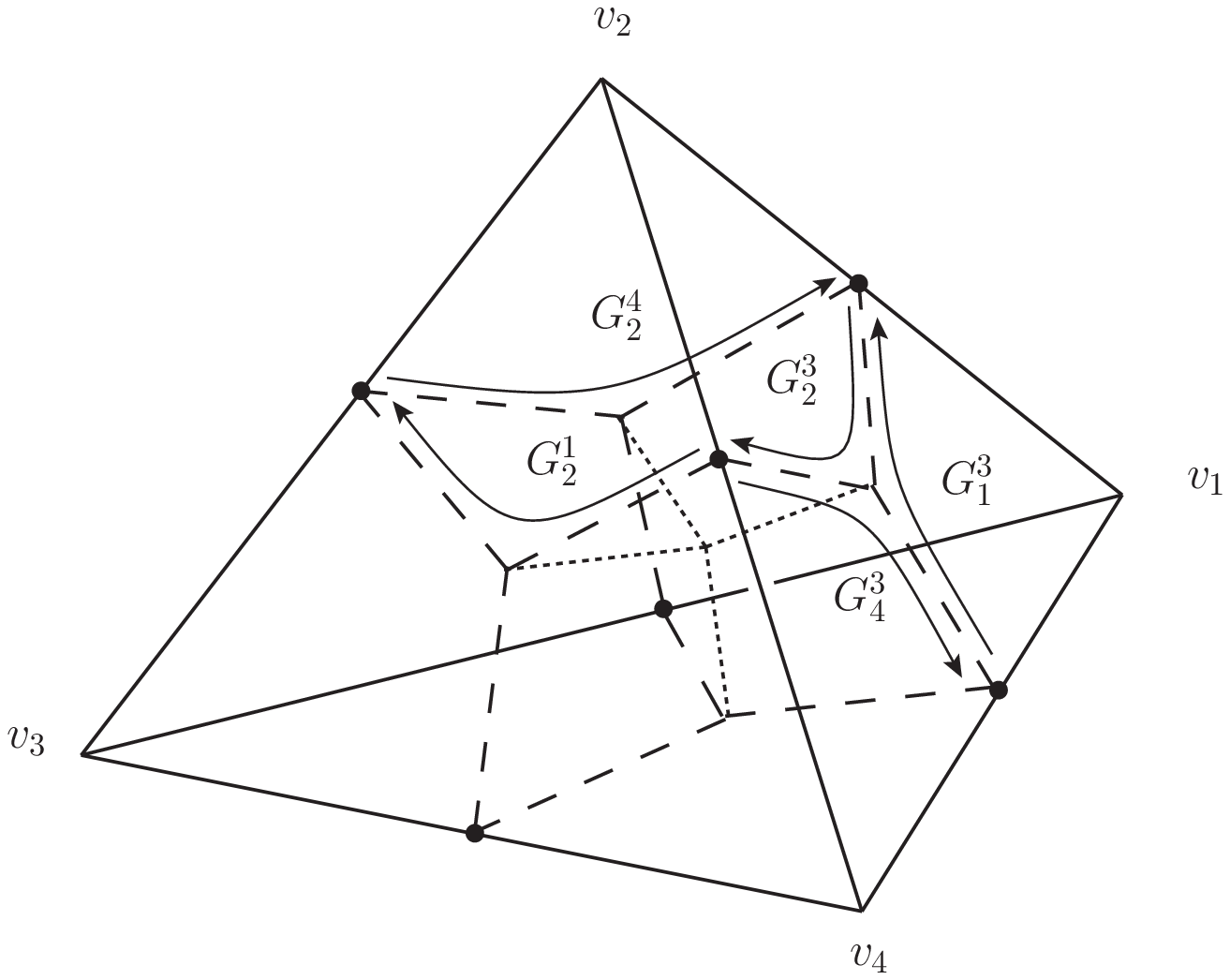}}
  \caption{Coloring conventions and group variables associated to one tetrahedral interaction. The amplitude imposes two kinds of conditions: consistency conditions on triangles, for instance $G^{3}_4 G^{3}_1 G^{3}_2 = \one$ in the triangle of color $3$; and flatness conditions around vertices, for example $G^{4}_2 G^{3}_2 G^{1}_2 = \one$ around the vertex $v_2$.}
  \label{tetrahedron_v}
\end{figure}

\subsubsection{Quantum amplitudes}

Now that we have well-understood the constraints that need to be imposed in vertex variables, we can provide a rigorous definition for the partition function $\cZ$ in these variables. We resort to a non degenerate Gaussian measure $\mu_{\cP}$ over the space of regular fields $\psi_\ell$, which combines the translation invariance constraint with the ill-defined Lebesgue measure into a well-defined Gaussian covariance: 
\beq\label{cov_v}
\int \extd \mu_{\cP}(\overline{\psi}, \psi) \, \overline{\psi_\ell(g_1, g_2, g_3)} \psi_{\ell'}(g_1', g_2' , g_3') \equiv \delta(g_1 g_2 g_3) \, \delta_{\ell, \ell'}  \prod_{i = 1}^{6} \delta(g_i
g_i'^{\inv}) \,.
\eeq
Only the exponential of
the interaction part of the action remains to be integrated, to give a suitable definition of $\cZ$:
\bes\label{partition_vertex}
\cZ &\equiv& \int \extd \mu_{\cP}(\overline{\psi}, \psi) \, \e^{- V[\overline{\psi}, \psi]} \\
V[\overline{\psi}, \psi] &\equiv&  \lambda \int [\extd G] \,\delta(G^{4}_{2} G^{3}_{2} G^{1}_{2})\delta(G^{4}_{3} G^{1}_{3} G^{2}_{3}) \delta(G^{1}_{4} G^{3}_{4} G^{2}_{4})\, \psi_1^{234} \psi_2^{431}
\psi_3^{412} \psi_4^{132}  \; \; + \; \; {\rm{c.c.}}\;
\ees
A couple of remarks are in order. First, only the flatness part of the kernel of the interaction has been used in the definition of $V$. This is because the distributional nature of the configuration
fields $\widetilde{\psi}_{\ell}$ has already been taken care of in the measure. Were we to integrate $S_{int}$ and not $V$, we would pick up products of equal distributions in the amplitudes, hence
further divergences. Second, at this formal level, which
of the four flatness constraints we use to define $V$ does not matter (see the resulting graphical representation in Figure \ref{int_vertex_v}). In the regularized theory, to which we turn in the next
subsection, this is not the case, and we expect different choices to give amplitudes with the same scaling behavior but differing by factors of order $1$ (in the cut-off). That is why, contrary to
the regularization chosen in the paper \cite{vertex}, we will use the symmetric regularization in the color indices of \cite{edge}.
Finally, it is important to stress that the vertex formulation of the path-integral we just introduced is strictly equivalent to the usual Boulatov model, written in terms of edge variables. At the level of gauge invariant fields, our construction amounts to a simple change of variables in terms of which the fields are expressed. Therefore the Jacobian of the transformation evaluates to one. 

\begin{figure}[h]
  \centering
  \subfloat[Full combinatorics (clockwise)]{\label{interaction_vertex}\includegraphics[scale=0.5]{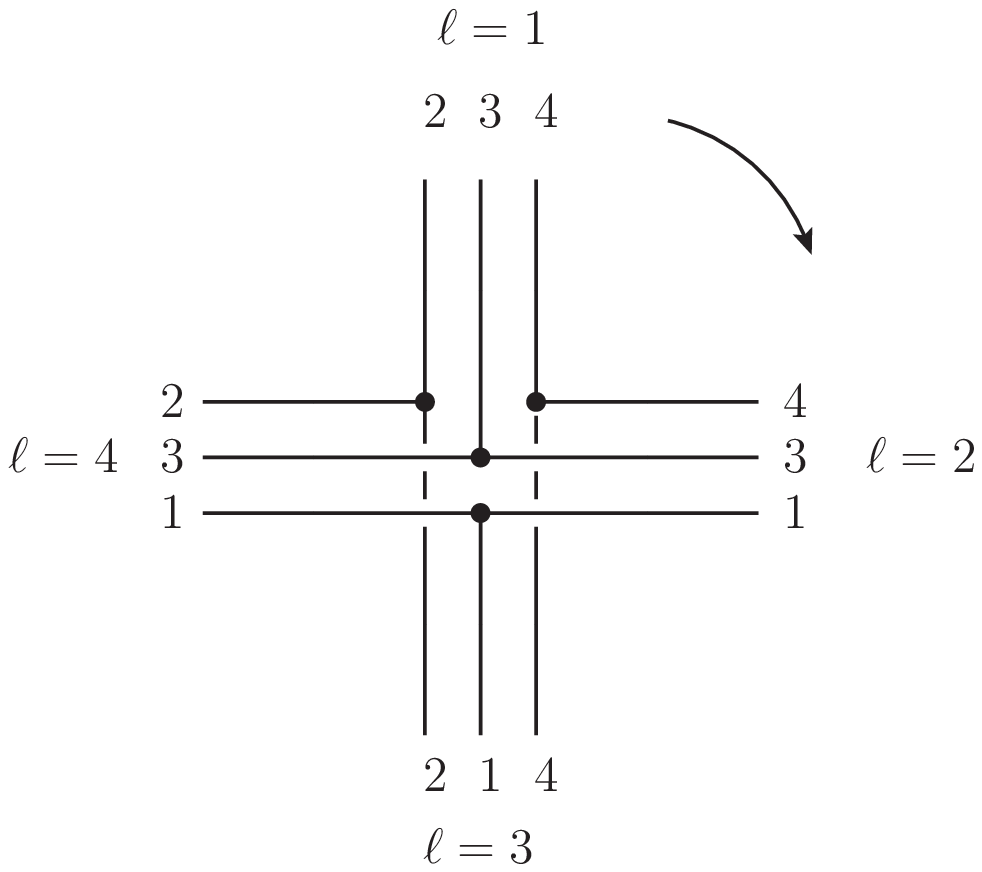}}
  \subfloat[Color 1 implicit]{\label{int_vertex_v2}\includegraphics[scale=0.5]{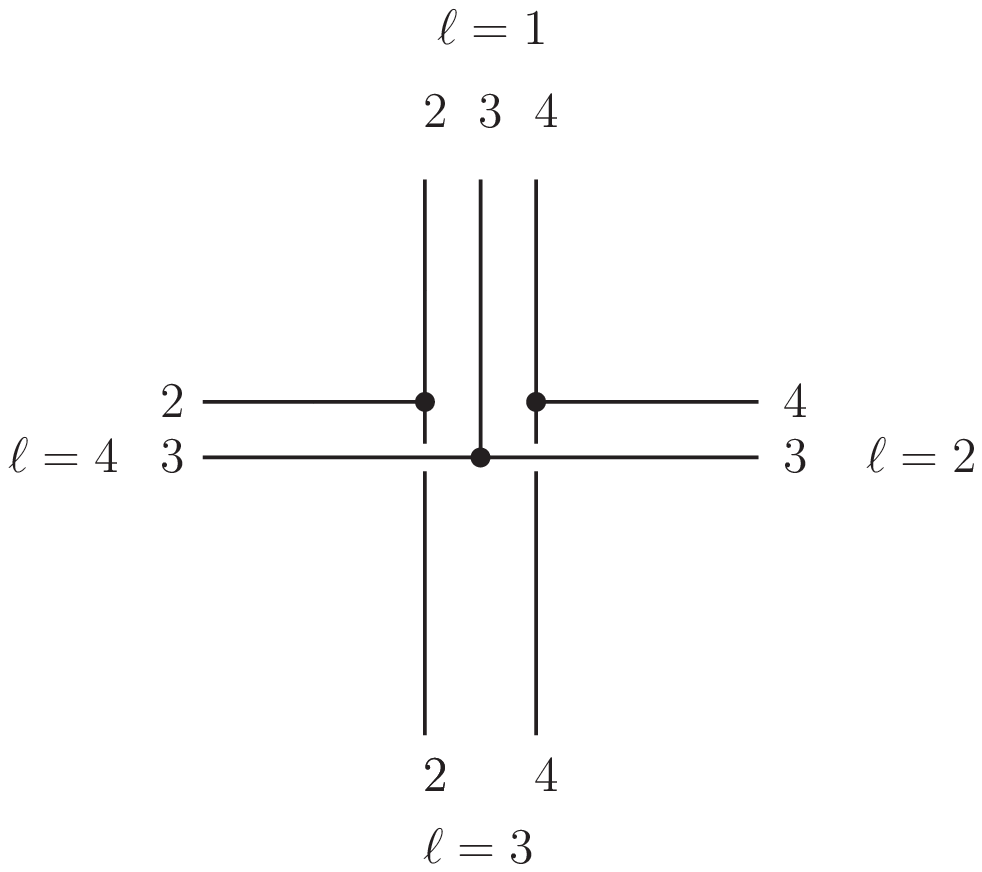}}
  \caption{Combinatorics of the interaction in vertex variables, in a form suitable for factorization of bubbles of color $1$.}
  \label{int_vertex_v}
\end{figure}

\

Formally, we can expand the partition function as a power series in $\lambda \overline{\lambda}$:
\beq
\cZ = \sum_{\cG} \frac{(- \lambda \overline{\lambda})^{\cN(\cG)/2}}{s(\cG)} \cA_{\cG}\,,
\eeq
where the amplitudes $\cA_{\cG}$ are labeled by closed colored graphs. White nodes are associated to anticlockwise vertices, and black nodes to clockwise ones. They contribute with an integrand respectively equal to $\lambda$ or $\overline{\lambda}$ times the kernel (\ref{kerint_v}). Each line of color $\ell$ is associated to a covariance (\ref{cov_v}), with $\ell' = \ell$. An example of the correspondence between colored and stranded representations is provided in Figure \ref{example_graph_v}, for the elementary melon graph. Also illustrated in this figure is the fact that the vertex variables make the $3$-bubbles explicit, as opposed to the faces in the edge variables. Indeed, in the stranded representation, let us consider a connected component of strands of color $\ell$. It is a graph of a non-commutative $\Phi^3$ scalar field theory on a Lie algebra space-time $\su(2)$, with momentum space $\SU(2)$, the interaction being essentially momentum conservation at each node. From the simplicial perspective, it is dual to the bubble around a vertex of color $\ell$, and encodes its topological structure. Precisely, each line of this $3$-graph can be assigned an additional individual color: that of the colored $4$-graph line it is part of. Thus we really have a colored $3$-graph, which is therefore dual to a closed and orientable triangulated surface \cite{Vince_2d,lrd,diffeos}: the bubble.
The overall amplitude associated to a 4-graph is therefore given by $\Phi^3$ graphs encoding the structure of the bubbles, glued to one another through propagators (associated to triangles).

\begin{figure}[h]
\begin{center}
\includegraphics[scale=0.5]{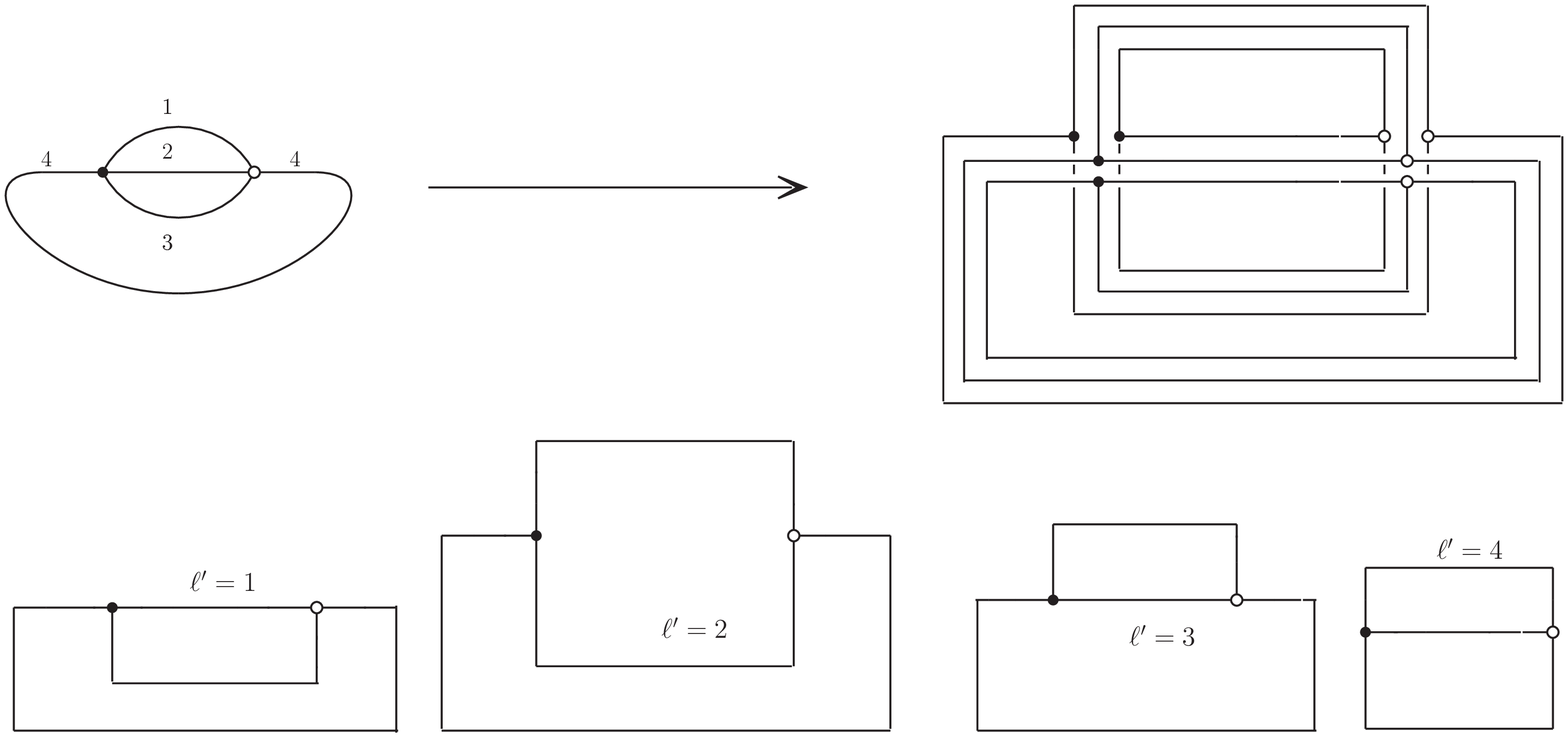} 
\caption{Combinatorial structure of the melon graph in vertex variables, and its four bubble graphs.} \label{example_graph_v}
\end{center}
\end{figure}

\

As far as computations are concerned, the main advantage of the vertex formulation, is that most of the non-trivial constraints encoded in the propagators can be integrated out straight away. This is because one of the four $\delta$-functions associated to stranded nodes is missing in each interaction kernel. A consistent choice throughout a given graph $\cG$ therefore allows to associate all the strands of a given color, say $1$, to integration variables which appear in no more than one propagator each. The corresponding $\delta$-functions can therefore be integrated to $1$, and one is left with reduced amplitudes built out of interaction vertices of the type shown in Figure \ref{int_vertex_v2}. One achieves in this case a factorization of the integrands of the amplitudes in terms of color-$1$ bubble contributions, initially proven in \cite{vertex}. This is of course possible with any color $\ell$, and yields the general formula:
\beq\label{amplitudes_v}
\cA_{\cG} = \int [\extd G]^{\frac{3 \cN}{2}} \left( \prod_{b \in \cB_{\ell}} \prod_{v \in V_{b}} \delta\left( \overrightarrow{\prod_{f \in
\triangle^{b}_{v}}} (G_{v}^{f})^{\epsilon^{f}_{v}}\right)  \right) \left( \prod_{f \in \cF_{\ell}} \delta\left( \overrightarrow{\prod_{v \in f}} G_{v}^{f}\right)\right)\, ,
\eeq
where $\cN$ is the order of $\cG$, $\cB_{\ell}$ is its set of bubbles of color $\ell$, $V_b$ the set of vertices in a bubble $b$, $\triangle^{b}_{v}$ the set of triangles in a bubble $b$ that share one of its vertices $v$, and
finally
$\cF_{\ell}$ is the set of triangles of color $\ell$ in the dual complex. 
The geometrical interpretation of this expression is very natural: the bubble terms in the first parenthesis encode flatness around
each of the vertices of the triangulated surface, and the terms in the second parenthesis encode the consistency conditions in triangles of color $\ell$. 
One thus obtains an effective description
of the model in terms of triangulated 3-cells whose boundaries have color $\ell$, the bubbles. 
In Figure \ref{ex_bubble}, we illustrate this result by showing a portion of a bubble $b$, dual to a vertex $v_b$. The effective interaction term associated to $b$ imposes a trivial holonomy around the
vertex $v \in b$.

\begin{figure}[h]
\centering%
\includegraphics[scale=0.5]{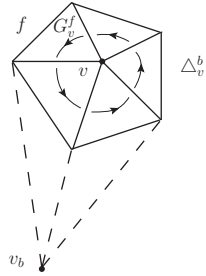}
\caption{Portion of a bubble $b$ dual to $v_b$: in $\triangle^{b}_{v}$, the holonomy around $v$ is imposed to be trivial.}
\label{ex_bubble}
\end{figure}

\

In the regulated theory, this factorization allows to derive nice bounds on the amplitudes: on the one hand, one can
understand the effect of singularities at vertices of color $\ell$; on the other hand, one can derive jacket bounds, which give information about the
global topology of the 4-colored graph, and of its dual simplicial complex.

		
		
		
		\subsection{Regularization and general scaling bounds}

Because the amplitudes of the Boulatov model are generically divergent, we need to regularize it. We adopt here the regularization scheme of \cite{edge}, which is well-adapted to the computation of scaling bounds in vertex variables.
This differs from the regularization used in \cite{vertex} in two respects: a) the cut-off is introduced at the level of the potential $V$; b) instead of a sharp cut-off in the harmonic
expansion of the $\delta$-functions, we will use heat kernels. The first point is in contrast with what has been and will be again advocated for renormalization, where keeping a fixed notion of locality in the interaction is very profitable. In the present situation, locality does not play any role, and on the contrary it is at the level of the covariances that the correspondence between edge and vertex variables has been established. It seems therefore preferable to keep the full Gaussian measures and regularize the interactions, this way the relation between scaling bounds in the two formulations is more transparent. Let us
elaborate a bit on this second point, and introduce the heat kernel. The $\delta$-function on $\SU(2)$ expands as:
\beq
\delta(g) = \sum_{j \in \frac{\mathbb{N}}{2}} (2 j + 1) \chi_{j}(g) \,.
\eeq

On the other hand, the heat kernel at time $\alpha > 0$ on $\SU(2)$ is given by:
\beq
K_{\alpha}(g) = \sum_{j \in \frac{\mathbb{N}}{2}} (2 j + 1) \e^{- \alpha j(j+1)} \chi_{j}(g) \,,
\eeq 
where $\chi_j$ are the characters of $\SU(2)$. We see that $K_{\alpha}$ converges to the $\delta$-distribution when $\alpha \to 0$, and therefore provides us with a regularization of the $\delta$-function.
The three properties of this regularization we will use in this and later chapters are the following. First, $K_{\alpha}$ is a positive function, which will be very convenient for bounding amplitudes. Second, these
functions behave nicely with respect to the convolution product:
\beq
\int \extd g K_{\alpha_1}(g_1 g^{\inv}) K_{\alpha_2}(g g_2) = K_{\alpha_1 + \alpha_2} (g_1 g_2)\,.
\eeq 
Finally, when $\alpha \to 0$, we have the following asymptotic formula\footnote{This formula holds on any compact excluding $-\one$, a simple derivation can be found in the Appendix \ref{app_heat}.}:
\beq
K_{\alpha} (g) \underset{\alpha \to 0}{\sim} {N_{\alpha}}^{3} \,  \exp\left( - \frac{|g|^2}{4 \alpha}\right) \frac{\vert g \vert / 2 }{\sin(\vert g \vert / 2)}\,,
\eeq
where $|g|$ is the Riemannian distance between $g$ and the identity, and ${N_{\alpha}}^{3} \equiv \sqrt{4 \pi} \alpha^{- 3 / 2}$. In particular we will use the fact that $K_{\alpha} (\one) \sim {N_{\alpha}}^{3}$ at
small $\alpha$. $N_\alpha$ is the parameter with respect to which the $1/N$ expansion will be constructed.  

\

We define the regularized theory by the following partition function:
\bes\label{partition_vertex_reg}
\cZ_{\alpha} &\equiv& \int \extd \mu_{\cP}(\overline{\psi}, \psi) \, \e^{- V_{\alpha}[\overline{\psi}, \psi]} \\
V_{\alpha}[\overline{\psi}, \psi] &\equiv&  \lambda \int [\extd G] \,\frac{K_{\alpha}(G^{2}_{1} G^{3}_{1} G^{4}_{1})K_{\alpha}(G^{4}_{2} G^{3}_{2} G^{1}_{2})K_{\alpha}(G^{4}_{3} G^{1}_{3} G^{2}_{3})
K_{\alpha}(G^{1}_{4} G^{3}_{4} G^{2}_{4})}{K_{\alpha}(\one)} \, \nn \\
&& \; \; \times \,\psi_1^{234} \psi_2^{431} \psi_3^{412} \psi_4^{132}  \;  +  \; {\rm{c.c}}.
\ees

\

Before moving on to the analysis of the regularized amplitudes, let us discuss briefly the role of the regularization chosen.
Note, in fact, that we have chosen a symmetric regularization in the colors, hence the re-introduction of the flatness constraint around the vertex of color $1$, together with the
appropriate rescaling. The main advantage of such a symmetric regularization is that discussing bounds from different bubble factorizations will be made easier, as will be detailed shortly. 

On the one hand, this regularization is slightly different from the natural scheme one would use in edge variables (see for instance \cite{ValentinJoseph}, which also largely motivated the heat kernel
regularization used in \cite{edge}). This is because it does not correspond to a regularization of the $\delta$-functions imposing flatness of the wedges dual to edges in the GFT interaction vertex
(thus encoding the piece-wise flatness of the simplicial complexes generated in the GFT expansion). On the other hand, regularizations that make bubble factorizations explicit after transformation to
vertex variables are very natural also when these last ones are used. For example, we could use a non-symmetric regularization of the interaction 
\beq\label{nonsym_reg}
V_{\alpha}^{1}[\overline{\psi}, \psi] \equiv  \lambda \int [\extd G] \, K_{\alpha}(G^{4}_{2} G^{3}_{2} G^{1}_{2})K_{\alpha}(G^{4}_{3} G^{1}_{3} G^{2}_{3})
K_{\alpha}(G^{1}_{4} G^{3}_{4} G^{2}_{4}) \, \psi_1^{234} \psi_2^{431} \psi_3^{412} \psi_4^{132}  \;  +  \; {\rm{c.c}},
\eeq
which allows to write the amplitudes in a factorized form similar to (\ref{amplitudes_v}), with $\ell = 1$. This corresponds exactly, in edge variables, to regularizing only the $\delta$-functions
associated to edges which do not contain the vertex of color $1$:
\bes
S_{int,1}^{\alpha}[\vphib , \vphi] &=& \lambda \int [\extd g ]^9 \, K_{\alpha}(g_4 g_4'^{\inv})
K_{\alpha}(g_5 g_5'^{\inv}) K_{\alpha}(g_6 g_6'^{\inv}) \nn \\
&&\times \,\vphi_1(g_1, g_2, g_3) \vphi_2(g_3, g_5, g_4) \vphi_3(g_5', g_2, g_6) \vphi_4(g_4', g_6', g_1)\;  +  \; {\rm{c.c}}.
\ees
We see, then, that equation (\ref{partition_vertex_reg}) interpolates between the four possible regularizations of this type. 

\

The main fact that makes this form of the regularized amplitude (\ref{partition_vertex_reg}) convenient for our purposes is that, using the positivity of the heat kernel, one can bound any of the four
flatness constraints by $K_{\alpha}(\one)$, and obtain bounds on amplitudes which have the same expression as with a regularization of the non-symmetric type (\ref{nonsym_reg}), but maintaining the
symmetry among colors manifest up to this last step.

One can then show that, for any graph $\cG$ and any choice of color $\ell_1$, the regularized amplitude $\cA_{\cG}^{\alpha}$ admits the
following bound:
\beq\label{amplitudes_v_reg}
|\cA_{\cG}^{\alpha}| \leq \int [\extd G]^{\frac{3 \cN}{2}} \left( \prod_{b \in \cB_{\ell_1}} \prod_{v \in V_{b}} K_{\langle v , b \rangle \alpha}\left(
\overrightarrow{\prod_{f \in \triangle^{b}_{v}}} (G_{v}^{f})^{\epsilon^{f}_{v}}\right)  \right) \left( \prod_{f \in \cF_{\ell_1}} \delta\left( \overrightarrow{\prod_{v \in f}}
G_{v}^{f}\right)\right)\, ,
\eeq
where $\langle v , b \rangle$ denotes the number of triangles in a bubble $b$ that contain the vertex $v$.  

\

Let us now see how this formula allows to derive interesting scaling bounds. 
The general idea is that we would like to remove the remaining propagator constraints, which are in a sense non-local quantities (from the point of view of the bubbles), and
prevent us from trivially integrating the amplitudes. 
A simple way of doing it is to pick up a second color label $\ell_2 \neq \ell_1$, and bound all the flatness constraints associated to vertices of color $\ell_2$  by their value at the identity. 
In the resulting bound, all
the propagator constraints will then have an independent variable of color $\ell_2$, allowing us to trivially integrate them. 
We are finally left with two $\phi^{3}$ graphs, corresponding to the strands in
the two remaining colors $\ell_3$ and $\ell_4$. 
Now, each connected component of such a graph is dual to a vertex (of the same color) of the simplicial complex. Therefore, integrating a tree in each of these components,
and bounding the final expression by its value at the identity, we arrive at:
\beq
|\cA_{\cG}^{\alpha}| \leq 
\left( \prod_{b \in \cB_{\ell_1}} \prod_{v \in V_{b}(\ell_2)} K_{\langle v , b \rangle \alpha}\left( \one \right)  \right)
\left( \prod_{v \in V_{\ell_3} \cup V_{\ell_4}} K_{ |v| \alpha}\left( \one \right) \right)\,,
\eeq
where for any $v \in V_{\ell_3} \cup V_{\ell_4}$
\beq
|v| \equiv \sum_{b \in \cB_{\ell_1}\, , \, b \supset v} \langle v , b \rangle\,.
\eeq
is equal to the number of tetrahedra in the simplicial complex that contain $v$. Finally, remarking that:
\beq
\forall a > 0, \qquad \frac{K_{a\alpha}(\one)}{K_{\alpha}(\one)} \underset{\alpha \to 0}{\longrightarrow} a^{- 3/2}
\eeq 
we can rewrite this bounds using powers of heat kernels with the same parameter, for instance $\alpha$. This allows to show that for any constant $K$ such that
\beq
K > K_0 \equiv \left( \prod_{b \in \cB_{\ell_1}} \prod_{v \in V_{b}(\ell_2)} \langle v , b \rangle^{- 3/2}  \right)
\left( \prod_{v \in V_{\ell_3} \cup V_{\ell_4}}  |v|^{- 3/2} \right)
\eeq
we asymptotically have:
\bes\label{bound_v}
|\cA_{\cG}^{\alpha}| &\leq& K \, [K_{\alpha}(\one)]^{\gamma} \nn \\
\gamma &=& \sum_{b \in \cB_{\ell_1}} |V_{b}(\ell_2)| + |\cB_{\ell_3}| + |\cB_{\ell_4}| \,.
\ees

\

In addition to the bound (\ref{bound_v}), we will also need the following combinatorial lemma from \cite{vertex}:
\begin{lemma}\label{lemma_comb}
Let $\cG$ be a connected vacuum graph. Then
\beq
\forall \ell \neq \ell'\,, \; |\cB_{\ell'}| + |\cB_{\ell}| - {\sum_{b \in \cB_{\ell}}} |V_{b}(\ell')| \leq 1\,.
\eeq
\end{lemma}
\begin{proof}
Choose two colors $\ell \neq \ell'$. From $\cG$ we construct a connectivity graph $\cC_{\ell, \ell'}(\cG)$, whose vertices are the bubbles of color $\ell$ and $\ell'$. Then for any $b' \in \cB_{\ell'}$ and $b \in \cB_{\ell}$ we draw a line between them if and only if $b$ has a vertex dual to $b'$ in its triangulation. We call $L$ the number of lines of $\cC_{\ell, \ell'}(\cG)$, and $N$ its number of vertices. 
Now remark that the fact that $\cG$ is connected implies that $\cC_{\ell, \ell'}(\cG)$ is also connected. In fact, the bubbles of color $\ell'$ are all connected in $\cG$ by lines of color $\ell'$. But these lines are themselves part of bubbles of color $\ell$, which means that two bubbles of color $\ell'$ are connected if and only if their dual vertices appear in a same bubble of color $\ell$, that is if and only if they are connected to a same element in the graph $\cC_{\ell, \ell'}(\cG)$. So $\cC_{\ell, \ell'}(\cG)$ is connected. A maximal tree in this graph has $N -1$ lines, which implies the simple inequality: $N- 1 \leq L$.
To conclude, first notice that by construction $N$ is equal to $|\cB_{\ell'}| + |\cB_{\ell}|$. Still by construction, for any $b \in \cB_{\ell}$, $|V_{b}(\ell')|$ has to be greater than the number of lines ending on $b$ in $\cC_{\ell, \ell'}(\cG)$. Therefore:
\beq
|\cB_{\ell'}| + |\cB_{\ell}| - {\sum_{b \in \cB_{\ell}}} |V_{b}(\ell')| \leq N - L \leq 1\,.
\eeq
\end{proof}

\

In the following paragraphs, two kinds of bounds on the divergence degree $\gamma_{\cG}$ of a connected vacuum graph $\cG$ will be provided. The notion of divergence degree we adopt is:
\beq
\gamma_{\cG} = \inf\{ \gamma \in \mathbb{R} \, / \,  \lim (K_{\alpha}(\one)^{- \gamma} \cA_{\cG}^{\alpha}) < + \infty \}\,.
\eeq 
which measures the divergence behavior of a graph $\cG$ in units of $K_{\alpha}(\one) \sim {N_\alpha}^3$.	
	
	
		\subsection{Topological singularities}

\subsubsection{Bubble bounds}
	
Let us first derive the bubble bounds, which allow to analyze the behavior of topologically singular configurations. These were first derive in \cite{vertex}, but we follow here the more direct proof of \cite{edge}. Starting from (\ref{bound_v}), the derivation is straightforward. 
As a first step, we simply apply Lemma \ref{lemma_comb} with $(\ell , \ell') = (\ell_1 , \ell_3)$ and $(\ell , \ell') = (\ell_1 , \ell_4)$. This
gives:
\bes
\gamma_{\cG} &\leq& 2 - 2 |\cB_{\ell_{1}}| + {\sum_{b \in \cB_{\ell_1}}} \left(|V_{b}(\ell_2)| + |V_{b}(\ell_3)| + |V_{b}(\ell_4)|\right)\\
&=& 2 - 2 |\cB_{\ell_{1}}| + {\sum_{b \in \cB_{\ell_1}}} |V_{b}|\,.
\ees
We then use the definition of the genus of a bubble $b \in \cB_{\ell}$
\beq
2 - 2 g_b \equiv |V_{b}| - |E_{b}| + |F_{b}|\,,
\eeq
and the combinatorial relation between the number of edges and faces
\beq
2 |E(b)| = 3 |F(b)|
\eeq
to write:
\beq
{\sum_{b \in \cB_{\ell_1}}} |V_{b}| = 
{\sum_{b \in \cB_{\ell_1}}} \left(2 - 2 g_b + \frac{|F_{b}|}{2}\right) = 
2 |\cB_{\ell_1}| + \frac{\cN}{2} - 2 {\sum_{b \in \cB_{\ell_1}}} g_b \,.
\eeq
Therefore:
\beq
\gamma_{\cG} \leq 2 + \frac{\cN}{2} - 2 {\sum_{b \in \cB_{\ell_1}}} g_b \,.
\eeq
In order to obtain a uniform bound in the number of interaction vertices (dual tetrahedra), one must rescale the GFT coupling constant with appropriate powers of $N_\alpha$, namely:
\beq\label{rescaling_v}
\lambda \to \frac{\lambda}{\sqrt{ K_{\alpha} (\one) }} = \frac{\lambda}{ {N_{\alpha}}^{3/2} } \,. 
\eeq

We summarize this result in the following proposition:
\begin{proposition}\label{propbubble_v}
Using the rescaling (\ref{rescaling_v}), the divergence degree $\gamma_{\cG}$ of any connected vacuum graph $\cG$ of the colored Boulatov model verifies, for any color $\ell$:
\beq
\gamma_{\cG} \leq 2 - 2 {\sum_{b \in \cB_{\ell}}} g_b \,.
\eeq 
\end{proposition}

As an immediate corollary of the previous proposition, one can give a bound in terms of the number of pointlike singularities of a given color:
\begin{corollary}
With the rescaling of the coupling constant (\ref{rescaling_v}), the divergence degree $\gamma_{\cG}$ of any connected vacuum graph $\cG$ verifies, for any color $\ell$:
\beq
\gamma_{\cG} \leq 2 (1 - N^{s}_{\ell}) \,,
\eeq
where $N^{s}_{\ell}$ is the number of singular vertices of color $\ell$. 
\end{corollary}

In summary: in the perturbative expansion of the Boulatov GFT model, singular simplicial complexes are generically suppressed, uniformly in the number of tetrahedra. Provided that these bounds are saturated, we have just established the existence of a $1/N$ expansion in which the leading order graphs represent regular manifolds. We therefore now turn to the question of the optimality of the bubble bounds. 

\subsubsection{Optimality of the bounds}

In order to address this question, we need to be able to compute the exact power-counting of a sufficiently rich set of graphs. In this respect, we propose to first design elementary pieces of graphs which have one unique bubble of color $\ell$, and a certain number of external legs. We will then be able to build connected vacuum graphs with any kinds of bubbles out of these elementary graphs. Of course we want to keep the combinatorics of these elementary graphs rather simple, to be able to do exact calculations. 

It is then natural to start from minimal $3$-graphs representing $2$-dimensional orientable surfaces of a given topology. They are called \textit{canonical graphs} in the mathematical literature (see \cite{Vince_2d} and references therein). A canonical graph of genus $g$ has $2(2g+1)$ nodes. Figure \ref{bubbles_v} shows the canonical graphs of genus $0$, $1$ and their generalization to any genus $g$. We refer to \cite{Vince_2d} for proofs of these statements and further comments.

\begin{figure}[h]
\begin{center}
\includegraphics[scale=0.5]{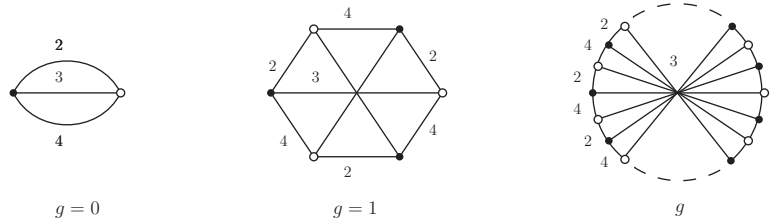} 
\caption{Canonical graphs for orientable surfaces: the sphere ($g=0$), the torus ($g=1$), and the general case of a genus $g$ surface.} \label{bubbles_v}
\end{center}
\end{figure}

We are ready to build our elementary graphs. For a given genus $g \in \mathbb{N}$ we start from the canonical graph shown in Figure \ref{bubbles_v}, and add external legs of color $1$ on every node. This gives a set of canonical bubbles with external legs, from which we can in principle construct any topology generated by the colored Boulatov model. 
To keep the combinatorics of the graphs simple, we will focus on linear chains of such bubbles. To this effect, we can first contract $2g$ pairs of external legs, keeping only two of them free, in each of these canonical bubbles. More precisely, for each genus $g$ we define the graph $\cC_{g}$ as shown in Figure \ref{bubbles_legs}. 

\begin{figure}[h]
\begin{center}
\includegraphics[scale=0.5]{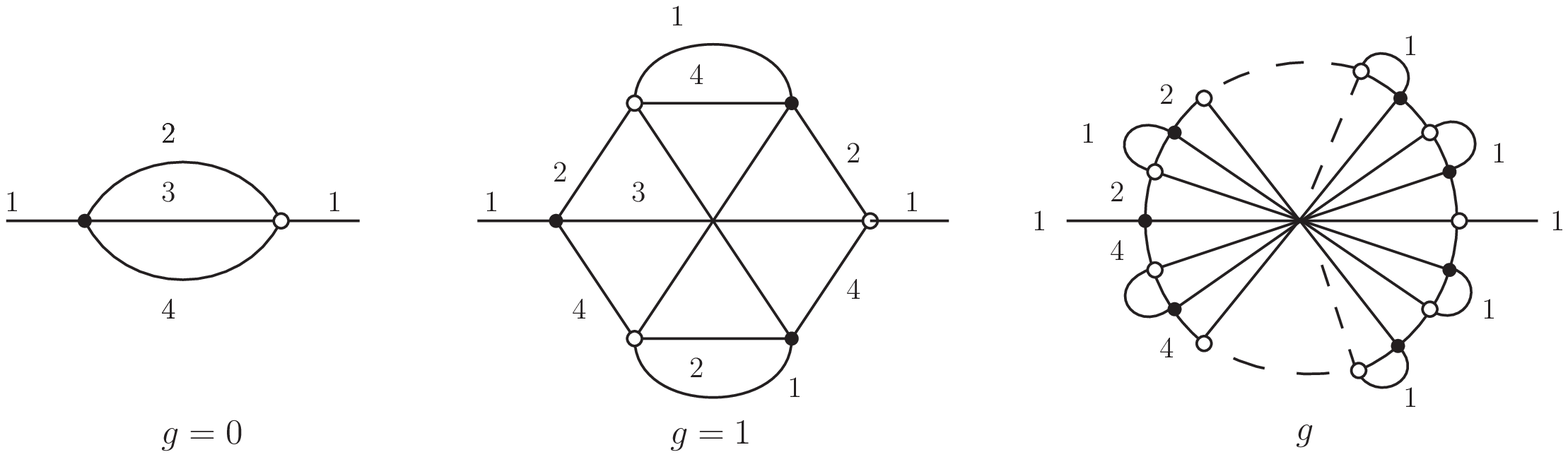} 
\caption{From left to right: $\cC_0$, $\cC_1$ and $\cC_g$.} \label{bubbles_legs}
\end{center}
\end{figure}

The reason why these graphs are useful lies in the following lemma (see Figure \ref{lemma_bubbles}):
\begin{lemma}\label{lemma_cg}
Let $\cG$ be a colored graph, with a subgraph $\cC_{g}$ for some $g \in \mathbb{N}$. Call $\cG_{\cC_{g} \rightarrow \cC_{0}}$ the graph obtained after replacement of $\cC_{g}$ by the graph $\cC_{0}$. Then, with the rescaling of the coupling constant (\ref{rescaling_v}), $\cA_{\cG}^{\alpha}$ scales like:
\beq
[K_{\alpha}(\one)]^{- 2 g} \cA_{\cG_{\cC_{g} \rightarrow \cC_{0}}}^{\alpha}\,.
\eeq
\end{lemma}
\begin{proof}
We first adopt a representation of the amplitude adapted to the color of the bubble we want to analyze. Let us assume that $\cC_g$ has external legs of color $1$, we then use inequality (\ref{amplitudes_v}) with $\ell = 1$, which is certainly an equality in the scaling sense. Now remark that a canonical triangulation of a bubble of genus $g$ has only three vertices. Indeed, dual of vertices (of the triangulation) of color $2$ are closed chains of strands alternatively of color $3$ and $4$. It is easy to see that there is only one such closed chain of strands in the canonical graph of genus $g$. So the dual triangulation has only one vertex of color $2$, and similarly for colors $3$ and $4$. This means that in the amplitude of $\cG$, the $\cC_{g}$ subgraph contributes with only three heat kernels associated to its dual vertices, and $2g$ $\delta$-functions associated to pairings of lines of color $1$. Moreover, these pairings are such that the arguments in the heat kernels associated to the vertices simplify upon integration (thanks to the convolution formula), and the contribution of $\cG_g$ to the amplitude of $\cG$ is equal, in the scaling sense\footnote{Since we are only interested in the scaling behavior here, we do not pay attention to the combinatorial factors resulting from the convolutions of heat kernels.}, to: 
\beq  
\int [\extd H]^{6 g} \left( [K_{\alpha} (\one)]^{2-2 g - 3} K_{\alpha}(G_{2} {\tilde{G}_{2}}^\inv) K_{\alpha}(G_{3} {\tilde{G}_{3}}^\inv) K_{\alpha}(G_{4} {\tilde{G}_{4}}^\inv) \right) \left( \prod_{i=1}^{2 g} K_{\alpha}(H^{(i)}_2 H^{(i)}_3  H^{(i)}_4 )\right) \,,
\eeq
where the variables $G_{\ell'}$ and $\tilde{G}_{\ell'}$ are that of the two external legs of $\cG_g$. The $H^{(i)}_{\ell'}$ are associated to the $2 g$ remaining lines of color $1$, and can be integrated. We obtain a term:
\beq
[K_{\alpha}(\one)]^{-2 g - 1} K_{\alpha}(G_{2} {\tilde{G}_{2}}^\inv) K_{\alpha}(G_{3} {\tilde{G}_{3}}^\inv) K_{\alpha}(G_{4} {\tilde{G}_{4}}^\inv)\,,
\eeq
which reduces to
\beq
[K_{\alpha}(\one)]^{- 1} K_{\alpha}(G_{2} {\tilde{G}_{2}}^\inv) K_{\alpha}(G_{3} {\tilde{G}_{3}}^\inv) K_{\alpha}(G_{4} {\tilde{G}_{4}}^\inv)
\eeq
when $g = 0$. These two terms differ by a factor $[K_{\alpha}(\one)]^{-2 g}$, which concludes the proof.
\end{proof}

\begin{figure}[h]
\begin{center}
\includegraphics[scale=0.5]{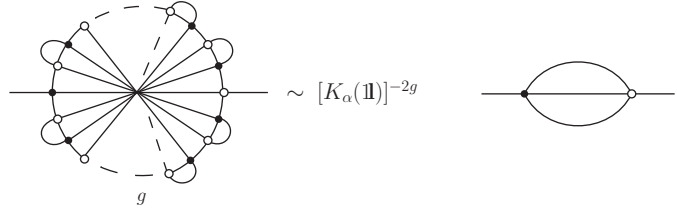} 
\caption{Graphical representation of lemma $1$.} \label{lemma_bubbles}
\end{center}
\end{figure}

This is all we need in order to saturate the bubble bounds. We call $\cC_{g_1,...,g_{n}}$ the chain of $n$ graphs $(\cC_{g_{1}}, \cdots, \cC_{g_n})$ as represented in Figure \ref{chains}. Chains of $\cC_{0}$ graphs being maximally divergent spheres \cite{scaling3d}, this suggests that the chain $\cC_{g_1,...,g_{n}}$ could be a dominant graph in the class of graphs with singularities $(g_{1}, ..., g_{n})$. So let us first compute these amplitudes.

\begin{figure}[h]
\begin{center}
\includegraphics[scale=0.5]{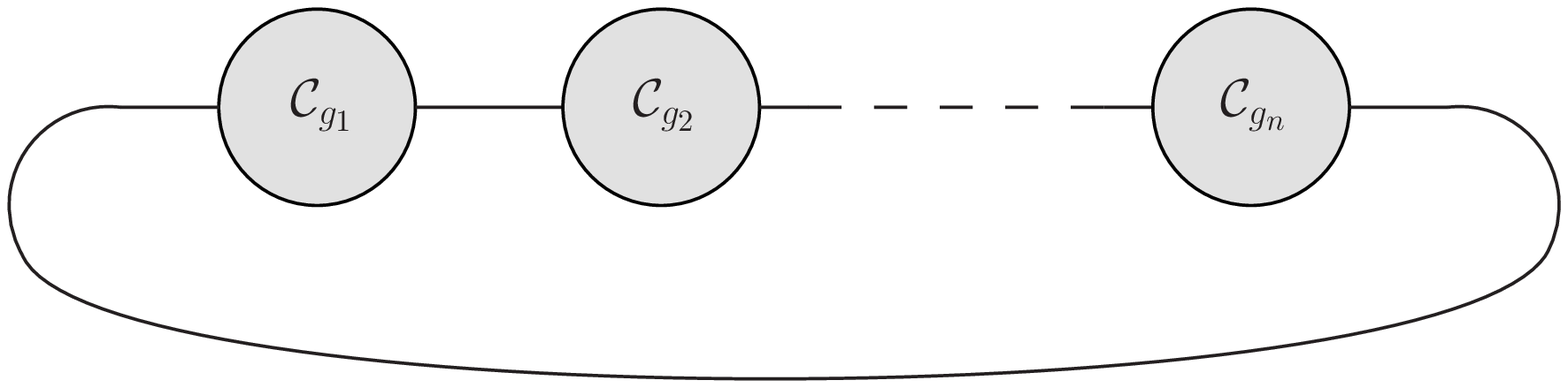} 
\caption{Chain $\cC_{g_1,...,g_{n}}$.} \label{chains}
\end{center}
\end{figure}

\begin{lemma}\label{lemma_chains} Let $n \in \mathbb{N^*}$, and $g_{1}, ..., g_{n} \in \mathbb{N}$. Then, with the rescaling of the coupling constant (\ref{rescaling_v}), $\cA_{\cC_{g_1,...,g_{n}}}^\alpha$ scales as:
\beq
[K_\alpha (\one)]^{2 - 2 {\sum_{i = 1}^{n}} g_i}\,.
\eeq 
\end{lemma}
\begin{proof}
Lemma \ref{lemma_cg} ensures that $\cC_{g_1,...,g_{n}}$ behaves like $\cC_{0,...,0}$ times $[K_\alpha (\one)]^{-2 {\sum_{i = 1}^{n}} g_i}$. As proved in \cite{scaling3d}, $\cC_{0,...,0}$ is dual to a sphere and maximally divergent. A way to see it here is to remark that with the scaling we chose for the coupling $\lambda$, a $\cC_{0}$ subgraph whose two external strands are not paired behaves like a propagator (straightforward calculation), so that we can replace all subgraphs $C_{0}$ but one by propagators. We are left with the simplest graph corresponding to a sphere, that is again the melon graph which behaves like $[K_\alpha (\one)]^{2}$. 
\end{proof}

\

This concludes the study of bubble bounds for the Boulatov model. We have confirmed that, upon rescaling of the coupling constant as given by (\ref{rescaling_v}), a $1/ N_\alpha$ expansion can be defined. The first term of this expansion, scaling like $[K_\alpha (\one)]^{2}$, only contains regular manifolds, while topological singularities are tamed according to Proposition \ref{propbubble_v}. This establishes the existence of a regime of the GFT in which the local topology of space-time emerges from a richer set of possibilities. To go further, and make contact with the way the existence of the $1/N$ expansion was first established, we focus next on another type of bounds, associated to jackets rather than bubbles.


	
		
		\subsection{Domination of melons}

\subsubsection{Jacket bounds}

The jacket bounds are crucial to the $1/N$ expansion. We will not only show that they can also be deduced from our framework, in a rather straightforward way, but also that the
vertex reformulation of the model allows to derive a stronger bound than the original one \cite{RazvanN}. This will also prove that the next-to-leading orders of the Boulatov model are most probably not populated by the same graphs as in the i.i.d. tensor model.

\

Jackets are two dimensional closed and orientable surfaces
embedded in a simplicial complex \cite{jimmy}. In the simplicial complex dual to a 4-colored graph $\cG$, we have three different jackets, each labeled by a pair $(\sigma , \sigma^{\inv})$ of cyclic
permutations of the color set \cite{FerriGagliardi, Vince_gene}. We can identify them as follows. We assign a color $(\ell \ell')$ to any edge between two vertices of colors $\ell$ and $\ell'$. Then,
to each
tetrahedron, one associates a rectangle, whose edges are so constructed: for any color $\ell$, the middle point of the edge of the tetrahedron of color $(\sigma(\ell) \sigma(\ell + 1))$ is
joined to the middle point of the edge of color $(\sigma(\ell +1) \sigma(\ell + 2))$\footnote{The addition in the color set is of course understood modulo $4$.}. An example is given in Figure \ref{jacket_v}.
Gluings of tetrahedra induce gluings of
these elementary rectangles, providing a quadrangulated surface: this is a jacket. The construction just outlined also makes clear why there are three possible jackets that can be embedded in the
simplicial complex. Also, one can show that these three jackets correspond to three possible reductions of the Boulatov model to a matrix model by reduction with respect to the diagonal gauge
invariance (closure constraint) at the level of the action \cite{jimmy}.

\begin{figure}[h]
\centering
\includegraphics[scale=0.5]{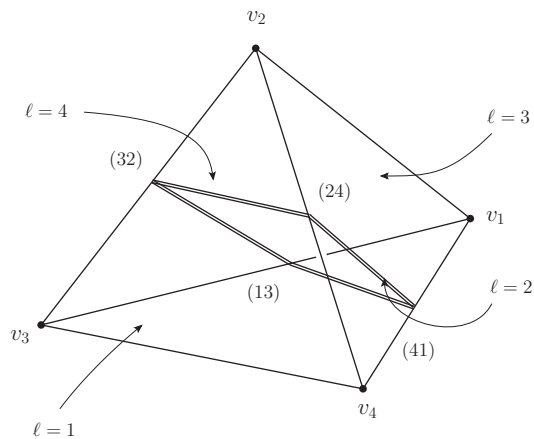}
\caption{Elementary building block of the jacket associated to the cycle $\sigma = (1\,4\,2\,3)$ (double lines).}
\label{jacket_v}
\end{figure}

\
Before moving on to the computation of the bound, we need to express the genus of a jacket $J = (\sigma , \sigma^{\inv})$ in terms of combinatorial quantities of $\cG$. This is the purpose of the
following lemma:
\begin{lemma}
The jacket of color $J = ( \sigma , \sigma^{\inv} )$ of a colored graph $\cG$ has genus:    
\beq\label{g_jacket}
g_J = 1 + \frac{1}{2} \left( |T| - \sum_{\ell} |E(\sigma(\ell) \sigma(\ell + 1))| \right) \,,
\eeq
with $T$ and $E(i)$ respectively the sets of tetrahedra and edges of color $i$ in the simplicial complex dual to $\cG$. 
\end{lemma}
\begin{proof}
The jacket $J$ is an orientable surface, hence its genus is related to its Euler characteristic by:
\beq
2 - 2 g_J = \chi_J = |\cF_J| - |\cE_J| + |\cV_J| \,,
\eeq 
where $\cF_J$,$\cE_J$ and $\cV_J$ are the sets of faces, edges and vertices of $J$. But, by construction:
\beq
|\cF_J| = |T| \;, \qquad |\cE_J| = |t| \;, \qquad |\cV_J| = \sum_{\ell} |E(\sigma(\ell) \sigma(\ell + 1))| \;,
\eeq
where $t$ is the set of triangles in the simplicial complex. Since each triangle is shared by two tetrahedra, we also have $|t| = 2 |T|$, and the result follows.
\end{proof}

\

Let us now consider a connected vacuum graph $\cG$, and one of its jackets $J = (\sigma , \sigma^{\inv})$. We can use the previous lemma to write $g_J$ with data suitable to bubble factorizations.
Indeed, for any distinct colors $\ell$ and $\ell'$, one immediately has:
\beq
|E(\ell \ell')| = \sum_{b \in \cB_{\ell}} |V_{b}(\ell')|\,, 
\eeq 
since to any bubble $b$ dual to the vertex $v_\ell$, and vertex $v_{\ell'} \in V_{b}(\ell')$, one uniquely associates the edge $(v_{\ell} v_{\ell'}) \in E(\ell \ell')$.
Therefore:
\beq
g_J = 1 + \frac{\cN}{2} - \frac{1}{2} \sum_{\ell} \sum_{b \in \cB_{\sigma(\ell)}} |V_{b}(\sigma(\ell + 1))|  \,.
\eeq

We can now try to make $g_J$ appear in the bounds we computed so far. Applying Lemma \ref{lemma_comb} to (\ref{bound_v}), with $(\ell , \ell') = (\ell_3 , \ell_3)$, we obtain:
\beq\label{ineq_3d}
\gamma_\cG \leq 1 + {\sum_{b \in \cB_{\ell_1}}} |V_{b}(\ell_2)| + {\sum_{b \in \cB_{\ell_3}}} |V_{b}(\ell_4)|\,.
\eeq
Averaging this expression and
\beq
\gamma_\cG \leq 1  + {\sum_{b \in \cB_{\ell_2}}} |V_{b}(\ell_3)| + {\sum_{b \in \cB_{\ell_4}}} |V_{b}(\ell_1)|\,,
\eeq
then yields:
\beq
\gamma_\cG \leq 2 + \frac{\cN}{2} - g_J\,, 
\eeq
with $\sigma = (\ell_1 \, \ell_2 \, \ell_3 \, \ell_4)$. As for the bubble bound, a uniform jacket bound in the number of GFT interaction vertices (tetrahedra of the simplicial complex) is
obtained by a simple rescaling of the coupling constant, which turns out to be identical to (\ref{rescaling_v}). This is not surprising, since from our analysis of the bubble bounds, this rescaling is the unique one supporting a non-trivial $1/N$ expansion. We
summarize our result in the following proposition: 
\begin{proposition}
With the rescaling of the coupling constant (\ref{rescaling_v}), the divergence degree $\gamma_{\cG}$ of any connected vacuum graph $\cG$ verifies, for {\it any} of its jackets $J$:
\beq
\gamma_{\cG} \leq 2 - g_J \,.
\eeq 

In particular, the following bound holds:
\beq\label{decay_jacket_v}
\gamma_{\cG} \leq 2 - \sup_{J} g_J \,.
\eeq 
\end{proposition}

We note that, as anticipated, this bound is stronger than the usual jacket bound, proven in \cite{RazvanN}:
\beq
\gamma_{\cG} \leq 2 - \frac{1}{3} \sum_{J} g_J \,.
\eeq 

\

We know already that, if a 3d complex has a jacket with genus zero, the complex is of spherical topology (trivial fundamental group) \cite{Gurau:2011xp}. Therefore, we conclude that the leading order graphs in the $1/N$ expansion not only index regular manifolds, as could be deduced from the bubble bounds, but also spherical ones. Moreover, since these leading order graphs also have degree $0$, they must be melonic, which proves that the same graphs populate the leading orders of the colored Boulatov model and of the i.i.d. rank-$3$ colored tensor model. We on the contrary expect to find differences at some point in the sub-dominant contributions, since the decay (\ref{decay_jacket_v}) is stronger than a decay in the degree of the graphs. 
 
We conclude by noting also that the same bound could give further insights into the topology of the higher order terms of the same expansion, due to the following fact.
Just as we know that $g = \inf_{\cG , J} g_J$ is a topological invariant, called regular genus, similarly $\tilde{g} = \inf_{\cG} \sup_{J \in \cG} g_J$ is also well-defined, and by definition a topological
invariant (the $\inf$ is taken over the equivalence class of graphs representing a given topology). If $\tilde{g}$ and $g$ are not identical, then our results allow to derive a non-trivial
topological bound in terms of $\tilde{g}$.

\subsubsection{1-dipole contractions}

We conclude this presentation of the $1/N$ expansion of the Boulatov model with a discussion of $1$-dipole contractions, as seen from the vertex variables. Dipole moves are used in crystallization theory \cite{FerriGagliardi} to map colored graphs encoding the same topology to one another. In tensor models, $1$-dipoles are of particular relevance, because they leave the amplitudes invariant (see for instance \cite{Gurau:2011xp}). In the Boulatov model, they only conserve the scalings of the amplitudes with the cut-off, not their exact values, which is the property on which Gurau relied in its derivation of the $1/N$ expansion \cite{RazvanN}. In our approach, it is again thanks to this approximate invariance that we can verify that melonic graphs indeed all saturate the jacket bounds, and are therefore the leading order graphs. A $1$-dipole of color $\ell$ is a line of same color which joins two nodes having no other line in common (see the left side of Figure \ref{dipole_largeN}). It is moreover said to be non-degenerate when it separates two distinct bubbles (of color $\ell$). Its contraction amounts to deleting the color-$\ell$ line and its two boundary nodes, and reconnecting the $6$ open lines thus created according to their colors. See Figure \ref{dipole_largeN}. In this chapter, we are only interested in non-degenerate $1$-dipoles, so we will often omit the qualifier.   

\begin{figure}[h]
\begin{center}
\includegraphics[scale=0.7]{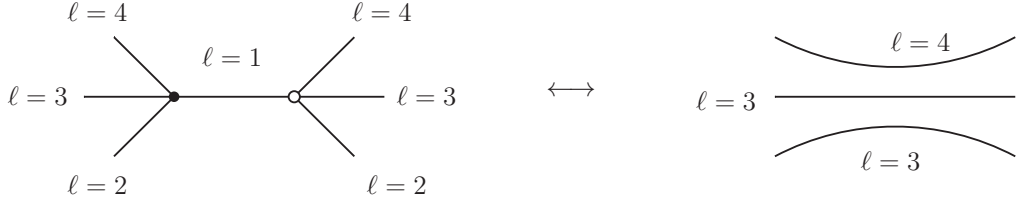} 
\caption{A $1$-dipole of color $1$ (left), and its contraction (right).} \label{dipole_largeN}
\end{center}
\end{figure}

\
Let us consider a graph $\cG$, and two different bubbles $b_1$ and $b_2$ in $\cB_{\ell}$, glued through a triangle $f_{0} \in F_{b_1} \cap F_{b_2}$. We make in addition the assumption that at least one of them, say $b_1$, is a sphere. The contribution of the two bubbles to the amplitude $\cA_\cG^\alpha$ is given by a factor of the form:
{\footnotesize\bes
\left( [K_\alpha (\one)]^{2 - |V_{b_1}|} \prod_{v \in V_{b_1}, v \notin f_0} K_{\alpha} \left( \overrightarrow{\prod_{f \in \triangle^{b_1}_{v}}} (G_{v}^{f})^{\epsilon^{f}_{v}}\right) \right)  \left( [K_{\alpha} (\one)]^{2-2 g_{b_2} - |V_{b_2}|} \prod_{v \in V_{b_2}, v \notin f_0} K_\alpha \left( \overrightarrow{\prod_{f \in \triangle^{b_2}_{v}}} (G_{v}^{f})^{\epsilon^{f}_{v}}\right)\right) \nn \\
\times \int \extd G_{u_1}^{f_0} \extd G_{u_2}^{f_0} \extd G_{u_3}^{f_0}  \, \delta \left( G_{u_{1}}^{f_0} G_{u_2}^{f_0} G_{u_3}^{f_0}\right) \prod_{i = 1}^{3} K_{\alpha}\left( \overrightarrow{\prod_{f \in \triangle^{b_1}_{u_i}}} (G_{u_i}^{f})^{\epsilon^{f}_{v}} \right) K_{\alpha}\left( \overrightarrow{\prod_{f \in \triangle^{b_2}_{u_i}}} (G_{u_i}^{f})^{\epsilon^{f}_{v}} \right) \, , \nn
\ees}
where $u_1$, $u_2$ and $u_3$ are the vertices of $f_0$. Before integrating with respect to $G_{u_{i}}^{f_0}$, we would like to get rid of $\delta\left( G_{u_{1}}^{f_0} G_{u_2}^{f_0} G_{u_3}^{f_0}\right)$, which imposes the closure of the triangle $f_0$. Because we are only interested in the scaling of $\cA_\cG^\alpha$, we can assume that this constraint is also imposed by a heat kernel at time of order $\alpha$. Using the other closure and flatness constraints in $b_1$, we see that it is equivalent to saying that the holonomy along a path circling $f_0$ in $b_1$ has to be flat (see Figure \ref{dipole_contraction_v}), again up to the width of the heat kernels. Iterating the process shows that this path can actually be deformed arbitrarily. But $b_1$ is a sphere, hence simply connected. We can therefore contract the path around another triangle of $b_1$, and write the constraint  $G_{u_{1}}^{f_0} G_{u_2}^{f_0} G_{u_3}^{f_0} = \one$ as the closure condition in this triangle. We thus see that $K_\alpha \left( G_{u_{1}}^{f_0} G_{u_2}^{f_0} G_{u_3}^{f_0}\right)$ is redundant and can be set to $K_\alpha \left( \one\right)$ without changing the scaling of the integral.
\begin{figure}[h]
\begin{center}
\includegraphics[scale=0.5]{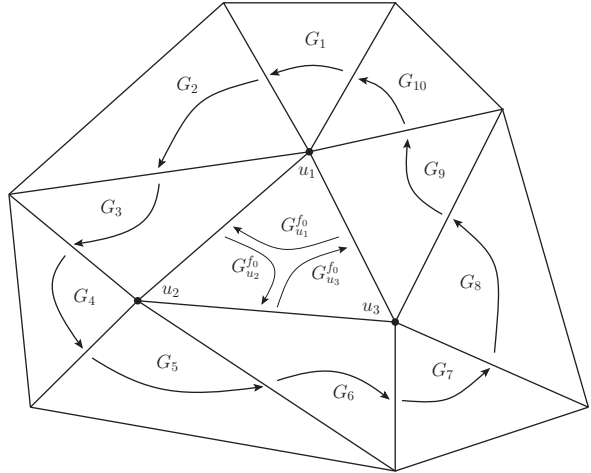} 
\caption{Triangle $f_0$ and its neighbors in $b_1$. Using flatness around $u_1$, $u_2$, $u_3$, and constraints in the three triangles sharing an edge with $f_0$, we show that $G_{u_{1}}^{f_0} G_{u_2}^{f_0} G_{u_3}^{f_0} = G_1 G_2 \cdots G_9 G_{10}$.} \label{dipole_contraction_v}
\end{center}
\end{figure}
We can now safely integrate the $G_{u_{i}}^{f_0}$ variables, which corresponds to removing $f_{0}$ and taking the connected sum of $b_1$ and $b_2$. We denote this connected sum $b_1 \# b_2$ and refer for instance to \cite{Vince_2d, francesco} for more details. This leads to:
\bes
\left( [K_\alpha (\one)]^{2 - |V_{b_1}| + 2 - |V_{b_2}| - 2 g_{b_2} } \prod_{v \in V_{b_1 \# b_2}} K_\alpha \left( \underset{f \in \triangle^{b_1 \# b_2}_{v}}{\overrightarrow{\prod}} (G_{v}^{f})^{\epsilon_{f}}\right) \right)  \times  K_\alpha \left( \one \right) \,.\nn
\ees
But we also have that (the second equality crucially depends on $b_1$ being a sphere):
\bes
|V_{{b_1} \# {b_2}}| = |V_{b_1}| + |V_{b_2}| - 3 \nn \\
g_{{b_1} \# {b_2}} = g_{b_2} \, ,
\ees
so that in the end the contribution of the two bubbles is that of their connected sum. We recover the relation found in \cite{RazvanN}, between the amplitude of the initial graph $\cG$ and the one after absorption of the planar bubble $b_1$ in $b_2$, noted $\cG_{b_{1} \rightarrow b_{2}}$:
\beq
\cA_{\cG}^{\alpha} \approx \cA_{\cG_{b_{1} \rightarrow b_{2}}}^{\alpha}\,,
\eeq
where the $\approx$ sign is meant to denote that the equality holds only in the scaling sense.

This shows that the amplitudes are invariant under $1$-dipole contractions in the power-counting sense. This can be applied to melonic graphs, which can be reduced to the elementary melon after a complete set of $1$-dipoles have been contracted. It is then easy to check that this elementary melon scales like $[K_\alpha (\one)]^2$, and thereby confirm that the leading order contributions of the $1/N$ expansion exactly consist in the melonic sector.

		
		
		
	\section{Colored Ooguri model}
		
		The purpose of this section is to extend the previous results to four dimensions, namely to the colored version of the Ooguri model \cite{Ooguri , cboulatov}. Just like the Boulatov model, it is a GFT quantization of $BF$ theory, this time with the group $\SO(4)$ (or $\SO(3) \times \SO(3)$ for simplicity). Because 4d gravity models are constructed by constraining the data appearing in such
a model, either at the level of the GFT action or directly at the level of its Feynman amplitudes, we see the results presented in this section as a first step towards performing a similar analysis in
4d gravity models.

\
We will first show that the colored Ooguri GFT model, usually formulated in terms of group-theoretic data associated to triangles in $4$-dimensional simplicial complexes \cite{cboulatov}, can equivalently be written with data associated to edges in the same simplicial complexes. 
Similarly to the Boulatov model, such a formulation will allow to factorize the amplitudes in terms of bubbles (here the $4$-bubbles), and to use new computation tools to derive bounds on the
regularized amplitudes. 
The two main results of this construction will be again: a) a bound on topologically singular vertices, resulting in a clear separation between leading order graphs corresponding to regular manifolds
and sub-dominant graphs associated to non-manifold configurations; b) a new proof and an improvement of the so-called jacket bound \cite{RazvanN,RazvanVincentN,Gurau:2011xq}, which moreover does not rely on topological moves (dipole contractions). 
		
		
		\subsection{Edge variables}
		
		\subsubsection{Action and partition function}
		
The colored Ooguri model is a field theory of five complex scalar fields $\{\vphi_\ell\ , \ell=1, \ldots ,5\}$, each of them defined over four copies of $\SO(4)$, which respect the following 
gauge invariance condition:
\beq \label{gaugeN}
\forall h \in \SO(4),  \qquad \vphi_\ell(hg_1, hg_2, hg_3, hg_4)  \, = \, \vphi_\ell(g_1, g_2, g_3, g_4) .
\eeq
Like in three dimensions, they are interpreted as quantized building blocks of spatial geometry, here tetrahedra. The $\SO(4)$ variables are interpreted as parallel transports of an 
$\SO(4)$ connection from the center of the tetrahedra to the centers of their boundary triangles. The action encodes the gluing of five tetrahedra to form 
a four-simplex via the interaction term, while the kinetic parts mimic the identification of two tetrahedra along their boundary triangles: 
\bes \label{action}
S[\vphi ]&=& S_{kin} [\vphi ] + S_{int}[\vphi ] ,\\
S_{kin} [\vphi] &=& \frac{1}{2} \int [\extd g_i]^4 \sum_{\ell=1}^5   \, \vphi_\ell(g_1, g_2, g_3, g_4) \overline{\vphi_{\ell}}(g_1, g_2, g_3, g_4) ,\\
S_{int}[\vphi] &=& \lambda \int [\extd g_{i} ]^{10} \, \vphi_1(g_1, g_2, g_3, g_4) \vphi_2(g_4, g_5, g_6, g_7) 
\vphi_3(g_7, g_3, g_8, g_9) \nn \\
&& \vphi_4(g_9, g_6, g_2, g_{10}) \vphi_5(g_{10}, g_8, g_5, g_1) + \; \; {\rm{c.c}}. 
\ees
A graphical representation of the two terms of this action is given in Figure \ref{vertex_propa_e}.

\begin{figure}[h]
  \centering
 \includegraphics[scale=0.5]{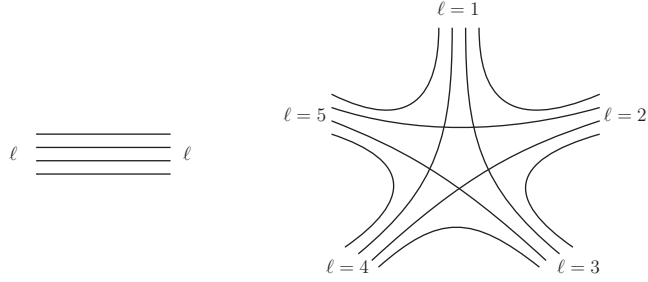}
  \caption{Combinatorics of the kinetic (left) and interaction (right) kernels in the usual (triangle) formulation.}
  \label{vertex_propa_e}
\end{figure}

There is also a metric representation \cite{Baratin:2010wi} of the same model in terms of Lie algebra variables, obtained via a group Fourier transform, this time for the group $\SO(4) \simeq ( \SU(2) \times
\SU(2) ) /
\mathbb{Z}_2$. As in the 3d case, we will however restrict ourselves to $\SO(3) \times \SO(3)$, for which a simple invertible group Fourier transform is available, and the generalized framework
\cite{karim} is not needed. As before, functions on $\SO(3)$ will be identified with functions $f$ on $\SU(2)$ such that $f(g) = f(-g)$. We adopt the notation $g = (g^{+}, g^{-}) \in \SU(2)
\times \SU(2)$,
and similarly for Lie algebra elements, and introduce the plane-waves
\beq
\forall g \in \SU(2) \times \SU(2) \,, \qquad \E_g \maps \, (\su(2) \times \su(2)) \ni x \mapsto \e_{g^{+}}(x^{+}) \, \e_{g^{-}}(x^{-})\,.
\eeq 
The group Fourier transform is given by
\beq
\widehat{f}(x) \equiv \int \extd g f(g) \E_{g}(x)\,, 
\eeq 
and sends the convolution product on $L^{2}(\SO(3) \times \SO(3))$ to a $\star$-product, defined on plane-waves as:
\beq
\forall g_1\,, g_2 \in \SU(2) \times \SU(2) \,, \qquad \E_{g_1} \star \E_{g_2} \equiv \E_{g_{1} g_{2}}\,.
\eeq
Besides these definitions, we will explicitly use the fact that the $\delta$-distributions on $\SO(3) \times \SO(3)$ can be decomposed in plane-waves as:
\beq
\delta(g) = \int \extd x \, \E_{g}(x)\,.
\eeq 
Dually, a non-commutative $\delta$-distribution on $\so(3) \times \so(3)$ can be defined by:
\beq
\delta (x) = \int \extd h \, \E_{h}(x)\,,
\eeq
which verifies, for any algebra function $f$:
\beq
\left( \delta \star f \right) (x) = f(0) \delta(x)\,.
\eeq

\

In analogy with what has been done in three dimensions, we would like to use the gauge invariance condition to re-express the action in terms of fields whose arguments are associated to simplices of
one dimension less: in this case, from triangles to edges. 
The main difference with the 3d situation however, is that the numbers of edges and triangles in a tetrahedron do not match: a tetrahedron consists in four
triangles,  but six edges. 

At the level of a field $\vphi_\ell(g_1, g_2, g_3, g_4)$, whose variables are associated to 4 different triangles, this translates into the fact that one can construct six independent
edge variables $G_{ij}$ from pairs of triangle variables $g_i$. For example: 
\bes
G_{41} = g_4^{\inv} g_1 \;, &\qquad G_{42} = g_4^{\inv} g_2 \;, &\qquad G_{43} = g_4^{\inv} g_3 \;, \nn \\
G_{12} = g_1^{\inv} g_2 \;, &\qquad G_{23} = g_2^{\inv} g_3 \;, &\qquad G_{31} = g_3^{\inv} g_1 \;.
\ees
This means that, in order to match the number of degrees of freedom in the two representations, we will have to use more constraints than in the Boulatov model. These constraints will reflect
geometrical 
conditions on the holonomies in a tetrahedron. Remarking that the variable $G_{ij}$ represents the holonomy from the center of the triangle $i$ to the center of the triangle $j$, we see that
for any distinct indices $i$, $j$ and $k$, we have:
\beq
G_{ij} G_{jk} G_{kl} = \one \,,
\eeq 
where from now on we use the notation: $G_{ij} \equiv G_{ji}^{\inv}$. There are a priori four such equations to impose (one for any triplet $\{ i, j, k\}$, i.e. one for any vertex of the tetrahedron).
However, only three of them are independent, since for example:
\beq
\left\{
    \begin{array}{lll}
        G_{12} G_{24} G_{41} &=& \one \\
        G_{23} G_{34} G_{42} &=& \one \; \Longrightarrow \; G_{12} G_{23} G_{31} = \one \\
	G_{31} G_{14} G_{43} &=& \one
    \end{array}
\right. 
\eeq

\
All this suggests to introduce new fields $\psi_\ell : \SO(4)^{\times 6} \rightarrow \mathbb{C}$, implicitly defined by\footnote{From now on, we will denote by $\delta$ the $\delta$-distribution on
$\SO(3) \times \SO(3)$, and by $\delta^{\SO(3)}$ (resp. $\delta^{\SU(2)}$) its counterpart on $\SO(3)$ (resp. $\SU(2)$).}:
\bes\label{naive}
\vphi_\ell(g_1, g_2, g_3, g_4) &=& \vphi_\ell(G_{41}, G_{42}, G_{43}, \one) \nn \\
 &\equiv& \int \extd G_{12} \extd G_{23} \extd G_{31} 
\delta(G_{12} G_{24} G_{41}) \delta(G_{23} G_{34} G_{42}) \delta(G_{31} G_{14} G_{43}) \\
&\times& \psi_\ell(G_{41}, G_{42}, G_{43}, G_{12}, G_{23}, G_{31})\,. \nn
\ees
This idea can be made precise using the group Fourier transform previously introduced. For any function $\phi \in L^{2}((\SO(3) \times \SO(3))^{\times 4})$, we define a function $\Upsilon[\phi]$ of
six $\so(3)
\times \so(3)$ elements:
\bes
\Upsilon[\phi]( x_{41}, x_{42} , x_{43} , x_{12} , x_{23} , x_{31})
&\equiv& \int [\extd g_i ] \, \phi(g_1 , g_2 , g_3 , g_4) \E_{{g_4}^{\inv} g_1}(x_{41}) \E_{{g_4}^{\inv} g_2}(x_{42}) \E_{{g_4}^{\inv} g_3}(x_{43}) \nn \\ 
&& \qquad \E_{{g_1}^{\inv} g_2}(x_{12}) \E_{{g_2}^{\inv} g_3}(x_{23}) \E_{{g_3}^{\inv} g_1}(x_{31})\,.
\ees
As in the 3d case, such a function is invariant under a simultaneous (deformed) translation of all its arguments; in this case, however, no geometric interpretation of the variables appearing as field
arguments nor of such invariance can be given, due to the fact that we are not dealing with geometric tetrahedra, but simply with combinatorial simplices to which variables from the classical phase
space of discrete BF theory are associated. The variables do not describe a geometric tetrahedron, and the translation symmetry of each field cannot be interpreted as the  translation of the vertices
(or the edges) of the tetrahedra in some embedding into $\mathbb{R}^4$; in fact, it is generated by a Lie algebra element of $\so(4)$ and not by a vector in $\mathbb{R}^4$. It is again the manifestation of the shift symmetry of $BF$ theory, but contrary to what happens in three dimensions, it cannot be traced back to a geometric action on the elements of the triangulations. In the GFT context, it has been studied in details in \cite{diffeos}, both in the triangle and edge formulations.
And as in the 3d case, a proper description of the invariance of the field under such deformed translations (in the Lie algebra) requires that we interpret the products of plane-waves as tensor
products, taken in
the order in which we wrote them. 

With this convention in mind, $\Upsilon[\phi]$ is invariant under the following symmetries (only three of them being independent):
\bes
\Upsilon[\phi] &\mapsto& \cT^{142}_{\ve} \act \Upsilon[\phi](x_{ij}) = \bigstar_{\ve} \Upsilon[\phi](x_{41} - \ve, x_{42} + \ve, x_{43} , x_{12} - \ve, x_{23} , x_{31} )\,, \\
\Upsilon[\phi] &\mapsto& \cT^{243}_{\ve} \act \Upsilon[\phi](x_{ij}) = \bigstar_{\ve}
\Upsilon[\phi](x_{41}, x_{42} - \ve, x_{43} + \ve , x_{12} , x_{23} - \ve, x_{31} )\,, \\
\Upsilon[\phi] &\mapsto& \cT^{143}_{\ve} \act \Upsilon[\phi](x_{ij}) = \bigstar_{\ve}
\Upsilon[\phi](x_{41} - \ve, x_{42}, x_{43} + \ve, x_{12}, x_{23} , x_{31} + \ve)\,, \\
\Upsilon[\phi] &\mapsto& \cT^{123}_{\ve} \act \Upsilon
[\phi](x_{ij}) = \bigstar_{\ve}
\Upsilon[\phi](x_{41}, x_{42}, x_{43} , x_{12} + \ve, x_{23} + \ve , x_{31} + \ve)\,.
\ees
These transformations correspond to a simultaneous translation of the Lie algebra variables associated to three edges sharing a vertex in a quantum tetrahedron by the same Lie algebra variables
associated to such common vertex. For instance, $\cT^{142}$ translates the three edges sharing the vertex
of color $3$.

\
Let us call the space of such invariant fields $\mathbb{T} \equiv {\rm{Im}}(\Upsilon)$, and $\mathbb{D}={\rm{Inv}}((\SO(3) \times \SO(3))^{\times 4})$ the space of fields in $L^{2}((\SO(3) \times
\SO(3))^{\times 4})$ that satisfy the
gauge invariance (\ref{gaugeN}). We now prove that the map from triangle group variables to the Lie algebra edge variables is one-to-one.
\begin{proposition}\label{bij_4d}
$\Upsilon$ is a bijection between $\mathbb{D}$ and $\mathbb{T}$. Its inverse maps any $\widetilde{\phi} \in \mathbb{T}$ to: 
\bes
\Upsilon^{\inv}[\widetilde{\phi}](g_i)
&\equiv& \int [\extd x_{ij}]^{3} 
\left( E_{g_1^{\inv} g_4}(x_{41}) E_{g_2^{\inv} g_4}(x_{42}) E_{g_3^{\inv} g_4}(x_{43}) E_{g_2^{\inv} g_1}(x_{12}) E_{g_3^{\inv} g_2}(x_{23}) E_{g_1^{\inv} g_3}(x_{31})\right) \nn \\
&& \qquad \star
\, \widetilde{\phi}(x_{ij})\,,
\ees
where only three $x_{ij}$ are being integrated, the others being fixed to any value.
\end{proposition}
\begin{proof}
Let us call $\widetilde{\Upsilon}$ the map defined by the previous formula, and show that $\widetilde{\Upsilon} \circ \Upsilon$ and $\Upsilon \circ \widetilde{\Upsilon}$ are the identity.

\
We first choose $\phi \in \mathbb{D}$, and check that $\widetilde{\Upsilon} \circ \Upsilon [\phi] = \phi$. Using the definitions, we immediately have:
\bes
\widetilde{\Upsilon} \circ \Upsilon [\phi](g_i) = \int [\extd g_{i}'] \phi(g_i') 
\int [\extd x_{ij}]^{3} E_{g_1^{\inv} g_4 g_4'^{\inv} g_1'}(x_{41}) E_{g_2^{\inv} g_4 g_4'^{\inv} g_2'}(x_{42}) E_{g_3^{\inv} g_4 g_4'^{\inv} g_3'}(x_{43}) \\
E_{g_2^{\inv} g_1 g_1'^{\inv} g_2'}(x_{12}) E_{g_3^{\inv} g_2 g_2'^{\inv} g_3'}(x_{23}) E_{g_1^{\inv} g_3 g_3'^{\inv} g_1'}(x_{31})
\ees
The integration over the $x_{ij}$ give three $\delta$-functions. For example, if we choose $x_{41}$, $x_{42}$ and $x_{43}$ as integrating variables, we obtain:
\bes
\widetilde{\Upsilon} \circ \Upsilon [\phi](g_i) = \int [\extd g_{i}'] \phi(g_i') 
 \delta(g_1^{\inv} g_4 g_4'^{\inv} g_1') \delta(g_2^{\inv} g_4 g_4'^{\inv} g_2') \delta(g_3^{\inv} g_4 g_4'^{\inv} g_3') \\
E_{g_2^{\inv} g_1 g_1'^{\inv} g_2'}(x_{12}) E_{g_3^{\inv} g_2 g_2'^{\inv} g_3'}(x_{23}) E_{g_1^{\inv} g_3 g_3'^{\inv} g_1'}(x_{31})\,.
\ees
We remark that the three $\delta$-functions impose that $g_i g_i'^{\inv}$ is independent of $i$, therefore the three remaining plane-waves are equal to $1$. We can finally introduce a resolution of
the identity $1 = \int \extd h \delta(h g_4 g_4'^{\inv})$, and obtain:
\bes
\widetilde{\Upsilon} \circ \Upsilon [\phi](g_i) &=& \int [\extd g_{i}']  \phi(g_i') \int \extd h \prod_{i = 1}^{4} \delta(h g_i g_i'^{\inv}) \\
&=& \int \extd h \, \phi(h g_1, h g_2, h g_3, h g_4) \\
&=& \phi(g_1, g_2, g_3, g_4)\,.
\ees
Note that we used the gauge invariance of $\phi$ in the last line.

\
Now, let us take $\widetilde{\phi} \in \mathbb{T}$, and show that $\Upsilon \circ \widetilde{\Upsilon} [\widetilde{\phi}] = \widetilde{\phi}$. We have:
\bes
\Upsilon \circ \widetilde{\Upsilon} [\widetilde{\phi}](x_{ij}) &=& \int [\extd x_{ij}']^{3} \int [\extd g_i]
\left( E_{g_1^{\inv} g_4}(x_{41}') E_{g_2^{\inv} g_4}(x_{42}') E_{g_3^{\inv} g_4}(x_{43}') \right. \nn \\
 &&\left. E_{g_2^{\inv} g_1}(x_{12}') E_{g_3^{\inv} g_2}(x_{23}') E_{g_1^{\inv} g_3}(x_{31}')\right) \star \widetilde{\phi}(x_{ij}')  \\
 &&\E_{{g_4}^{\inv} g_1}(x_{41}) \E_{{g_4}^{\inv} g_2}(x_{42}) \E_{{g_4}^{\inv} g_3}(x_{43}) \E_{{g_1}^{\inv} g_2}(x_{12}) \E_{{g_2}^{\inv} g_3}(x_{23})
\E_{{g_3}^{\inv} g_1}(x_{31})\,. \nn
\ees
Each integral with respect to a variable $g_i$ gives a non-commutative $\delta$-function involving six different Lie algebra elements. For example, the integral over $g_1$ gives a $\delta (
x_{41} - x_{41}' + x_{12}' - x_{12} + x_{31} - x_{31}')$. The three others are $\delta ( x_{42} - x_{42}' + x_{12} - x_{12}' + x_{23}' - x_{23})$, $\delta ( x_{43} - x_{43}' + x_{23} -
x_{23}' + x_{31}' - x_{31})$, and $\delta ( x_{41}' - x_{41} + x_{42}' - x_{42} + x_{43}' - x_{43})$. After integration of variables $x_{4j}'$, one obtains:
\beq
\Upsilon \circ \widetilde{\Upsilon} [\widetilde{\phi}](x_{ij}) = ( \cT^{142}_{x_{12} - x_{12}'} \cT^{243}_{x_{23} - x_{23}'} \cT^{243}_{- x_{31} + x_{31}'} ) \act \widetilde{\phi}(x_{ij})\,,
\eeq
which, thanks to the invariance of the field $\widetilde{\phi}$, ends the proof.
\end{proof}

\

This proposition ensures that an edge formulation is indeed possible. Moreover, the translation invariance of the fields in $\mathbb{T}$ guarantees that, to construct the edge representation in group
space, one just needs to plug (\ref{naive}) in (\ref{action}). In terms of the new fields $\psi_\ell$, the action can be written:
\bes \label{action_edge}
S_{kin} [\psi] &=& \frac{1}{2} \int [\extd G]^6 \sum_{\ell=1}^5   \, \psi_\ell(G^\ell_{41}, G^\ell_{42}, G^\ell_{43}, G^\ell_{12}, G^\ell_{23}, G^\ell_{31}) 
\overline{\psi_{\ell}}(G^\ell_{41}, G^\ell_{42}, G^\ell_{43}, G^\ell_{12}, G^\ell_{23}, G^\ell_{31}) \nn \\
&& \times \delta(G^\ell_{12} G^\ell_{24} G^\ell_{41}) \delta(G^\ell_{23} G^\ell_{34} G^\ell_{42}) \delta(G^\ell_{31} G^\ell_{14} G^\ell_{43}) ,\\
S_{int}[\psi] &=& \lambda \int [\extd G]^{30} \, \psi_1(G_{25}^{1}, G_{24}^{1}, G_{23}^{1}, G_{54}^{1}, G_{43}^{1}, G_{35}^{1}) \psi_2(G_{31}^{2}, G_{35}^{2}, G_{34}^{2}, G_{15}^{2}, G_{54}^{2},
G_{41}^{2}) \nn \\
&& \psi_3(G_{42}^{3}, G_{41}^{3}, G_{45}^{3}, G_{21}^{3}, G_{15}^{3}, G_{52}^{3}) \psi_4(G_{53}^{4}, G_{52}^{4}, G_{51}^{4}, G_{32}^{4}, G_{21}^{4}, G_{13}^{4}) \nn \\  
&& \psi_5(G_{14}^{5}, G_{13}^{5}, 
G_{12}^{5}, G_{43}^{5}, G_{32}^{5}, G_{24}^{5}) \nn \\
&&\times \left( \prod_{\ell = 1}^{5} \nu(G^{\ell}) \right) \delta(H_{345}) \, \delta(H_{514}) \, \delta(H_{125}) \, \delta(H_{123}) \, \delta(H_{234}) \, \delta(H_{253}) \nn \\
&+& \; \; {\rm{c.c}}, \nn
\ees
where of course $G_{ij}^\ell$ is the (group) variable associated to the edge shared by the triangles $i$ and $j$ in the tetrahedron of color $\ell$ (thus opposite to the vertex of the same color).
In this formula, $\nu(G^{\ell})$ is the measure factor associated to the field $\ell$ (i.e. a product of three $\delta$-functions as in the kinetic term), and $H_{ijk}$ is defined by:
$H_{ijk} \equiv G_{jk}^{i} G_{ij}^{k} G_{ki}^{j}$. These $H$'s are holonomies around edges in the boundaries of the 4-simplex corresponding to the GFT interaction vertex, in the very same way as the
3d case was giving flatness conditions around vertices in
boundaries of tetrahedra. As remarked earlier, there are ten edges in a 4-simplex, hence ten flatness conditions to impose (one for each choice of triplet of distinct colors $i$, $j$ and $k$).
However,
only six of them are independent, which is why the same number of $\delta$-functions of this type appear in the interaction. For the same reason, one is free to choose any set of six independent
$H_{ijk}$ to write the distribution encoding flatness. As in 3d, this freedom will prove very useful.

\

The partition function is defined through a path integral, with the propagator encoded in the covariance of a Gaussian measure $\mu_{\cP}$:
\bes
\int \extd \mu_{\cP}(\overline{\psi}, \psi) \, \overline{\psi}_\ell(g_1, \cdots, g_6) \psi_{\ell'}(g_1', \cdots, g_6') &\equiv& \nu(g_1, \cdots , g_6) \, \delta_{\ell, \ell'}  \prod_{i = 1}^{6}
\delta(g_i g_i'^{\inv}) \,, \\
\nu(g_1, \cdots , g_6) &\equiv& \delta(g_4 g_2^{\inv} g_1) \delta(g_5 g_3^{\inv} g_2) \delta(g_6 g_1^{\inv} g_3) \,, \nn
\ees
with respect to which we integrate the exponential of the interaction part of the action\footnote{Note however that the interaction part does not include the constraints $\nu$, since they are already
imposed in the propagator.}:
\bes\label{partition_edge}
\cZ &\equiv& \int \extd \mu_{\cP}(\overline{\psi}, \psi) \, \e^{- V[\overline{\psi}, \psi]} \\
V[\overline{\psi}, \psi] &\equiv&  \lambda \int [\extd G] \, \cV(\{ H \}) \, \psi_1 \psi_2 \psi_3 \psi_4 \psi_5  \; \; + \; \; {\rm{c.c}}.
\ees
We kept variables of integration implicit in $V$, and called $\cV$ the distribution encoding flatness in (\ref{action_edge}). The kernel $\cV$ is represented as a stranded graph in Figure
\ref{int_e_fig}\footnote{In order to limit the number of crossings, we have reorganized the strands of the different fields.}.

\begin{figure}
\centering
\includegraphics[scale=0.5]{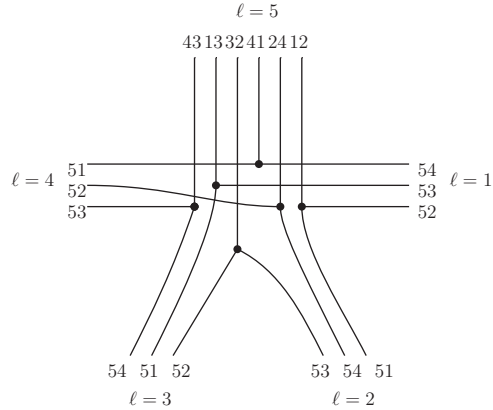}
\caption{Combinatorics of the interaction in edge variables, in a form suitable for factorization of bubbles of color $5$.}\label{int_e_fig}
\end{figure}

\subsubsection{Amplitudes}
In this section, we explore a similar route as in the 3d case, and propose a way to write amplitudes in such a way that the integrand factorizes into bubble contributions. At this stage, everything is
formal (the amplitudes are generically divergent in the absence of regularization), but will be a precious guide in order to derive bounds once a cut-off is added.

\
We first pick up a given color (say $\ell=5$) and choose the six edges of the tetrahedron labeled by $5$ to impose the flatness conditions in any given 4-simplex. This corresponds to working with
holonomies $H_{ijk}$
such that one of the indices is equal to 5. We can for example express the distribution $\cV$ in (\ref{action_edge}) as:
\beq
\cV = \delta(H_{345}) \, \delta(H_{514}) \, \delta(H_{125}) \, \delta(H_{315}) \, \delta(H_{245}) \, \delta(H_{253})\,.
\eeq
This expression involves 18 independent variables $G^{\ell}_{i j}$, 6 of them with $\ell = 5$, and 3 for each $\ell \neq 5$. Therefore, when computing the amplitude of a colored graph $\cG$, all the
propagators with $\ell \neq 5$ will have three strands with free endpoints, that can be integrated straightforwardly. As a result, the closure conditions associated to tetrahedra of color $5$ will
have trivial contributions. Let us verify it for a tetrahedron of color $1$. The integral that we have to compute in this case is of the form:
\beq
\cA_{\cG} = \int [\extd G] \delta(G^{1}_{54} G^{1}_{42} G^{1}_{25}) \delta(G^{1}_{43} G^{1}_{32} G^{1}_{24}) \delta(G^{1}_{35} G^{1}_{52} G^{1}_{23}) R(\{ G \}) \,,
\eeq  
where $R$ does not depend on $G^{1}_{42}$, $G^{1}_{23}$ and $G^{1}_{34}$, since $H_{423}$ does not explicitly appears in the expression we chose for $\cV$. Successively integrating these variables we
have: 
\bes
\cA_{\cG} &=& \int [\extd G]  \delta(G^{1}_{43} G^{1}_{32} G^{1}_{25} G^{1}_{54}) \delta(G^{1}_{35} G^{1}_{52} G^{1}_{23}) R(\{ G \}) \nn \\
&=& \int [\extd G]  \delta(G^{1}_{43} G^{1}_{35} G^{1}_{52} G^{1}_{25} G^{1}_{54})  R(\{ G \})\nn \\
&=& \int [\extd G]  R(\{ G \})\,,
\ees
which shows that the amplitude is unchanged if all the propagators of color $1$ are replaced by trivial ones.
\
Therefore, as in the 3d case, we are lead to simplified expressions for the amplitudes, where all propagators of color $\ell \neq 5$ are trivially integrated. This allows to factorize the integrand of
the amplitude into bubbles (of color $5$) contributions. Since within these bubbles all the propagators are effectively trivial, we can contract each connected component of strands in a bubble to one
node. This is easily understood by looking at Figure \ref{int_e_fig}: in each bubble, all the internal strands are part of lines of colors $\ell \neq 5$, that is those that effectively contain only
three strands; since the constraints associated to these propagators have been integrated with respect to the deleted strands, the remaining strands encode simple convolutions of $\delta$-functions
that can be successively integrated out. The last $\delta$-distribution in a connected component of strands encodes flatness of its dual edge. We call $\cB_5$ the set of bubbles of color $5$, $E_b$ the set
of edges in a bubble $b$, $T_5$
the set of tetrahedra of color $5$. With these notations, the amplitude of $\cG$ (with $\cN$ nodes) can be written as:
\beq\label{amplitude_edge}
\cA_{\cG} = \int [\extd G]^{3 \cN} \left(\prod_{b \in \cB_5}  \prod_{e \in E_b}  \delta( \overrightarrow{\prod_{\tau \supset e}}
(G^{\tau}_{e})^{\epsilon^{\tau}_{e}})\right) 
\left( \prod_{\tau \in T_5} \delta(G^{\tau}_{43} G^{\tau}_{31} G^{\tau}_{14}) \delta(G^{\tau}_{32} G^{\tau}_{21} G^{\tau}_{13}) \delta(G^{\tau}_{24} G^{\tau}_{41} G^{\tau}_{12}) \right) \,,
\eeq
where $\epsilon^{\tau}_{e} = \pm 1$ depending on orientation conventions. 

\

The different types of bounds we will prove in the next paragraphs will rely on two different ways of trivializing the interaction kernels, so that in the above expression, the integration of the last propagators coupling the
variables of different bubbles can be performed. The two strategies give an expression for the amplitudes in which different combinatorial substructures in a
5-colored graph (and in its dual simplicial complex) are singled out. They will respectively be used to obtain bounds referring to the 4-bubbles, or bounds in terms of jackets.  

The first possibility we will describe consists in trivializing three constraints associated to a tree of edges in each tetrahedron, and this will lead
straightforwardly to a bubble bound. The second strategy amounts on the contrary to trivializing the constraints associated to the three edges of a same triangle in each tetrahedron. Only two $\delta$-functions per
propagator will be easily integrable in this case, and the remaining ones will allow to factorize the integrand in terms of Boulatov integrands, henceforth giving a bound involving Boulatov amplitudes
of
$4$-bubbles. This will instead lead to a jacket bound.
		
		
		\subsection{Regularization and general scaling bounds}
	
	The amplitudes as written are of course divergent and need to be regularized to be given rigorous meaning.  

A nice aspect of the latter formulation of the Ooguri model lies in the fact that the constraints associated to edges need not to be regularized in order to make the amplitudes finite, as it will be
shortly proven. In other words, only the dynamics of the theory is affected by the cut-off procedure, and not the kinematical space of fields in terms of which the theory is defined. As in the 3d
case, we will use a heat kernel regularization of the $\delta$-distributions, that with respect to a sharp cut-off will have the main advantage of being positive. The $\SO(3) \times \SO(3)$
$\delta$-distribution splits into a product of two $\SO(3)$ terms: for any $g = (g^{+}, g^{-}) \in \SO(3) \times \SO(3)$, $\delta(g) = \delta^{\SO(3)}(g^{+}) \delta^{\SO(3)}(g^{-})$. Given our
parametrization of the space of functions on $\SO(3) \times \SO(3)$, we can define a regularized distribution for $\SO(3) \times \SO(3)$ using $\SU(2)$ heat kernels:
\beq
\forall g=(g^{+}, g^{-}) \in \SO(3) \times \SO(3), \qquad \delta_{\alpha}(g) \equiv K_{\alpha}(g^{+}) K_{\alpha}(g^{-})\,.
\eeq

We therefore define the regulated partition function as:
\beq\label{partition_edge_reg}
\cZ_{\alpha} \equiv \int \extd \mu_{\cP}(\overline{\psi}, \psi) \, \e^{- V_{\alpha}[\overline{\psi}, \psi]} \,.
\eeq
$V_{\alpha}$ is the regulated interaction, associated to the kernel:
\bes
\cV_{\alpha} \equiv \frac{1}{[\delta_{\alpha}(\one)]^{4}} \prod_{\{i j k \}} \delta_{\alpha}(G_{jk}^{i} G_{ij}^{k} G_{ki}^{j}) 
= \frac{1}{[\delta_{\alpha}(\one)]^{4}} \prod_{\{i j k \}} \delta_{\alpha}(H_{ijk}) \,   
\ees
where the product runs over the 10 possible choices of 3 colors among 5. Note that we again chose a symmetric regularization in the colors, hence the rescaling by $\frac{1}{[\delta_{\alpha}(\one)]^{4}}$. 

The same kind of comments as in 3d apply here. First, this
choice of symmetric regularization is convenient as it will allow to easily average over color attributions. Second, we could have as well chosen a non-symmetric regularization, compatible with the
map $\Upsilon$. For example, a regularized interaction
\bes
S^{\alpha}_{int,5}[\vphi] &=& \lambda \int [\extd g_{i} ]^{10} \, \delta_{\alpha}(g_{2} g_{2}'^{\inv})  \delta_{\alpha}(g_{3} g_{3}'^{\inv})  \delta_{\alpha}(g_{4} g_{4}'^{\inv})  \delta_{\alpha}(g_{6} g_{6}'^{\inv}) 
\delta_{\alpha}(g_{7} g_{7}'^{\inv})  \delta_{\alpha}(g_{9} g_{9}'^{\inv}) \\ 
&\times& \vphi_1(g_1, g_2', g_3', g_4') \vphi_2(g_4, g_5, g_6', g_7') 
\vphi_3(g_7, g_3, g_8, g_9') \vphi_4(g_9, g_6, g_2, g_{10}) \vphi_5(g_{10}, g_8, g_5, g_1) \nn \\
&+& {\rm{c.c.}} \nn
\ees
in triangle variables corresponds to a regularized interaction kernel 
\beq
\cV^{5}_{\alpha} \equiv \prod_{\{i j 5 \}} \delta_{\alpha}(H_{ijk})
\eeq
in edge variables, which in turn implies an explicit factorization of the amplitudes in terms of bubbles of color $5$. This is the exact analogue of (\ref{amplitude_edge}) in the theory with cut-off. 

\
With the symmetric regularization, this type of factorization is recovered as a bound only, but for arbitrary color $\ell$. Moreover this bound will be always saturated in power-counting. It
is obtained by bounding four redundant flatness conditions in all the interactions. For instance, if we use the color $5$ as before, we have:
\beq
| \cV_{\alpha} | \leq \prod_{\{i j 5 \}} \delta_{\alpha}(H_{ij5}) \,,
\eeq 
where now the product runs over the 6 flatness conditions involving the color $5$. Thanks to the positivity of the regularization, and the convolution properties of the heat kernel, this immediately yields:
\bes\label{amplitude_edge_reg}
|\cA_{\cG}^{\alpha}| &\leq& \int [\extd G]^{3 \cN} \left(\prod_{b \in \cB_5}  \prod_{e \in E_b}  \delta_{\langle e,b \rangle \alpha}( \overrightarrow{\prod_{\tau
\supset e}} (G^{\tau}_{e})^{\epsilon^{\tau}_{e}})\right) \nn \\ 
&& \qquad \times \left( \prod_{\tau \in T_5} \delta(G^{\tau}_{43} G^{\tau}_{31} G^{\tau}_{14}) \delta(G^{\tau}_{32} G^{\tau}_{21} G^{\tau}_{13}) \delta(G^{\tau}_{24} G^{\tau}_{41} G^{\tau}_{12}) \right) 
\ees
where for any edge $e$ in a bubble $b$, we denote by $\langle e,b \rangle$ the number of tetrahedra in $b$ that contain $e$. 
\begin{figure}[h]\label{trees}
  \centering
  \subfloat[A tree of edges.]{\label{tree_tetra1}\includegraphics[scale=0.4]{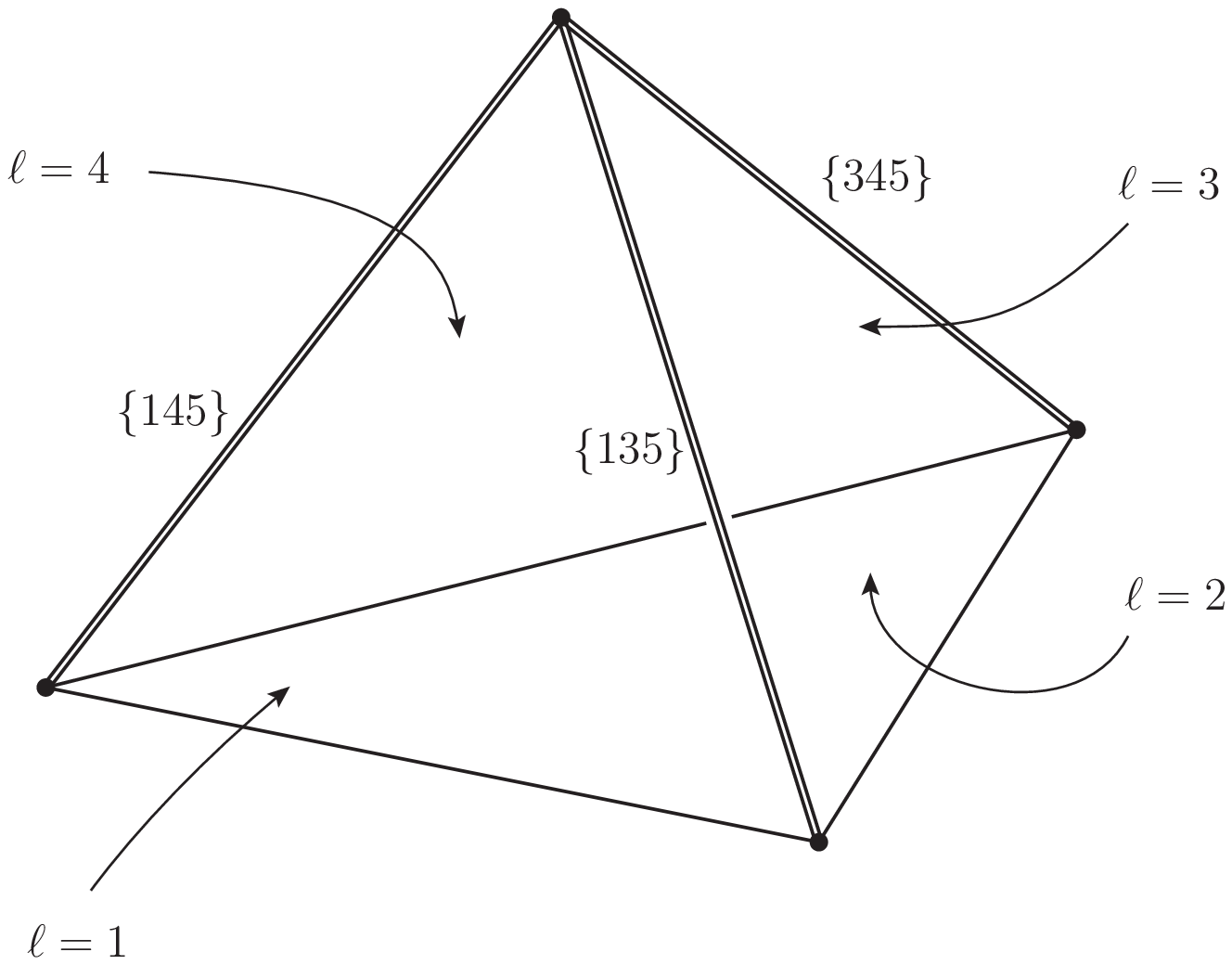}}                
  \subfloat[Another tree.]
{\label{tree_tetra2}\includegraphics[scale=0.4]{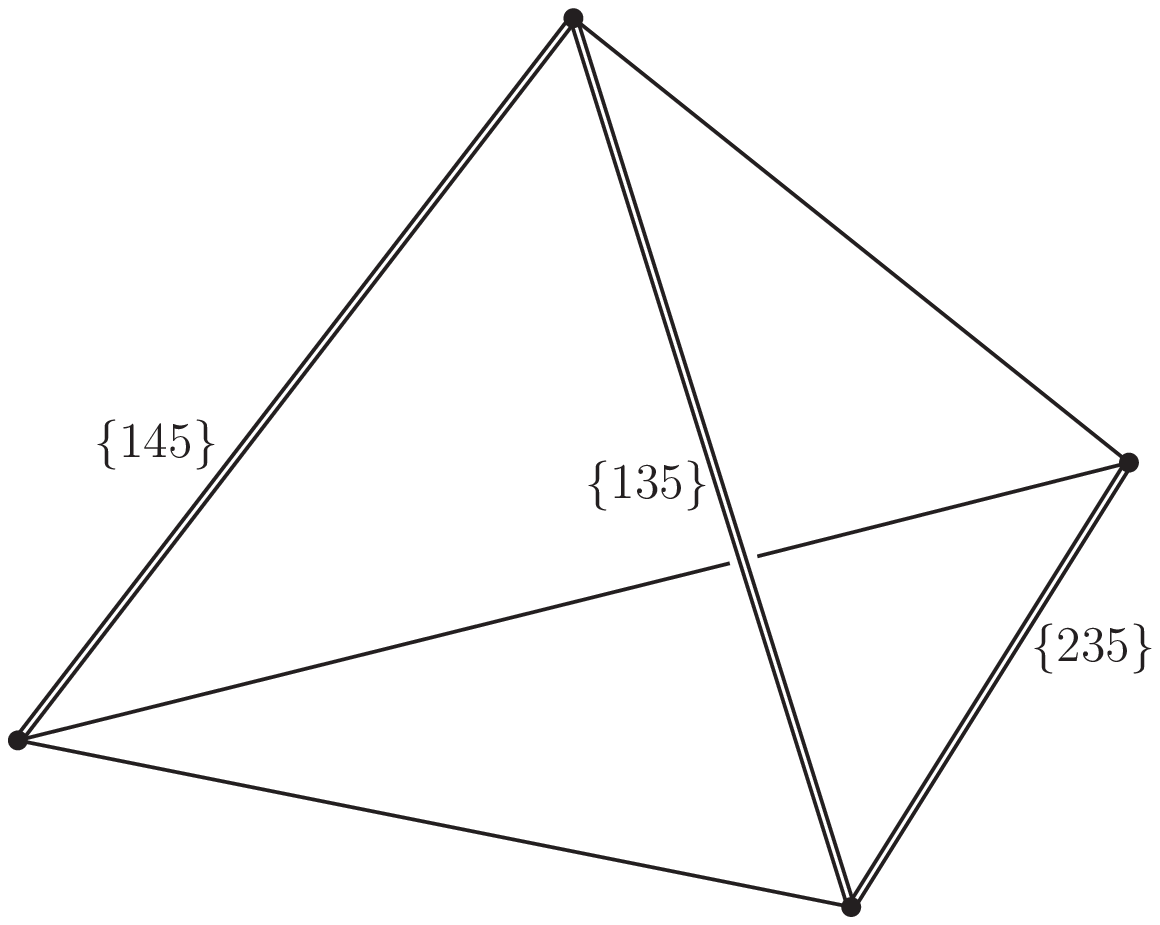}}
\subfloat[Not a tree.]
{\label{not_tree_tetra}\includegraphics[scale=0.4]{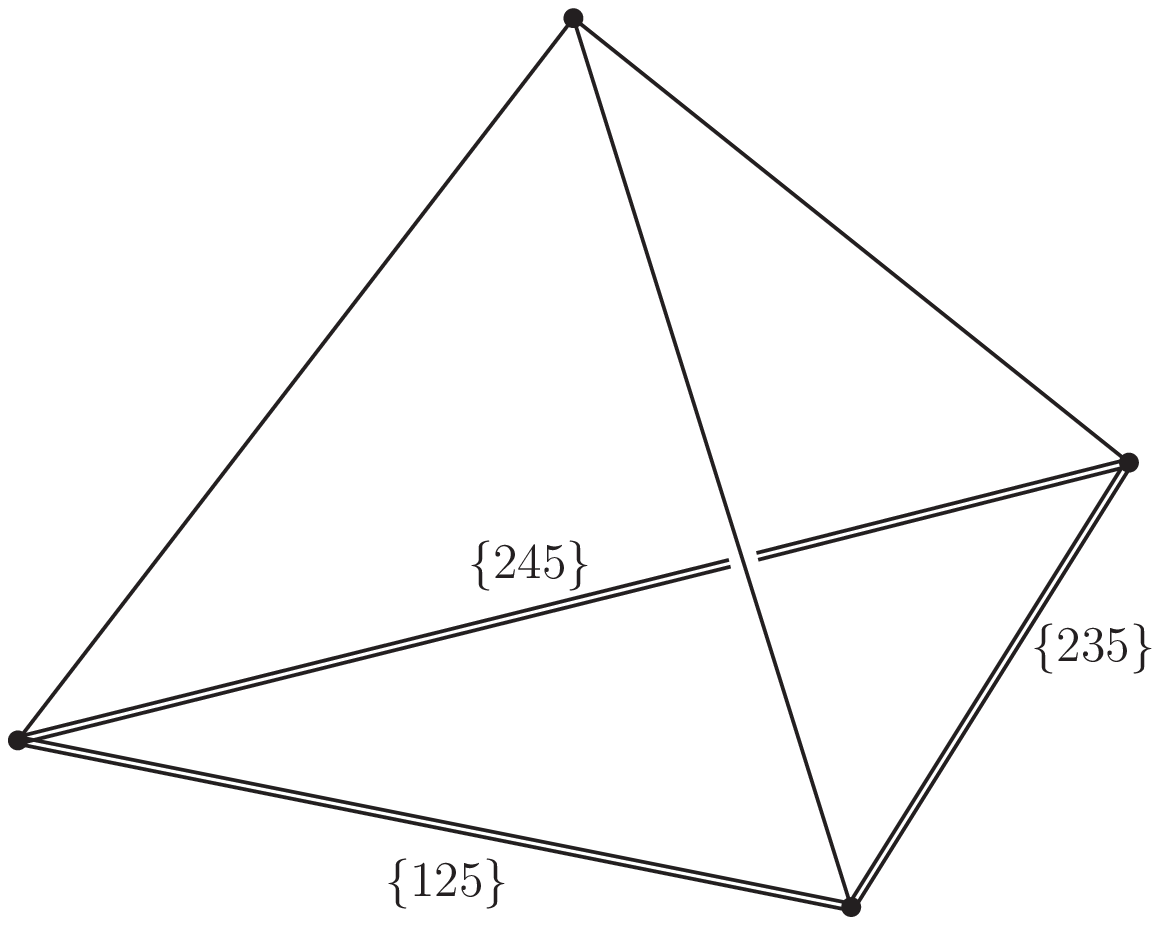}}
  \caption{Different possibilities for trivializing edge interactions in a tetrahedron of color $5$ (double lines): a tree in Figure \ref{tree_tetra1} and Figure \ref{tree_tetra2}; edges associated to a
same triangle in Figure \ref{not_tree_tetra}.}
\end{figure}
We now have to integrate the remaining propagators, using the two possible sets of variables evoked before. 

\
We start with the first strategy and look for integrating variables associated to a tree of edges in each tetrahedron. 
There are two kinds of such trees: three edges sharing a same vertex; or three edges such that the first one shares a vertex with the second, the second with the third, but the third does not share
any vertex with the first one (see Figures \ref{tree_tetra1} and \ref{tree_tetra2}). Of course we want to use the colors to define these trees, so that the same simplification takes place in all the
tetrahedra of a given simplicial complex. The notations are as follows: we will associate the color $\{\ell_1 \ell_2 \ell_3\}$ to an edge involving flatness constraints $H_{\ell_1 \ell_2 \ell_3}$. Such an
edge is therefore dual to a (maximal) connected subgraph of $\cG$, involving only lines of colors $\ell_1$, $\ell_2$ and $\ell_3$.\footnote{Note that this convention differs from the one used in 3d to
label vertices, as here the latter would amount to labeling an edge by the two line colors on which its dual graph does {\it{not}} have support.}
For definiteness, we will use variables involved in flatness conditions around edges of colors $\{345\}$, $\{145\}$ and $\{135\}$, a choice corresponding to a tree of the first kind (i.e. Figure
\ref{tree_tetra1}). 

Since each strand is connected to two interactions, the integrals to
compute are not simple convolutions, therefore difficult. To circumvent this problem we simply bound all the heat kernels implementing flatness constraints of colors $\{345\}$, $\{145\}$ and $\{135\}$
by their value at the identity, and then integrate the propagators. This yields:
\bes
|\cA_{\cG}^{\alpha}| &\leq&  \int [\extd G]^{3 \cN} \left(\prod_{b \in \cB_5}  \prod_{e \in E_{b}(345)\cup E_{b}(145) \cup E_{b}(135)}  \delta_{\langle e,b
\rangle \alpha}(\one) \right) \nn \\
&&\times \left(\prod_{b \in \cB_5}  \prod_{e \in E_{b}(125)\cup E_{b}(235) \cup E_{b}(245)}  \delta_{\langle e,b \rangle \alpha}( \overrightarrow{\prod_{\tau \supset e}}
(G^{\tau}_{e})^{\epsilon^{\tau}_{e}})\right) \nn \\
&&\times \left( \prod_{\tau \in T_5} \delta(G^{\tau}_{43} G^{\tau}_{31} G^{\tau}_{14}) \delta(G^{\tau}_{32} G^{\tau}_{21} G^{\tau}_{13}) \delta(G^{\tau}_{24} G^{\tau}_{41} G^{\tau}_{12}) \right) \nn \\
&\leq& \left(\prod_{b \in \cB_5}  \prod_{e \in E_{b}(345)\cup E_{b}(145) \cup E_{b}(135)}  \delta_{\langle e,b \rangle \alpha}(\one) \right)  \nn \\
&& \times \int [\extd G]^{\frac{3}{2} \cN} \left( \prod_{b \in \cB_5}  \prod_{e \in E_{b}(125)\cup E_{b}(235) \cup E_{b}(245)}  \delta_{\langle e,b \rangle \alpha}( \overrightarrow{\prod_{\tau \supset e}}
(G^{\tau}_{e})^{\epsilon^{\tau}_{e}}) \right)\,.
\ees
The only term left to integrate is a product of integrals associated to connected $\phi^3$ graphs (whose lines are strands of the initial graph). Each of these is dual to an edge of color $\{125\}$,
$\{235\}$ or $\{245\}$. Integrating a maximal tree
of strands in each of these graphs, then bounding the last $\delta$-function by its value at the identity, we obtain the general bound:
\beq
|\cA_{\cG}^{\alpha}| \leq \left(\prod_{b \in \cB_5}  \prod_{e \in E_{b}(345)\cup E_{b}(145) \cup E_{b}(135)}  \delta_{\langle e,b \rangle \alpha}(\one) \right)
\left( \prod_{e \in E(125)\cup E(235) \cup E(245)}  \delta_{|e| \alpha}( \one) \right)\,,
\eeq
where 
\beq
|e| \equiv \sum_{b \in \cB_\ell \, , \, b \supset e}  \langle e,b \rangle.
\eeq 
coincides with the number of $4$-simplices containing the edge $e$. 

Remarking that, when $\alpha$ goes to zero:
\beq
\forall a > 0, \qquad \frac{\delta_{a \alpha}(\one)}{\delta_{\alpha}(\one)} \rightarrow a^{- 3} \,,
\eeq
we can rewrite this bounds using powers of heat kernels with the same parameter, for instance $\alpha$. This allows to show that for any constant $K$ such that
\beq
K > K_0 \equiv \left(\prod_{b \in \cB_5}  \prod_{e \in E_{b}(345)\cup E_{b}(145) \cup E_{b}(135)}  \langle e,b \rangle^{- 3} \right)
\left( \prod_{e \in E(125)\cup E(235) \cup E(245)}  |e|^{- 3} \right)
\eeq
we asymptotically have:
\bes\label{bound1}
|\cA_{\cG}^{\alpha}| &\leq& K  \, [\delta_{\alpha}(\one)]^{\gamma}  \\
\gamma &=& |E(125)| + |E(235)| + |E(245)| + \sum_{b \in \cB_5} \left(|E_{b}(345)| + |E_{b}(145)| + |E_{b}(135)|\right) \,. \nn
\ees
This formula will be the central tool in the bounds on pseudo-manifolds we will derive below. Before moving to a second important formula following from the factorized expression for the amplitudes, we
remark that we can simply choose $K = 1$.

\

The second strategy we suggested to bound formula (\ref{amplitude_edge_reg}) also starts from a splitting of edge colors into two parts:
\bes
|\cA_{\cG}^{\alpha}| &\leq&  \int [\extd G]^{3 \cN} \left(\prod_{b \in \cB_5} 
\prod_{e \in E_{b}(125)\cup E_{b}(235) \cup E_{b}(245)}  \delta_{\langle e,b \rangle \alpha}(\overrightarrow{\prod_{\tau \supset e}} (G^{\tau}_{e})^{\epsilon^{\tau}_{e}}) \right) \nn \\
&&\times  \left(\prod_{b \in \cB_5}  \prod_{e \in E_{b}(345)\cup E_{b}(145) \cup E_{b}(135)}  \delta_{\langle e,b \rangle \alpha}( \overrightarrow{\prod_{\tau \supset e}}
(G^{\tau}_{e})^{\epsilon^{\tau}_{e}})\right) \\
&& \left( \prod_{\tau \in T_5} \delta(G^{\tau}_{43} G^{\tau}_{31} G^{\tau}_{14}) \delta(G^{\tau}_{32} G^{\tau}_{21} G^{\tau}_{13}) \delta(G^{\tau}_{24} G^{\tau}_{41} G^{\tau}_{12}) \right) \nn 
\ees
We can integrate the propagators with respect to variables $G^{\tau}_{23}$ and $G^{\tau}_{24}$, which are only involved in the first factor:
\bes
|\cA_{\cG}^{\alpha}| &\leq& \int [\extd G]^{2 \cN} \left( \prod_{b \in \cB_5} \prod_{e \in E_{b}(125)} 
\delta_{\langle e,b \rangle \alpha}(\overrightarrow{\prod_{\tau \supset e}} (G^{\tau}_{21})^{\epsilon^{\tau}_{e}}) \right) \nn \\
&&  \left( \prod_{b \in \cB_5}  
  \left( \prod_{e \in E_{b}(235)} \delta_{\langle e,b \rangle \alpha}(\overrightarrow{\prod_{\tau \supset e}} (G^{\tau}_{21} G^{\tau}_{13})^{\epsilon^{\tau}_{e}}) \right)
  \left( \prod_{e \in E_{b}(245)} \delta_{\langle e,b \rangle \alpha}(\overrightarrow{\prod_{\tau \supset e}} (G^{\tau}_{21} G^{\tau}_{14})^{\epsilon^{\tau}_{e}}) \right) \right)  \nn \\
&&\times   \left(\prod_{b \in \cB_5}  \prod_{e \in E_{b}(345)\cup E_{b}(145) \cup E_{b}(135)}  \delta_{\langle e,b \rangle \alpha}( \overrightarrow{\prod_{\tau \supset e}}
(G^{\tau}_{e})^{\epsilon^{\tau}_{e}})\right)
\left( \prod_{\tau \in T_5} \delta(G^{\tau}_{43} G^{\tau}_{31} G^{\tau}_{14}) \right) 
\ees
The interesting feature of this formula is the following: the last line is the integrand of the Boulatov amplitude of the 3d colored graph obtained from $\cG$ by deleting the lines of color $2$! 
Since the connected components of this graph are the bubbles in $\cB_{2}$, it is tempting to factorize their 3d amplitudes.
The simplest way do so is to bound the $\delta$-functions appearing in the first two lines that involve variables from the third. We wrote the last inequality in such a way that these are exactly the
terms in the second line. 
Therefore:
\bes
|\cA_{\cG}^{\alpha}| &\leq& \left( \prod_{b \in \cB_5} \prod_{e \in E_{b}(235) \cup E_{b}(245)} \delta_{\langle e,b \rangle \alpha}(\one) \right) 
\int [\extd G]^{\frac{\cN}{2}} \left( \prod_{b \in \cB_5} \prod_{e \in E_{b}(125)}  \delta_{\langle e,b \rangle \alpha}(\overrightarrow{\prod_{\tau \supset e}} (G^{\tau}_{21})^{\epsilon^{\tau}_{e}})
\right) \nn \\
&&  \int [\extd G]^{\frac{3 \cN}{2}} \left(\prod_{b \in \cB_5}  \prod_{e \in E_{b}(345)\cup E_{b}(145) \cup E_{b}(135)}  \delta_{\langle e,b \rangle \alpha}( \overrightarrow{\prod_{\tau \supset e}}
(G^{\tau}_{e})^{\epsilon^{\tau}_{e}})\right)
\left( \prod_{\tau \in T_5} \delta(G^{\tau}_{43} G^{\tau}_{31} G^{\tau}_{14}) \right) \nn
\ees
The last line is now exactly a product of bubble 3d amplitudes: $\prod_{b \in \cB_2} \cA_{b}^{\alpha}$. As for the integral in the first line, it is associated with the graph made of all the strands of
color $(125)$ which, as in the previous case, we bound by:
\beq
\int [\extd G]^{\frac{\cN}{2}} \left( \prod_{b \in \cB_5} \prod_{e \in E_{b}(125)}  \delta_{\langle e,b \rangle \alpha}(\overrightarrow{\prod_{\tau \supset e}} (G^{\tau}_{21})^{\epsilon^{\tau}_{e}})
\right)
\leq \prod_{e \in E(125)}  \delta_{|e| \alpha}( \one) \, .
\eeq

The net result is that, for any constant Q such that:
\beq
Q > Q_0 \equiv \left(\prod_{b \in \cB_5}  \prod_{e \in E_{b}(235) \cup E_{b}(245)}  \langle e,b \rangle^{- 3} \right) \left( \prod_{e \in E(125)} |e|^{-3} \right)
\eeq
we asymptotically have:
\bes\label{bound2}
|\cA_{\cG}^{\alpha}| &\leq& Q \, [\delta_{\alpha}(\one)]^{\eta} \prod_{b \in \cB_2} |\cA_{b}^{\alpha}| \nn \\
\eta &=& |E(125)| +\sum_{b \in \cB_5} \left(|E_{b}(235)| + |E_{b}(245)|\right) \,.
\ees
As before, we remark that $Q = 1$ is a valid choice.

		
		\subsection{Topological singularities}
		
		In this section, we would like to use equation (\ref{bound1}) to give a bound on the amplitudes of simplicial complexes with pointlike singularities. As in 3d, we need to introduce a combinatorial quantity which 
determines whether a bubble is spherical or not. Unlike in 3d however, this cannot be the genus, since the Euler characteristic of a manifold in odd dimensions is 0, irrespectively of its topology. We therefore propose 
to use the notion of degree introduced in tensor models (see previous chapter and \cite{RazvanN,RazvanVincentN,Gurau:2011xq}):  $\omega (b) = \sum_{J} g_J$, where the sum runs over the jackets of the $4$-bubble $b$, and $g_J$ is the genus of the jacket $J$. Indeed, the degree is 0 only for a specific set of simplicial decompositions of the sphere (those associated to melonic
graphs \cite{critical}). As a result, in this 4d case we will prove a bound on simplicial complexes which have non-melonic bubbles, that is on a subclass of manifolds as well as on singular pseudomanifolds.
But, remarking that {\it{all}} the jackets of a non-spherical bubble have genus bigger than $1$, we will refine the bound for singular pseudomanifolds, which will in the end take the same form as in 3d.

\
Let us consider a connected and closed colored graph $\cG$. We propose to use the bound (\ref{bound1}) to deduce a bound on the divergence degree $\gamma_{\cG}$ of $\cG$. Since the pre-factor $K$ can be
chosen to be $1$, only the parameter $\gamma$ in (\ref{bound1}) matters, and $\gamma_{\cG} \leq \gamma$, that is:
\beq\label{bound1'}
\gamma_{\cG} \leq |E(125)| + |E(235)| + |E(245)| + \sum_{b \in \cB_5} \left(|E_{b}(345)| + |E_{b}(145)| + |E_{b}(135)|\right) \,.
\eeq
Since this expression involves only quantities associated to edges and $4$-bubbles of the simplicial complex, let us show that the sum of the degrees of the same bubbles can also be expressed in terms
of similar quantities, namely we first prove:
\begin{lemma}\label{lemmaAdd}
\beq\label{sum_deg}
\sum_{b \in \cB_5} \omega(b) = 3 |\cB_{5}| +  \frac{3 \cN}{2} - \sum_{b \in \cB_5} |E_{b}|
\eeq
\end{lemma}
\begin{proof}
Indeed, by definition of the degree of a bubble $b \in \cB_{5}$, and by formula (\ref{g_jacket}), we have:
\beq		
\omega(b) = \sum_{ J= (\sigma, \sigma^{- 1})} \left( 1 + \frac{1}{2} \left( |T_{b}^{5}| - \sum_{\ell = 1}^{4} |E_{b}(\sigma(\ell) \sigma(\ell + 1) 5 )| \right) \right)\,,
\eeq
where $\sigma$ are cycles over $\{ 1, \ldots, 4\}$.\footnote{Note that we also used the fact that labeling jackets by cycles over vertex colors as in (\ref{g_jacket}), is equivalent to labeling them
by cycles over edge colors.} There are three different jackets in $b$, and each of them has support over four edge colors out of six, therefore:
\bes
\omega(b) &=& 3 +  \frac{1}{2} \left( 3 |T_{b}^{5}| - \frac{3 \times 4}{6} |E_b| \right) \nn \\
&=& 3 +  \frac{3 |T_{b}^{5}|}{2} - |E_b|\,.
\ees
Summing over the bubbles, and using
\beq
\sum_{b \in \cB_5} |T_{b}^{5}| = 2 |T^{5}| = \cN
\eeq
we finally obtain the claimed equality. \end{proof}
We therefore just need to make $\sum_{b \in \cB_5} |E_{b}|$ appear on the right-hand-side of inequality (\ref{bound1'}). To this aim, and as in the 3d case, we prove a combinatorial lemma:
\begin{lemma}\label{comb_lemma2}
 For any edge coloring $i = (\ell_1 \ell_2 5)$:
\beq
|\cB_{5}| + |E(i)| - \sum_{b \in \cB_{5}} |E_{b}(i)| \leq 1 \, .
\eeq
\end{lemma}
\begin{proof} 
Consider the graph $\cC_{i,5}(\cG)$ whose nodes are the elements of $\cB_5 \cup E(i)$, and links are constructed as follows: there is a link between a bubble $b \in \cB_5$ and $e \in E(i)$ if and only
if $e$ appears in the triangulation
of $b$.
\
We first show that, because $\cG$ itself is assumed to be connected, $\cC_{i,5}(\cG)$ is also connected. Let $b_i$ and $b_f$ be two elements of $\cC_{i,5}(\cG) \cap \cB_5$. The connectivity of $\cG$
ensures that there exists
a sequence of $p$ bubbles $(b_i \equiv b_0 , b_1 , \cdots , b_{p-1} , b_p \equiv b_f)$, such that: for any $k \in \llbracket 0 , p\rrbracket$, $b_k$ and $b_{k + 1}$ share a tetrahedron $\tau_k$. 
Hence: for any $k \in \llbracket 0 , p\rrbracket$, $b_k$ and $b_{k + 1}$ share an edge $e_k$. So, in $\cC_{i, 5}(\cG)$, $(b_i \equiv b_0 , e_0 , b_1 , e_1, \cdots , b_{p-1} , e_{p - 1}, b_p \equiv
b_f)$ is a connected path
from $b_i$ to $b_f$. Finally, remarking that any element of $\cC_{i, 5}(\cG) \cap E(i)$ is by definition connected to at least one element of $\cC_{i, 5}(\cG) \cap \cB_5$, we conclude that any two
nodes of $\cC_{i, 5}(\cG)$
are connected.
\
Call $L$ the number of links of $\cC_{i, 5}(\cG)$. Its number of nodes being equal to $|\cB_5| + |E(i)|$, the connectivity implies that:
\beq
|\cB_5| + |E(i)| \leq 1 + L\,.
\eeq
But for any $b \in \cB_5$, the number of lines $L(b)$ connected to $b$ in $\cC_{i, 5}(\cG)$ verifies: $L(b) \leq |E_{b}(i)|$. This is because $E_{b}(i)$ is the set of available edges of color $i$ in
the bubble (some of which
might be identified when gluing the different bubbles together), and $L(b)$ is the true number of edges of color $i$ in $b$, once all the identifications of edges have been taken into account. Since
by construction
\beq
L= \sum_{b \in \cB_5} L(b)\,,
\eeq 
we conclude that:
\beq
|\cB_5| + |E(i)| \leq 1 + \sum_{b \in \cB_5} L(b) \leq 1 + \sum_{b \in \cB_5} |E_{b}(i)|\,.
\eeq
\end{proof}

This, together with (\ref{bound1'}) and (\ref{sum_deg}), immediately yields the wanted result:
\beq
\gamma_{\cG} \leq 3 + \frac{3 \cN}{2}  - \sum_{b \in \cB_{5}} \omega(b) \, .
\eeq
Of course, we have a similar bound for any color $\ell$. 

This formula is particularly simple when the coupling constant $\lambda$ is appropriately rescaled:
\beq\label{rescaling_4d}
 \lambda \rightarrow \frac{\lambda}{\delta_{\alpha}(\one)^{3/2}}\,,
\eeq
which is also the setting of the $1/N$ expansion of the Ooguri model, as initially found in \cite{Gurau:2011xq}, and as will be confirmed shortly. Thus we have obtained the following:
\begin{proposition}
With the rescaling of the coupling constant (\ref{rescaling_4d}), the divergence degree $\gamma_{\cG}$ of a graph $\cG$ verifies, for any color $\ell$:
\beq
\gamma_{\cG} \leq 3  - \sum_{b \in \cB_{\ell}} \omega(b) \, \label{bound}.
\eeq
\end{proposition}

Given that any pointlike singularity implies a degree greater than zero, we therefore see that, like in 3d, the more pointlike singularities of the same color a simplicial complex has, the more its
amplitude is suppressed with respect to the leading order 
($\sim \delta_{\alpha}(\one)^{3}$). 
We can make this point a bit more precise, making use of a lemma, a simple proof of which can be found in the review \cite{Gurau:2011xp}\footnote{This result is actually generalizable to any dimension, since a graph with a planar
jacket has a trivial \textit{regular genus}, a property that in turn characterizes spheres (see \cite{FerriGagliardi, Vince_gene} and references therein).}:
\begin{lemma}
 If a connected $4$-colored graph $\cG$ possesses a spherical jacket, then $\cG$ is dual to a sphere.
\end{lemma}
\begin{proof}
 See Proposition 3 in \cite{Gurau:2011xp}.
\end{proof}
In particular, this implies that if a bubble $b$ is not spherical, all its jackets have genus at least $1$, and therefore: $\omega(b) \geq 3$. This allows to prove the following corollary:
\begin{corollary}
 With the rescaling of the coupling constant (\ref{rescaling_4d}), the divergence degree $\gamma_{\cG}$ of any connected vacuum graph $\cG$ verifies, for any color $\ell$:
\beq
\gamma_{\cG} \leq 3 (1 - N^{s}_{\ell}) \,,
\eeq
where $N^{s}_{\ell}$ is the number of singular vertices of color $\ell$. 
\end{corollary}
This result is very similar to what we have proven in 3d, showing suppression of singular pseudomanifolds with respect to the leading order ($\gamma_{\cG} = 3$), therefore consisting only of
manifolds. In particular, in both cases singular pseudomanifolds all have convergent amplitudes. 

Finally, as anticipated in the beginning of this section, we notice that equation (\ref{bound}) also constrains the amplitudes of a special class of manifolds: those which have non-melonic bubbles.
\begin{corollary}
With the rescaling of the coupling constant (\ref{rescaling_4d}), the divergence degree $\gamma_{\cG}$ of any connected vacuum graph $\cG$ verifies, for any color $\ell$:
\beq
\gamma_{\cG} \leq 3 - N^{nm}_{\ell} \,,
\eeq
where $N^{nm}_{\ell}$ is the number of non-melonic bubbles of color $\ell$. 
\end{corollary}
		
		
		\subsection{Domination of melons}
		
\subsubsection{Jacket bounds}		
	The notion of jacket of a graph $\cG$ we used in 3d generalizes to any dimension. They are closed surfaces labeled by pairs $(\sigma , \sigma^{\inv})$ of cyclic permutations of $\{ 1 , \ldots , 5
\}$, 
in which the graph $\cG$ can be regularly embedded \cite{FerriGagliardi, Vince_gene}\footnote{In contrast with the 3d case, we construct the jackets in terms of data related to the graph $\cG$ itself,
as opposed to its dual simplicial complex. 
If the two descriptions are of course equivalent, we find more convenient to work with $\cG$ itself in 4d, since representing or even giving an intuitive picture of simplicial complexes is necessarily
more difficult than in 3d.}.
The jacket $J = (\sigma , \sigma^{\inv})$ is the closed surface constituted of all the faces of colors $(\sigma(\ell) \sigma(\ell + 1))$ in $\cG$, glued along there common links. Its genus can be
easily computed
in terms of combinatorial data associated to $\cG$:
\begin{lemma}
 The jacket $J = (\sigma , \sigma^{\inv})$ of a $5$-colored graph $\cG$ has genus:
\beq
g_{J} = 1 + \frac{3 \cN}{4} - \frac{1}{2} \sum_{\ell = 1}^{5} |\cF(\sigma(\ell) \sigma(\ell + 1))|\,,
\eeq
where $\cN$ is the order of $\cG$, and $\cF( i j )$ is the set of faces of color $(i j)$ (dual to the set of triangles of color $(ij)$ in the simplicial complex).
\end{lemma}
 \begin{proof}
By definition, the Euler characteristic of $J$ is:
\beq
\chi_J = 2 - 2 g_J = |\cF_J| - |\cE_J| + |\cV_J|
\eeq
where $\cF_J$, $\cE_J$ and $\cV_J$ are the set of faces, edges and vertices of $J$. But, by construction:
\beq
|\cF_J| = \sum_{\ell} |\cF(\sigma(\ell) \sigma(\ell + 1))| \;, \qquad |\cE_J| = \cL \;, \qquad |\cV_J| = \cN \;,
\eeq
where $\cL$ is the number of lines in $\cG$. Since moreover $\cL = 2 \cN$, we conclude that:
\beq
2 - 2 g_J = - \frac{3 \cN}{2} + \sum_{\ell} |\cF(\sigma(\ell) \sigma(\ell + 1))| \,.
\eeq 
 \end{proof}

Jackets have been used in the Ooguri model to compute bounds on amplitudes \cite{Gurau:2011xq}. Interestingly, our construction allows to slightly strengthen these results, which we demonstrate now.
Discarding pre-factors, formula (\ref{bound2}) immediately implies the following bound on the degree of divergence of a graph $\cG$:
\beq
\gamma_{\cG} \leq |E(125)| + \sum_{b \in \cB_5} \left(|E_{b}(235)| + |E_{b}(245)|\right) + \sum_{b \in \cB_2} \gamma_{3d}(b)\,,
\eeq 
where $\gamma_{3d}(b)$ is the degree of divergence of the $3d$ amplitudes represented by the bubble graph $b$.
Since the choice of colors is arbitrary this generalizes immediately. For any cycle $\sigma = (\ell_1  \ell_2  \ell_3  \ell_4 \ell_5)$ of $(12345)$:
\beq\label{ineq}
\gamma_{\cG} \leq |E(\ell_1 \ell_3 \ell_5)| + \sum_{b \in \cB_{\ell_1}} \left(|E_{b}(\ell_1 \ell_2 \ell_3)| + |E_{b}(\ell_1 \ell_3 \ell_4)|\right) + \sum_{b \in \cB_{\ell_3}} \gamma_{3d}(b)\,.
\eeq 

From this, we would like to deduce an inequality involving the genus of the jacket $J$ associated to the cycle $\sigma$. This can be done in three steps. First remark that for any $i \neq j$, 
and $\ell$ different from both $i$ and $j$, the set of faces $\cF(i j)$ can be partitioned in terms of the faces $\cF_{b}(i j)$ of the bubble graphs $b \in \cB_{\ell}$:
\beq
\cF(ij) = \underset{b \in \cB_{\ell}}{\cup} \cF_{b}(i j)\,. 
\eeq
But since the set of faces $\cF_{b}(ij)$ is dual to the set of edges $\E_{b}(\ell ij)$ in the simplicial complex associated to $b$, we immediately obtain the following equality of cardinals:
\beq
|\cF(i j)| = \sum_{b \in \cB_\ell} |E_{b}(\ell i j)| \,.
\eeq
In particular, the sum over $\cB_{\ell_1}$ in (\ref{ineq}) is equal to $|\cF(\ell_2 \ell_3)| + |\cF(\ell_3 \ell_4)|$. Second, we can use lemma \ref{comb_lemma2} to bound $|E(\ell_1 \ell_3 \ell_5)|$:
\beq
|E(\ell_1 \ell_3 \ell_5)| \leq 1 - |\cB_{\ell_3}| + \sum_{b \in \cB_{\ell_3}} |E_{b}(\ell_1 \ell_3 \ell_5)| = 1 - |\cB_{\ell_3}| + |\cF(\ell_1 \ell_5)|\,.
\eeq
Finally, we can use a bound of the type (\ref{ineq_3d}) to bound the $\gamma_{3d}(b)$ terms:
\beq
\sum_{b \in \cB_{\ell_3}} \gamma_{3d}(b) \leq \sum_{b \in \cB_{\ell_3}} \left( 1 + |E_{b}(\ell_3 \ell_1 \ell_2)| + |E_{b}(\ell_3 \ell_4 \ell_5)| \right) = |\cB_{\ell_3}| + |\cF(\ell_1 \ell_2)| +
|\cF(\ell_4 \ell_5)|\,.
\eeq
All in all we obtain:
\beq
\gamma_{\cG} \leq 1 + \sum_{\ell = 1}^{5} |\cF( \sigma(\ell) \sigma(\ell + 1))| = 3 + \frac{3 \cN}{2} - 2 g_J
\eeq
This is our final result, which as in 3d is particularly nice once the coupling constant has been appropriately rescaled.
\begin{proposition}\label{jacket_4d}
With the rescaling of the coupling constant (\ref{rescaling_4d}), the divergence degree $\gamma_{\cG}$ of a connected vacuum graph $\cG$ verifies, for any jacket $g_J$:
\beq
\gamma_{\cG} \leq 3  - 2 g_J\,.
\eeq
\end{proposition}

In particular, we notice that the strongest bound is obtained by considering the {\it maximum} of the genera of all the 12 jackets of the graph $\cG$.
Also, we see that all graphs with a jacket of genus greater than $2$ have convergent amplitudes, just like singular pseudomanifolds. 

\

This has to be compared with the known jacket bound, which is weaker and is a direct corollary of the previous result:
\begin{corollary}
 With the rescaling of the coupling constant (\ref{rescaling_4d}), the divergence degree $\gamma_{\cG}$ of a connected vacuum graph $\cG$ verifies:
\beq
\gamma_{\cG} \leq 3  - \frac{1}{6} \omega(\cG)\,.
\eeq
\end{corollary}

\begin{proof}
 The degree of $\cG$ is defined by:
\beq
\omega(\cG) = \sum_{J} g_J\,,
\eeq
so that averaging the previous bound over the $12$ jackets of $\cG$ immediately gives the standard jacket bound. 
\end{proof}

\subsubsection{$1/N$ expansion}
	
	Similarly to the Boulatov model, we have just seen how focusing on the formal symmetries of the Ooguri model can lead to interesting scaling bounds on the amplitudes. While at first sight the edge formulation in 4d seems to be tailored to the analysis of bubble divergences only, it has been demonstrated that jacket bounds are also contained in the basic scaling bounds (\ref{bound1}) and (\ref{bound2}). Therefore we again recover from this analysis that leading order graphs must be degree $0$, hence melonic. The reciprocal can be proved in the same way as in 3d, by analyzing how the scaling behavior of the amplitudes is left invariant by non-degenerate $1$-dipole moves. Therefore, the edge formulation of the Ooguri model provides an independent proof of the existence of the $1/N$ expansion, and of the fact that the first term in this expansion is provided by the melonic sector of the theory.

\
Methodologically, it is interesting to point out the recursive nature of this strategy: scaling bounds for the Boulatov model enter the proof of scaling bounds for the Ooguri model. It is therefore tempting to conjecture that such a pattern could be used to devise a general proof of the existence of the $1/N$ expansion, for topological GFTs of arbitrary dimensions. We however do not want to explore this path further, first because the existence of such an expansion has already been settled down \cite{Gurau:2011xq}, and more importantly because we are ultimately interested in $4d$ quantum gravity models.

\
The most interesting outcome of the edge formulation of the Ooguri model is a refinement of the jacket bounds, turning the average of the genera of the jackets into a supremum. While this does not play any role at leading order, this certainly affects the type of graphs contributing to lower orders in $N$. We could therefore expect specific signatures of the Ooguri model, different from that of the i.i.d. $4d$ tensor model, at lower orders in $N$. Where such differences exactly lie, and whether they have the potential to affect the critical behavior of such models has yet to be determined, but it is already clear that the edge formulation could be used to advance further in this direction. 
	
%
%
%
%
%
%
%
%
%
%
%
%
%

\chapter{Renormalization of Tensorial Group Field Theories: generalities}\label{renormalization}

\hfill\begin{minipage}{10cm}
{\footnotesize {\it
The "renormalization group" approach is a strategy for dealing with problems involving many length scales. The strategy is to tackle the problem in steps, one step for each length scale. 
}}
\vspace{0.2cm}
\newline {\footnotesize {\bf Kenneth G. Wilson}, Nobel Lecture, December 1982.}
\end{minipage}

\vspace{0.4cm}

In this chapter, we finally move to the renormalization aspects of this manuscript. We first focus on the general formalism, and highlight the new features one needs to introduce in order to handle TGFTs with connection degrees of freedom. Two examples, taken from \cite{u1} and \cite{su2}, will be treated in details in the last two chapters of this thesis.

\section{Preliminaries: renormalization of local field theories}

For the reader's convenience, we start with a brief introduction of standard material about renormalization theory for local quantum field theories \cite{vincent_book , vincent_2002}. This allows us to introduce the main concepts and multiscale tools we will then generalize to TGFTs.

\subsection{Locality, scales and divergences}

Let us, for the sake of simplicity, consider a bosonic scalar field $\phi$ in $D$-dimensional flat (Euclidean) space-time. The free theory is specified by a kinetic action 
\beq
S_{kin}(\phi) = \frac{1}{2} \int \extd^D x \left( m^2 \phi(x)^2 - \phi(x) \Delta \phi(x) \right)\,,  
\eeq
where $m>0$ is the mass and $\Delta$ the Laplacian on $\mathbb{R}^D$. A quantum interacting theory for such a scalar field can be formally specified through the partition function:
\beq
\cZ = \int \cD \phi \, \e^{- S_{kin} (\phi) - S_{int} (\phi)} \,,
\eeq
where $\cD \phi$ is the (ill-defined) Lebesgue measure on the space of fields, and $S_{int}$ is the interaction part of the action. The latter is constrained by the physical symmetries of the model one is considering, among which Poincaré invariance when space-time is flat, or rather Euclidean invariance in the Wick rotated version we are focusing one. The interactions are required to be \textit{local} in space-time, which means that $S_{int}$ should be a linear combinations of monomials of the form 
\beq
\int \extd^D x \, \phi(x)^k \,, \; k \in \mathbb{N}^*.
\eeq  
We consider the simplest situation here, but of course further restrictions coming from different types of symmetries, for instance gauge symmetries, should be implemented in more complicated theories. In addition to that, we will also require $S_{int}$ to be positive, to avoid quantum instabilities of the vacuum. Hence we assume:
\beq
S_{int} (\phi) = \sum_{k \in \mathbb{N}^{*}} \lambda_{2k} \int \extd^D x \, \phi(x)^{2 k} \,,
\eeq  
where $\lambda_{2k}$ are the (real) bare coupling constants. 

\
Note that the derivatives introduced in the kinetic term softly break the locality of the interactions. This is one aspect we will generalize to the tensorial world, as advocated in \cite{vincent_tt1 , vincent_tt2}. 

\
The kinetic part of the action can be combined with the Lebesgue measure into a well-defined Gaussian measure $\extd \mu_C (\phi)$ whose covariance $C$ encodes the propagator of the theory:  
\beq
\int \extd \mu_C (\phi) \phi(x) \phi(y) = C(x ; y)\,. 
\eeq
In our case, $C$ is invariant under translations, we will therefore use the notation $C(x ; y) = C(x - y)$. The Fourier transform of this function is proportional to $(p^2 + m^2)^{\inv}$, where $p$ is the momentum, thereby introducing an energy scale into the theory. By means of Schwinger's trick, which consists in the simple equation
\beq
\frac{1}{p^2 + m^2} = \int_{0}^{+ \infty} \extd \alpha\, \e^{- \alpha (p^2 + m^2)},
\eeq
we can encode the scale in the real parameter $\alpha$ and show by Fourier transforming back to configuration space that $C$ is equal to:
\beq
C(x - y) = \int_{0}^{+ \infty} \extd \alpha \, \e^{- m^2 \alpha } \frac{\e^{-\vert x - y \vert^2 / 4 \alpha}}{(4 \pi \alpha)^{D/2}}\,.
\eeq
The second factor in the integrand is nothing but the heat kernel on $\mathbb{R}^D$ at time $\alpha$. When $D \geq 2$, this covariance is ill-defined at coinciding points ($x=y$), since in this case the integrand is not integrable in the neighborhood of $\alpha = 0$. This is the source of UV divergences, arising from infinitely close points in configuration space, or equivalently modes with arbitrary high momenta. To make sense of the theory, one should therefore first regularize and then renormalize it. Since the divergences have been identified as coming from the neighborhood of $\alpha = 0$, a convenient way of regularizing the covariance is provided by a cut-off $\Lambda > 0$ on the Schwinger parameter $\alpha$:
\beq
C^\Lambda (x) = \int_{\Lambda}^{+ \infty} \extd \alpha \, \e^{- m^2 \alpha } \frac{\e^{- x^2 / 4 \alpha}}{(4 \pi \alpha)^{D/2}}\,. 
\eeq

\
One can then expand perturbatively this theory with cut-off around the free theory ($\lambda_{2k} = 0$), and for definiteness we can assume that only a finite number of coupling constants $\lambda_4 , \ldots , \lambda_{v_{max}}$ are turned on. By Wick's theorem, the partition function takes the form
\beq
\cZ_\Lambda = \sum_{\cG} \frac{1}{s(\cG)} \prod_{k = 2}^{v_{max}/2} (-\lambda_{2k})^{n_{2k}(\cG)} \, \cA_\cG \,,
\eeq
where $\cG$ are (vacuum) Feynman graphs, $s(\cG)$ a symmetry factor associated to a graph $\cG$, $n_{2k} (\cG)$ its number of vertices of valency $2 k$, and $\cA_\cG$ its amplitude. The amplitude $\cA_\cG$ is computed as usual, by associating (cut-off) propagators to lines, interaction kernels to vertices, and convoluting them according to the pattern of $\cG$. The purpose of perturbative renormalization theory is to give a meaning to this perturbative expansion (as a formal multi-series) in the limit $\Lambda \to 0$, that is when the divergences contained in the propagators start manifesting themselves in the amplitudes. This procedure will in particular imply further constraints on the parameters of the theory, such as the dimension $D$ or the coupling constants.




\subsection{Perturbative renormalization through a multiscale decomposition}

The modern understanding of renormalization theory dates back to Wilson, who provided a deep physical understanding of the mathematical procedures underlying it. We introduce here one such procedure, the multiscale analysis, together with the Wilsonian perspective. 

The key idea of renormalization is to organize physical processes according to the scales at which they happen. The first assumption one can make in this respect, is that a theory with cut-off $\Lambda$ can be trusted when it comes to computing scattering amplitudes between boundary states with small enough energies i.e. such that $(p^2 + m^2) \Lambda \ll 1$. The key questions Wilson asked were: how is this approximation affected when one moves the cut-off up? is it self-reproducing?

\
The multi-scale representation of field theories investigates these questions in a discrete setting. Rather than changing the cut-off in a continuous fashion, one instead uses slices of scales with non-zero width. To this effect, let us fix an arbitrary constant $M>1$, and assume the cut-off to be of the form $\Lambda = M^{-2 \rho}$, with $\rho$ an integer. The half-line $\left[ M^{-2 \rho} , +\infty \right[$ is then partitioned into slices $\left[ M^{-2 i} , M^{-2 (i - 1)} \right[$, with a geometric progression, except for the last $i = 0$ slice which is simply $\left[ 1 , + \infty \right[$. This allows to split the propagation of the field into contributions from these different $i$-labeled scales. Using the simpler notation $C^\rho \equiv C^{\Lambda}$, we can define the propagator at scale $i$ by
\bes\label{ci1}
C_0 (x) &=& \int_{1}^{+ \infty} \extd \alpha \, \e^{- m^2 \alpha } \frac{\e^{- x^2 / 4 \alpha}}{(4 \pi \alpha)^{D/2}}\,, \\
\forall i \geq 1\,, \quad C_i (x) &=& \int_{M^{-2i}}^{M^{-2(i-1)}} \extd \alpha \, \e^{- m^2 \alpha } \frac{\e^{- x^2 / 4 \alpha}}{(4 \pi \alpha)^{D/2}}\,, \\ 
\ees
in such a way that\footnote{Notice the convention adopted from now on: lower indices label the slices, while upper ones indicate the cut-off.}:
\beq\label{cov_scale1}
C^{\rho} = \sum_{i = 0}^{\rho} C_i\,.
\eeq
The random field can then be decomposed into independent contributions $\phi_i$ from each scale $i$, distributed according to the Gaussian measures $\extd \mu_{C_{i}} (\phi_i )$:
\bes
\phi &=& \sum_{i = 0}^{\rho} \phi_i\,,\\
\extd \mu_{C^{\rho}} (\phi) &=& \prod_{i = 0}^{\rho} \extd \mu_{C_i} (\phi_i)\,.
\ees
The renormalization group flow allows to integrate out the highest energy scales, slice by slice. As far as low energy processes are concerned, the fine UV effects can be captured by new effective coupling constants $\lambda_{2k, i}$ and mass $m_i$, which are the new parameters of the theory once the cut-off $\rho$ has been lowered down to $i$. More precisely, discrete flow equations for these coupling constants can be deduced from the ansatz:
\bes\label{flow_gene}
\e^{- S_{int}^{\{ \lambda_{2k , i} \}} (\Phi) } 
&\approx& \int 
\extd \mu_{C_{i+1 
}} (\phi_{i+1})\, \e^{- S_{int}^{\{ \lambda_{2k , i + 1} \} } (\Phi + \phi_{i+1})}\,.
\ees
In this formula, $\Phi = \sum_{j \leq i} \phi_j$ contains the slow moving modes of the field, and new indices have been introduced to take the fact that the coupling constants 
must be redefined at each step into account. To compute the effective action $S_{int}^{^{\{ \lambda_{2k , i} \}}}$ at scale $i$ from its counter-part at scale $i+1$, one needs to first integrate the high energy shell $i+1$, and then compute a logarithm. These two steps can be best understood at the level of the Feynman expansion, where taking the logarithm is equivalent to restricting to connected graphs. This is what makes the latter so useful to renormalization. The approximation procedure itself consists in expanding the leading-order terms in this expansion, the divergent graphs, around their local contributions. These can be reabsorbed into the coupling constants at scale $i$, while the finite corrections are irrelevant in the UV and are therefore discarded. This defines a discrete renormalization group flow, which allows to deduce renormalizable coupling constants (at scale $i = 0$) from the bare ones (at scale $i = \rho$ or $+ \infty$ when this makes sense) to $i = 0$. A theory is called renormalizable if solutions to the general ansatz (\ref{flow_gene}) which involve only a finite number of non-zero coupling constants (at a given scale) can be found. This is nothing but asking the theory to be predictive: with a renormalizable theory, once the coupling constants at a given scale have been measured all (lower energy) scattering processes can be computed in principle. 

\
We already see that renormalizability depends on a subtle interplay between scales and locality. In particular, only the high energy pieces of the divergent graphs will need to be renormalized. It is explicit from the renormalization group ansatz (\ref{flow_gene}), in which by definition only the high energy slices contribute to the corrections to the coupling constants. Once these have been taken care of, no UV divergence remains in a renormalizable theory. 
To see this at the level of the Feynman amplitude of a connected graph $\cG$, one can expand each propagator as in (\ref{cov_scale1}), which gives rise to the multiscale expansion of the amplitude $\cA_\cG$:
\beq
\cA_\cG = \sum_{\mu} \cA_{\cG, \mu}\,,
\eeq 
where $\mu \equiv \{ i_l \,,\; l \in L(\cG)\}$ is an attribution of one scale index per line of $\cG$, and in $\cA_{\cG , \mu}$ the $\alpha$ integrals have been restricted to the appropriate slices. To give full justice to the renormalization group perspective, it is then preferable to bound $\cA_{\cG , \mu}$ in a $\mu$-dependent and optimal way, rather than relying on crude bounds on $\cA_\cG$. It is particularly important if one wants to prove renormalizability at all orders in perturbations. Technically, such bounds are obtained in two steps. One first focuses on the covariance, and proves that there exist constants $K > 0$ and $\delta > 0$ such that
\beq
\forall i \in \mathbb{N}\,, \; \vert C_i (x, y) \vert \leq K M^{(D - 2) i}\, \e^{- \delta M^i \vert x - y \vert}\,.
\eeq 
This bound is a simple consequence of the definitions (\ref{ci1}), and captures the peakedness properties of the propagators at high energy in an optimal way \cite{vincent_book}. One then plugs this bound in the amplitude $\cA_{\cG, \mu}$, and optimize it according to the scale attribution $\mu$. To this effect, one introduces the notion of \textit{high subgraph}, constructed as follows: call $\cG_i$ the set of lines of $(\cG , \mu)$ with scales higher or equal to $i$; we label its connected components $\cG_i^{(k)}$, with $k$ running from $1$ to some integer $k(i)$. The $\cG_i^{(k)}$'s are the high subgraphs of $(\cG , \mu)$, that is its connected subgraphs with internal scales strictly higher than the external scales (i.e. the scales labeling the external legs). Seen from their external legs, such subgraphs 'look' local, because of the high energy carried by the internal propagators. From the general argument above, they are therefore expected to be responsible for the divergences, and for the flow of the coupling constants. It can indeed be shown \cite{vincent_book} that:
\beq\label{pc1}
\vert \cA_{\cG , \mu } \vert \leq K^{L(\cG)}  \prod_{i \in \mathbb{N}} \prod_{ k \in \llbracket 1 , k(i) \rrbracket } M^{\omega [  \cG_{i}^{(k)}]}\,,
\eeq
where $\omega(\cH)$ is the degree of divergence of a subgraph $\cH \subset \cG$. In this particular case:
\beq
\omega ( \cH ) = D F(\cH) - 2 L(\cH)\,, 
\eeq 
where $F$ is the number of loops, and $L$ still denotes the number of internal lines. Equation (\ref{pc1}) is a multiscale generalization of the usual superficial power-counting argument, which allows to systematically investigate renormalizability at all orders in perturbation theory. The \textit{divergent subgraphs} are those with $\omega \geq 0$, and can generate divergences when they are high, while the contributions from \textit{convergent subgraphs} ($\omega < 0$) are always finite. Therefore only the former need to be renormalized. 

We will introduce the full machinery of perturbative renormalization directly at the level of the TGFTs. We only recall here how renormalizability can be inferred from the fundamental bound (\ref{pc1}). The key point is to recast the divergent degree in an appropriate form. In the present situation, we can use the combinatorial relations:
\beq
F = L - n + 1\,, \qquad L = \sum_{k = 2}^{v_{max} / 2 } k n_{2k} - \frac{N}{2} \,,
\eeq   
where $N$ is the number of external legs. This allows to show that
\beq
\omega = D - \frac{D-2}{2} N + \sum_{k = 2}^{v_{max} / 2 } \left[ (D-2) k - D \right] n_{2k}\,.
\eeq
To ensure renormalizability, $\omega$ must first be bounded from above, otherwise arbitrarily fast divergences would appear in the UV. Second, the number of external legs of a divergent graph must also be bounded, otherwise the divergences would turn on local couplings $\lambda_{2 k}$ with arbitrarily high $k$. These two conditions require the coefficient in front of $N$ to be positive, and the coefficients in front of each $n_k$ to be negative, hence:
\beq
D \geq 2 \,, \qquad v_{max} \leq \frac{2 D}{D - 2}\,.
\eeq
When $D = 2$, $\omega = D (1 - n)$, therefore only single-vertex graphs (called \textit{tadpoles}) contribute to the divergences. There is only a finite number of divergent graphs in this case, irrespectively of the value of $v_{max}$. Such a theory is called \textit{super-renormalizable}, and can be renormalized via a Wick ordering procedure, for any polynomial interaction \cite{Simon}. The models presented in Chapter \ref{chap:u1} are tensorial generalizations of these $P(\Phi)_2$ theories. When $D \geq 5$, $v_{max} < 4$, and therefore no renormalizable interacting theory can exist. The only two remaining possibilities are: $D=3$ and $v_{max} = 6$, or $D = 4$ and $v_{max} = 4$. They are called \textit{just-renormalizable} \cite{salmhofer , vincent_book}, because they generate an infinite number of divergent graphs. As far as fundamental interactions are concerned, just-renormalizablity turns out to be the norm, and we will therefore pay them special attention in the tensorial world as well. On general grounds, we can also expect a close parallel between models supporting a $1/N$ expansion and just-renormalizable field theories. The first condition for the existence of a $1/N$ expansion, an upper bound on the divergence degree, is the analogue of the renormalizability condition; and the second, the existence of an infinite number of leading-order graphs, is similar to the additional just-renormalizability condition. Field theories are only more refined in the sense that the specific scaling of the coupling constants one chooses by hand in a $1/N$ expansion is instead implemented through a non-trivial propagator kernel. The large $N$ phase is so to speak dynamically encoded in the renormalization group flow.  

\
A last point we would like to comment on is the question of \textit{asymptotic freedom}. The $\phi^4$ theory in four dimensions we have just introduced suffers from a \textit{Landau pole} effect: having the renormalized coupling constant ($\lambda_{4,0}$) fixed to a finite value requires the effective coupling constants $\lambda_{4, i}$ to blow up when $i \to + \infty$. Therefore the perturbative expansion cannot be trusted all the way down to the UV. Such theories cannot easily be expected to describe fundamental interactions, in particular gravity. QCD on the contrary is an example of asymptotically free theory, which means that its coupling constant flows to $0$ in the UV. This property is responsible for the phase transition from quarks to hadrons, and is in this sense also desirable in TGFT: not only it could make the perturbative treatment of a quantum theory of gravity well-defined all the way down to the UV, but it could also dynamically generate phase transitions such as discrete-to-continuum ones \cite{vincent_tt1 , vincent_tt2}.  









\section{Locality and propagation in GFT}

In order to extend the perturbative techniques outlined in the previous section to GFTs, one needs to equip oneself with appropriate generalizations of the key structures entering renormalization theory. Of particular importance is the splitting of the action into kinetic and interaction terms, which is the guiding thread we decide to follow here. If we look at space-time based quantum field theories from a more abstract perspective, the notion of scale encoded in the spectrum of the propagator is to be seen as primitive, irrespectively of its physical interpretation. It is nothing more than an abstract parameter allowing to decimate the degrees of freedom, going from regions with many degrees of freedom to regions with few. Such a definition of scale survives in our background independent context, even in the absence of any clear physical interpretation. As explained in the introduction of this thesis, our attitude at this stage is to regard the specific notion of scale we will use as a fundamental postulate of the models we will consider. In a sense we are reversing the usual perspective: instead of 
relying on a predefined and physically transparent notion of scale to construct and interpret renormalization theory, we introduce an abstract renormalization group as a primary structure and look for its consequences. It remains to be seen how fruitful such a point of view can be on the physical side, but what will already be clear in this thesis is that it allows to construct mathematically consistent perturbative GFTs. In addition to the propagator, entailing the definition of scale, we will need to adopt a locality principle, understood as a prescription for the possible non-Gaussian perturbations introduced in the interaction. This will again have nothing to do with a space-time based notion of locality, but will play the same role in the subtraction of infinities, hence keeping the same nomenclature seems justified. In this section, we review possible choices of propagators and locality principles, which also gives us the opportunity to place the results of this thesis in the context of past and future works.

 
\subsection{Simplicial and tensorial interactions}

As already explained, essentially two types of interactions have been explored so far in the GFT literature: simplicial interactions and tensor invariants. While the first choice seems natural in the discrete gravity context of spin foam models, it seems rather awkward when it comes to GFT. Indeed, when interpreted as Feynman amplitudes of a quantum field theory, spin foam amplitudes have no reason to be built out of a single vertex interaction. The type of interactions one allows in a field theory should be specified by a locality principle, as opposed to (heuristically motivated but) essentially arbitrary combinatorial restrictions. On top of clarifying the assumptions made in the construction of a model, a locality principle is essential to renormalization, providing the necessary flexibility to reabsorb divergences in effective couplings. Complementarily, it is in this context that renormalization proves its full power, leading in the best case scenario to a finite set of relevant couplings. A last important point which can be achieved with tensorial interactions but seems out of reach with simplicial ones is positivity, which is at the core of the stability of the vacuum. 

\
Since the aim of this thesis is to explore how GFTs can be turned into rigorous QFTs, adopting tensorial invariance as a locality principle seems obviously preferred to the rigid framework of simplicial interactions. Motivations from the gravity side are so far missing, but given the dramatic progress entailed by tensor invariance in tensor models, and the arbitrariness of simplicial restrictions, it is at the very least one first interesting prescription to explore further. Moreover, there is a limited relation between color models and tensor invariant ones: i.i.d tensor invariant models are the single field effective versions of colored i.i.d
models. This exact correspondence does not survive when spin foam constraints are added, as is for example illustrated by the vertex (resp. edge) formulation of the Boulatov (resp. Ooguri) model. Although a bubble factorization is achieved with this method, thus providing an effective single-field model where interactions are labeled by the bubbles, these interactions are not tensor invariants. This is due to the presence of $\delta$-functions enforcing flatness conditions inside each bubble. These enriched bubble interactions can provide an alternative locality principle, however presumably leading to a more difficult renormalizability analysis. On top of that, it is not clear so far whether such complications are relevant, we therefore decide to stick to tensor invariance for now. Moreover, if one modifies the colored Boulatov model in such a way that only one field is subject to the Gauss constraint, the other three being i.i.d, one can integrate out the i.i.d fields and obtain a single-field tensor invariant version of the Boulatov model. Therefore the correspondence between colored and tensor invariant models is not completely lost when leaving the i.i.d world.

\
Let us summarize the arguments invoked to use tensor invariant rather than simplicial interactions:
\begin{itemize}
\item tensor invariant interactions are specified by a locality principle;
\item they are infinitely many;
\item positivity of the interactions can be achieved with tensor invariant interactions;
\item a tensor invariant GFT can be obtained as a single-field effective model of a colored GFT. 
\end{itemize}
Group field theories based on such a locality principle have been called Tensorial Group Field Theories (TGFTs). A large variety of such models can be considered, depending on the rank of the tensors, the space in which the indices leave, and the propagator. 
%
%
%
%
%

\subsection{Constraints and propagation}

The second element entering the definition of a GFT we would like to restrict is the propagator, or equivalently the Gaussian measure around which the local interactions introduce perturbations. In all the GFTs encountered so far in this thesis, the propagator is a projector, hence ultra-local. For instance, the propagator of i.i.d tensor models reduces to the identity operator, while that of the Boulatov and Ooguri models is the projector on gauge invariant fields. Both have spectra contained in  $\{ 0 , 1 \}$, and therefore cannot provide any non-trivial notion of scale. A natural way of changing this situation, originally advocated in \cite{ValentinJoseph, tensor_4d}, is to simply add non-trivial differential operators on top of such structures. The specific operator chosen in \cite{tensor_4d}, concerned with a $\U(1)$ rank-$4$ TGFT without constraints was the Laplacian. While this conservative choice is certainly worth-exploring, it is not clear yet whether it is unique in any sense. Motivations from the quantum gravity side are inexistent, but the Laplacian is to some extent called for by quantum field theory. First, the analysis of \cite{ValentinJoseph} suggests that the radiative divergences of the $2$-point functions of the Boulatov and Ooguri models generate an additional Laplacian term, which should therefore be included in the bare theory to start with. Second, it was advocated in \cite{vincent_tt2} that some generalization of Osterwalder-Schrader positivity, which is one of the axioms of space-time based Euclidean quantum field theories, should also hold in the context of TGFTs. With this assumption, Laplace-type propagators are singled out as the positive propagators with the strongest decay in the UV. Hence they are the 'best' operators one can consider from the point of view of renormalization\footnote{One can also very well consider models with lower order differential operators, as for example in \cite{josephsamary}, but using second order operators leads to a maximal set of renormalizable couplings.}. In our opinion, these partial arguments give enough support to motivate studies of TGFTs with Laplacian propagators, which will be the focus of the remainder of this thesis. However, a sense of necessity is certainly lacking, and complementary motivations coming from the gravity side would be most welcomed. 

\

The work on TGFTs we are about to introduce was motivated by one key breakthrough: a four-dimensional model with group $\U(1)$ of the purely tensor type (i.e. without any constraint) could be proven perturbatively renormalizable at all orders \cite{tensor_4d}, and even asymptotically free \cite{josephaf} soon after! This model has moreover up to $\vphi^6$ terms, a property which together with the absence of Landau pole in the UV, contrasts with the ordinary $\vphi^4_4$ local quantum field theory. It was already clear at this stage that introducing tensor invariance and Laplace terms in GFTs was a good move, with the first example of perturbatively well-defined GFT all the way down to the UV. This theory is however purely combinatorial in nature, with no discrete geometric data carried by the group elements at the level of the amplitudes. The $\U(1)$ group arises on the contrary merely as the Fourier dual of $\mathbb{Z}$, providing continuous labels for infinite size tensors. The regularization of such a model in group space is therefore akin to a large $N$ cut-off on the size of the tensors, and the renormalization group flows from larger to smaller values of $N$. 

A program aiming at bridging the gap with spin foam 4d quantum gravity models was initiated in \cite{u1}. The first step in this program consists in reintroducing connection degrees of freedom in the amplitudes, thanks to a Gauss constraint. This step is now essentially completed, and the subject of the end of this manuscript. In \cite{u1} the general class of field theories with gauge invariance and Laplace terms both implemented in the propagator was introduced. For simplicity a $\U(1)$ four-dimensional model, with no geometric meaning but same structure as more involved models, was investigated in details, and proven super-renormalizable at all orders, for any finite number of tensorial interactions. On top of introducing the main techniques allowing to handle the holonomy contributions in the amplitudes, this example confirmed that gauge invariance models are generically 'more' renormalizable than their purely tensorial counter-parts. Such results were extended to similar Abelian models in \cite{fabien_dine}, which were proven just-renormalizable. Meanwhile, a more systematic study of just-renormalizablity was explored in \cite{su2}, and the first non-Abelian model could be proven just-renormalizable. This latter model possesses the additional (very welcomed) following properties: a) it is three-dimensional and therefore the symmetry group $\SU(2)$ matches the dimensionality of quantum space-time; b) it is the only potentially just-renormalizable model for which this geometric interpretation holds. From this point of view, the first result of \cite{tensor_4d} generalizes very well to GFTs related to non-trivial spin foam models, and in particular to a variant of the Boulatov model for Euclidean gravity in three dimensions. The second aspect, asymptotic freedom, was first settled for two Abelian models in \cite{dineaf}, and is currently under investigation for the rank-$3$ $\SU(2)$ model of \cite{su2}. We will present partial calculations at the very end of the last chapter, suggesting that the latter is again asymptotically free \cite{beta_su2}. This is very encouraging, and supports the idea that asymptotic freedom might be generic in TGFTs \cite{joseph_d2}. All in all, it seems that TGFTs are now mature enough to be applied to four-dimensional quantum gravity models. The next step in the near future will therefore be to introduce some version of the simplicity constraints entailing geometricity in 4d. First calculations have been already done for the EPRL model \cite{aldo}, of the kind which in our view should be more systematically performed in a TGFT context. 

\
The remainder of this chapter is concerned with the general structure of TGFTs with gauge invariance, and their renormalization. We will in particular conclude with the general classification of potentially just-renormalizable models established in \cite{su2}. In the next two chapters we will move to concrete examples, first the Abelian models from \cite{u1}, and finally the more physically relevant TGFT of \cite{su2}. 
%
%
%
%
%

\section{A class of models with closure constraint}

\subsection{Definition}

A generic TGFT is a quantum field theory of a single tensorial field, with entries in a Lie group. We assume $G$ to be a compact Lie group of dimension $D$, and the field to be a rank-$d$ complex function $\vphi( g_1 ,
\dots , g_d )$, with $d \geq 3$. The statistics is then defined by a partition function
\beq
\cZ = \int \extd \mu_C (\vphi , \vphib) \, \e^{- S(\vphi , \vphib )}\,,
\eeq
where $\extd \mu_C (\vphi , \vphib)$ is a Gaussian measure characterized by its covariance $C$, and $S$ is the interaction part of the action. 

\
The first ingredient, locality as tensor invariance, can be thought of as a limit of a $U(N)^{\otimes d}$ invariance,
where $N$ is a cut-off on representation labels (e.g. spins) in the harmonic expansion of the field. To avoid mathematical complications, we will not try to characterize this symmetry group in the $N \to + \infty$ limit, but simply define tensor invariants as convolutions of a certain number of fields $\vphi$ and $\vphib$ such that any $k$-th index of a field $\vphi$ is
contracted with a $k$-th index of a conjugate field $\overline{\vphi}$. As explained in Chapter \ref{color_tensor}, they are dual to $d$-colored graphs, built from two types of nodes and $d$ types of colored edges: each white (resp. black) dot
represents a field $\vphi$ (resp. $\vphib$), while a contraction of two indices in position $k$ is associated to an edge with color label $k$. Connected such graphs, the $d$-bubbles, generate the set of 
connected tensor invariants. 
See Figure \ref{tensorgraphs} for examples in dimension $d = 4$. 
We assume that the interaction part of the action is a sum of such connected invariants
\beq
S(\vphi , \vphib) = \sum_{b \in \cB} t_b I_b (\vphi , \vphib)\,, 
\eeq
where $\cB$ is a finite set of $d$-bubbles, and $I_b$ is the connected invariant encoded by the bubble $b$.
\begin{figure}[h]
\begin{center}
\includegraphics[scale=1]{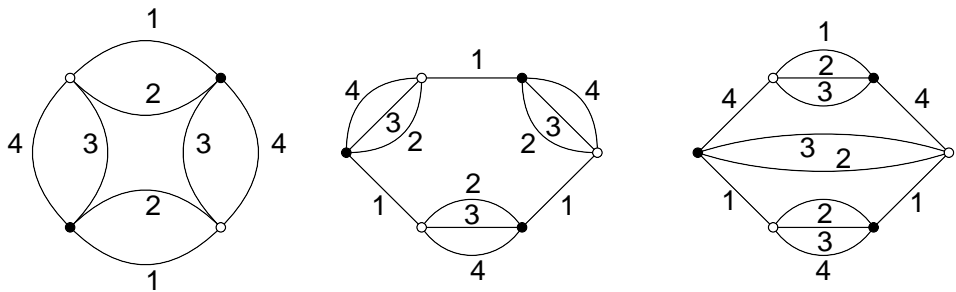}
\caption{Three connected tensor invariants in $d=4$.}
\label{tensorgraphs}
\end{center}
\end{figure}

\
The Gaussian measure $\extd \mu_C$ must first implement a non-trivial dynamics, through a propagator 
\beq\label{propa_nogauge}
\left( m^2 - \sum_{\ell = 1}^{d} \Delta_\ell \right)^{-1} \,,
\eeq
where $\Delta_\ell$ is the Laplace-Beltrami operator on $G$, acting on color-$\ell$ indices. On top of that it must take the gauge invariance condition
\beq \label{gauge}
\forall h \in G \,, \qquad \vphi(h g_1, \dots , h g_d ) = \vphi(g_1, \dots , g_d )\,
\eeq
into account.
The resulting covariance can be expressed as an integral over a Schwinger parameter $\alpha$ of a product of heat kernels on $G$ at time $\alpha$:
\bes
\int \extd \mu_C (\vphi , \vphib) \, \vphi(g_1 , \dots , g_d ) \vphib(g_1' , \dots , g_d' ) &=& C(g_1, \dots , g_d ; g_1' , \dots , g_d' ) \\
&\equiv& \int_{0}^{+ \infty} \extd \alpha \, \e^{- \alpha m^2} \int \extd h \prod_{\ell = 1}^{d} K_{\alpha} (g_\ell h g_\ell'^{\inv})\,.
\ees
This is nothing but a gauge-averaged version of the Schwinger representation of (\ref{propa_nogauge}). This decomposition of the propagator provides an intrinsic notion of scale, parametrized by $\alpha$. Divergences result from the UV region (i.e. $\alpha \to 0$), hence the need to introduce a cut-off ($\alpha \geq \Lambda$), and subsequently to remove
it via renormalization. Note also that the definition of high energy region is canonical once the propagator is specified: it simply corresponds to this corner of $\alpha$-space from which divergences originate, and should not be endowed with any a priori geometric meaning. 
\begin{figure}[b]
\begin{center}
\includegraphics[scale=0.5]{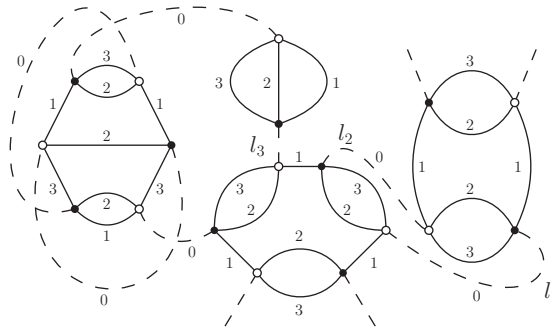}
\caption{A graph with $4$ vertices, $6$ lines and $4$ external legs in $d = 3$.}
\label{example_graph}
\end{center}
\end{figure}
\begin{figure}[b]
\begin{center}
\includegraphics[scale=0.7]{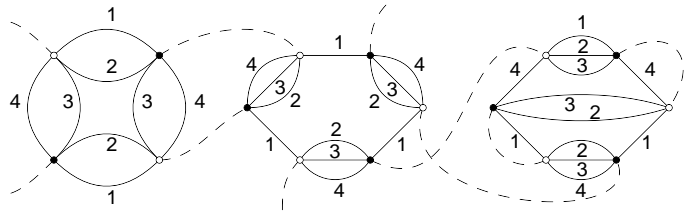}
\caption{A graph with $3$ vertices, $6$ lines and $4$ external legs in $d = 4$.}
\label{coloredgraph}
\end{center}
\end{figure}
Like uncolored tensor models, the perturbative expansion of such a theory is captured by Feynman graphs whose vertices are $d$-bubbles, and whose propagators are associated to an additional type of colored edges, of color $\ell = 0$, represented as dashed lines. When seen on the same footing,
these $d+1$ types of colored edges form $(d+1)$-colored graphs. 
To a Feynman \textit{graph} $\cG$, whose elements are $d$-bubble vertices ($V(\cG)$) and color-$0$ lines ($L(\cG)$), is therefore uniquely associated a $(d+1)$-colored graph $\cG_c$, called the
\textit{colored extension} of $\cG$. See Figures \ref{example_graph} and \ref{coloredgraph} for examples of Feynman graphs in $d = 3$ and $4$.
The connected Schwinger functions are given by a sum over line-connected Feynman graphs:
\beq
\cS_N = \sum_{\cG \; \mathrm{connected}, N(\cG)= N} \frac{1}{s(\cG)} \left(\prod_{b \in \cB} (- t_b)^{n_b (\cG)}\right) \cA_\cG \,,
\eeq
where $N(\cG)$ is the number of external legs of a graph $\cG$, $n_b (\cG)$ its number of vertices of type $b$, and $s(\cG)$ a symmetry factor. The amplitude $\cA_\cG$ of $\cG$ is expressed in terms of holonomies along its faces,
which can be easily defined in the colored extension $\cG_c$: a \textit{face} $f$ of color $\ell$ is a maximal connected subset of edges of color $0$ and $\ell$. In $\cG$, $f$ is a set of color-$0$ lines, from which the holonomies are constructed.
We finally use the following additional notations: $\alpha(f) \equiv \underset{e \in f}{\sum} \alpha_e$ is the sum of the Schwinger parameters appearing in the face $f$; $\epsilon_{ef} = \pm 1$ or $0$ is the adjacency or incidence matrix, encoding the line content of faces
and their relative orientations; the faces are split into closed ($F$) and opened ones ($F_{ext}$); $g_{s(f)}$ and $g_{t(f)}$ denote boundary variables in open faces, with functions $s$ and $t$ mapping open faces to their ``source" and "target" boundary variables. 
The amplitude $\cA_\cG$ takes the form:
\begin{eqnarray}\label{ampl_ab}
\cA_\cG &=& \left[ \prod_{e \in L(\cG)} \int \extd \alpha_{e} \, e^{- m^2 \alpha_e} \int \extd h_e \right] 
\left( \prod_{f \in F (\cG)} K_{\alpha(f)}\left( \overrightarrow{\prod_{e \in f}} {h_e}^{\epsilon_{ef}} \right) \right) \nn\\
&&\left( \prod_{f \in F_{ext}(\cG)} K_{\alpha(f)} \left( g_{s(f)}
\left[\overrightarrow{\prod_{e \in f}} {h_e}^{\epsilon_{ef}}\right] g_{t(f)}^{\inv} \right) \right) \,.
\end{eqnarray}

\
An important feature of the amplitude of $\cG$ is a $G^{V(\cG)}$ gauge symmetry:
\beq
h_e \mapsto g_{t(e)} h_e g_{s(e)}^{\inv}\,,
\eeq
where $t(e)$ (resp. $s(e)$) is the target (resp. source) vertex of an (oriented) edge $e$, and one of the two group elements is trivial for open lines. It is the gauge invariance \eqref{gauge} imposed on the TGFT field that is responsible of this gauge invariance at the level of the Feynman amplitudes, and for their expression \eqref{ampl} as a lattice gauge theory on $\cG$. It is completely analogous to the gauge symmetry of the Boulatov and Ooguri models, except that it is associated to bubbles rather than individual nodes of the colored graphs. When $\cG$ is connected, it is convenient to gauge fix the $h$ variables along a spanning
tree $\cT$ of the graph:
$$
h_e = \one
$$
in the integrand of (\ref{ampl_ab}), for every line $e \in \cT$. We will use such gauge fixings in the following.



\subsection{Graph-theoretic and combinatorial tools}

We collect here a number of definitions and results, first introduced in \cite{u1} and subsequently refined in \cite{su2}, which are key to the analysis we will perform in the following. 
The shift from the usual QFT notion of locality to tensor invariance requires non-trivial generalizations of the ordinary graph-theoretic notions underlying renormalization theory. Probably
the most important one in this respect is that of \textit{quasi-local subgraphs}, that is \textit{connected subgraphs} which, from the point of view of their external legs, \textit{look} local. 

A major role
will be played by the faces of the graph, which is where the curvature of the discrete connection introduced by the new gauge invariance condition is assigned. To be clear, the faces are followed easily by drawing the 
colored extension $\cG_c$ of the graph $\cG$. The color-$\ell$ faces of $\cG$ are the alternating circuits of lines of color 0 and $\ell$ in $\cG_c$, and can be either closed 
(internal) or open (external). Rather than the usual incidence matrix $\epsilon_{ev}$ between lines and vertices of ordinary graph theory, it is the incidence matrix of lines and closed faces $\epsilon_{ef}$ in $\cG$ which plays the leading role in TGFTs \cite{lrd, vincent_renGFT, lin, tensor_4d}, as shown by formula (\ref{ampl_ab}). To precisely define this matrix one needs an orientation of both the lines and the faces. Then $\epsilon_{ef}$ is $+1$ if the face $f$ goes through the line $e$ with the same orientation, $-1$ if the face $f$ goes through the line $e$ with 
opposite orientation, and finally $0$ otherwise. The colored structure ensures the absence of tadfaces, i.e. faces which pass several times through the same line, hence the $\epsilon_{ef}$ is well-defined.

\
We start with the notion of subgraph. In ordinary graph theory a subgraph of a graph $\cG$ is most conveniently defined as a {\emph{subset} $\cH$ \emph{of lines}} of $\cG$, so that a graph with $L$ lines has exactly $2^L$
subgraphs. Such a subset of lines is then completed canonically by adding the vertices attached to the lines and the external lines, also called legs. The latter are defined by first cutting in the middle all
lines of $\cG \setminus \cH$. Legs of $\cH$ then correspond 
either to true legs of $\cG$ attached to vertices of $\cH$ or to half-lines of $G \setminus \cH$ attached to the vertices of $\cH$. Finally, ordinary connectedness of $\cH$ can be defined
in terms of the ordinary incidence matrix $\epsilon_{ev}$ of $\cH$: the connected components of $\cH$ correspond to the maximal factorized rectangular blocks of this matrix. Hence 
elementary connections between lines come from their common attached vertices.

Recalling that a tensorial graph $\cG$ has color-$0$ internal lines and external legs, $d$-bubbles as vertices, and faces, 
the definition of a subgraph for TGFTs is a natural generalization of the ordinary definition.

\begin{definition}
A \textit{subgraph} $\cH$ of a graph $\cG$ is a subset of  lines of $\cG$, hence $\cG$ has exactly $2^{L(\cG)}$ subgraphs. 
$\cH$ is then completed by first adding the vertices that touch its lines. 
The faces closed in $\cG$ which  pass only through lines of $\cH$ form the set of \textit{internal faces} of $\cH$. The 
external faces of $\cH$ are the maximal open connected pieces of either open or closed faces of $\cG$ that pass through lines of $\cH$. 
Finally all the external legs or half-lines of $\cG \setminus \cH$ touching the vertices of $\cH$ are considered \textit{external legs} of $\cH$.
\end{definition}

We denote $L(\cH)$ and $F(\cH)$ the set of lines and internal faces of $\cH$, and $N(\cH)$ and $F_{ext}(\cH)$ the set of external legs and external faces. When no confusion is possible we also write $L$, $F$ etc for the cardinality 
of the corresponding sets. Moreover, the subgraph made of the lines $l_1, \ldots, l_k$ will simply be denoted $\{l_1 , \ldots , l_k \}$.

\

\noindent {\bf{Example.}} In Figure \ref{example_graph}, $\cH_{12} = \{ l_1 , l_2 \}$ has two lines ($L(\cH_{12})=2$) which touch two vertices, giving $V(\cH_{12})=2$. Six additional half-lines are hooked up to these two bubbles, giving a total of $N( \cH_{12} )=6$ external legs. Finally, $\cH_{12}$ has four faces in total: two of them are internal, of color $2$ and $3$ respectively, hence $F(\cH_{12}) = 2$; the two others are external faces of color $1$, hence $F_{ext}(\cH_{12}) = 2$. Note that the connected pieces of (the colored extension of) $\cH_{12}$ which consist of two external legs and a single colored line should not be considered as external faces.   

\

On top of the usual notion of connectedness of subgraphs, to which we will refer as \textit{vertex-connectedness} in order to avoid any confusion, we will heavily rely on the similar concept of \textit{face-connectedness}. While the former focuses on incidence relations between lines and vertices, the latter puts the emphasis on incidence relations between lines and faces.

\begin{definition}
\begin{enumerate}[(i)]
\item The \textit{face-connected components} of a subgraph $\cH$ are defined as the subsets of lines of the maximal factorized rectangular blocks of its $\epsilon_{ef}$ incidence matrix (with entries in $L(\cH) \times F(\cH)$). 
\item A subgraph $\cH$ is called \emph{face-connected} if it has a single face-connected component.
\item Let $\cG$ be a graph. The face-connected subgraphs $\cH_1 , \ldots , \cH_k \subset \cG$ are said to be \textit{face-disjoint} if they form exactly $k$ face-connected components in their union $\cH_1 \cup \dots \cup \cH_k$. 
\end{enumerate}
\end{definition}
The notion of face-connectedness is finer than vertex-connectedness, in the sense that any face-connected subgraph is also vertex-connected. It should also be noted that with the previous definition, the face-disjoint subgraphs $\cH_1, \ldots , \cH_k \subset \cG$ can consist of strictly less than $k$ face-connected components in $\cG$ itself. What really matters is that there \textit{exists} a subgraph of $\cG$ into which $\cH_1 , \ldots , \cH_k$ form $k$ face-connected components. 

\

\noindent {\bf{Examples.}} In Figure \ref{example_graph}, $\cH_{12} = \{ l_1 , l_2 \}$ and $\cH_{123} = \{ l_1 , l_2 , l_3 \}$ are both vertex-connected, while only $\cH_{12}$ is face-connected. $\cH_{123}$ has two face-connected components: $\{ l_3 \}$ and $\{ l_1 , l_2 \}$. In Figure \ref{ex_con_comp}, $\cH_{1} = \{l_1\}$ and $\cH_2 = \{ l_2 \}$ are face-disjoint because they are their own face-connected components in $\cH_1 \cup \cH_2 = \{ l_1 , l_2\}$. On the other hand, they are not face-connected components of $\cH_{123}$, which is itself face-connected. This illustrates the subtelty in the definition of face-disjointness we just pointed out.

\begin{figure}[ht]
\begin{center}
\includegraphics[scale=0.5]{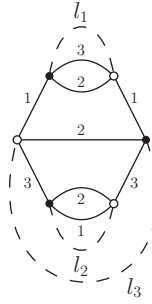}
\caption{$\cH_{1} = \{l_1\}$ and $\cH_2 = \{ l_2 \}$ are face-disjoint.}
\label{ex_con_comp}
\end{center}
\end{figure}

\

It is convenient to define elementary operations on TGFT graphs at the level of their underlying colored graphs, where dipole moves play a prominent role. In the renormalization context, dipole moves will be used as a way to consistently erase faces with high scales. Therefore topological considerations will be essentially irrelevant, and we adopt from now on a new definition of dipoles, which does not distinguish degenerate from non-degenerate ones. Moreover, since only color-$0$ lines carry propagators, we also only consider those dipoles having an internal line of color $0$.
\begin{definition}
Let $\cG$ be a graph, and $\cG_c$ its colored extension. For any integer $k$ such that $1 \leq k \leq d+1$, a $k$-dipole is a line of $\cG$ whose image in $\cG_c$ links two nodes $n$ and
$\overline{n}$ which are connected by exactly $k - 1$ additional colored lines.
\end{definition}

\begin{definition}
Let $\cG$ be a graph, and $\cG_c$ its colored extension. The contraction of a $k$-dipole $d_k$ is an operation in $\cG_c$ that consists in:
\begin{enumerate}[(i)]
 \item deleting the two nodes $n$ and $\overline{n}$ linked by $d_k$, together with the $k$ lines that connect them;
 \item reconnecting the resulting $d - k + 1$ pairs of open legs according to their colors.
\end{enumerate}
We call $\cG_c / d_k$ the resulting colored graph, and $\cG / d_k$ its pre-image. See Figure \ref{k_dipole}.
\end{definition}

\begin{figure}[h]
\begin{center}
\includegraphics[scale=0.6]{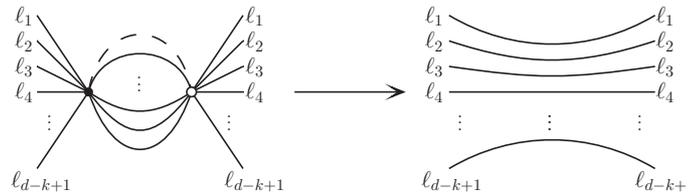}
\caption{Contraction of a $k$-dipole line.}
\label{k_dipole}
\end{center}
\end{figure}

This purely combinatorial notion of contraction of lines can be extended to arbitrary subgraphs. 
\begin{definition}
We call contraction of a subgraph $\cH \subset \cG$ the successive contractions of all the lines of $\cH$. The resulting graph is independent of the order in which the lines of $\cH$ are contracted,
and is noted $\cG / \cH$.
\end{definition}
\begin{proof}
To confirm that this definition is consistent, we need to prove that dipole contractions are commuting operations. Consider two distinct lines $e_1$ and $e_2$ in a graph $\cG$, and call $\cH$ the
subgraph made of 
$e_1$ and $e_2$. We distinguish three cases.
\begin{enumerate}[(i)]
 \item $\cH$ is disconnected. This means that $e_1$ and $e_2$ are part of two independent dipoles, with no colored line in common, and the two contraction operations obviously commute.
 \item $\cH$ is connected, and none of its internal faces contain both $e_1$ and $e_2$. This means that $e_1$ and $e_2$ are contained in two dipoles $d_1$ and $d_2$, such that for each
color $i$, at most one line of color $i$ connects $d_1$ to $d_2$. Hence contracting $d_1$ (resp. $d_2$) does not change the nature of the dipole in which $e_2$ (resp. $e_1$) is contained. 
So here again, $d_1$ and $d_2$ are local objects which can be contracted independently. 
 \item $\cH$ has $q \geq 1$ internal faces containing both $e_1$ and $e_2$. In this case, the contraction of $e_1$ (resp. $e_2$) changes the nature of the dipole in which $e_2$ (resp. $e_1$) is
contained: $q$ internal faces are added to it. However, when contracting the second tadpole, these faces are deleted, so for any order in which the contractions are performed, 
all the internal faces are deleted. As for the external faces, the situation is the same as in the previous case, and the two contractions commute.
\end{enumerate}
\end{proof}
We can also give a more global characterization of the contraction operation.
\begin{proposition}
Let $\cH$ be a subgraph of $\cG$, and $\cH_c$ its colored extension. The contracted graph $\cG / \cH$ is obtained by:
\begin{enumerate}[(a)]
 \item deleting all the internal faces of $\cH$;
 \item replacing all the external faces of $\cH_c$ by single lines of the appropriate color.
\end{enumerate}
\end{proposition}
\begin{proof}
We prove this by induction on the number of lines in $\cH$. If $\cH$ contains one single line, then it is a dipole, and the proposition is true according to the very definition of a dipole
contraction. Now, suppose that $\cH$
is made of $n>1$ lines, $n-1$ of them being contained in the subgraph $\cH_0 \subset \cH$, and call the last one $e$. The set of internal faces in $\cH$ decomposes into several subsets. The faces
which are internal to $\cH_0$ are deleted 
by hypothesis when contracting $\cH_0$. Those common to $\cH_0$ and $e$ become internal dipole faces once $\cH_0$ is contracted, so they are deleted when $e$ is contracted, and the same is of course
true for the remaining internal faces which have $e$ as single line. The same distinction of cases applied to external faces of $\cH_c$ allows to prove that they are replaced by single lines of the
appropriate color, which achieves
the proof. 
\end{proof}

Contracting a subgraph $\cH \subset \cG$ can heavily modify the connectivity properties of $\cG$, depending on the nature of the dipoles this operation involves. It is indeed easy to check that:
\begin{proposition}\label{disconnected}
\begin{enumerate}[(i)]
 \item For any vertex-connected graph $\cG$, if $e$ is a line of $\cG$ contained in a $d$-dipole, then $\cG / e$ is vertex-connected.
 \item For any $1 \leq q \leq d - k + 1$, there exists a connected graph $\cG$ and a $k$-dipole $e$ such that $\cG / e$ has exactly $q$ connected components. 
\end{enumerate}
\end{proposition}
This point is to be contrasted with usual graph theory, where an elementary contraction moves simply amounts to shrinking a line until its two end vertices get identified, and therefore conserves (vertex-)connectedness.  

\

Let us now understand how these combinatorial moves affect the amplitudes.
$1$-dipole contractions play a particularly important role in colored tensor models and GFTs, because they implement the topological notion of connected sum of $d$-bubbles. This remains true in our context, the
relevant move being the contraction of a $1$-dipole which is not a tadpole (i.e. a non-degenerate $1$-dipole in the usual nomenclature). Interestingly, the contraction of a full set of such
$1$-dipoles is intimately related to the gauge-fixing procedure sketched before. In a connected graph $\cG$, a maximal set of $1$-dipoles can be successively contracted, by picking up a maximal tree of lines $\cT \subset \cG$.
But we also know that in the Feynman amplitude of $\cG$, the group elements associated to this tree can be set to $\one$. Incidentally, the purely combinatorial notion of contraction of lines of $\cT$
is nothing but the result of trivial convolutions in the amplitude $\cA_\cG$. The only difference between $\cA_\cG$ and $\cA_{\cG/\cT}$, where $\cG / \cT$ denotes the fully contracted graph, is a set
of simple integrals with respect to Schwinger parameters, while their integrands have exactly the same structure. This observation is crucial to the definition of quasi-locality in TGFTs, which must be based on the properties of the \textit{reduced graphs} $\cG / \cT$ rather than $\cG$ itself.

\
In usual field theories, contraction moves are used to approximate high energy pieces of the amplitudes by local vertices. This is easily understood since propagators tend to bring neighboring vertices together, and increasingly so the higher the energy. In TGFTs, a similar picture can be proposed by focusing on faces rather than lines. A closed face $f$ will tend to be associated to a trivial holonomy when the Schwinger parameter $\alpha (f)$ gets small, that is when all the propagators along $f$ have high scales. Therefore, in high subgraphs one will always be able to approximate the closed face holonomies by $\one$, which can be pictured as shrinking these faces to points. However, this does not mean that this approximation can be recast as the amplitude of a simplified graphs, for elementary holonomies along individual lines are not necessarily themselves close to $\one$. 
Contrary to usual quantum field theories, in which high subgraphs automatically look local (i.e. point-like), there is an additional tension between global and local properties in TGFTs, which precludes the automatic understanding of high subgraphs as quasi-local (i.e. tensor invariant like) objects. A simple way of understanding this point is that a subgraph whose closed faces all have trivial holonomies can have non-trivial holonomies associated to its external faces: therefore, from the point of view of its external legs, such a subgraph introduces non-trivial parallel transports which distinguish it from an elementary bubble vertex. We shall therefore introduce combinatorial conditions which, when satisfied, allow to interpret high subgraphs as quasi-local objects. In order to disentangle the loss of tensorial invariance from the loss of connectedness, we define two classes of subgraphs, the \textit{contractible} and the \textit{tracial} ones \cite{u1}.    
\begin{definition}
Let $\cG$ be a vertex-connected graph, and $\cH$ be one of its face-connected subgraphs. 
\begin{enumerate}[(i)]
\item If $\cH$ is a tadpole, $\cH$ is \textit{contractible} if, for any group elements assignment $(h_e)_{e \in L(\cH)}$:
\beq
\left( \forall f \in F(\cH)\,, \;  \overrightarrow{\prod_{e \in f}} {h_e}^{\epsilon_{ef}} = \one \right) \Rightarrow \left( \forall e \in L(\cH) \, , \; h_e = \one \right)\,.
\eeq
\item In general, $\cH$ is contractible if it admits a spanning tree $\cT$ such that 
$\cH / \cT$ is a contractible tadpole.
\item $\cH$ is \textit{tracial} if it is contractible and the contracted graph $\cG / \cH$ is vertex-connected.
\end{enumerate}
\end{definition}

A contractible graph is therefore a subgraph on which any flat connection is trivial up to a gauge transformation. Note that this gauge freedom is what makes the contraction with respect to a spanning tree an essential feature of the definition. On the other hand, the notion of traciality is independent of the choice of tree, as it is a statement about $\cG / \cH$, in which all internal lines of $\cH$ have been contracted.

\
In the multiscale effective expansion, high divergent subgraphs will give rise to effective couplings. To apply this procedure in our context, such subgraphs need to be
tracial, or at least contractible. Traciality ensures that the divergence of a high subgraph can be factorized into a divergent coefficient times a \emph{connected}
invariant. For  high divergent subgraphs which are contractible but not tracial, a factorization of the divergences is still possible, but 
in terms of disconnected invariants; these have been called anomalous terms in \cite{tensor_4d}. It is not clear yet whether this is a major issue and how these
anomalies should be interpreted physically, but in the models considered below all the divergent high subgraphs are tracial. Indeed
we already noticed that any $k$-dipole with $k >1$ is contractible, and that any $d$-dipole is tracial, as its contraction also preserves connectedness. 
Combined, these two facts provide us already with an interesting class of tracial subgraphs. They are called {\it melopoles} \cite{u1} because they combine the idea of melonic graphs
and tadpoles. The high divergent subgraphs appearing in the models studied in \cite{u1}, and to which Chapter \ref{chap:u1} will be devoted, are melopoles.

\begin{definition}
In a graph $\cG$, a \emph{melopole}
is a single-vertex subgraph $\cH$ (hence $\cH$ is made of tadpole lines attached to a single vertex in the ordinary sense), 
such that there is at least one ordering (or ``Hepp's sector") of its $k$ lines as $l_1, \cdots , l_k$ such that $ \{l_1 ,  \dots , l_{i} \} / \{l_1 , \dots , l_{i-1} \} $ is a $d$-dipole for $1 \le i \le k$.
\end{definition}

\begin{figure}
\begin{center}
\includegraphics[scale=0.6]{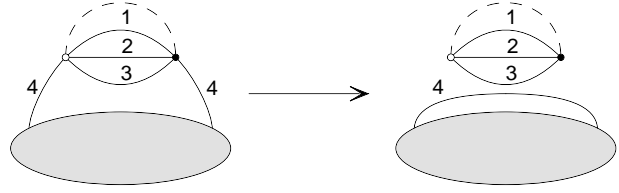}
\caption{A single-line melopole (left), and the result of its contraction}
\label{melo1}
\end{center}
\end{figure}

The simplest melopole has just one line and is shown in Figure \ref{melo1} (for $d=4$). Its contraction within a connected graph (grey blob) results in a connected graph times a coefficient (of which a graphical representation is given).

\begin{proposition}
Any face-connected melopole is tracial.
\end{proposition} 
\begin{proof}
Obviously it is contractible; and vertex-connectedness cannot be lost at any stage if one contracts in the order of the correct Hepp's sector.
\end{proof}

In just-renormalizable models \cite{su2}, a larger class of tracial subgraphs will dominate, which extend the notion of melopole to an arbitrary number of vertices. They are called \textit{melonic}, and generalize the melonic graphs encountered so far.

\begin{definition}
In a graph $\cG$, a \emph{melonic subgraph}
is a face-connected subgraph $\cH$ containing at least one maximal tree $\cT$ such that $\cH / \cT$ is a melopole.\footnote{This definition is chosen so that at least one internal face of $\cG$ runs through any line of any melonic subgraph. $\cG$ itself is considered melonic if it is melonic as a subgraph of itself. This definition will ensure Lemma \ref{1PI}.}
\end{definition}

\begin{proposition}
Any melonic subgraph is tracial.
\end{proposition}



\section{Multiscale expansion and power-counting}

We are now ready to introduce the multiscale expansion, and the power-counting theorem that result. The latter was first proved for an Abelian group $G = \U(1)^{D}$ in \cite{u1}, and provides an optimal bound in this case. It was then extended to non-Abelian compact groups in \cite{su2}, with full details provided for $G = \SU(2)$ only. When the group is non-Abelian, the Abelian power-counting theorem remains valid, though it does not provide an optimal bound in general. As we will see, it is however optimal for contractible subgraphs, and hence for all the subgraphs we will need to renormalize. For clarity of the presentation, we focus on the
Abelian case in this chapter, and even restrict to $G = \U(1)$ in the proofs. The non-Abelian case will be reported on in more details in Chapter \ref{chap:su2}.

 
\subsection{Multiscale decomposition}

As in usual local field theories, the multi-scale expansion relies on a slicing of the propagator in the Schwinger parameter $\alpha$, according to a geometric progression. We therefore fix an arbitrary constant $M > 1$ and for any integer $i \geq 0$, we define the slice of covariance $C_i$ as:
\bes
C_{0}(g_1, \ldots , g_d ; g_1' , \ldots , g_d' ) &=& \int_{1}^{+ \infty} \extd \alpha \, \e^{- \alpha m^2} \int \extd h \prod_{\ell = 1}^{d} K_{\alpha} (g_\ell h g_\ell'^{\inv})\,, \\
\forall i \geq 1\,, \quad C_{i}(g_1, \ldots , g_d ; g_1' , \ldots , g_d' ) &=& \int_{M^{ - 2 i}}^{M^{ - 2(i-1)}} \extd \alpha \, \e^{- \alpha m^2} \int \extd h \prod_{\ell = 1}^{d} K_{\alpha} (g_\ell h
g_\ell'^{\inv})\,.
\ees  
In order to be compatible with the slicing, we choose a UV regulator of the form $\Lambda = M^{-2 \rho}$, and denote the cut-off covariance by:
\beq  \label{decomposi}
C^{\rho} = \sum_{0 \le i \leq \rho} C_i \,.
\eeq  
We can then decompose the amplitudes themselves, according to scale attributions $\mu = \{ i_e \}$ where $i_e$ are integers associated to each line, determining the slice attribution of its propagator. The full amplitude $\cA_\cG$ of $\cG$ is then reconstructed from the sliced amplitudes $\cA_{\cG, \mu}$ by simply summing over the scale attribution $\mu$:
\beq
\cA_\cG = \sum_{\mu} \cA_{\cG , \mu}\,.
\eeq

\
We can then introduce high subgraphs, and their inclusion tree: the Gallavotti-Nicol\`{o} tree. The only difference with usual quantum field theories, though a very important one, is that we will require high subgraphs to be face-connected rather than only vertex-connected. This is again in agreement with the structure of the amplitudes, whose integrands factorize over face-connected components. While not absolutely necessary, this choice is physically meaningful, and as will be clear in the following will technically ease the analysis of the renormalized series. It would however be possible to stick to the usual notion of vertex-connected high subgraph, as is done in \cite{fabien_dine}. 
\begin{definition}
Let $\cG$ be a vertex-connected graph, with scale attribution $\mu$.
\begin{enumerate}[(i)]
\item Given a subgraph $\cH \in \cG$, one defines internal and external scales:
\beq
i_{\cH}(\mu) = \inf_{e \in L(\cH)} i_e (\mu)\,, \qquad e_{\cH}(\mu) = \sup_{e \in N_{ext}(\cH)} i_e (\mu)\,,
\eeq
where $N_{ext}(\cH)$ are the external legs of $\cH$ which are hooked to external faces.
\item A \textit{high subgraph} of $(\cG , \mu)$ is a face-connected subgraph $\cH \subset \cG$ such that
$e_{\cH}(\mu) < i_{\cH}(\mu)$. We label them as follows. For any $i$, $\cG_i$ is defined as the set of lines of $\cG$ with scales higher or equal to $i$. We call $k(i)$ its number of face-connected components, and $\{ \cG_{i}^{(k)} | 1 \leq k \leq k(i) \}$ its face-connected components. The subgraphs $\cG_{i}^{(k)}$ are exactly the high subgraphs.
\item Two high subgraphs are either included into another or line-disjoint, therefore the inclusion relations of the subgraphs $\cG_{i}^{(k)}$ can be represented as an abstract graph, whose root is the whole graph $\cG$. This is the \textit{Gallavotti-Nicol\`o tree} or simply \textit{GN tree}.
\end{enumerate}
\end{definition}
 



\subsection{Propagator bounds}

The idea of the multiscale analysis is to bound sliced propagators, and deduce an optimized bound for each $\cA_{\cG , \mu}$ separately. To this effect, we first need to capture the peakedness properties of the propagators into Gaussian bounds.
They can be deduced from a general fact about heat kernels on curved manifolds: at small times, they look just the same as their flat counterparts, and can therefore be bounded by suitable Gaussian functions. In particular, the heat kernel on $\U(1)$ can be parametrized by one angle $\theta \in \left[ 0 , 2 \pi \right[$, and expanded as
\beq
 K_{\alpha}(\theta) = \frac{\e^{- \frac{1}{4 \alpha} \theta^2 }}{\sqrt{\alpha}} \left( 1 + 2 \sum_{n = 1}^{\infty} \e^{- \frac{\pi^2 n^2}{\alpha}} \cosh\left( \frac{n \pi}{\alpha} \theta \right) \right)\,,
 \eeq
 so that the propagator can explicitly be written as:
 \bes
 C(\theta_1, \dots , \theta_d ; \theta_1' , \dots , \theta_d' ) &=& \int_{0}^{+ \infty} \extd \alpha \, \frac{\e^{- \alpha m^2}}{\alpha^{d/2}}  \int_{0}^{2 \pi} \extd \lambda \, \e^{- \frac{1}{4 \alpha} \sum_\ell (\theta_\ell - \theta_\ell' + \lambda)^{2} } \nn\\
 && \; \times \; T(\alpha ; \theta_1 - \theta_1' + \lambda , \dots , \theta_d - \theta_d' + \lambda)\,,
 \ees
 with 
 \beq
 T(\alpha ; \theta_1 , \dots , \theta_d ) \equiv \prod_{\ell = 1}^{d} \left( 
 1 + 2 \sum_{n = 1}^{\infty} \e^{- \frac{\pi^2 n^2}{\alpha}} \cosh\left( \frac{n \pi}{\alpha} \theta_\ell \right) \right)\,.
 \eeq

One can then prove generic bounds for the sliced propagators $C_i$, on which the whole power-counting will rely. Such bounds have been computed in the model of \cite{tensor_4d}, and immediately imply the following:
 \begin{proposition}
 There exist constants $K > 0$ and $\delta > 0$, such that for all $i \in \mathbb{N}$:
 \bes\label{propa_bound_ab} 
 C_i (\theta_1 , \dots , \theta_d ; \theta_1' , \dots , \theta_d') &\leq& K M^{(d-2) i } \int \extd \lambda \,
 \e^{- \delta M^{i} \sum_\ell |\theta_\ell - \theta_\ell' + \lambda| }\,, \\
\forall \ell \in \llbracket 1 , d \rrbracket \,, \quad\frac{\partial}{\partial \theta_\ell} C_i (\theta_1 , \dots , \theta_d ; \theta_1' , \dots , \theta_d') &\leq& K M^{(d-1) i } \int \extd \lambda \,
 \e^{- \delta M^{i} \sum_\ell |\theta_\ell - \theta_\ell' + \lambda| }\,. \label{deriv_bound_ab}
 \ees
 \end{proposition}
The bound on the derivatives of $C_i$ is not necessary to the proof of the power-counting theorem, but are crucial to establish the finite character of the renormalized amplitudes of the models presented in Chapter \ref{chap:u1}. They are also generalizable to higher numbers of derivatives, as will be necessary in Chapter \ref{chap:su2}. For Abelian compact Lie groups of dimension $D$, $\theta$'s and $\lambda$ are $D$-dimensional and, for example, the first bound becomes
 \beq\label{propa_boundD} 
 C_i (\vec \theta_1 , \dots , \vec\theta_d ; \vec\theta_1' , \dots , \vec\theta_d') \leq K M^{(dD-2) i } \int \extd \vec\lambda \,
 \e^{- \delta M^{i} \sum_e |\vec\theta_e - \vec\theta_e' + \vec\lambda| }\,.
 \eeq 


\subsection{Abelian power-counting}

\subsubsection{Power-counting in a slice}
  
The divergence degree of Abelian TGFT subgraphs in a single slice and
with a heat kernel regularization has been established and analyzed in \cite{lin}. For an Abelian compact group of dimension $D$ it gives\footnote{The degree of a colored graph will not make any apparition in the last three chapters of this thesis, we therefore allow ourselves to use the same notation $\omega$ for the divergence degree of TGFTs.}
 \beq  
\omega (\cH)  = -2 L (\cH) + D (F(\cH) - R(\cH)) 
 \eeq
where $R( \cH )$ is the rank of the $\epsilon_{ef}$ incidence matrix of $\cH$ (which we recall takes only internal lines and faces into account). As we will see, the $\theta$ integrations transform the $(dD-2)L$ term coming from the propagator decays into a $-2 L + D F$ exponent. The gauge invariance, which manifests itself in the $\lambda$ integrals (and is absent in \cite{tensor_4d}), is responsible for the additional rank contribution. 

If the subgraph $\cH$ is the union of several face-connected components $\cH_k$, the divergence degree factorizes as the sum of the divergence degrees of the face-connected components. Indeed, from the very definition of face-connectedness, the face-connected components are the smallest pieces of $\cH$ on which the $\epsilon_{ef}$ incidence matrix is block-rectangular. Therefore:
\beq  
\omega (\cH)  = \sum_k \omega (\cH_k) \,,
\eeq
which provides the finest understanding of the divergences. 

In the case of non-commutative TGFTs the ordering of faces results in a more subtle 
single-slice power-counting, established in \cite{vm1, vm2, vm3}. The divergences are exactly captured by a \emph{twisted} divergence degree $\omega_t$, which is bounded by the Abelian one. Importantly for us, Abelian and twisted divergence degree coincide for contractible subgraphs. 

\subsubsection{Multiscale power-counting}
  
Consider a graph $\cG$, and the multi-scale decomposition of its amplitude $\cA_\cG = \underset{\mu}{\sum}  \cA_{\cG, \mu}$.
The multiscale power-counting is a bound which, at fixed momentum attribution $\mu$, factorizes over all the $ \cG_{i}^{(k)}$ nodes of the Gallavotti-Nicol\'o tree (hereafter GN tree).

 \begin{proposition}\label{prop:abpc}[Abelian power-counting]   
There exists a constant $K > 0$ such that the following bound holds:
\beq\label{fund_ab}
 \vert \cA_{\cG, \mu}  \vert   \leq K^{L(\cG)}  \prod_{i \in \mathbb{N}} \prod_{ k \in \llbracket 1 , k(i) \rrbracket } M^{\omega [  \cG_{i}^{(k)}]}
\eeq
\end{proposition}
\begin{proof}
The idea is to combine the single-scale bound and the GN tree to obtain a new optimized bound over scales. 
First we collect all the powers of $M$ in front of the propagator bounds, and rewrite them as
$$\prod_i \prod_{k  \in \llbracket 1 , k(i) \rrbracket } M^{(d D -2) L [  \cG_{i}^{(k)}]}$$
through the usual trivial identity $M^i = \prod_{j=1}^i M$.
We then integrate the $\theta$ variables in an optimal way, as was done in \cite{tensor_4d}. In each face $f$, a maximal tree of lines $T_f$ is chosen to perform the $\theta$ integrations. Optimality is ensured by requiring the trees $T_f$ to be compatible with the abstract GN tree, in the sense that $T_f \cap G_{i}^{(k)}$ has to be a tree itself, for any $f$ and $G_{i}^{(k)}$.
This results in a factor 
$$ K^{L(\cG)} \prod_i \prod_{ k \in \llbracket 1 , k(i) \rrbracket } M^{D \left( -d  L ( \cG_{i}^{(k)})  + F( \cG_{i}^{(k)}) \right) } \,,$$
that is one decay per elementary strand ($- d L$ term), except for one strand in each closed face ($+ F$ term). 
Combined with the previous term, this gives a bound 
$$ K^{L(\cG)}   \prod_i \prod_{ k \in \llbracket 1 , k(i) \rrbracket } M^{-2 L ( \cG_{i}^{(k)})  +  DF( \cG_{i}^{(k)})  } \,.$$

It remains to perform the $\lambda$ integrals, using the remaining decay, that is
\beq\label{facetree} 
\prod_f \e^{- \delta M^{i(f)} \vert \sum_e \epsilon_{ef} \lambda_e \vert}, 
\eeq
where $i(f)$ is the lowest scale in the face $f$. These integrals should give the rank contribution to $\omega$. 
In order to optimize this effect, we should select a restricted set of faces $F_\mu$ 
such that the submatrix $\epsilon_{ef}$ with $f$ 
restricted to $F_\mu$ has rank $R (\cG_{i}^{(k)})$ in each $\cG_{i}^{(k)}$ node, and forget the (redundant) decay factors from the other faces in \eqref{facetree}. 
This is analogue to selecting a spanning tree $T_\mu$ and neglecting the loop lines decays in ordinary field theories.

To select $F_\mu$, we start from the leaves of the GN tree and proceed towards its root $\cG$.
In a leaf $\cH$ we select a first subset of faces such that the restricted submatrix $\epsilon_{ef}$ with $f$
and $e$ in $\cH$ has maximal rank; then we \emph{contract} $\cH$ and continue the procedure for the reduced
graph and the reduced GN tree, until the root is reached. At the end we obtain a particular set of faces $F_\mu$.

At each node $\cG_{i}^{(k)}$ we have discarded the full incidence columns for internal faces 
which were combinations of other columns of that node. But because such faces were internal, these full columns have zeros outside the $ \cG_{i}^{(k)}$ block. Hence removing them cannot have any effect on the lower GN nodes' ranks. The conclusion is that the incidence matrix reduced to $F_\mu$, that is for which all internal faces not contained in $F_\mu$ have been discarded,
has still rank $R ( \cG_{i}^{(k)} )$ in each $\cG_{i}^{(k)}$ node.

Discarding the decay factors for faces not in $F_\mu$, we now need to analyze the result of the integral 
\begin{equation} 
\int \prod_{e \in L(\cG)} \extd^D \lambda_e   \prod_{f \in F_\mu} \e^{- \delta M^{i(f)} \vert \sum_e \epsilon_{ef} \lambda_e \vert}, 
\end{equation}
and prove that it gives $\prod_i \prod_{ k \in \llbracket 1 , k(i) \rrbracket } M^{-D R ( \cG_{i}^{(k)} )}$. 
To this effect, in the graph $\cG$ we can pick up a set $L_\mu$ of  exactly $\vert F_\mu \vert$ \emph{lines} such that the (square) $F_\mu \times L_\mu$ minor of $\epsilon_{ef}$ has non-zero determinant. The $L_\mu \times F_\mu$ square incidence matrix $\epsilon_{ef}$ must still have exactly $R ( \cG_{i}^{(k)} )$ rank in each $\cG_{i}^{(k)}$ node (otherwise the $R ( \cG_{i}^{(k)} )$ columns $\epsilon_{ef}$ for $f \in F_\mu \cap \cG_{i}^{(k)}$
would not generate a space of dimension $R ( \cG_{i}^{(k)} )$, and the rank of the selected $F_\mu \times L_\mu$ square matrix would be strictly smaller than $F_\mu$).

We can now fix all values of the $\lambda_e$ parameters of the lines not in $L_\mu$ and consider the integrals 
\beq \int \prod_{e \in L_\mu} \extd^D \lambda_e \prod_{f \in F_\mu} \e^{- \delta M^{i(f)} \vert \sum_e \epsilon_{ef} \lambda_e \vert}, 
\eeq
We change variables so that the integral becomes
\beq \int J \prod_{f \in F_\mu} d^D x_f  \e^{- \delta M^{i(f)} \vert  x_f - y_f \vert}, 
\eeq
where the $y_f$ variables are functions of the fixed $\lambda_e$ parameters of the lines not in $L_\mu$ and $J$ is a Jacobian.
This integral gives $\prod_{f \in F_\mu}  M^{-D i(f)}$, which by the condition on $F_\mu$ turns into $\prod_i \prod_{ k \in \llbracket 1 , k(i) \rrbracket } M^{-D r_{i,k}}$ as expected.

Remark finally that by Hadamard's bound, since each column of this determinant is made of at most $d$ factors $\pm 1$ (a line containing at most $d$ internal faces), the Jacobian $J$ of the corresponding change of variables 
is at most $\sqrt{d} ^{F_\mu}$, hence can be absorbed in the $K^{L(\cG)}$ factor. 

Finally we can integrate the fixed $\lambda_e$ parameters for $e \not \in L_\mu$ at a cost bounded by $K^{L (\cG)}$ since we assumed the group to be compact.
\end{proof}

It is important to note that, when the group is commutative (as we assumed here), the multiscale Abelian bound is optimal, in the sense that a lower bound of the same form could be derived, with only a different constant $K > 0$. This can be checked step by step in the derivation, starting from the propagator bounds on which we already commented.
We will come back to this power-counting theorem in Chapter \ref{chap:su2}, where we will explain why it still holds in the non-Abelian case, and why it is optimal for contractible subgraphs. 



\section{Classification of just-renormalizable models}

In this section, we discuss the classification of possible just-renormalizable models derived in \cite{su2}. On top of being very general and constraining, therefore interesting, it requires a detailed analysis of the Abelian divergence degree and of the properties of melonic subgraphs. Hence it is a good way to get more familiar with these structures and collect mathematical facts along the way. 

\subsection{Analysis of the Abelian divergence degree}\label{sec:degree}

We come back to the general situation of a compact Lie group $G$ of dimension $D$, not necessarily commutative, and a rank-$d$ field with $d \geq 3$. We assume for the moment that the Abelian power-counting theorem holds in general, and for simplicity we call \textit{divergent} (resp. \textit{convergent}) a subgraph with $\omega \geq 0$ (resp. $\omega < 0$). Because the divergent subgraphs in the generalized non-Abelian sense will turn out to be contractible, this nomenclature is perfectly consistent. We also denote by $v_{max}$ the maximal valency of $d$-bubble interactions appearing in the action (i.e. in the set $\cB$). 
The question we would like to address is the following: 

\

{\emph{Which values of $d$, $D$ and $v_{max}$ are likely to support just-renormalizable theories?}}

\

The first step we need to take is to write $\omega$ in a form which makes some key numbers associated to the subgraphs explicit. Let us consider a face-connected subgraph $\cH \subset \cG$ with $V$ vertices, $L$ lines, $F$ faces, and $N$ external legs. $R$ is the rank of the $\epsilon_{lf}$ incidence matrix of $\cH$. When $F = 0$, face-connectedness imposes $L=1$, and one trivially has $\omega(\cH) = -2$, which makes $\cH$ convergent. From now on, we therefore assume $F \geq 1$. Face-connectedness imposes that each line of $\cH$ appears in at least one of its internal faces. For $1 \leq k \leq v_{max} / 2$, $n_{2 k}$ is the number of bubbles with valency $2 k$ in $\cG$. 

Remember that the incidence matrix has entries $0, +1$ or $-1$ since the graphs we consider have no tadfaces.

Since we are going to make extensive use of contractions of graphs along trees, as a way to gauge fix the amplitudes, we first establish the change in divergence degree under such a contraction. It should come as no surprise that contracting a tree in $\cH$ only affects $\omega$ through its $- 2 L$ term, and this can be proven very concretely, in a similar way as one would proceed to fully justify the gauge fixing procedure \cite{pr1}.

\begin{lemma}
Under contraction of a tree $\cT$, $F$ and $R$ each do not change so that 
$[F-R] (\cH) =[ F-R] (\cH / \cT)$.
\end{lemma}
\begin{proof} That $F$ does not change is easy to show: existing faces can only get shorter under contraction of a tree line but cannot disappear (this is true also for \emph{open} faces). 

$R$ does not change because of the gauge invariance. Given a tree $\cT$ with $\vert\cT \vert = T= V-1$ lines, we can define the $L \times T$ matrix $\eta_{l,\ell}$ which has entries $0, +1$ or $-1$ in the following way: for any oriented line $l=(v,v')$ we consider the \emph{unique path} $P_\cT(l)$ in the tree $\cT$ going from vertex $v$ to $v'$ and define $\eta_{l,\ell}$ to be zero if this path does not contain $\ell$ and $\pm 1$ if it does, the sign taking into account the orientations of the path and of the line $\ell$. Remark that $\eta_{\ell\ell}=1$ for all $\ell$.

Then for each (closed) face $f$, made of $l_1,\ldots , l_p$, it is easy to check that the induced loop on $\cT$, made by gluing the paths $P_\cT(l_1), \dots ,  P_\cT(l_p)$, which is contractible, must take each tree line $\ell$ an equal number of times and with opposite signs. Therefore 
\beq 
E(f, \ell) =\sum_l \epsilon_{lf}\eta_{l,\ell} = 0, \quad \to  \epsilon_{\ell f} = -  \sum_{l \not = \ell} \epsilon_{lf}\eta_{l,\ell}\,.
\eeq 
Hence the line $\epsilon_{\ell f}$ is a combination of the other lines, and the incidence matrix after contracting $\ell$ maintains the same rank.
\end{proof}

We shall consider now a \textit{tensorial rosette} \cite{addendum}, namely the subgraph $\cH / \cT$ obtained after contraction of a \textit{spanning tree}; it has $L- (V-1)$ lines and a single vertex. The goal is to gain a better control over its degree of divergence and the various contributions to it. The key procedure to achieve the goal is to apply $k$-dipole contractions to the tensorial rosette, and establish how they affect the divergence degree. 

Note that $\cH / \cT$ is not necessarily face-connected, since the contraction of tree lines affects how faces are connected to one another. 
Recall also that a line is a $k$-dipole if it belongs to exactly $(k-1)$ faces of length $1$.

Under a $k$-dipole contraction we know that a single line and possibly several faces disappear, hence the rank
of the incidence matrix can either remain the same or go down by 1 unit. Moreover
only the faces of length $1$ can eventually disappear. And if there exist such faces
the rank must go down by exactly $1$, since we delete a column which is not a combination of the others.
\begin{itemize}  
\item $F \to F- (k - 1)$ and $R \to R-1$, hence $F-R \to F- R- (k-2)$  if $k \ge 2$,

\item $F \to F$, and $R \to R$ or $R \to R-1$,  hence $F-R \to F- R$ or $F-R \to F- R+1$ 
if $k=1$.
\end{itemize}

By definition, a rosette (with external legs) is a melopole if and only if there is an ordering of its lines such that they can all be successively contracted as $d$-dipoles. In that case, we find that $F-R =  (d-2)[L- (V-1)]$. If the rosette is not a melopole, there is at least one step where $F-R$ decreases by less than $(d-2)$, so we expect such a subgraph to be suppressed with respect to a melopole. However, $k$-dipole contractions with $k < d$ need not conserve vertex-connectedness, so we need to refine this argument. To do so, we write the divergence degree of any rosette in terms of the quantity
\beq
\rho \equiv F - R - \left( d-2 \right) \widetilde{L} \,,
\eeq
where $\widetilde{L}$ is the number of lines of the rosette. It will be convenient in the following to consider (vertex)-disjoint unions of rosettes, to which $\rho$ is extended by linearity. These disjoint unions of rosettes will simply be called rosettes from now on, and their single-vertex components will be said to be \textit{connected}.

Since $\widetilde{L} = L - V + 1$ is the number of lines of any rosette of the graph $\cH$, and $F-R$ does not depend on $\cT$ either, we know that $\rho(\cH / \cT)$ is independent of $\cT$. It is therefore a function of equivalent classes of rosettes. This way we obtain a nice splitting of $\omega$, between a rosette dependent contribution and additional combinatorial terms capturing the characteristics of the initial graph:
\beq\label{rholv}
\omega(\cH) = D \left(d - 2 \right) + \left[ D \left(d - 2 \right) - 2 \right] L - D \left(d - 2 \right) V + D \rho( \cH / \cT ) \,.
\eeq
The first three terms do not depend on the rank $R$, and provided that $\rho$ can be understood, will give a simple classification of divergences. To establish this central result about the values of $\rho$, one first needs to prove a technical lemma, about $1$-dipole contractions. 
\begin{lemma}\label{d1}
Let $\cG$ be a face-connected rosette (with $F(\cG) \geq 1$), and $\ell$ a $1$-dipole line in $\cG$. If $\cG / \ell$ has more vacuum connected components than $\cG$, then
\beq
R( \cG / \ell ) = R ( \cG ) - 1 \,.
\eeq 
\end{lemma}
\begin{proof}
As stated before, such a move either lowers $R$ by $1$ or leaves it unchanged. We just have to show that given our hypothesis, we are in the first situation. 
We first remark that lines and faces can be oriented in such a way that $\epsilon_{lf} = +1$ or $0$. We can for instance positively orient lines from white to black nodes, and faces accordingly. 
With this convention, we can exploit the colored structure of the graphs in the following way: for any color $1 \leq i \leq d$, each line appears in exactly one face of color $i$. For vacuum graphs, all these faces are closed
and correspond to entries in the $\epsilon_{lf}$ matrix, implying
\beq
\sum_{f \, {\rm{of} \, \rm{color}} \, i } \epsilon_{l f} = 1
\eeq
for any $i$ and any $l$.
Given the hypothesis on $\cG / \ell$, we know that up to permutations of lines and columns, $\epsilon_{lf}$ takes the form:
\[
\left(
\begin{array}{c|c}
  \raisebox{-10pt}{{\huge\mbox{{$M_1$}}}} & \raisebox{-10pt}{{\huge\mbox{{$0$}}}} \\ \hline
  \ast \, \cdots \, \ast \, 1 & 1 \, \varepsilon_2 \, \cdots \, \varepsilon_{d - 1}\,  0\,  \cdots\,  0 \\ \hline
  \raisebox{-10pt}{{\huge\mbox{{$0$}}}} & \raisebox{-10pt}{{\huge\mbox{{$M_2$}}}} \\
\end{array}
\right)
\]
where $M_2$ is the $\epsilon_{lf}$ matrix of a vacuum graph, one of the additional vacuum components created by the contraction of $\ell$. $M_1$ is the $\epsilon_{lf}$ matrix associated to the complement (possibly several connected components) in $\cG / \ell$.
The additional line corresponds to $\ell$, and because $\cG$ is face-connected, it must contain at least a $1$ under $M_1$, and a $1$ above $M_2$. This leaves up to $d-2$ additional non-trivial entries in this line above $M_2$, denoted by the variables $\varepsilon_i = 0$ or $1$. Let us call
$i_1$ the color of the face associated to the first column of $M_2$. Non-zero $\varepsilon$'s are necessarily associated to different colors: call them $i_2$ up to $i_{d-1}$. This implies that the remaining color, $i_d$, only appears in faces
of $M_2$ that do not intersect with $\ell$. Calling $C_f$ the columns of $M_2$, and $C_{f_1}$ its first column, one has:
$$
C_{f_1} + \sum_{f \, {\rm{of} \, \rm{color}} \, i_1 \,; \,f \neq f_1} C_f = \sum_{f \, {\rm{of} \, \rm{color}} \, i_d \,} C_f\,.
$$     
The operation $$C_{f_1} \to C_{f_1} + \sum_{f \, {\rm{of} \, \rm{color}} \, i_1 \,; \,f \neq f_1} C_f - \sum_{f \, {\rm{of} \, \rm{color}} \, i_d \,} C_f$$
cancels the first column of $M_2$, and when operated on the whole matrix does not change the line $\ell$. 
We conclude that $R(\cG) = {\rm{rank}}(M_1) + {\rm{rank}}(M_2) + 1 = R( \cG / \ell) + 1$. 
\end{proof}

The essential property of the quantity $\rho$ is that it is bounded from above, and is extremal for melopoles. More precisely we have:
\begin{proposition}\label{rho}
Let $\cG$ be a connected rosette. 
\begin{enumerate}[(i)]
\item If $\cG$ is a vacuum graph, then 
$$ \rho(\cG) \leq 1$$
and
$$ \rho (\cG) = 1 \Leftrightarrow \cG \; \mathrm{is} \; \mathrm{a} \; \mathrm{melopole}\,.$$
\item If $\cG$ is not a vacuum graph, i.e. has external legs, then 
$$ \rho(\cG) \leq 0$$
and
$$ \rho (\cG) = 0 \Leftrightarrow \cG \; \mathrm{is} \; \mathrm{a} \; \mathrm{melopole}\,.$$
\end{enumerate}
\end{proposition}
\begin{proof}
It is easy to see that $\rho$ is conserved under $d$-dipole contractions. In particular a simple computation shows that $\rho(\cG) = 1$ when $\cG$ is a vacuum melopole, and $\rho(\cG) = 0$ when $\cG$ is a non-vacuum melopole. 
We can prove the general bounds and the two remaining implications in (i) and (ii) by induction on the number of lines $L$ of the rosette $\cG$. 
\begin{itemize}
 \item If $L = 1$, $\cG$ can be both vacuum or non-vacuum. In the first situation, $\cG$ cannot be anything else than the
fundamental melon with 2 nodes. It has exactly $d$ faces, a rank $R = 1$, so that $\rho (\cG) = 1$. In the second situation,
namely when $\cG$ is non-vacuum, the number of faces is strictly smaller than $d$, as at least one strand running through the single line of $\cG$ must correspond to an external face. Since on the other hand 
the rank is $0$ when $F(\cG) = 0$ and $1$ otherwise, we see that $\rho (\cG) \leq 0$, and $\rho (\cG) = 0$ whenever the number of faces is exactly $(d-1)$. In this case, the unique line of $\cG$ is a $d$-dipole, therefore $\cG$ is a melopole.
 \item Let us now assume that $L \geq 2$ and that properties (i) and (ii) hold for a number of lines $L' \leq L -1$. If $\cG$ is not face-connected (and therefore non-vacuum), we can decompose it into face-connected components $\cG_1 , \ldots , \cG_k$ with $k \geq 2$. Each of these components has a number of lines strictly smaller than $L$, so by the induction hypothesis $\rho(\cG) = \underset{i}{\sum} \rho( \cG_i ) \leq 0$. Moreover, $\rho( \cG ) = 0$ if and only if $\rho ( \cG_i ) = 0$ for any $i$, in which case $\cG$ is a melopole since all the $\cG_i$'s are themselves melopoles. This being said, we assume from now on that $\cG$ is face-connected, and pick up a $k$-dipole line $\ell$ in $\cG$ ($1 \leq k \leq d$). 

Let us first suppose that $k \geq 2$. 
$\cG / \ell$ has $\widetilde{L} \leq L - 1$ lines in its rosettes, and $(F-R) (\cG / \ell ) = (F - R) (\cG) - (k - 2)$, which implies $\rho(\cG / \ell ) \geq \rho( \cG ) + (d - k)$ (with equality if and only if $\cG / \ell$ is itself a rosette). Moreover, $\cG / \ell$ is possibly disconnected and consists in $q$ vertex-connected components with $1 \leq q \leq d-k + 1$, yielding $q$ connected rosettes (after possible contractions of tree lines).
By the induction hypothesis, we therefore have $\rho( \cG ) \leq q - (d -k) \leq 1$, and $\rho(\cG) = 1$ if and only
if $\cG / \ell$ consists of $d - k + 1$ connected vacuum melopoles, in which case $\cG$ itself is a vacuum melopole. Similarly, $\rho(\cG) = 0$ if and only if $\cG / \ell$ consists of $d-k$ vacuum melopoles
and $1$ non-vacuum melopole, in which case $\cG$ is a non-vacuum melopole. 

If $k =1$, we either have $R (\cG / \ell) = R ( \cG ) - 1$ or $R (\cG / \ell) = R ( \cG )$, which respectively imply $\rho( \cG ) \leq \rho( \cG / \ell) - (d - 1)$ or $\rho( \cG ) \leq \rho( \cG / \ell) - (d - 2)$. 
The first situation is strictly analogous to the $k \geq 2$ case, therefore the same conclusions follow. In the second
situation, we resort to lemma \ref{d1}. Since $\cG$ has been assumed face-connected, and $L \geq 2$ implies $F(\cG) \geq 1$, the lemma is applicable: $\cG / \ell$ cannot have more vacuum connected components than $\cG$. In particular, if $\cG$ is non-vacuum, $\rho( \cG / \ell ) \leq 0$, 
therefore $\rho( \cG ) \leq - (d - 2) < 0$. Likewise, $\rho(\cG) \leq 0$ when $\cG$ is vacuum. 

We conclude that the two properties (i) and (ii) are true at rank $L$.
\end{itemize}
\end{proof}

\begin{corollary}
Let $\cH$ be a vertex-connected subgraph. If $\cH$ admits a melopole rosette (in particular, if $\cH$ is melonic), then all its rosettes are melopoles.
\end{corollary}
\begin{proof}
The quantity $\rho(\cH / \cT)$ is independent of the particular spanning tree $\cT$ one is considering. Therefore, if $\cH / \cT$ is a melopole then this holds for any other spanning tree $\cT'$. 
\end{proof}

\subsection{Just-renormalizable models}

We are now in good position to establish a list of potentially just-renormalizable theories. Indeed, by simply rewriting $L$ and $V$ as
\beq
L = \sum_{k = 1}^{v_{max} / 2} k \, n_{2k} - \frac{N}{2} \;, \qquad V = \sum_{k = 1}^{v_{max} / 2} n_{2 k}\,,
\eeq
one obtains the following bound on the degree of non-vacuum face-connected subgraphs:

\beq
\omega \leq D \left( d - 2 \right) - \frac{ D(d - 2) - 2}{2} N 
+ \sum_{k = 1}^{v_{max} / 2} \left[ \left( D(d-2) - 2 \right) k - D \left( d - 2 \right) \right] n_{2 k}\,.
\eeq
Since we also know this inequality to be saturated (by melonic graphs), it yields a necessary condition for just-renormalizable theories:
\beq\label{condition}
v_{max} = \frac{2 D (d-2)}{D(d-2) - 2}\,,
\eeq
and in such cases
\beq\label{vertex_true}
\omega = \frac{D (d - 2) - 2}{2}\left( v_{max} - N \right) - \sum_{k = 1}^{v_{max}/2 - 1} \left[ D \left( d - 2 \right) - \left( D(d-2) - 2 \right) k\right] n_{2 k} + D \rho \,.
\eeq
We immediately deduce that only $n$-point functions with $n \leq v_{max}$ can diverge, which is a necessary condition for renormalization. Equation (\ref{condition}) has exactly five non-trivial solutions (i.e. $v_{max} > 2$), which yields five classes of potentially
just-renormalizable interacting theories. Two of them are $\vphi^6$ models, the three others being of the $\vphi^4$ type. A particularly interesting model from a quantum gravity perspective is the $\vphi^6$ theory with $d=3$ and $D=3$,
which can incorporate the essential structures of 3d quantum gravity (model A)
. We will focus on this case in Chapter \ref{chap:su2}, but we already notice that the same methods could as well be applied to any of the four other types of
candidate theories. Table \ref{theories} summarizes the essential properties of these would-be just-renormalizable theories, called of type A up to E.

\begin{table}[h]
\centering
\begin{tabular}{| c || c | c | c | c |}
    \hline
  Type & $d$ & $D$ & $v_{max}$ & $\omega$  \\ \hline\hline
A & 3 & 3 & 6 & $3 - N/ 2 - 2 n_2 - n_4 + 3 \rho$ \\ \hline
B & 3 & 4 & 4 & $4 - N - 2 n_2 + 4 \rho$ \\ \hline
C & 4 & 2 & 4 & $4 - N - 2 n_2 + 2 \rho$\\ \hline
D & 5 & 1 & 6 & $3 - N/ 2 - 2 n_2 - n_4 + \rho$ \\ \hline
E & 6 & 1 & 4 & $4 - N - 2 n_2 + \rho$\\ 
    \hline
  \end{tabular}
\caption{Classification of potentially just-renormalizable models.}
\label{theories}
\end{table}

Models D and E have actually been studied and shown renormalizable in \cite{fabien_dine}, in the Abelian case, and even asymptotically free in the UV \cite{dineaf}.
Non-vacuum divergences of models A and B will only have melonic contributions, while models C, D and E can also include non-melonic terms. From our analysis, there could be: up to $\rho = -1$ divergent $2$-point graphs in model C; up to
$\rho = -2$ divergent $2$-point graphs and $\rho = -1$ divergent $4$-point graphs in model D; up to $\rho = -2$ divergent $2$-point graphs in model C. These require a (presumably simple) refinement of proposition \ref{rho}.
As for models of type A and B, we would not need any further understanding of $\rho$ to explore them further.

\

Finally, one also remarks that face-connectedness did not play any role in the derivation of expression (\ref{vertex_true}). Indeed, it is as well valid for vertex-connected unions of non-trivial face-connected subgraphs, which as we will see in Chapter \ref{chap:su2}, is also relevant to renormalizability.

\subsection{Properties of melonic subgraphs}

Since they will play a central role in the remainder of this thesis, we conclude this section by a set of properties verified by melonic subgraphs, especially non-vacuum ones.

\

The first thing one can notice is that by mere definition, any line in a melonic subgraph $\cH$ is part of an internal face in $F( \cH )$. This means in particular that $\cH$ cannot be split in two vertex-connected parts connected by a single $1$-dipole line $e$, since the three faces running through $e$ would then necessarily be external to $\cH$. In other words:
\begin{lemma}\label{1PI}
Any melonic subgraph $\cH \subset \cG$ is $1$-particle irreducible.
\end{lemma}
From the point of view of renormalization theory, this is already interesting, as $2$-point divergences in particular will not require any further decomposition into $1$-particle irreducible components.

We now turn to specific properties of non-vacuum melonic subgraphs. In order to understand further their possible structures, it is natural to first focus on their rosettes. The following proposition shows that they cannot be arbitrary melopoles.

\begin{proposition}\label{face_rosette}
Let $\cH \subset \cG$ be a non-vacuum melonic subgraph. For any spanning tree $\cT$ in $\cH$, the rosette $\cH / \cT$ is face-connected. 
\end{proposition}
\begin{proof}
Let $\cT$ be a spanning tree in $\cH$ and let us suppose that $\cH / \cT$ has $k \geq 2$ face-connected components. In order to find a contradiction, one first remarks that the contraction of a tree conserves the number of faces, and even elementary face connections. That is to say: if $l_1 , l_2 \in L( \cH ) \setminus \cT$ share a face in $\cH$, they also share a face in $\cH / \cT$. Therefore, the lines of $L ( \cH ) \setminus \cT$ can be split into $k$ subsets, such that each one of them does not share any face of $\cH$ with any other. We give a pictorial representation of what we mean on the left side of Figure \ref{propo6_1}, with $k = 4$. The internal structure of the vertices is ommited, the $4$ subsets of lines are marked with different symbols, while the tree lines are left unmarked.    
\begin{figure}[h]
\begin{center}
\includegraphics[scale=0.6]{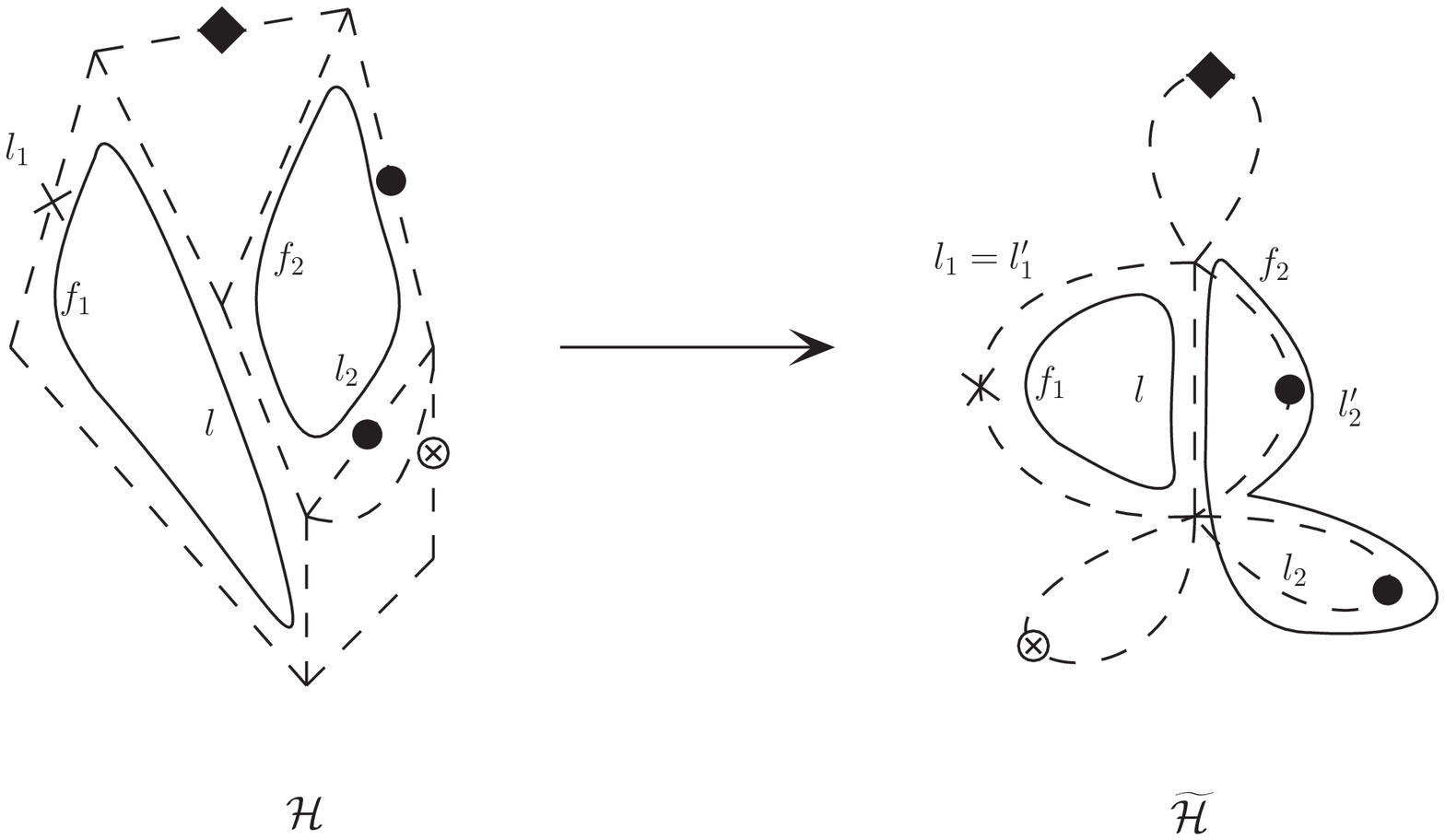}
\caption{Simplified representation of a melonic graph $\cH$ and its contraction $\widetilde{\cH}$.}
\label{propo6_1}
\end{center}
\end{figure}
The face-connectedness of $\cH$ is ensured by the tree lines, which must connect together these $k$ subsets. Incidentally, there must be at least one line $l \in \cT$ which is face-connected to two or more of these subsets. In particular\footnote{At this point we rely on $F(\cT) = \emptyset$, which holds because $\cT$ is a tree.}, we can find two faces $f_1$ and $f_2$ which are face-disconnected in $\cH / \cT$, and two lines $l_1 , l_2 \in \cH \setminus \cT$ such that: $l \in f_1 \cap f_2$, $l_1 \in f_1$ and $l_2 \in f_2$. See again the left side of Figure \ref{propo6_1}, where $f_1$ and $f_2$ are explictly represented as undashed closed loops.  

\begin{figure}[h]
\begin{center}
\includegraphics[scale=0.6]{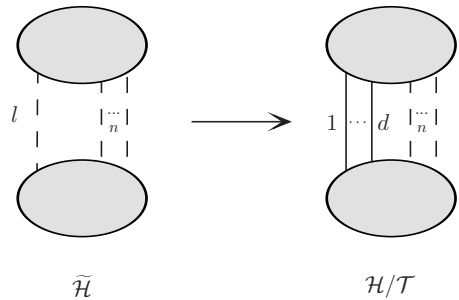}
\caption{Last step of a contraction of a spanning tree in a melonic subgraph.}
\label{melonic_last_step}
\end{center}
\end{figure}

Now, call $\widetilde{\cH}$ the subgraph obtained after contraction of all the tree lines but $l$, i.e. $\widetilde{\cH} \equiv \cH / (\cT \setminus \{ l \})$. $\widetilde{\cH}$ consists of two vertices, connected by $l$ and a certain number $n$ of lines from $L(\cH) \setminus \cT$ (see the right part of Figure \ref{propo6_1} and the left part of Figure \ref{melonic_last_step}). Through $l$ run at least two faces, $f_1$ and $f_2$. $l$ is their single connection, since they are disconnected in $\widetilde{\cH} / \{ l \} = \cH / \cT$. This requires the existence of two $1$-dipole lines $l_1'$ and $l_2'$ in $\widetilde{\cH} \setminus \{ l \}$, through which $f_1$ and $f_2$ respectively run. In Figure \ref{propo6_1} we see that $l_1 ' = l_1$, but because $l_2$ is a tadpole line in $\widetilde{\cH}$, we must choose $l_2 ' \neq l_2$. Otherwise, $f_1$ and $f_2$ could not close without being connected in $\cH  / \cT$. The colored extension of $\cH / \cT$ can thus be split into two groups of nodes, connected by $n \geq 2$ lines and $d$ colored lines (created by the contraction of $l$, see the right part of Figure \ref{melonic_last_step}). It is easy to understand that such a drawing cannot correspond to a melopole. Indeed, the number of colored lines connecting the two groups of nodes would need to be at least $n (d - 2) + 1$.\footnote{A simple way to understand this last point is the following. Suppose there are $p$ colored lines between the two groups of nodes. If none of the $n$ lines between the two groups of nodes are elementary melons, an elementary melon can be contracted in one of them, without affecting the $n$ lines nor the $p$ colored lines between them. If on the contrary one of the $n$ lines is an elementary melon, it can be contracted. This cancels $d-1$ colored lines connecting the two groups of nodes, and replaces it by a single one. Hence $n \to n-1$ and $p \to p - (d - 2)$. By induction, one must therefore have $p - n (d - 2) \geq 1$, where the $1$ on the right side is due to the last step $n = 1 \to n = 0$.} Hence $d \geq n (d - 2) + 1$, from which we deduce: 
\beq
d \leq \frac{2n - 1}{n -1}\,.
\eeq
When $n \geq 3$, this is incompatible with $d \geq 3$, and $n=2$ is also incompatible with $d \geq 4$. If $n = 2$ and $d = 3$, a contradiction also arises, thanks to the colors. In the process of elementary melon contractions, the first of the two lines to become elementary will delete $2$ colored lines, say with colors $1$ and $2$, and replace it by a color-$3$ line. One therefore obtains two groups of nodes connected by two color-$3$ lines and a single color-$0$ line, which cannot form an elementary melon. See Figure \ref{propo6_2}.
\begin{figure}[h]
\begin{center}
\includegraphics[scale=0.6]{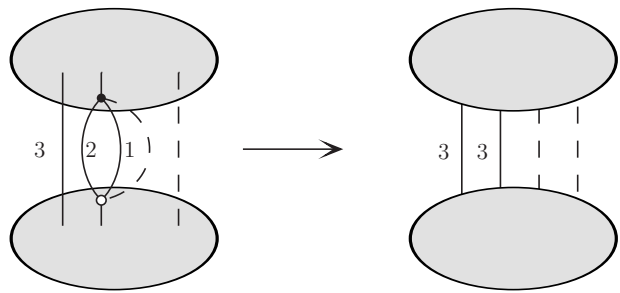}
\caption{Contraction of an elementary melon in a $3$-colored rosette with $n = 2$.}
\label{propo6_2}
\end{center}
\end{figure} 
\end{proof}

One immediately notices that this proposition also holds for any forest $\cF \subset \cH$, that is any set of lines without loops, be it a spanning tree or not. Indeed, any such $\cF$ is included in a spanning tree $\cT$. The contraction of $\cF$ on the one hand can only increase the number of face-connected components, and on the other hand the full contraction of $\cT$ leads to a single face-connected components, hence the contraction of $\cF$ also leads to a single face-connected component. 

\
We provide an illustration of this result in Figure \ref{rosette_propo6}, representing a melonic graph and one of its rosettes, which is face-connected.
\begin{figure}[h]
\begin{center}
\includegraphics[scale=0.6]{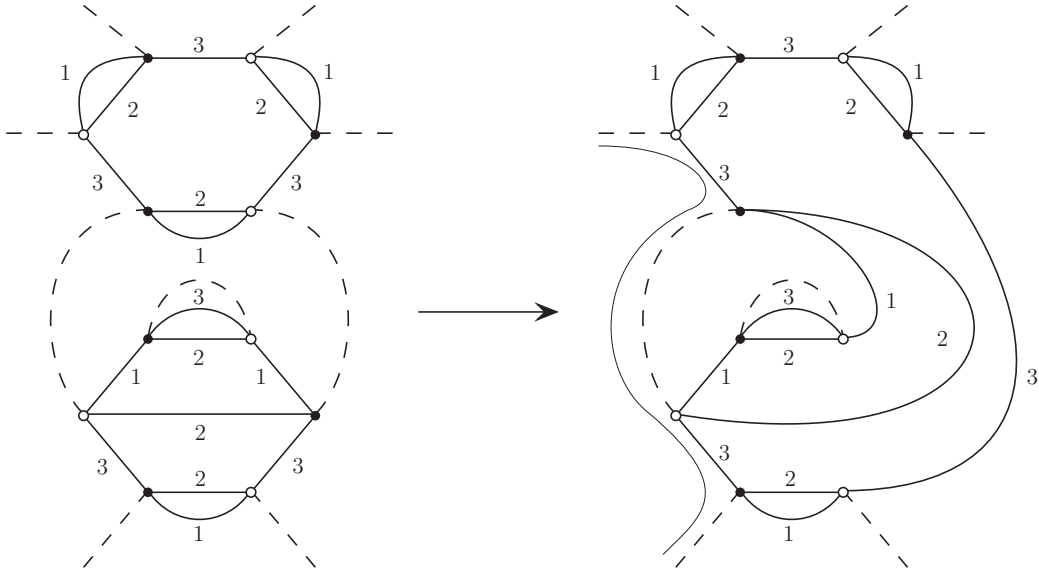}
\caption{A melonic graph (left) and one of its rosettes (right). The latter is face-connected (Proposition \ref{face_rosette}) and has a single external face (Corollary \ref{coro2}), represented as a thin line.}
\label{rosette_propo6}
\end{center}
\end{figure} 
\

A more important consequence of this statement is a restriction on the number of external faces of the rosettes:
\begin{corollary}\label{coro2}
Let $\cH \subset \cG$ be a non-vacuum melonic subgraph. For any spanning tree $\cT$ in $\cH$, $F_{ext}( \cH / \cT ) = 1$.
\end{corollary}
\begin{proof}
Let us prove that any face-connected melopole $\widetilde{H}$ has a single external face. $\cH / \cT$ being itself face-connected thanks to the previous proposition, the result will immediately follow. We proceed by induction on $L(\widetilde{\cH})$. The elementary melon has $d-1$ internal faces and $1$ external face, hence the property holds when $L(\widetilde{\cH}) = 1$. If $L(\widetilde{\cH}) \geq 2$, we can contract an elementary $d$-dipole line $l$ in $\widetilde{\cH}$. The subgraph $\{ l \}$ has $1$ external face, but it is internal in $\widetilde{\cH}$, otherwise the latter would not be face-connected. Hence $\widetilde{\cH}$ and $\widetilde{\cH} / \{ l \}$ have the same number of external faces. By the induction hypothesis, $\widetilde{\cH} / \{ l \}$ (which is a face-connected melopole) has a single external face, and so do $\widetilde{\cH}$.
\end{proof}

We illustrate again this result in Figure \ref{rosette_propo6}. Such restrictions on the rosettes constrain the face structure of the initial melonic graphs themselves.

\begin{proposition}\label{propo_color}
Let $\cH \subset \cG$ be a non-vacuum melonic subgraph. All the external faces of $\cH$ have the same color.
\end{proposition}
\begin{proof}
Suppose $F_{ext} ( \cH ) \geq 2$. Let us choose two distinct external faces $f_1$ and $f_2$, and show that they are of the same color. We furthermore select a line $l_1 \in f_1$, and a spanning tree $\cT$ in $\cH$ such that $l_1 \notin \cT$. This is possible thanks to lemma $\ref{1PI}$, and this guarantees that the unique external face of $\cH / \cT$ is $f_1$. This also means that in $\cH$, $f_2$ only runs through $\cT$, otherwise it would constitute a second face in $\cH / \cT$. We can in particular pick a line $l_2 \in f_2 \cap \cT$. See Figure \ref{propo7} for an example, in which we use the same simplified representation as before, except that the external faces we are interested in have open ends. 
\begin{figure}[h]
\begin{center}
\includegraphics[scale=0.6]{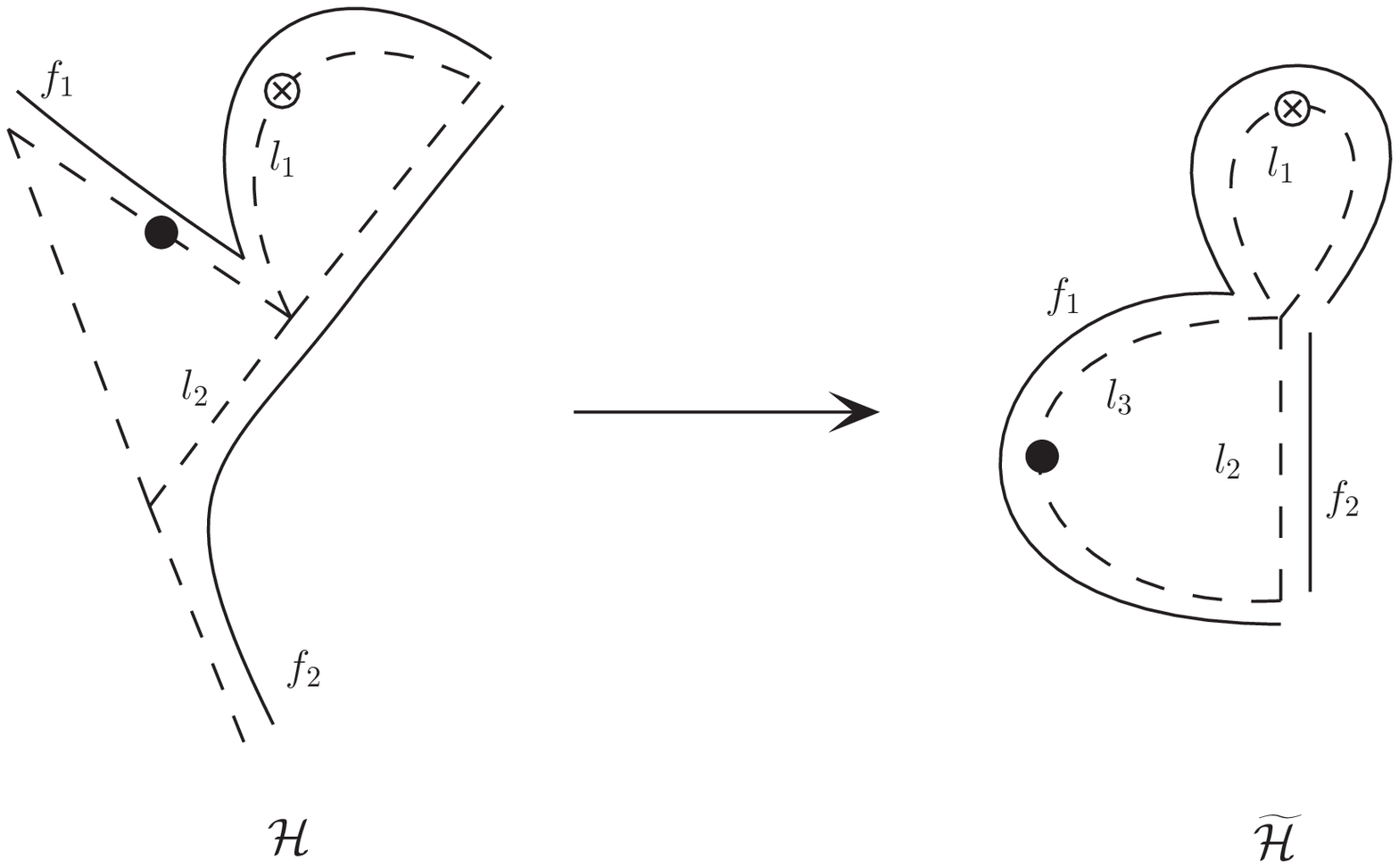}
\caption{A melonic graph $\cH$ and its contraction $\widetilde{\cH}$.}
\label{propo7}
\end{center}
\end{figure} 
 Similarly to the strategy followed in the proof of Proposition \ref{face_rosette}, define $\widetilde{\cH} \equiv \cH / ( \cT \setminus \{ l_2 \})$. As was already explained, $\widetilde{\cH}$ consists of two vertices, connected by $l_2$ and at most one extra line (see Figure \ref{melonic_last_step}, with $n = 1$). There cannot be just $l_2$ connecting these two vertices, because $\widetilde{\cH}$ is $1$-particle irreducible, hence there are exactly two such lines. Call $l_3$ the second of these lines (it is not necessarily possible to choose $l_3 = l_1$, see Figure \ref{propo7}). They must have at least $(d-1)$ internal faces in common, otherwise $\widetilde{\cH} / \{ l_2 \} = \cH / \cT$ would not be a melopole. They moreoever cannot have $d$ internal faces in common, otherwise $\cH / \cT$ would be vacuum. This means that both appear in external faces of the same color. One of them is of course $f_2$ (which goes through $l_2$), and the second (which goes through $l_3$) is either $f_1$, again $f_2$, or yet another external face. $f_2$ is excluded because by construction it had no support on $\cH \setminus \cT$. Moreover, $\widetilde{\cH}$ must have exactly two external faces, since only one is deleted when contracting $l_2$ and the resulting rosette $\cH / \cT$ has itself a single external face (by Corollary \ref{coro2}). Hence the external face running through $l_3$ can only be $f_1$, and we conclude that it has the same color than $f_2$.
\end{proof}
This property is quite useful in practice because it implies a restrictive bound on the number of external faces of a melonic subgraph in terms of its number of external legs.
\begin{corollary}\label{faces_legs}
A melonic subgraph with $N$ external legs has at most $\frac{N}{2}$ external faces.
\end{corollary}
\begin{proof}
In any vertex-connected graph with $N$ external legs, the number of external faces of a given color is bounded by $\frac{N}{2}$.
\end{proof}
Figure \ref{rosette_propo6} provides a good example of a melonic graph having more external faces that its rosettes: while the rosette on the right side has a single external face (in agreement with Corollary \ref{coro2}), the graph on the left side hase two external faces, and they both have the same color $3$ (in agreement with Proposition \ref{propo_color}). 

\

Finally one would like to understand the inclusion and connectivity relations between all divergent subgraphs of a given non-vacuum graph. This is a very important point to address in view of the perturbative renormalization of such models, in which divergent subgraphs are inductively integrated out. As usual, the central notion in this respect is that of a "Zimmermann" forest, which we will generalize to our situation (where face-connectedness replaces vertex-connectedness) in the next chapters. At this stage, we just elaborate on some properties of melonic subgraphs which will later on help simplifying the analysis of "Zimmermann" forests of divergent subgraphs. 
%



\begin{proposition}\label{curiosity}
Let $\cG$ be a non-vacuum vertex-connected graph. If $\cH_1 , \cH_2 \subset \cG$ are two melonic subgraphs, then: 
\begin{enumerate}[(i)]
\item $\cH_1$ and $\cH_2$ are line-disjoint, or one is included into the other.
\item If $\cH_1 \cup \cH_2$ is melonic, then: $\cH_1 \subset \cH_2$ or $\cH_2 \subset \cH_1$.
\end{enumerate} 
Moreover, any $\cH_1 , \ldots , \cH_k  \subset \cG$ melonic are necessarily face-disjoint if their union $\cH_1 \cup \ldots \cup \cH_k$ is also melonic.
\end{proposition}
\begin{proof}
Let us first focus on (i) and (ii). To this effect, we assume that: (i) $\cH_1 \cap \cH_2 \neq \emptyset$ (and in particular $\cH_1$ and $\cH_2$ are face-connected in their union); (ii) $\cH_1$ and $\cH_2$ are face-connected in $\cH_1 \cup \cH_2$, and the latter is also melonic. We need to prove that in these two situations, $\cH_1 \subset \cH_2$ or $\cH_2 \subset \cH_1$. In order to achieve this, we suppose that both $\widetilde{\cH}_1 \equiv \cH_1 \setminus ( \cH_1 \cap \cH_2 )$ and $\widetilde{\cH}_2 \equiv \cH_2 \setminus ( \cH_1 \cap \cH_2 )$ are non-empty, and look for a contradiction. 

Let $f_1$ be an arbitrary external face of $\cH_1$. Choose a line $l_1 \in f_1$, and a spanning tree $\cT_1$ in $\cH_1$, such that $l_1 \notin \cT_1$. Then the unique face of $\cH_1 / \cT_1$ is $f_1$. We want to argue that $( \cH_1 \cup \cH_2 ) / \cT_1 = ( \cH_1 / \cT_1 ) \cup \widetilde{\cH}_2$ is face-connected. In situation (ii), this is guaranteed by Proposition \ref{face_rosette} (applied to $\cH_1 \cup \cH_2$). In situation (i) on the other hand, one can decompose it as a disjoint union of subgraphs as follows:
\beq
( \cH_1 \cup \cH_2 ) / \cT_1 = \widetilde{\cH}_1 / (\cT_1 \cap \widetilde{\cH}_1) \sqcup (\cH_1 \cap \cH_2) / (\cT_1 \cap \cH_2) \sqcup \widetilde{\cH}_2 \,. 
\eeq 
The key thing to remark is that through each line of $\cH_1 \cap \cH_2$ run at least $d-1$ faces from $F(\cH_1)$, and at least $d-1$ from $F(\cH_2)$. Since at most a total of $d$ faces run through each line (and $d \geq 3$), we conclude that each line of $\cH_1 \cap \cH_2$ appears in at least one face of $F(\cH_1) \cap F(\cH_2)$. Therefore $(\cH_1 \cap \cH_2) / (\cT_1 \cap \cH_2)$ has at least one face, and is in particular non-empty. We also know that $\widetilde{\cH}_1 / (\cT_1 \cap \widetilde{\cH}_1) \sqcup (\cH_1 \cap \cH_2) / (\cT_1 \cap \cH_2) = \cH_1 / \cT_1$ is face-connected, as well as $(\cH_1 \cap \cH_2) / (\cT_1 \cap \cH_2) \sqcup \widetilde{\cH}_2 = \cH_2 / ( \cT_1 \cap \cH_2 )$. Therefore $( \cH_1 \cup \cH_2 ) / \cT_1$ is itself face-connected.
Finally, since $\widetilde{\cH}_2 \neq \emptyset$, this is only possible if an external face of $\cH_1 / \cT_1$ is internal in $( \cH_1 \cup \cH_2 ) / \cT_1$. We conclude that $f_1$ is internal to $( \cH_1 \cup \cH_2 ) / \cT_1$, hence to $\cH_1 \cup \cH_2$.

We have just shown that all the external faces of $\cH_1$ are internal to $\cH_1 \cup \cH_2$. Likewise, all the external faces of $\cH_2$ are internal to $\cH_1 \cup \cH_2$. Therefore $F_{ext} ( \cH_1 \cup \cH_2 ) = \emptyset$, which implies that $\cH_1 \cup \cH_2 = \cG$ is vacuum, and contradicts our hypotheses.

\

We can proceed in a similar way than for (ii) to prove the last statement. Assume $\cH_1 , \ldots , \cH_k$ to be melonic, line-disjoint, and face-connected in their union. The connectedness of $\cH_1 \cup \dots \cup \cH_k$ and any of its reduction by a forest implies that all the external faces of $\cH_i$ are internal in $\cH_1 \cup \dots \cup \cH_k$, for any $1 \leq i \leq k$. Therefore the latter is vacuum, and this again contradicts the fact that $\cG$ is not.  
\end{proof} 

\

\noindent{\bf{Example.}} Figure \ref{overlap} represents two non-trivial melonic graphs $\cH_1$ and $\cH_2$ which are line-disjoint but face-connected in their union. Accordingly, their union is not melonic, as can be checked explicitly. 






\chapter{Super-renormalizable $\U(1)$ models in four dimensions}\label{chap:u1}

In this chapter, we illustrate the general TGFT formalism introduced previously, specializing to $d = 4$ and $G = \U(1)$. We will use angle coordinates $\theta_\ell \in \left[ 0 , 2 \pi \right[$ for the group elements $g_\ell = \e^{\rm{i} \theta_\ell}$, parameterizing the field $\vphi(\theta_1 , \dots , \theta_4)$. We do not make any additional hypothesis on the set of bubble interactions $\cB$ other than assuming it finite.

This specific set of models was introduced in \cite{u1} in order to explore how gauge invariance affects renormalizability in TGFT. It should be stressed that it is not expected to have physical relevance, but it allows to understand the general structures gauge invariance requires. The situation is simplified in two respects as compared to the more interesting model we will report on in the next chapter: the gauge group is Abelian, which eases the understanding of the divergences; and as we will see, this model is super-renormalizable, therefore the renormalization procedure itself is more straightforward. Finally, on the mathematical side, we will explain how the Wick ordering procedure can be generalized to the tensorial case, which is non-trivial and interesting in itself.

\section{Divergent subgraphs and Wick ordering}

Since $D = 1$, the divergence degree of a face-connected subgraph $\cH \subset \cG$ is given by
\bes\label{omega_u1su2}
\omega (\cH) &=& - 2 L(\cH) + F(\cH) - R(\cH) \\
&=& 2 \left( 1 - V (\cH) \right) + \rho( \cH / \cT )\,.
\ees
We see in particular that the $L$ contribution from formula (\ref{rholv}) vanishes. Since $\rho$ is bounded by $1$, it immediately follows that only tadpoles can be divergent, and therefore we are in presence of a super-renormalizable model (if renormalizable at all). If we were to focus on non-vacuum divergences only, we could already conclude that they all come from melopoles, thanks to proposition \ref{rho}. However, we will see that vacuum divergences are not all melonic, and therefore classifying them from such a perspective would require a refined understanding of $\rho$. Moreover, when this model was studied in \cite{u1}, the expression of the divergence degree in terms of $\rho$ was not available yet, and a different route was therefore explored, relying on a bound on $\omega$ rather than an identity. Since we will dispense ourselves from analyzing the vacuum divergences of the $\SU(2)$ model \cite{su2} in the next chapter, we feel that those of the present Abelian case provide an interesting example. For this reason we need in any case the bound on $\omega$ derived in \cite{u1}, and therefore decide to follow the original classification of divergences, which does not rely on the second line of equation (\ref{omega_u1su2}).

\subsection{A bound on the divergence degree}

We need to determine the set of divergent subgraphs, that is those $\cH$ such that $\omega (\cH) \geq 0$. In order to prove the model to be renormalizable, it will also be necessary to find a uniform decay of the amplitude associated to convergent graphs ($\omega < 0$), with respect to their external legs. In this respect, a suitable bound on $\omega$ in terms of simple combinatorial quantities is sufficient. We can for instance decompose the number of faces with respect to the number of lines they consist of. We call $F_k$ the number of internal faces with $k$ lines, and $F_{ext,k}$ the number of external faces with $k$ lines, so that:
\beq
F = \sum_{k \geq 1} F_k\,, \qquad F_{ext} = \sum_{k \geq 1} F_{ext, k}\,.
\eeq
We can also express the number of lines in terms of these quantities. Since $4$ different faces run through each line of $\cH$, we have:
\beq
4 L = \sum_{k \geq 1} k F_k + \sum_{k \geq 1} k F_{ext, k}\,,
\eeq
where in this formula both sums start with $k = 1$. We can therefore rewrite $\omega$ as
\beq
\omega = \sum_{k \geq 1} \left( 1 - \frac{k}{2} \right) F_k - \sum_{k \geq 1} \frac{k}{2} F_{ext, k} - R \,.
\eeq

We remark that the only positive contribution in this sum is given by $F_1$, to which only $p$-dipoles with $p \geq 2$ contribute. More precisely, 
\beq
F_1 = D_2 + 2 D_3 + 3 D_4 + 4 D_5\,,
\eeq
where $D_p$ is the number of $p$-dipole lines in $\cH$. We are thus lead to find a bound on $R$ in terms of these numbers of dipoles, which is the purpose of the following lemma.
\begin{lemma}
The rank of the incidence matrix associated to a face-connected graph $\cH$ verifies:
\beq
R \geq D_2 + D_3 + D_4 + D_5\,.
\eeq
\end{lemma}
\begin{proof}
Each $p$-dipole with $p \geq 2$ contains at least one internal face, which is independent of all the faces appearing in other lines. 
\end{proof}

Plugging this inequality into the expression of $\omega$ yields the following bound:
\beq\label{bound_om}
\omega \leq D_5 + \frac{D_4}{2} - \frac{D_2}{2} - \sum_{k \geq 3} \left( \frac{k}{2} - 1 \right) F_k - \sum_{k \geq 1} \frac{k}{2} F_{ext, k}\,.
\eeq
Note also that $D_5$ is always $0$, unless $\cH$ is the unique vacuum graph with a single line (sometimes called \textit{supermelon}). So the only non-trivial positive contribution comes from the $4$-dipoles. This way we already guess that melopoles are responsible for most of the divergences, which we confirm below.

\subsection{Classification of divergences}
To control the contribution of $D_4$ in (\ref{bound_om}), we take a step back and analyze the (exact) effect on $\omega$ of a $4$-dipole contraction in a face-connected graph $\cH$. 

\begin{proposition}
Let $\cH$ be a face-connected subgraph, and $l$ a $4$-dipole line. Then
\beq
\omega(\cH) = \omega(\cH / l)\,.
\eeq
\end{proposition}
\begin{proof}
We have immediately $L (\cH / l) = L(\cH) - 1$ and $F (\cH / l) = F(\cH) - 3$. As for the rank of the incidence matrix, it was already pointed out that: $R(\cH / l) = R(\cH) - 1$. Therefore:
\beq
\omega(\cH / l) = \omega(\cH) + 2 - 3 + 1 = \omega( \cH )\,.
\eeq
\end{proof}

This property can be used to recursively reduce the analysis to that of graphs with a few melonic lines. For such graphs, (\ref{bound_om}) is constraining enough, and we can obtain the following classification.
\begin{proposition}\label{fund_u1}
Let $\cH \subset \cG$ be a face-connected subgraph.
\begin{itemize}
\item If $\omega(\cH) = 1$, then $\cH$ is a vacuum melopole.
\item If $\omega(\cH) = 0$, then $\cH$ is either a non-vacuum melopole, or a \textit{submelonic vacuum graph} (see Figure \ref{submelonic}). 
\item Otherwise, $\omega(\cH) \leq -1$ and $\omega(\cH) \leq - \frac{N(\cH)}{4}$.
\end{itemize}
\end{proposition}
\begin{proof}
Let us first assume that $\cH$ is a vacuum graph. We can perform a maximal set of successive $4$-dipole contractions, so as to obtain a graph $\widetilde{\cH}$ with $D_4 = 0$ and same power-counting as $\cH$. If $D_5 (\widetilde{\cH}) = 1$, then $\widetilde{\cH}$ is the supermelon graph, which means that $\cH$ is a melopole, and $\omega(\cH) = \omega(\widetilde{\cH}) = - 2 + 4 - 1 = 1$. On the other hand, when $D_5 (\widetilde{\cH}) = 0$, equation (\ref{bound_om}) gives
\beq
\omega( \widetilde{\cH} ) \leq  - \frac{D_2 (\widetilde{\cH}) }{2} - \sum_{k \geq 3} \left( \frac{k}{2} - 1 \right) F_k ( \widetilde{\cH} )\,,
\eeq
from which we infer that $\omega( \widetilde{\cH} ) \leq -1$ unless perhaps when $D_2 (\widetilde{\cH}) = F_k (\widetilde{\cH} ) = 0$ for any $k \geq 3$. But it is easy to see that these conditions immediately imply that $\widetilde{\cH}$ has one of the structures shown in Figure \ref{2cases}. A direct calculation then confirms that $\omega = 0$ for the left drawing, but $\omega = 1$ for the drawing on the right. This finally shows that $\omega = 0$ graphs are exactly the minimal graph on the left side of Figure \ref{2cases} dressed with additional melopoles, as shown in Figure \ref{submelonic}. We propose to call them \textit{submelonic vacuum graphs}. 

\begin{figure}
\begin{center}
\includegraphics[scale=0.5]{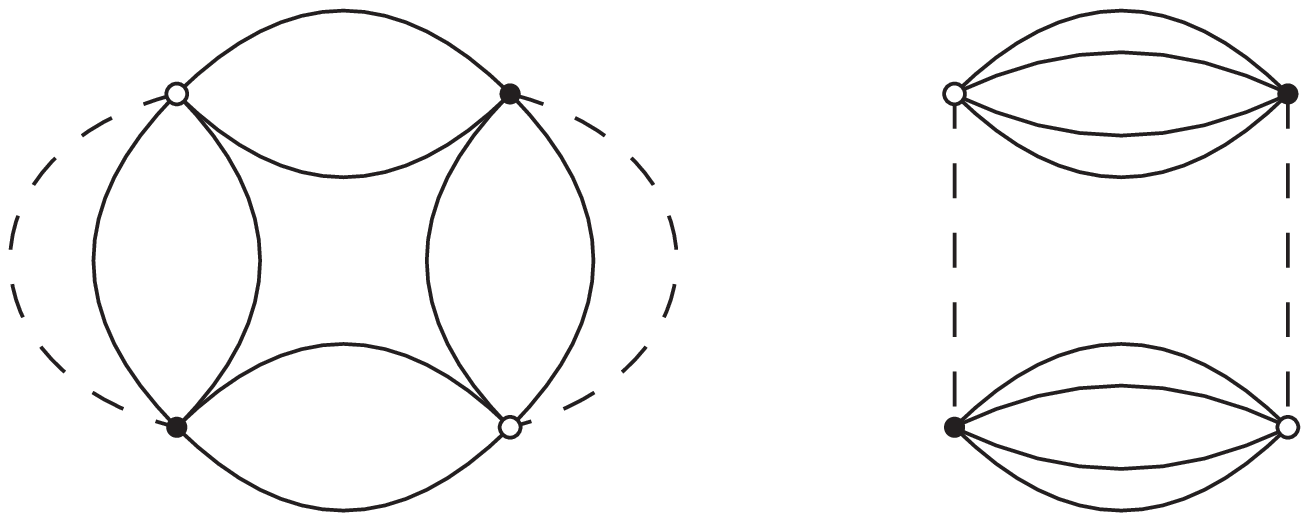}
\caption{Two vacuum graphs with $D_2 (\widetilde{\cH}) = F_k (\widetilde{\cH} ) = 0$ for any $k \geq 3$.}
\label{2cases}
\end{center}
\end{figure}

\
Let us now consider the case of a non-vacuum graph $\cH$. We can again perform a maximal set of $4$-dipole contractions 
and construct a new graph $\widetilde{\cH}$ verifying either: a) $L(\widetilde{\cH}) = D_4 (\widetilde{\cH}) = 1$; or b) $ D_4 (\widetilde{\cH}) = 0 $. In situation a), $\widetilde{\cH}$ reduces to a single $4$-dipole line, $\cH$ itself is a melopole, and $\omega(\cH) = \omega(\widetilde{\cH}) = - 2 + 3 - 1 = 0$. In situation b), the bound on $\omega$ gives
\beq
\omega( \widetilde{\cH} ) \leq  - \frac{D_2 (\widetilde{\cH}) }{2} - \sum_{k \geq 3} \left( \frac{k}{2} - 1 \right) F_k ( \widetilde{\cH} ) - \sum_{k \geq 1} \frac{k}{2} F_{ext , k} (\widetilde{\cH}) < 0\,,
\eeq
which shows that $\omega ( \widetilde{\cH} ) = \omega ( \cH ) \leq - 1$. We can finally prove a decay in terms of the number of external lines. For instance, we remark that the connectedness of $\widetilde{\cH}$ implies that at least one face going through a given external leg is of the type $F_{ext , k}$ with $k \geq 1$. And because each of these faces contains two external legs, we have $\sum_{k \geq 1} F_{ext , k} \geq \frac{N}{2}$. So we finally obtain
\beq
\omega(\cH) = \omega( \widetilde{\cH} ) \leq - \sum_{k \geq 1} \frac{k}{2} F_{ext , k} (\widetilde{\cH}) 
\leq - \frac{1}{2} \sum_{k \geq 1} F_{ext , k} \leq - \frac{N}{4}\,.
\eeq
 
\
All possible situations have been scanned, which ends the proof. 
\end{proof}

This classification allows to identify melopoles as the only source of divergences in the scale decomposition of non-vacuum (vertex-)connected amplitudes. Any model with a finite set of $4$-bubble interactions comes with a finite number of melopoles, and is therefore expected to be super-renormalizable. The purpose of the next sections is to prove that this is indeed the case.

\begin{figure}
\begin{center}
\includegraphics[scale=0.5]{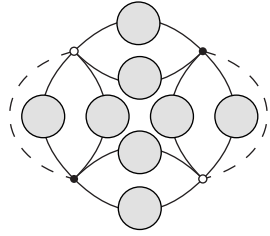}
\caption{The class of submelonic vacuum graphs: grey blobs represent melopole insertions.}
\label{submelonic}
\end{center}
\end{figure}

\subsection{Localization operators}

Combinatorial contractions of graphs can be represented as localization operators acting on the amplitudes. They are the technical tools allowing to erase high energy information and reabsorb it into effective local counter-terms. To define them, let us consider a graph $\cG$, and a face-connected subgraph $\cH \subset \cG$. We define an operator $\tau_\cH$ by its action on the integrand of $\cG$. The amplitude $\cA_\cG$ is of the form:
\bes\label{ampl}
\cA_\cG &=& \left[ \prod_{e \in L(\cH)} \int \extd \alpha_{e} \, e^{- m^2 \alpha_e} \int \extd \lambda_e \right] 
\left( \prod_{f \in F(\cH)} K_{\alpha(f)}\left( \sum_{e \in \partial f} \epsilon_{ef} \lambda_e \right) \right) \\ 
 && \left[ \prod_{e \in N(\cH)} \int \extd \theta_e \right] 
 \left( \prod_{f \in F_{ext}(\cH)} K_{\alpha(f)} \left( \theta_{s(f)}
+ \sum_{e \in \partial f} \epsilon_{ef} \lambda_e - \theta_{t(f)}\right) \right)
\, \cR_{\cG \setminus \cH}\left(\{\theta_{s(f)}, \theta_{t(f)}\}\right)\,, \nn
\ees
where $\cR_{\cG \setminus \cH}$ only depends on $\theta$ variables appearing in the external faces of $\cH$. We then define $\tau_\cH$ as:
\beq
\tau_\cH \cR_{\cG \setminus \cH}\left(\{\theta_{s(f)}, \theta_{t(f)}\}\right) \equiv \cR_{\cG \setminus \cH}\left(\{\theta_{s(f)}, \theta_{s(f)}\}\right)\,,
\eeq 
that is by moving all target variables of the external faces of $\cH$ to the sources. This definition is motivated by the fact that when the parallel transports inside $\cH$ are negligible, holonomies along external faces can be well approximated by directly connecting the two points at the boundary of $\cH$.

\begin{proposition} 
 Let $\cH \subset \cG$ be a face-connected subgraph. The action of $\tau_{\cH}$ on $\cA_\cG$ factorizes as:
\beq
\tau_\cH \cA_\cG = \nu_{\rho}(\cH) \cA_{\cG / \cH} \,,
\eeq
where $\nu_{\rho}(\cH)$ is a numerical coefficient depending on the cut-off $\rho$, and given by the following integral:
\beq
\nu_{\rho}(\cH) \equiv  \left[ \prod_{e \in L(\cH)} \int_{M^{-2 \rho}}^{+ \infty} \extd \alpha_{e} \, e^{- m^2 \alpha_e} \int \extd \lambda_e \right] 
\left( \prod_{f \in F(\cH)} K_{\alpha(f)}\left( \sum_{e \in \partial f} \epsilon_{ef} \lambda_e \right) \right)\,.
\eeq
\end{proposition}
\begin{proof}
 Applying $\tau_\cH$ in equation (\ref{ampl}), one remarks that heat kernels associated to external amplitudes can readily be integrated with respect to the variables $\theta_{t(f)}$. These integrals give trivial contributions, thanks to the normalization of the heat kernel. We are therefore left with an integral over the internal faces of $\cH$, giving $\nu_{\rho}(\cH)$, times an amplitude which is immediately identified to be that of the contracted graph $\cG / \cH$.
\end{proof}
{\bf Remark. } One can also use the same kind of factorization for the action of $\tau_\cH$ on $\cA_{\cG , \mu}$, in which case we will use the notation $\nu_\mu (\cH)$: 
\beq
\tau_\cH \cA_{\cG , \mu} = \nu_{\mu}(\cH) \cA_{\cG / \cH , \mu}\,.
\eeq

An illustration of the factorizability property is given in Figure \ref{melo1}. This definition of the localization operators will be sufficient for the renormalization of logarithmic divergences, hence for all the cases we will have to tackle in this chapter. As usual, the renormalization of power-like divergences in more complicated models will require to push to a higher order the Taylor expansion around localized terms. An example will be provided by the model of Chapter \ref{chap:su2}, where not only the contraction operators will take higher order terms into account, but also in a non-commutative context. 

\subsection{Melordering}

In the usual super-renormalizable $P(\phi)_2$ field theory \cite{Simon}, the finite set of counter-terms that are needed to tame divergences is simply provided by Wick ordering. It consists in a simple change of basis of interaction invariants, the coupling constants in this new basis being the renormalized ones. The net effect of Wick ordering at the level of the Feynman expansion is to simply cancel the contributions of graphs with tadpoles. This suggests a similar strategy to remove the special kind of tadpoles that are responsible for the divergences of our tensorial models, that is the melopoles. We will call this particular version of Wick ordering the \textit{melordering}.

\
Before going to the details of melordering, a few preliminary remarks are in order, as we have to face a few subtleties introduced by the refined notion of connectedness on which TGFT relies.
In scalar theories, tadpole lines are exactly local objects, in the sense that their contributions can be factorized exactly. This is the reason why Wick ordering can be defined as a choice of a family of orthogonal polynomials with
respect to the regularized covariance \cite{salmhofer}. When such invariants are used as a basis to express the interaction part of the action, their expectation values in the vacuum is zero, and more generally all
tadpole contributions cancel out exactly. In tensorial theories however, we have seen that tadpoles can only be approximately local, at the condition of them being tracial (which melopoles are).
We therefore cannot hope to cancel them exactly, but only to eliminate their local divergent part. An important consequence is for example that melordered invariants will not necessarily have a zero expectation value in the vacuum, but only a finite one (at the additional condition that submelonic vacuum counter-terms are added when needed, see section \ref{sec:sub}).

\
We now proceed with the definition of melordering.
Let us call $\Inv$ the vector space of connected tensor invariants, generated by the $4$-bubbles. Associated to the regularized covariance $C^{\rho}$, we want to define a linear and bijective map
$\Omega_{\rho}: \Inv \mapsto \Inv$ 
that maps any $4$-bubble to a suitably weighted sum of lower order $4$-bubbles. Getting inspiration from the scalar case, one should define $\Omega_{\rho}(I_{b})$ as a sum over pairings of the
external legs of $b$. The relevant pairings will be those resulting in melopoles. As we will see in an explicit example (see section \ref{sec:example}), a single connected invariant can give rise to several face-disjoint melopoles. For this reason, and despite the super-renormalizable nature of the model, the counter-terms have already a rich structure, only captured by the full machinery of Zimmermann forests. In such an approach, the renormalized amplitudes are given by sums over inclusion forests of divergent subgraphs $\cF$, of contractions of the bare amplitudes
\beq
\cA_\cG^R = \sum_\cF \prod_{\cH \in \cF} (- \tau_\cH ) \cA_\cG \,.
\eeq
In our case, the relevant structure is given by inclusion forests of face-connected melopoles, which we call \textit{meloforests} and define with respect to both subgraphs and bubble invariants.
\begin{definition}
\begin{enumerate}[(i)]
\item Let $\cH \subset \cG$ be a subgraph. A meloforest $\cM$ of $\cH$ is a set of non-empty and face-connected melopoles of $\cH$, such that: for any $m , m' \in \cM$, either $m$ and $m'$ are line-disjoint and face-disjoint, or $m \subset m'$ or $m' \subset m$. We note $\cM(\cH)$ the set of meloforests of $\cH$.
\item Let $b$ be a $4$-bubble. A meloforest $\cM$ of $b$ is a meloforest for a graph made of a single vertex $b$. We call $I_{b, \rho}^{\cM}$ the observable associated to the smallest such graph, namely $\underset{m \in \cM}{\cup} m$. We note $\cM(b)$ the set of meloforests of $b$.
\end{enumerate}
\end{definition}
Meloforests have a relatively simple structure, due to a uniqueness property \cite{uncoloring,universality}. 
\begin{lemma}
Let $b$ be a $4$-bubble. There exists a unique graph $\cG$ such that any meloforest of $b$ is a meloforest of $\cG$.
\end{lemma}
\begin{proof}
As remarked in \cite{uncoloring,universality}, only melonic $2$-point subgraphs (in the sense of colored graphs) of $b$ can be closed in face-connected melopoles, and there is a unique way of doing so. Closing the maximal $2$-point subgraphs of $b$ in such a way results therefore in the unique graph $\cG$.
\end{proof}

We can now proceed with the definition of the melordering map.
\begin{definition}
For any $4$-bubble $b$, associated to the invariant $I_b$, and a cut-off $\rho$, we define the \textit{melordered invariant} $\Omega_{\rho}(I_{b})$ as
\beq
\Omega_{\rho}(I_{b}) \equiv \sum_{\cM \in \cM(b)} \prod_{m \in \cM} \left( - \tau_{m} \right) I_{b, \rho}^{\cM}\,.
\eeq
\end{definition}
By convention, the sum over meloforests includes the empty one, so that $\Omega_{\rho}(I_{b})$ has same order as $I_{b}$. Products of contraction operators are commutative, the definition is therefore unambiguous. These (non-trivial) products of contractions ensure that each term in the sum is a weighted $4$-bubble invariant, making $\Omega_{\rho}$ a well-defined linear map from $\Inv$ to itself. An example is worked out explicitly in section \ref{sec:example}.

\
Consider now the theory defined in terms of melordered interactions at cut-off $\rho$, with partition function:
\bes
\cZ_{\Omega_\rho} &=& \int \extd \mu_{C_\rho} (\vphi , \vphib) \, \e^{- S_{\Omega_\rho}(\vphi , \vphib )} \,, \\
S_{\Omega_\rho}(\vphi , \vphib ) &=& \sum_{b \in \cB} t_b^R \, \Omega_{\rho}(I_b )(\vphi , \vphib).
\ees
We shall then consider the perturbative expansion in the renormalized couplings $t^R_b$ and prove that the corresponding Feynman amplitudes are finite. Let us call $\cS_N^{\Omega_\rho}$ the $N$-point Schwinger function of the melordered model. The next proposition shows that renormalized amplitudes have the expected form.
\begin{proposition}
The $N$-point Schwinger function $\cS_N^{\Omega_\rho}$ expands as:
\beq
\cS_N^{\Omega_\rho} = \sum_{\cG \; \mathrm{connected}, N(\cG)= N} \frac{1}{s(\cG)} \left(\prod_{b \in \cB} (- t_b^R )^{n_b (\cG)}\right) \cA_\cG^{R} \,,
\eeq
where the renormalized amplitudes can be expressed in terms of the bare ones as
\beq
\cA_\cG^R = \left( \sum_{\cM \in \cM(\cG)} \prod_{m \in \cM} \left( - \tau_{m} \right) 
\right) \cA_\cG \,.
\eeq
\end{proposition}
\begin{proof}
We first remark that the set $\cM (\cG)$ of meloforests of $\cG$ can be described according to meloforests of bubble vertices $b \in \cB(\cG)$:
\beq
\cM (\cG) = \left\{ \underset{b \in \cB}{\bigcup} \cM_b | \cM_b \; { \rm meloforest} \; {\rm of} \; b \in \cB(\cG) \right\}\,.
\eeq 
$\cA_\cG^R$ as defined above can therefore be written
\bes
\cA_\cG^R &=& \left( \sum_{ (\cM_b)_{b \in \cB(\cG)} } \prod_{b \in \cB(\cG)} \prod_{m \in \cM_b} \left( - \tau_{m} \right) 
\right) \cA_\cG \\
&=&  \prod_{b \in \cB(\cG)} \left( \sum_{\cM_b} \prod_{m \in \cM_b} \left( - \tau_{m} \right) 
\right) \cA_\cG \,.
\ees
Each element of the product over $b \in \cB(\cG)$ is a contraction operator taking all melopoles associated to $b$ into account. Let us fix a graph $\cG$ and a bubble $b$. Among the set of Wick contractions appearing in $\cS_N^{\Omega_\rho}$, the operator $\underset{\cM_b}{\sum} \underset{m \in \cM_b}{\prod} \left( - \tau_{m} \right)$ encodes all the terms due to the interaction $\Omega_\rho (I_b)$ that are compatible with the combinatorics of the external legs of $b$ in $\cG$ and the structure of the rest of the graph. We therefore understand that $\cS_N^{\Omega_\rho}$ as written above is a valid repackaging of all the Wick contractions generated by the melordered interaction.
\end{proof}
We will devote the whole section \ref{sec:finiteness} to proving that the renormalized amplitudes are indeed finite. Before that, we return to submelonic vacuum divergences.

  \subsection{Vacuum submelonic counter-terms}\label{sec:sub}

The melordering we just introduced is designed to remove melopole divergences, including logarithmic divergences of non-vacuum graphs and linear divergences resulting from vacuum melopoles. However, we have seen that a third source of divergences is given by submelonic vacuum graphs. They again concern tadpole graphs, so they can also be removed by adding extra counter-terms to the melordering of some of the bubbles. As long as we are concerned with computations of transition amplitudes, they are irrelevant since they will only affect $\cZ$ and none of the connected Schwinger functions.  But we include them here for completeness.  
 
\
We can define an \textit{extended melordering} $\overline{\Omega}_\rho$ that coincides with $\Omega_\rho$ for bubbles which cannot be closed in a submelonic vacuum graph, and adds additional counter-terms to those which can. We can call the latter \textit{submelonic bubbles}. They are exactly the bubbles that reduce to a four-point graph as in Figure \ref{sub_open} once all the melonic parts have been closed into melopoles and contracted. Such bubbles generate additional divergent forests, which we can call \textit{submelonic forests}:
\begin{definition}
Let $b$ be a submelonic bubble. A submelonic forest of $b$ is a forest $\cS = \cM \cup \{ \cG \}$, where $\cM$ is a meloforest and $\cG$ is a vacuum graph with a single vertex $b$. We call $I_{b, \rho}^{\cS}$ the amplitude associated to the graph $\cG$. We call $\cS(b)$ the set of submelonic forests of $b$.
\end{definition}   
{\bf Remark.} Given a submelonic bubble, there are exactly two possible choices for $\cG$, which correspond to the two possible ways of closing the melopole-free graph of Figure \ref{sub_open}.

\begin{figure}
\begin{center}
\includegraphics[scale=0.5]{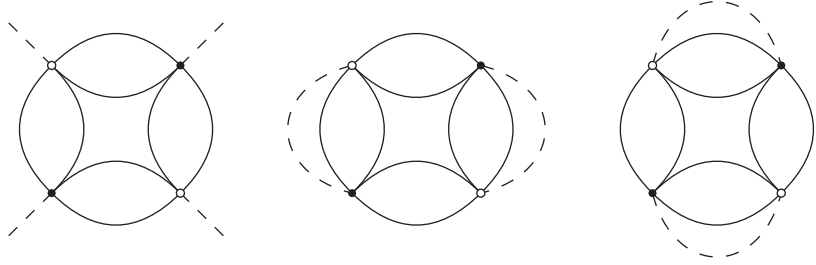}
\caption{On the left: structure of a submelonic bubble once all melonic parts have been closed into melopoles and contracted. On the right: the two ways of obtaining a submelonic vacuum graph.}
\label{sub_open}
\end{center}
\end{figure}

The extended melordering is finally defined by
\beq
\overline{\Omega}_{\rho}(I_{b}) \equiv \sum_{\cM \in \cM(b)} \prod_{m \in \cM} \left( - \tau_{m} \right) I_{b, \rho}^{\cM} + \sum_{\cS \in \cS(b)} \prod_{s \in \cS} \left( - \tau_{s} \right) I_{b, \rho}^{\cS}
\eeq
when $b$ is submelonic. This implies similar formulas for renormalized amplitudes in the extended melordered model, which in particular do not affect the expression for melordered connected Schwinger functions. The only difference will be that the partition function of the extended melordered model will be well-defined as a formal series, contrary to the simple melordering for which $\cZ$ will have some logarithmically divergent coefficients. 

\section{Finiteness of the renormalized series}\label{sec:finiteness}

In this section, we prove that melordered models with maximal interaction order $p < + \infty$ are perturbatively finite at any order. To avoid dealing with submelonic vacuum graphs, we will only focus on the connected Schwinger functions, which are the physically meaningful quantities after all. They are well-defined formal series in the renormalized couplings if all non-vacuum and connected renormalized amplitudes $\cA_\cG^{R}$ are finite.

\
We will again rely on the multi-scale analysis, following the usual procedure of \cite{vincent_book}, which consists in two steps. We first need to show that renormalized amplitudes associated to bare divergent graphs verify multi-scale convergent bounds. This is most conveniently done through a classification of divergent forests (in our case meloforests), which for a given scale attribution $\mu$, splits these in two families: the \textit{dangerous} ones, associated to high subgraphs, that cancel genuine divergences; and on the other hand \textit{inoffensive} divergent forests that do not have any quasi-locality property, henceforth do not serve any purpose. The inoffensive forests bring finite contributions that do not ruin the power-counting, and can rather be interpreted as a drawback of the renormalized series: they have no physically meaningful consequence, and in addition (in just renormalizable models, but not for super-renormalizable models like the ones treated here) results in "renormalon effects" that typically prevents from constructing a convergent series. In a second step, we will prove that the sum over scale attributions can be performed, and the cut-off $\rho$ sent to infinity while keeping the amplitudes finite. 

  \subsection{Classification of forests}

We follow the general classification procedure of \cite{vincent_book}, which at each scale attribution allows to factorize the contraction operators defining the renormalized amplitude. Let $\cG$ be a vertex-connected (non-vacuum) graph. We can decompose the renormalized amplitude $\cA_\cG^R$ in terms of its scale attributions:
\beq\label{2sums}
\cA_\cG^R = \sum_\mu \sum_{\cM \in \cM(\cG)} \prod_{m \in \cM}(- \tau_m) \cA_{\cG , \mu} \,. 
\eeq
The classification of forests is a reshuffling of the sum over meloforests that allows to permute the two sums. We know that for a given scale attribution $\mu$, the forests that contribute to the divergences are those containing high melopoles. We therefore need to define the notion of high meloforest, and reorganize the sum in terms of these quantities. We follow the standard procedure \cite{vincent_book}, and start with the following set of definitions.

\begin{definition}
Let $\cG$ be a vertex-connected graph, $\mu$ a scale attribution, and $\cM$ a meloforest of $\cG$.
\begin{enumerate}[(i)]
\item We say that a subgraph $g \subset \cG$ is \textit{compatible} with a meloforest $\cM$ if $\cM \cup \{ g \}$ is a forest.
\item If $g$ is compatible with a meloforest $\cM$, we note $B_{\cM} (m)$ the \textit{ancestor} of $g$ in $\cM \cup \{ g \}$, and we similarly call $A_\cM (g) \equiv \{m \subset g | m \in \cM \}$ the \textit{descendants}.
\item \textit{Internal and external scales} of a compatible graph $g$ in a meloforest $\cM$ are defined by:
\beq
i_{g , \cM}(\mu) = \inf_{e \in L( g \backslash A_\cM (g) )} i_e(\mu)\,, 
\qquad e_{g , \cM}(\mu) = \sup_{e \in N(g) \cap B_\cM (g)} i_e(\mu).
\eeq
\item The \textit{dangerous} part of a meloforest $\cM$ with respect to $\mu$ is:
\beq
D_\mu (\cM) = \{ m \in \cM | i_{m , \cM}(\mu) > e_{m , \cM}(\mu) \}\,,
\eeq 
and the \textit{inoffensive} part is the complement $I_\mu (\cM) = \cM \backslash D_\mu (\cM)$. Finally $I(\mu)$ is the set of all inoffensive forests in $\cG$.
\end{enumerate}
\end{definition} 
{\bf{Remarks.}} The notions of internal and external scales with respect to a meloforest are consistent with the previous definitions, since $i_{m , \emptyset} = i_m$ and $e_{m, \emptyset} = e_m$. Moreover, (non-vacuum) melopoles have exactly two external legs, which makes the situation relatively simple.

The following important lemma leads to the partition of forests.
\begin{lemma}
Given a meloforest $\cM$,
\beq
I_\mu ( I_\mu (\cM) ) =  I_\mu (\cM)\,. 
\eeq
\end{lemma}
\begin{proof}
Similar to \cite{vincent_book}, but simpler. 
\end{proof}

This implies that the set of meloforests $\cM( \cG )$ of a vertex-connected graph $\cG$ can
be partitioned, according to the inoffensive forests associated to any scale attribution $\mu$:
\beq
\cM( \cG ) = \underset{\cM | I_\mu (\cM) = \cM }{\bigcup} \{ \cM' | I_\mu ( \cM' ) = \cM \}\,.
\eeq
We can finally characterize the equivalence class of a meloforest $\cM$ by introducing its maximal forest $\cM \cup H_\mu (\cM)$, where
\beq
H_\mu (\cM) = \{ m \; {\rm compatible} \; {\rm with} \; \cM \, |  \; m \in D_\mu ( \cM \cup \{ m \})\}\,.
\eeq
We can indeed show that:
\begin{proposition}
For any $\cM \in I(\mu)$, $\cM \cup H_\mu (\cM)$ is a meloforest, and moreover:
\beq
\forall \cM' \in \cM(\cG), \, I_\mu(\cM') = \cM \Longleftrightarrow \cM \subset \cM' \subset \cM \cup H_\mu (\cM)\,.
\eeq
\end{proposition}
\begin{proof}
Similar to \cite{vincent_book}, but simpler. 
\end{proof}

This finally allows to reorganize the operator defining the renormalized amplitude as
\beq
\sum_{\cM \in \cM(\cG)} \prod_{m \in \cM}(- \tau_m) = \sum_{\cM \in I(\mu)} \prod_{m \in \cM} (- \tau_m) \prod_{h \in H_\mu (\cM)} (1 - \tau_h)\,,
\eeq
which decomposes the product of contraction operators into inoffensive parts and high parts. And since it holds for any $\mu$, we can use this formula to invert the two sums in (\ref{2sums}) and obtain:
\bes\label{classified}
\cA_\cG^R &=& \sum_{\cM \in \cM(\cG)} \cA_{\cG , \cM}^R\,,\\
\cA_{\cG , \cM}^R &=& \sum_{\mu | \cM \in I(\mu)} \prod_{m \in \cM} (- \tau_m) \prod_{h \in H_\mu (\cM)} (1 - \tau_h) \cA_{\cG , \mu}\,.
\ees
The factorization (\ref{classified}) is key to the proof of finiteness. We shall first show that, with respect to the bare theory, the power-counting of $\cA_{\cG , \cM}^{R}$ for a given scale attribution is improved, and is always convergent. We will then explain why the sum over scales is finite given such convergent multiscale bounds. The final sum over meloforests will not bring more divergences, since their cardinal is finite (and even bounded by $K^{n(\cG)}$
for some $K>0$).
	
  \subsection{Power-counting of renormalized amplitudes}	

Let us fix a meloforest $\cM$ and a scale attribution $\mu$ such that $\cM \in I(\mu)$. The product of operators acting on $\cA_{\cG , \mu}$ in (\ref{classified}) can be computed explicitly. We can for example first act with $\underset{m \in \cM}{\prod} \tau_m$ which evaluates as
\beq
\prod_{m \in \cM} \tau_m \, \cA_{\cG , \mu} = \left( \prod_{m \in \cM} \nu_\mu (m / {A_\cM(m)}) \right) \cA_{\cG / \cM , \mu}\,,
\eeq
where $\cG / \cM$ is the graph obtained from $\cG$ once all the subgraphs of $\cM$ have been contracted. This graph is nothing but $\cG / {A_{\cM}(\cG)}$. 
$\nu_\mu$ is a generalized notion of amplitude associated to subgraphs $\cH \subset \cG$, that just discards the contributions of external faces. In this sense, it is analogue to an amputated amplitude in usual field theories. In particular, we can assume that $\cA_{\cG / \cM , \mu}$ is an amputated amplitude and write:
\beq\label{intermediate}
\prod_{m \in \cM} \tau_m \, \cA_{\cG , \mu} = \prod_{ g \in \cM \cup \{ \cG \}} \nu_\mu (g / {A_{\cM }(g)})\,.
\eeq
The power-counting, which only depends on internal faces, is unaffected by the fact that we are working with such amputated amplitudes, and we conclude that
\beq
\vert \prod_{m \in \cM} (- \tau_m) \cA_{\cG , \mu} \vert \leq K^{L(\cG)} \prod_{ g \in \cM \cup \{ \cG \}} \prod_{(i , k)} M^{\omega[ ( g / {A_{\cM}(g)} )_i^{(k)} ]}\,.
\eeq 
This is a generalization of the power-counting (\ref{fund_ab}), and reduces to it when $\cM = \emptyset$. This proves that the sum over inoffensive forests does not improve nor worsen the power-counting, as was expected. Finiteness is entirely implemented by the useful part of the contraction operators, namely $\underset{h \in H_\mu (\cM)}{\prod} (1 - \tau_h)$. To make this apparent, we first write it as
\beq
\prod_{h \in H_\mu (\cM)} (1 - \tau_h) = \prod_{g \in \cM \cup \{ \cG \}} \prod_{h \in H_\mu (\cM) | B_\cM(h) = g} (1 - \tau_h)
\eeq
and act on (\ref{intermediate}) to get
\beq\label{intermediate2}
\vert \prod_{h \in H_\mu (\cM)} (1 - \tau_h) \prod_{m \in \cM} (- \tau_m) \cA_{\cG , \mu}  \vert
=  \prod_{ g \in \cM \cup \{ \cG \}} \prod_{h \in H_\mu (\cM) | B_\cM(h) = g} \vert (1 - \tau_h) \,  \nu_\mu (g / {A_{\cM}(g)}) \vert\,.
\eeq
Now, the effect of each $(1 - \tau_h)$ is to interpolate one of the variables of (at most two) external propagators in $N (h) \cap ( g / {A_{\cM}(g)} )$. For example, assuming the fourth variable is concerned (that is $h$ is a melopole that has been inserted on an colored line of color $4$), we have something of the form
\beq
C_{i}( \theta_1 , \dots , \theta_4 ; \theta_\ell') - C_i( \theta_1 , \dots , \tilde{\theta}_4 ; \theta_\ell')
= \int_{0}^{1} \extd t \left( \theta_4 - \tilde{\theta}_4  \right) \frac{\partial}{ \partial \theta_4} C_i( \theta_1 , \dots , \tilde{\theta}_4 + t (\theta_4 - \tilde{\theta}_4) ; \theta_\ell')\,,
\eeq
with $i \leq e_{\cM , h} (\mu)$. Moreover, since $h$ is high in $g / {A_{\cM}(g)}$, $| \theta_4 - \tilde{\theta}_4 |$ is at most of order $M^{- i_{\cM , h} (\mu)} $. So using the bound (\ref{deriv_bound_ab}) on derivatives of the propagator, we conclude that $(1 - \tau_h)$ improves the bare power-counting by a factor:
\beq
M^{i} |\theta_4 - \tilde{\theta}_4| \leq K M^{ e_{\cM , h} (\mu) - i_{\cM , h} (\mu) }\,.
\eeq

This additional decay allows to prove the following proposition.
\begin{proposition}\label{pc_r}
There exists a constant $K$ such that for any graph $\cG$ and meloforest $\cM$:
\beq
| \cA_{\cG, \cM}^{R} | \leq  K^{L(\cG)}  \sum_{\mu | \cM \in I(\mu)} \prod_{ g \in \cM \cup \{ \cG \}} \prod_{ (i , k) } M^{\omega'[ ( g / {A_{\cM}(g)} )_i^{(k)} ]}\,,
\eeq
where
\beq
\omega'[ ( g / {A_{\cM}(g)} )_i^{(k)} ] = \min \{ - 1 , \omega[ ( g / {A_{\cM}(g)} )_i^{(k)} ] \}\,,
\eeq
except if $g \in \cM$ and $( g / {A_{\cM}(g)} )_i^{(k)} = g / {A_{\cM}(g)}$, in which case $\omega'(( g / {A_{\cM}(g)} )_i^{(k)}) = 0$.
\end{proposition}
\begin{proof}
From (\ref{intermediate2}), and using the additional decays from operators $(1 - \tau_h)$, one improves the degree by a factor $-1$ for most of the high subgraphs. More precisely, this is possible for any high subgraph that has external legs in a contraction $g / {A_{\cM}(g)}$, that is any high subgraph 
$( g / {A_{\cM}(g)} )_i^{(k)}$ different from a root  $g / {A_{\cM}(g)}$.
\end{proof}

  \subsection{Sum over scale attributions}

Equipped with this improved power-counting, we can finally prove that the renormalized amplitudes are finite. For clarity of the presentation, let us first show it for a fully convergent graph $\cG$, that is a graph with no melopole.   
In this case, we know that:
\beq
| \cA_{\cG , \mu} | \leq K^{L(\cG)}  \prod_{ (i , k ) } M^{- N (\cG_i^{(k)}) / 4}\,,
\eeq
from which we need to extract enough decay in $\mu$ to sum over the scale attributions. Let $\cB(\cG)$ be the set of vertices (i.e. $4$-bubbles) of $\cG$, and for $b \in \cB(\cG)$ let us call $L_b (\cG)$ the set of lines that are hooked to it. We can define notions of internal and external scales associated to a bubble $b$:
\beq
i_b (\mu) = \sup_{l \in L_b (\cG)} i_l (\mu)  \,, \qquad  e_b (\mu) = \inf_{l \in L_b (\cG)} i_l (\mu)\,.
\eeq
We then remark that for any $i \in \mathbb{N}$ and $b \in \cB(\cG)$, $b$ touches a high subgraph $\cG_i^{(k)}$ if and only if $i \leq i_b (\mu)$. Moreover when it does, the number of high subgraphs $\cG_i^{(k)}$ that touch $b$ is certainly bounded by its number of external legs, and therefore by $p$. Hence we can assign a fraction $1 / p$ of the decay of a bubble to every high subgraph with respect to which it is external. This yields  
\beq
 \prod_{ (i , k ) } M^{- N (\cG_i^{(k)}) / 4} \leq  \prod_{ (i , k ) } \prod_{b \in \cB(\cG_i^{(k)})| e_b (\mu) < i \leq i_b (\mu)}
 M^{-  \frac{1}{4p}}\,,
\eeq 
by using the fact that $b$ is an external vertex of $\cG_i^{(k)}$ exactly when $e_b (\mu) < i \leq i_b (\mu)$. We can then invert the two products and obtain
\beq
| \cA_{\cG , \mu} | \leq K^{L(\cG)}  \prod_{ b \in \cB(\cG) } \prod_{ (i , k) | e_b (\mu) < i \leq i_b (\mu)} M^{-  \frac{1}{4p}} = K^{L(\cG)}  \prod_{ b \in \cB(\cG) } M^{- \frac{i_b (\mu) - e_b (\mu)}{3 p}}\,.
\eeq
Finally, since the number of pairs of legs hooked to a given vertex $b$ is bounded by $p (p -1) / 2$, we can finally conclude that
\beq
| \cA_{\cG , \mu} | \leq K^{L(\cG)}  \prod_{ b \in \cB(\cG) } \prod_{(l , l') \in L_b (\cG) \times L_b (\cG)} M^{- \frac{ 2 | i_{l'} (\mu) - i_l (\mu) | }{3 p^2 (p - 1)}}\,.
\eeq
With this decay at hand, the sum over scales can be performed by picking a 'tree of scales', very similarly to the choice of a tree adapted to the GN tree that establishes the power-counting. We refer to \cite{vincent_book} or to Section \ref{su2_cvpc} for more details about this procedure, and just state the resulting proposition.

\begin{proposition}
There exists a constant $K > 0$ such that the amplitude of any fully convergent graph $\cG$ is absolutely convergent with respect to $\mu$, and moreover
\beq
\sum_{\mu} | \cA_{\cG , \mu} | \leq K^{L (\cG)} \,.
\eeq
\end{proposition}

We now explain why similarly, when $\cG$ contains melopoles, the sum over $\mu$ in (\ref{classified}) can be performed without cut-off. From the power-counting of Proposition \ref{pc_r}, and given that melopoles have at most two external legs, one notices 
that
\beq
\omega[( g / {A_{\cM}(g)} )_i^{(k)}] \leq - \frac{N (( g / {A_{\cM}(g)} )_i^{(k)}) }{2}\,.
\eeq
So the decay that was proven for convergent graphs (\ref{fund_u1}) generalizes to
\beq
| \cA_{\cG, \cM}^{R} | \leq  K^{L(\cG)}  \sum_{\mu | \cM \in I(\mu)} \prod_{ g \in \cM \cup \{ \cG \}} \prod_{ (i , k) } M^{- \frac{N ( ( g / {A_{\cM}(g)} )_i^{(k)} )}{4} }\,.
\eeq  
The strategy used for proving convergence of fully convergent graphs is therefore applicable. We conclude that $\cA_{\cG , \cM }$ is an absolutely convergent series in $\mu$, and even bounded by $K^{L(\cG)}$ for some constant $K$. The final sum over meloforests is not problematic, as the number of melopoles associated to a bubble $b$ is clearly bounded by a constant (for example $2^{p/2}$). This means that the number of meloforests associated to a vertex $b$ is also bounded by a constant $K_1 > 0$, and since meloforests of graphs are by definition unions of meloforests associated to single vertices, the number of meloforests of $\cG$ is itself bounded by $K_1 ^{n (\cG)}$. Overall, we conclude that:
\begin{proposition}
There exists a constant $K > 0$, such that the renormalized amplitude of any (non-vacuum) graph $\cG$ verifies:
\beq
| \cA_{\cG}^{R} | \leq K^{L( \cG )} \,.
\eeq
\end{proposition}
This not only proves renormalizability of the model, but also that there is no renormalon effect. The latter is a specific feature of our super-renormalizable model, that would not hold for more complicated just-renormalizable models, as will be confirmed in the next chapter. In such situations, it will be preferable to resort to the \textit{effective series}, because it is the unphysical sum over inoffensive forests automatically generated in the renormalized series that is responsible for this undesirable effect (see \cite{vincent_book}). 

\
We finally state the main theorem of this chapter.
\begin{theorem}
The melordered $U(1)$ model in $d = 4$, with an arbitrary finite set of $4$-bubble interactions, is perturbatively finite at any order.
\end{theorem}

\section{Example: Wick-ordering of a $\vphi^6$ interaction}\label{sec:example}

To illustrate the results of this chapter, we provide here some more details about a particular simple model, with a single $\vphi^6$ interaction:
\bes
S(\vphi , \vphib) &=& \int [\extd g_i]^{12} \vphi( g_{1} , g_{2} , g_{3} , g_{4}) \vphib( g_{1} , g_{2} , g_{3} , g_{5}) \vphi( g_{8} , g_{7} , g_{6} , g_{5}) \\ 
&& \vphib( g_{8} , g_{9} , g_{10} , g_{11}) \vphi( g_{12} , g_{9} , g_{10} , g_{11}) \vphib( g_{12} , g_{7} , g_{6} , g_{4})\,.
\ees

It is an invariant, represented by the $4$-colored graph of Figure \ref{int6_ab}.
It is moreover melonic, and its external legs can be paired so as to form the vacuum melopole shown in Figure \ref{melop6}. This melopole strictly contains four
non-empty melopoles: $S_1 = \{ l_1\}$, $S_3 = \{ l_3\}$, $S_{12} = \{ l_1 , l_2\}$, $S_{23} = \{ l_2 , l_3 \}$. On the other hand, $\{ l_1 , l_3 \}$ and $\{ l_2 \}$ are not melopoles. 

\begin{figure}[h]
  \centering
  \subfloat[Interaction]{\label{int6_ab}\includegraphics[scale=0.5]{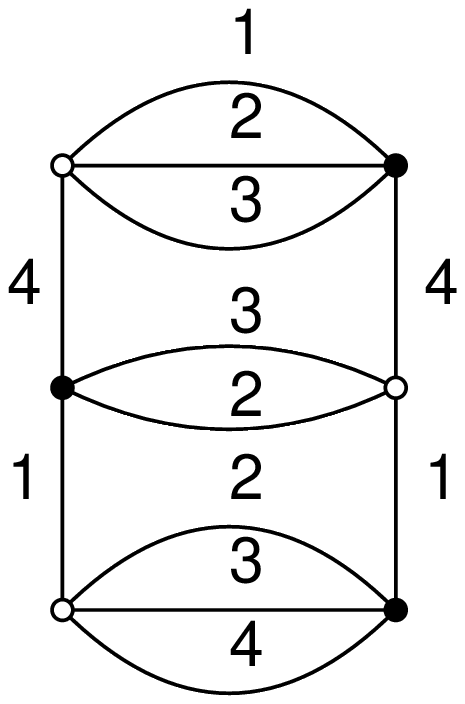}}                
  \subfloat[Vacuum melopole]
{\label{melop6}\includegraphics[scale=0.5]{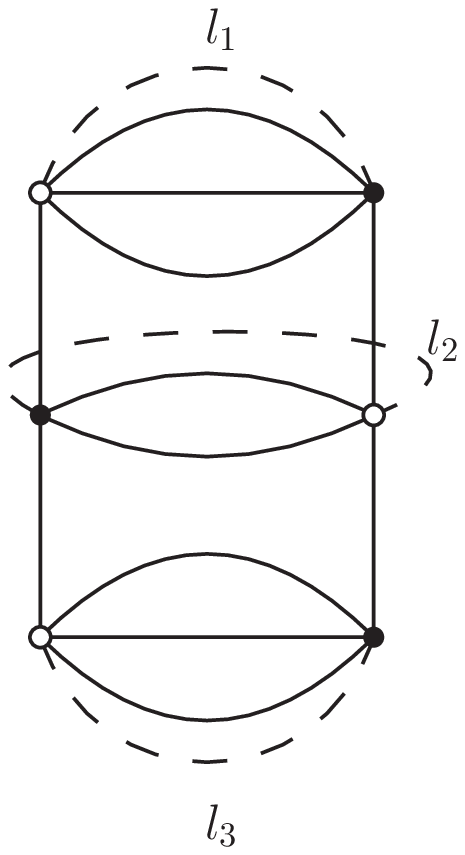}}
  \caption{Bubble interaction with color labels, and unique vacuum melopole that can be obtained from it.}
\end{figure}

We can construct $16$ meloforests out of these melopoles. Half of them, hence 8 do not contain the full graph $S_{123}$.
They are listed below according to the number of subgraphs:
\begin{itemize}
\item the empty forest $\emptyset$;
\item 4 forests with $1$ subgraph: $\{ S_{1} \}$, $\{ S_{3} \}$, $\{ S_{12} \}$, $\{ S_{23} \}$;
\item 3 forests with $2$ subgraphs: $\{ S_{1} , S_{12} \}$, $\{ S_{3} , S_{23} \}$, $\{ S_{1} , S_{3} \}$.
\end{itemize} 
The other half is simply obtained by adding $S_{123}$ to all of these forests. 

\
The melordering generates three kinds of counter-terms: vacuum terms, $2$-point function terms, and two types of $4$-point function terms. We call $b_2$ the $2$-point effective bubble, $b_{4,1}$ and $b_{4,4}$ the two $4$-point effective bubbles, as shown in Figure \ref{eff6}. The melordered interaction will take the form
\beq
\Omega_\rho (S) = S + t_{4,1} (\rho) \, b_{4,1} + t_{4,4} (\rho) \, b_{4,4} + t_{2} (\rho) \, b_{2} + t_\emptyset (\rho) \,,
\eeq
where $t$ are sums of products of coefficients $\nu$. To determine them, we need to analyze the contraction operators they correspond to.

\begin{figure}
\begin{center}
\includegraphics[scale=0.5]{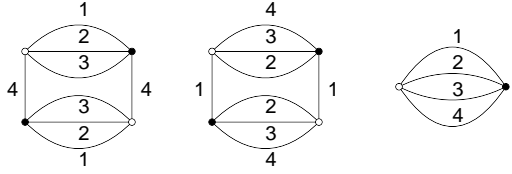}
\caption{Effective interactions generated by melordering. From left to right: $b_{4,4}$, $b_{4,1}$ and $b_2$.}
\label{eff6}
\end{center}
\end{figure}

\
The $4$-point interaction terms are simple, since they are generated by forests $\{ S_1 \}$ and $\{ S_3 \}$. $t_{4,1} (\rho)$ and $t_{4,4} (\rho)$ are therefore both given by the evaluation of $\nu_\rho$ on a single-line melopole (noted $\nu_\rho (1)$) 
\beq
t_{4,1} (\rho) = t_{4,4} (\rho) = - \nu_\rho (1) = - \int_{M^{-2 \rho}}^{+ \infty} \extd \alpha e^{-\alpha m^2} \int \extd \lambda \left( K_\alpha (\lambda) \right)^{3}\,,
\eeq
which is proportional to $\rho$ in the large $\rho$ limit. As expected, these are log-divergent terms.

\
The $2$-point interaction is generated by $\{ S_{12} \}$, $\{ S_{23} \}$, $\{ S_{1} , S_{12} \}$, $\{ S_{3} , S_{23} \}$ and
$\{ S_{1} , S_{3} \}$. $\{ S_{12} \}$ and $\{ S_{23} \}$ contribute with a minus sign, and with an absolute value given by the evaluation of a two-line melopole, that is
\beq
- \nu_\rho (2) = \int_{M^{-2 \rho}}^{+ \infty} \extd \alpha_1 \extd \alpha_2 e^{-(\alpha_1 + \alpha_2 ) m^2} 
\int \extd \lambda_1 \int \extd \lambda_2 \left( K_{\alpha_1} (\lambda_1) \right)^{2} \left( K_{\alpha_2} (\lambda_2) \right)^{3} K_{\alpha_1 + \alpha_2} (\lambda_1 + \lambda_2)
\eeq
each. The three other terms come with a plus sign, and factorize as the square of a single line melopole. Therefore:
\beq
t_2 (\rho) = - 2 \nu_\rho (2) + 3 ( \nu_\rho (1) )^2 \,.
\eeq

\
All the other forests contribute to the vacuum counter-term. There are eight of them. It is then easy to see that:
\beq
t_\emptyset (\rho) = - \mu_\rho (3) + 2 \nu_\rho (1) \mu_\rho (2)  + 2 \nu_\rho (2) \mu_\rho (1)  - 3 ( \nu_\rho (1) )^2 \mu_\rho (1)  \,.
\eeq
where the $\nu_\rho$ are the logarithmically divergent previous integrals and the 
$\mu_\rho(1,2,3)$ are full vacuum melopoles amplitudes (with respectively 1 2 and 3 lines), each diverging 
linearly in $M^{ \rho}$. One can check that the integral over the Gaussian measure
of the full melordered combination is then finite as all divergent contributions cancel out.     

\chapter{Just-renormalizable $\SU(2)$ model in three dimensions}\label{chap:su2}


In this chapter, we focus on the $\vphi^6$ rank-$3$ model based on the group $\SU(2)$, of type A in Table \ref{theories}. A detailed proof of its renormalizability, based on \cite{su2}, is provided. We then present preliminary calculations of the renormalization group flow \cite{beta_su2}.   

\section{The model and its divergences}
\label{su2_model}

\subsection{Regularization and counter-terms}

From now on, $G = \SU(2)$ and $K_\alpha$ is the corresponding heat kernel at time $\alpha$, which we recall explicitly writes
\beq
K_{\alpha} = \sum_{j \in \mathbb{N}/2} (2 j + 1)\e^{- \alpha j (j+1) } \chi_{j}
\eeq 
in terms of the characters $\chi_j$. The cut-off covariance is given by
\beq \label{paracut}
C^{\Lambda}(g_1, g_2 , g_3 ; g_1' , g_2' , g_3' )
\equiv \int_{\Lambda}^{+ \infty} \extd \alpha \, \e^{- \alpha m_{phys,i}^2} \int \extd h \prod_{\ell = 1}^{3} K_{\alpha} (g_\ell h g_\ell'^{\inv})\,,
\eeq
defined for any $\Lambda > 0$. This allows to define a UV regularized theory, with partition function
\beq
\cZ_{\Lambda} = \int \extd \mu_{C^\Lambda} (\vphi , \vphib) \, \e^{- S_{\Lambda} (\vphi , \vphib )}\,.
\eeq
According to our analysis of the Abelian divergence degree, $S_{\Lambda}$ can contain only up to $\vphi^6$ $3$-bubbles. This gives exactly $5$ possible patterns of contractions (up to color permutations): one $\vphi^2$ interaction, one $\vphi^{4}$ interaction, and three $\vphi^{6}$ interactions. They are represented in Figure \ref{int}. 

Among the three types of interactions of order $6$, only the first two can constitute melonic subgraphs. Indeed, an interaction of the type $(6,3)$ cannot be part of a melonic subgraph, therefore cannot give any contribution to the renormalization of coupling constants. Reciprocally, the contraction of a melonic subgraph in a graph built from vertices of the type $(2)$, $(4)$, $(6,1)$ and $(6,2)$ cannot create an effective $(6,3)$-vertex. This is due to the fact that a $(6,3)$-bubble  is dual to the triangulation of a torus, while the other four interactions represent spheres, and the topology of $3$-bubbles is conserved under contraction of melonic subgraphs \cite{FerriGagliardi, Vince_gene}. 

Therefore, we can and we shall exclude interactions of the type $(6,3)$ from $S_{\Lambda}$ from now on. This is a very nice feature of the model, for essentially two reasons. First, from a discrete geometric perspective, $(6,3)$ interactions would introduce topological singularities that would be difficult to interpret in a quantum gravity context, so it is good that they are not needed for renormalization. Second, contrary to the other interactions, they are not positive and could therefore induce non-perturbative quantum instabilities.
\begin{figure}[h]
\begin{center}
\includegraphics[scale=0.5]{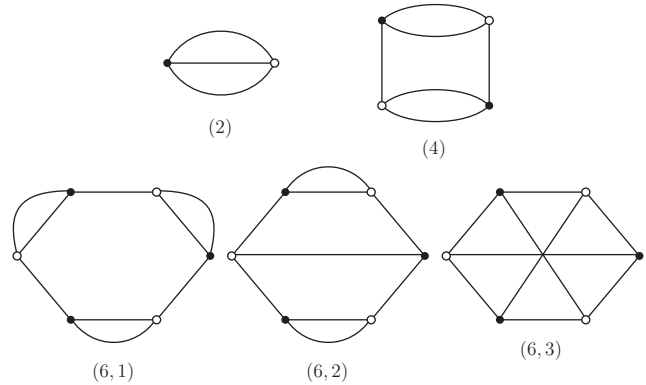}
\caption{Possible $d$-bubble interactions.}
\label{int}
\end{center}
\end{figure}
\
The $2$-point interaction is identical to a mass term, and will therefore be used to implement the mass renormalization counter-terms. Since the model will also generate quadratically divergent $2$-point functions, we also need to include wave function counter-terms in $S_\Lambda$. Finally, we require color permutation invariance of the $4$- and $6$-point interactions. All in all, this gives
\beq\label{drawing_sym}
S_\Lambda  = \frac{t_4^{\Lambda}}{2} S_{4} + \frac{t_{6,1}^{\Lambda}}{3} S_{6,1} + t_{6,2}^{\Lambda} S_{6,2} + CT_{m}^{\Lambda} S_{m} + CT_{\vphi}^{\Lambda} S_{\vphi}\,,
\eeq
where:
\bes
S_{4} (\vphi , \vphib) &=& \int [\extd g ]^6 \, \vphi(g_1 , g_2 , g_3 ) \vphib(g_1 , g_2 , g_4 ) \vphi(g_5 , g_6 , g_3 ) \vphib(g_5 , g_6 , g_4 ) \nn  \\
&& + \; {\rm color} \; {\rm permutations}  \,, \label{color_sym1} \\
S_{6,1} (\vphi , \vphib) &=&  \int [\extd g ]^9 \, \vphi(g_1 , g_2 , g_7 ) \vphib(g_1 , g_2 , g_9 ) \vphi(g_3 , g_4 , g_9 ) \vphib(g_3 , g_4 , g_8 ) \vphi(g_5 , g_6 , g_8 ) \vphib(g_5 , g_6 , g_7 ) \nn \\
&& + \; {\rm color} \; {\rm permutations}  \,, \label{color_sym2} \\
S_{6, 2} (\vphi , \vphib) &=& \int [\extd g ]^9 \, \vphi(g_1 , g_2 , g_3 ) \vphib(g_1 , g_2 , g_4 ) \vphi(g_8 , g_9 , g_4 ) \vphib(g_7 , g_9 , g_3 ) \vphi(g_7 , g_5 , g_6 ) \vphib(g_8 , g_5 , g_6 ) \nn \\
&& + \; {\rm color} \; {\rm permutations} \,,  \label{color_sym3} \\
S_{m} (\vphi , \vphib) &=& \int [\extd g ]^3 \, \vphi(g_1 , g_2 , g_3 ) \vphib(g_1 , g_2 , g_3 )\,, \label{color_sym4}\\
S_{\vphi} (\vphi , \vphib) &=& \int [\extd g ]^3 \, \vphi(g_1 , g_2 , g_3 ) \left( - \sum_{l = 1}^{3} \Delta_\ell \right) \vphib(g_1 , g_2 , g_3 )\,. \label{color_sym5}
\ees

Two types of symmetries have to be kept in mind. 
In equations (\ref{color_sym1}-\ref{color_sym5}), we just averaged over color permutations. This gives a priori $6$ terms for each bubble type, but some of them are identical. It turns out that for each type of interaction, we have exactly $3$ distinct bubbles. Similarly, $S_\vphi$ is a sum of three term, which we can consider as new bubbles. With the mass term, we therefore have a total number of $13$ different bubbles in the theory. From now on, $\cB$ has to be understood in this extended sense. We could as well work with independent couplings for each bubble $b \in \cB$, but we decide to consider the symmetric model only, which seems to us the most relevant situation. However, it is convenient to work with notations adapted to the more general situations, because this allows to write most of the equations in a more condensed fashion. In the following, we will work with coupling constants $t_b^\Lambda$ for any $b \in \cB$, which has to be understood as $t_{4}^{\Lambda}$, $t_{6,1}^{\Lambda}$, $t_{6,2}^{\Lambda}$, $CT_m^{\Lambda}$ or $CT_\vphi^{\Lambda}$ depending on the nature of $b$.

\

In (\ref{drawing_sym}), we divided each coupling constant by a certain number of permutations of labels on the external legs of a bubble associated to this coupling. More precisely, it is the order of the subgroup of the permutations of these labels leaving the labeled colored graph invariant. Note that a first look at $(6,2)$ interactions suggests an order $2$ symmetry, but it is incompatible with any coloring. The role of such rescalings of the coupling constants is, as usual, to make the symmetry factors appearing in the perturbative expansions more transparent. The symmetry factor $s(\cG)$ associated to a Feynman graph $\cG$ becomes the number of its automorphisms. All these conventions will be useful when discussing in details how divergences can be absorbed into new effective coupling constants.

\

Finally, the reader might wonder whether it is appropriate to include the $2$-point function counter-terms in the interaction part of the action, rather than associating flowing parameters to the covariance itself. This question is particularly pressing for wave-function counter-terms, since they break the tensorial invariance of the interaction action. One might worry that the degenerate nature of the covariance could prevent a Laplacian interaction with no projector from being reabsorbed in a modification of the wave-function parameter of the covariance. However, it is not difficult to understand that the situation is identical to that of a non-degenerate covariance. At fixed cut-off, modifying the covariance is not exactly the same as adding $2$-point function counter-terms in the action, but the two prescriptions coincide in the $\Lambda \to 0$ limit, as will be commented on in Section \ref{sec:rg_flow}. Thus, it is perfectly safe to work in the second setting. Moreover, this has the main advantage of being compatible with a fixed slicing of the covariance according to scales, which is the central technical tool of the work presented in this thesis. 

\subsection{List of divergent subgraphs}

From the previous sections, and as we will confirm later on, the Abelian divergence degree of a subgraph $\cH$ will allow to classify the divergences. When $\cH$ does not contain any wave-function counter-terms, one has\footnote{One also assumes $F(\cH) \geq 1$, as in \ref{sec:degree}.}:
\beq\label{def_omega1}
\omega (\cH) = 3 - \frac{N}{2} - 2 n_2 - n_4 + 3 \rho(\cH / \cT) \,.
\eeq
We will moreover see in the next section that wave-function counter-terms are neutral with respect to power-counting arguments. We can therefore extend the definition (\ref{def_omega1}) of $\omega$ to arbitrary subgraphs if $n_2$ is understood as the number of $2$-valent bubbles 
\textit{which are not of the wave-function counter-term type}, and the contraction of a tree is also understood in a general sense: $\cH / \cT$ is the subgraph obtained by \textit{first collapsing all chains of wave-function counter-terms}, and then contracting a tree $\cT$ in the collapsed graph. Alternatively, $\omega$ takes the generalized form:
\beq
\omega(\cH) = -2 ( L - W ) + 3 ( F - R ) \,,
\eeq
where $W$ is the number of wave-function counter-terms in $\cH$. \textit{This formula holds also when $F(\cH) = 0$}.

\
   
Let us focus on non-vacuum connected subgraphs with $F \geq 1$, which are the physically relevant ones. In this case $\rho = 0$ for melonic subgraphs and $\rho \leq -1$ otherwise. Therefore 
$$\omega (\cH) \leq - \frac{N}{2}$$
if $\cH$ is not melonic. As a result, divergences are entirely due to melonic subgraphs. They are in particular tracial, which means their Abelian power-counting is optimal. We therefore obtain an exact classification of divergent subgraphs, provided in table \ref{div}. It tells us that $6$-point functions have logarithmic divergences, $4$-point functions linear divergences as well as possible logarithmic ones, that will have to be absorbed in the constants $t_{4}^{\Lambda}$, $t_{6,1}^{\Lambda}$ and $t_{6,2}^{\Lambda}$. The full $2$-point function will be quadratically divergent, generating the constants $CT_{m}^{\Lambda}$ and $CT_{\vphi}^{\Lambda}$.

\begin{table}[h]
\centering
\begin{tabular}{| l | c | c | c || r |}
    \hline
    $N$ & $n_2$ & $n_4$ & $\rho$ & $\omega$  \\ \hline\hline
 6 & 0 & 0 & 0 & 0 \\ \hline
 4 & 0 & 0 & 0 & 1 \\ 
 4 & 0 & 1 & 0 & 0 \\ \hline
 2 & 0 & 0 & 0 & 2 \\ 
 2 & 0 & 1 & 0 & 1 \\
 2 & 0 & 2 & 0 & 0 \\
 2 & 1 & 0 & 0 & 0 \\ 
    \hline
  \end{tabular}
\caption{Classification of non-vacuum divergent graphs for $d=D=3$. All of them are melonic.}
\label{div}
\end{table}

{\bf Remark.}
There are a lot more cases to consider for vacuum divergences, including non-melonic contributions. However, they are irrelevant 
to perturbative renormalization, and since the previous chapter provides an example of how these can be analyzed, we discard them here.

\
In light of Corollary \ref{faces_legs}, we also notice that $2$-point divergent subgraphs, hence all degree $2$ subgraphs, have a single external face. This is a useful point to keep in mind as far as wave-function renormalization is concerned. As for $4$- and $6$-point divergent subgraphs, they have at most $2$ and $3$ external faces respectively. It is also not difficult to find examples saturating these two bounds, as shown in Figures \ref{div_4} and \ref{div_6}. 

\begin{figure}[h]
  \centering
  \subfloat[$\omega = 1$]{\label{div_4}\includegraphics[scale=0.6]{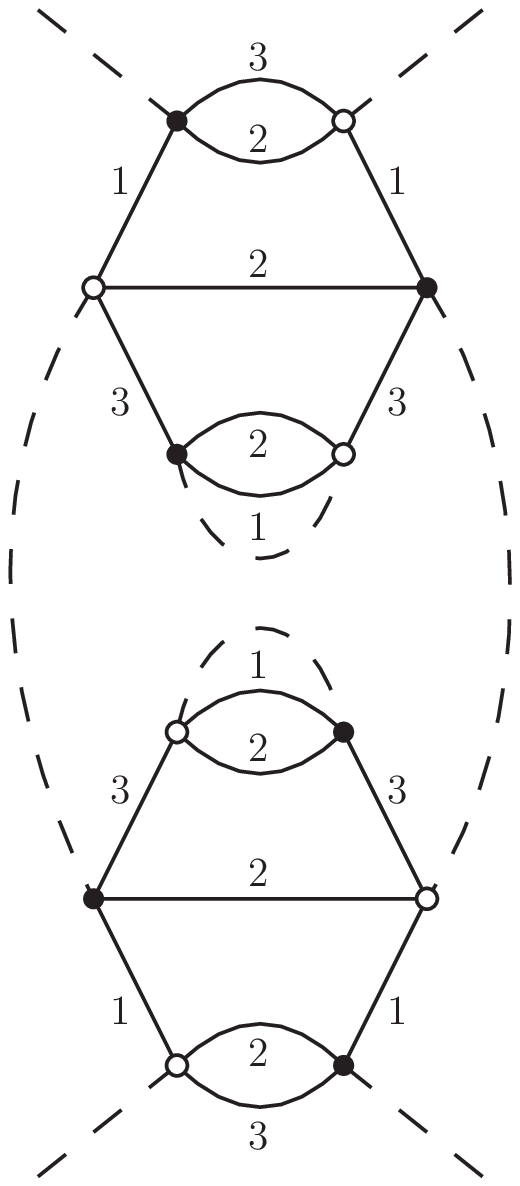}}   	       
  \subfloat[$\omega = 0$]{\label{div_6}\includegraphics[scale=0.6]{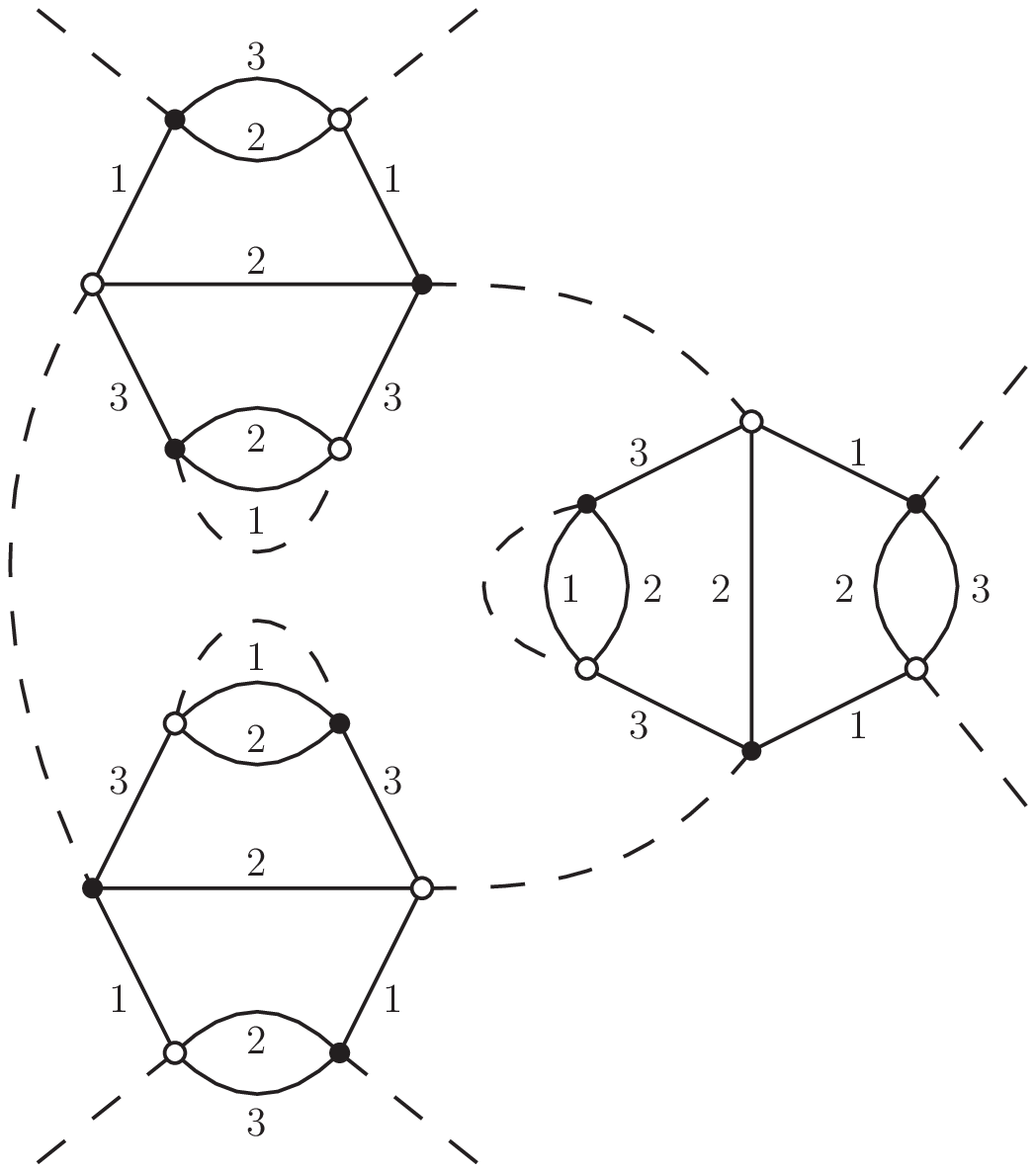}}
  \caption{Divergent subgraphs with respectively $2$ and $3$ external faces.}
\end{figure}

\section{Non-Abelian multiscale expansion}

In this section, we extend the multiscale tools already encountered to the present non-commutative and just-renormalizable model. This includes a proof of the Abelian bound (\ref{fund_ab}), and a definition of suitable localization operators.
We again set the cut-off to $\Lambda = M^{- 2 \rho}$ and slice the propagators in such a way as to express the amplitudes as in equation (\ref{decomposi}).

\subsection{Power-counting theorem}

The results we need to extend to our $\SU(2)$ context are the propagator bounds, which themselves rely on peakedness properties of the heat kernel. Let us denote $\vert X \vert$ the norm of a Lie algebra element $X \in \su(2)$, and $|g|$ the geodesic distance between a Lie group element $g \in \SU(2)$ and the identity $\one$. We can prove the following bounds on $K_\alpha$ and its Lie derivatives\footnote{We define the Lie derivative of a function $f$ as:
\beq
\cL_X f(g) \equiv \frac{\extd}{\extd t} f(g \e^{t X}) |_{t = 0} \,.
\eeq}.
\begin{lemma}\label{heat}
There exists a set of constants $\delta > 0$ and $K_n > 0$ , such that for any $n \in \mathbb{N}$ the following holds:
\beq
\forall \alpha \in \left] 0 , 1 \right] , \quad \forall g \in \SU(2), \quad \forall X \in \su(2),\, |X| = 1, \qquad
\vert ( \cL_X )^n K_\alpha (g) \vert \leq K_n \alpha^{-\frac{3 + n}{2}} \e^{- \delta \frac{|g|^2}{\alpha}} 
\eeq
\end{lemma}
\begin{proof}
See the Appendix.
\end{proof}

As a consequence, the divergences associated to the propagators and their derivatives can be captured in the following bounds.
 \begin{proposition}
 There exist constants $K > 0$ and $\delta > 0$, such that for all $i \in \mathbb{N}$:
 \beq\label{propa_bound} 
 C_i (g_1, g_2 , g_3 ; g_1' , g_2' , g_3' ) \leq K M^{7 i } \int \extd h \,
 \e^{- \delta M^{i} \sum_{\ell =1}^{3} |g_\ell h g_\ell'^{\inv}| }\,\,.
 \eeq
 Moreover, for any integer $k \geq 1$, there exists a constant $K_k$, such that for any $i \in \mathbb{N}$, any choices of colors $\ell_p$ and Lie algebra elements $X_p \in \su(2)$ of unit norms ($1 \leq p \leq k$):
 \beq\label{deriv_bound}
\left( \prod_{p = 1}^{k} \cL_{X_{p} , g_{\ell_p}} \right) C_i (g_1, g_2 , g_3 ; g_1' , g_2' , g_3' ) \leq K M^{ (7 + k) i } \int \extd h \,
 \e^{- \delta M^{i} \sum_{\ell =1}^{3} |g_\ell h g_\ell'^{\inv}| }\,,
\eeq
where $\cL_{X_p , g_{\ell_p}}$ is the Lie derivative with respect to the variable $g_{\ell_p}$ in direction $X_p$.
\end{proposition}
\begin{proof}
For $i \geq 1$, the previous lemma immediately shows that:
\bes
C_i (g_1, g_2 , g_3 ; g_1' , g_2' , g_3' ) &\leq& K_1 \int_{M^{- 2 i}}^{M^{- 2 (i- 1)}} \int \extd h \, \frac{\e^{- \frac{\delta_1}{\alpha} \sum_{\ell =1}^{3} |g_\ell h g_\ell'^{\inv}|^2 }}{\alpha^{9/2}} \\
&\leq& K_1 M^{- 2 (i- 1)} (M^{2 i})^{9/2} \int \extd h \, \e^{- \delta_1 M^{- 2 i} \sum_{\ell =1}^{3} |g_\ell h g_\ell'^{\inv}|^2 } \\
&\leq & K M^{7 i } \int \extd h \, \e^{- \delta M^{i} \sum_{\ell =1}^{3} |g_\ell h g_\ell'^{\inv}| }\,,
\ees
for some strictly positive constants $K_1$, $\delta_1$, $K$ and $\delta$. And similarly for Lie derivatives of $C_i$.

\
When $i = 0$, equations (\ref{nice_form}), (\ref{bound_F}) and (\ref{ext_bound}), together with the fact that $m \neq 0$ allow to bound the integrand of $C_0$ by an integrable function of $\alpha \in [ 1 , + \infty [$. $C_0$ is therefore bounded from above by a constant, and due to the compact nature of $\SU(2)$ we can immediately deduce a bound of the form
\beq   
C_0 (g_1, g_2 , g_3 ; g_1' , g_2' , g_3' ) \leq K \int \extd h \,
 \e^{- \delta \sum_{\ell =1}^{3} |g_\ell h g_\ell'^{\inv}| }\,.
\eeq
Again, the same idea applies to the bound on the Lie derivatives of $C_0$, which concludes the proof.
\end{proof}

We can now extend the multiscale power-counting of Proposition \ref{prop:abpc} to our non-Abelian model.

\begin{proposition}
There exists a constant $K > 0$, such that for any connected graph $\cG$ with scale attribution $\mu$, the following bound holds: 
\beq\label{fund}
 \vert \cA_{\cG, \mu}  \vert   \leq K^{L(\cG)}  \prod_{i \in \mathbb{N}} \prod_{ k \in \llbracket 1 , k(i) \rrbracket } M^{\omega [  \cG_{i}^{(k)}]}\,,
\eeq
where $\omega$ is the Abelian degree of divergence
\beq
\omega(\cH) = - 2 ( L(\cH) - W(\cH) ) + 3 ( F(\cH) - R(\cH) ) \,.
\eeq
\end{proposition}
\begin{proof}
Let us first assume $W(\cG) = 0$. In this case, we follow and adapt the proof of Abelian power-counting of Proposition \ref{prop:abpc}. We first integrate the $g$ variables in an optimal way, as was done with the $\theta$ variables for Abelian models. In each face $f$, a maximal tree of lines $T_f$ is chosen to perform the $g$ integrations. Optimality is ensured by requiring the trees $T_f$ to be compatible with the abstract GN tree. This yields:
\bes\label{step1}
\vert \cA_{\cG, \mu} \vert &\leq& K^{L(\cG)}  \prod_{i \in \mathbb{N}} \prod_{ k \in \llbracket 1 , k(i) \rrbracket } M^{- 2 L( \cG_{i}^{(k)} ) + 3 F( \cG_{i}^{(k)} )} \\
&& \times \int [\extd h]^{L(\cG)} \prod_{f} \e^{- \delta M^{i(f)} \vert \overrightarrow{\prod_e} h_e^{\epsilon_{ef}} \vert}\,,
\ees
where $i(f) = \min \{ i_e \vert e \in f \}$.

\
The main difference with Proposition \ref{prop:abpc} is that now the holonomy variables are non-commuting, which prevents us from easily integrating them out. We can however rely on the methods developed in \cite{vm1, vm2, vm3}, which provide an exact power-counting theorem for $BF$ spin foam models. In particular, one can show that for any $2$-complex with $E$ edges and $F$ faces, the expression
\beq
\int [d g_e ]^E \exp \left( - \Lambda \sum_f | \prod_{e \in f} g_e^{\epsilon_{ef}}| \right)
\eeq
scales as $\Lambda^{- \rm{rk} \, \delta^{1}_{\phi}}$ when $\Lambda \to 0$. $\delta^{1}_{\phi}$ is the twisted boundary map associated to a (non-singular) flat connection $\phi$\footnote{The explicit construction of this map can be found in \cite{vm3}. With the notations of the present thesis, it is defined as 
\bes
\delta^{1}_\phi : \quad E \otimes su(2) &\rightarrow& F \otimes su(2) \\
												 e \otimes X &\mapsto & \sum_f \epsilon_{ef} f \otimes Ad_{P \phi (v_e, v_f)} (X)
\ees
where $P \phi (v_e, v_f)$ is a path from a reference vertex $v_e$ in the edge $e$ to a reference vertex $v_f$ in the face $f$. The adjoint action encodes parallel transport with respect to $\phi$, and is full rank.}, which takes the non-commutativity of the group into account. Remarkably, this boundary map verifies:
\beq
{\rm{rk}} \, \delta^1_\phi \geq 3 {\rm{rk}} \, \epsilon_{ef}\,.
\eeq
As a result, the contribution of the closed faces of a $\cG_i^{(k)}$ can be bounded by
\beq
\int [\extd g_e ]^{L(\cG_i^{(k)})} \exp \left( - M^{i(f)} \underset{f \in F(\cG_i^{(k)})}{\sum} | \prod_{e \in f} g_e^{\epsilon_{ef}}| \right) \leq K_1^{L(\cG_i^{(k)})} M^{- 3 R(\cG_i^{(k)}) i}\,.
\eeq
The power-counting (\ref{fund}) is recovered by recursively applying this bound, from the leaves to the root of the GN tree. 

\

The $W(\cG) \neq 0$ case is an immediate consequence of the $W(\cG) = 0$ one. Indeed, one just needs to understand how the insertion of a wave-function counter-term in a graph $\cG$ affects its amplitude $\cA_\cG$. While it adds one line to $\cG$, it does not change its number of faces, nor their connectivity structure, hence the rank $R$ is not modified either. The line being created is responsible for an additional $M^{-2 i}$ factor in the power-counting, with $i$ its scale. On the other hand, it is acted upon by a Laplace operator, that is two derivatives, which according to (\ref{deriv_bound}) generate an extra $M^{2 i}$. The two contributions cancel out, which shows that wave-function counter-terms are neutral to power-counting. The $L$ contribution to $\omega$ has therefore to be compensated by a $W$ term with the opposite sign.
\end{proof}

Notice that all the steps in the derivation of the bound are optimal, except for the last integrations of face contributions. In this last step we discarded the fine effects of the non-commutative nature of $\SU(2)$, encoded in the rank of $\delta^{1}_{\phi}$. Remark however that no such effect is present for a contractible $\cG_{i}^{(k)}$, since the $2$-complex formed by its internal faces is simply connected \cite{vm3}. Indeed, such a subgraph supports a unique flat connection (the trivial one), which means that the integrand in equation (\ref{step1}) can be linearized around $h_e = \one$, showing the equivalence between Abelian and non-Abelian power-countings in this case.   
Since melonic subgraphs are contractible, this confirms our previous claim: the Abelian power-counting exactly captures the divergences of the model presented in this chapter.

\subsection{Contraction of high melonic subgraphs}\label{sec:contraction}

We close this section with a discussion of the key ingredients entering the renormalization of this model, by explaining how local approximations to high melonic subgraphs are extracted from high slices to lower slices of the amplitudes. A full account of the renormalization procedure, including rigorous finiteness results, will be detailed in the next and final section. 
 
\

\begin{figure}[h]
\begin{center}
\includegraphics[scale=0.6]{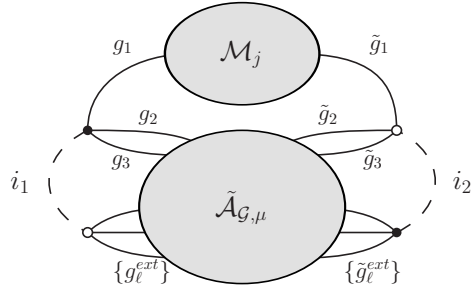}
\caption{A graph with high melonic subgraph $\cM_j$.}
\label{contract_melon}
\end{center}
\end{figure}
Let us consider a non-vacuum graph with scale attribution $( \cG , \mu )$, containing a melonic high subgraph $\cM_j \subset \cG$ at scale $j$. For the convenience of the reader, we first focus on the case $F_{ext}(\cM_j) = 1$, which encompasses all the $2$-point divergent subgraphs, therefore all the degree $2$ subgraphs. We also first assume that no wave-function counter-term is present in $\cM_j$.

\subsubsection{Divergent subgraphs with a single external face, and no wave-function counter-terms}

Since $F_{ext} (\cM_j) = 1$, $N_{ext} ( \cM_j )$ contains two external propagators, labeled by external variables $\{ g_\ell^{ext} \,, \ell = 1 , \ldots , 3 \}$ and $\{ \tilde{g}_\ell^{ext} \,, \ell = 1 , \ldots , 3 \}$, and scales $i_1 < j$ and $i_2 < j$ respectively. We can assume (without loss of generality) that the melonic subgraph $\cM_j$ is inserted on a color line of color $\ell = 1$. The amplitude of $\cG$, pictured in Figure \ref{contract_melon}, takes the form:
\bes
\cA_{\cG , \mu} &=& \int [\extd g_\ell ]^3 [\extd \tilde{g}_\ell ]^3 [\extd g_\ell^{ext} ]^3 [\extd \tilde{g}_\ell^{ext} ]^3 \widetilde{\cA}_{\cG , \mu}  ( g_2 , g_3 ; \tilde{g}_2 , \tilde{g}_3 ; \{ g_\ell^{ext} \} ; \{ \tilde{g}_\ell^{ext} \} )  \\
&& \times C_{i_1} ( g_1^{ext} , g_2^{ext} , g_3^{ext} ; g_1 , g_2 , g_3 ) \cM_j ( g_1 , \tilde{g}_1 ) C_{i_2} ( \tilde{g}_1 , \tilde{g}_2 , \tilde{g}_3 ; g_1^{ext} , g_2^{ext} , g_3^{ext} ) \, . \nn
\ees
The idea is then to approximate the value of $\cA_{\cG , \mu}$ by an amplitude associated to the contracted graph $\cG / \cM_j$. This can be realized by "moving" one of the two external propagators towards the other. In practice, we can use the interpolation\footnote{$X_g$ denotes the Lie algebra element with the smallest norm such that $\e^{X_g} = g$.}
\bes
g_1 (t) = \tilde{g}_1 e^{t X_{\tilde{g}_1^{\inv} g_1}} \,, \qquad t \in [0 , 1]\, 
\ees 
and define:
\bes\label{sep_scales}
\cA_{\cG , \mu} (t) &=& \int [\extd g_\ell ]^3 [\extd \tilde{g}_\ell ]^3 [\extd g_\ell^{ext} ]^3 [\extd \tilde{g}_\ell^{ext} ]^3 \widetilde{\cA}_{\cG , \mu}  ( g_2, g_3 ; \tilde{g}_2 , \tilde{g}_3 ; \{ g_\ell^{ext} \} ; \{ \tilde{g}_\ell^{ext} \} ) \\
&& \times C_{i_1} ( g_1^{ext} , g_2^{ext} , g_3^{ext} ; g_1 (t) , g_2 , g_3 ) \cM_j ( g_1 , \tilde{g}_1 ) C_{i_2} ( \tilde{g}_1 , \tilde{g}_2 , \tilde{g}_3 , \tilde{g}_1^{ext} , \tilde{g}_2^{ext} , \tilde{g}_3^{ext} ) \,. \nn 
\ees
This formula together with a Taylor expansion allows to approximate $\cA_\cG = \cA_\cG (1)$ by $\cA_\cG (0)$ and its derivatives. The order at which the approximation should be pushed is determined by the degree of divergence $\omega(\cM_j )$ of $\cM_j$ and the power-counting theorem: we should use the lowest order ensuring that the remainder in the Taylor expansion has a convergent power-counting. Roughly speaking each derivative in $t$ decreases the degree of divergence by $1$, therefore the Taylor expansion needs to be performed up to order $\omega(\cM_j )$:
\beq\label{taylor}
\cA_{\cG , \mu} = \cA_{\cG , \mu} (1) = \cA_{\cG , \mu} (0) + \sum_{k = 1}^{\omega(\cM_j )} \frac{1}{k !} \cA_{\cG , \mu}^{(k)} (0) + \int_{0}^{1} \extd t \, \frac{(1 - t)^{\omega(\cM_j )}}{\omega(\cM_j ) !} \cA_{\cG , \mu}^{(\omega(\cM_j ) + 1)} (t) .
\eeq 

\

Before analyzing further the form of each of these terms, we point out a few interesting properties verified by the function $\cM_j$. First, since by definition the variables $g_1$ and $\tilde{g}_1$ are boundary variables for a same face (and because the heat kernel is a central function), $\cM_j( g_1 , \tilde{g}_1)$ can only depend on $\tilde{g}_1^{\inv} g_1$. From now on, we therefore use the notation:
\beq
\cM_j( g_1 , \tilde{g}_1) = \cM_j ( \tilde{g}_1^{\inv} g_1 )\,.
\eeq
We can then prove the following lemma.
\begin{lemma}
\begin{enumerate}[(i)]
\item $\cM_j$ is invariant under inversion:
\beq
\forall g \in \SU(2) \, , \qquad \cM_j (g^{\inv}) = \cM_j (g) \,. 
\eeq 
\item $\cM_j$ is central:
\beq
\forall g \, , h \in \SU(2) \, , \qquad \cM_j ( h g h^{\inv}) = \cM_j (g) \,.
\eeq
\end{enumerate}
\end{lemma}
\begin{proof}
We can proceed by induction on the number of lines $\widetilde{L}$ of a rosette of $\cM_j$. 

\
When $\widetilde{L} = 1$, $\cM_j$ can be cast as an integral over a single Schwinger parameter $\alpha$ of an integrand of the form:
\beq
\int \extd h \,  K_\alpha ( \tilde{g}_1^{\inv}  g_1 h) \, (K_\alpha ( h ))^2 \,.
\eeq
By invariance of the heat kernels and the Haar measure under inversion and conjugation, the invariance of $\cM_j$ immediately follows. 

\
Suppose now that $\widetilde{L} \geq 2$. A rosette of $\cM_j$ can be thought of as an elementary melon decorated with two melonic insertions of size strictly smaller than $\widetilde{L}$ (at least one of them being non-empty). We therefore have:
\bes
\cM_j ( \tilde{g}_1^{\inv} g_1 ) &=& \int [ \extd g_2 \extd \tilde{g}_2 \extd g_3 \extd \tilde{g}_3 ] \, m^{(2)} ( g_2^{\inv} g_2 ) \, m^{(3)} ( g_3^{\inv} g_3 ) \\
&&\times \int \extd \alpha \int \extd h \, K_\alpha ( \tilde{g}_1^{\inv}  g_1 h) K_\alpha ( \tilde{g}_2^{\inv}  g_2 h ) K_\alpha ( \tilde{g}_3^{\inv}  g_3 h) \,, \nn
\ees
in which we did not specify the integration domain of $\alpha$, since it does not play any role here. $m^{(2)}$ and $m^{(3)}$ are associated to melonic subgraphs of size strictly smaller than $L$, we can therefore assume that they are invariant under conjugation and inversion\footnote{If one $m^{(i)}$ is an empty melon, then $m^{(i)} (g) = \delta(g)$, and is trivially invariant.}. Using again the invariance of the heat kernels and the Haar measure, we immediately conclude that $\cM_j$ itself is invariant. 
   
\end{proof}

\

We now come back to (\ref{taylor}). The degree of divergence being bounded by $2$, it contains terms $\cA_{\cG , \mu}^{(k)} (0)$ with $k \leq 2$. We now show that $\cA_{\cG , \mu} (0)$ gives mass counter-terms, $\cA_{\cG , \mu}^{(1)} (0)$ is identically zero, and $\cA_{\cG , \mu}^{(2)} (0)$ implies wave-function counter-terms. This is stated in the following proposition.

\begin{proposition}
\begin{enumerate}[(i)]
\item $\cA_{\cG , \mu} (0)$ is proportional to the amplitude of the contracted graph $\cG / \cM_j$, with the same scale attribution:
\beq
\cA_{\cG , \mu} (0) = \left( \int \extd g \cM_j (g) \right) \cA_{\cG / \cM_j , \mu} \,.
\eeq
\item Due to the symmetries of $\cM_j$, $\cA_{\cG , \mu}^{(1)} (0)$ vanishes: 
\beq
\cA_{\cG , \mu}^{(1)} (0) = 0\,.
\eeq
\item $\cA_{\cG , \mu}^{(2)}$ is proportional to an amplitude in which a Laplace operator has been inserted in place of $\cM_j$:
\bes\label{wf_ct}
\cA_{\cG , \mu}^{(2)} (0) &=& \left( \frac{1}{3} \int \extd g \cM_j (g) | X_g | ^{2} \right)  \\ 
&& \times \int [\extd g_\ell ]^3 [\extd \tilde{g}_\ell ]^3 \int [\extd g_\ell^{ext} ]^3 [\extd \tilde{g}_\ell^{ext} ]^3 \widetilde{\cA}_{\cG , \mu}  ( g_2 , g_3 ; \tilde{g}_2 , \tilde{g}_3 ; \{ g_\ell^{ext} \} ; \{ \tilde{g}_\ell^{ext} \} ) \nn \\
&& \times \left( \Delta_{\tilde{g}_1} C_{i_1} ( g_1^{ext} , g_2^{ext} , g_3^{ext} ; \tilde{g}_1 , g_2 , g_3 ) \right) C_{i_2} ( \tilde{g}_1 , \tilde{g}_2 , \tilde{g}_3 ; \tilde{g}_1^{ext} , \tilde{g}_2^{ext} , \tilde{g}_3^{ext} )  \,. \nn
\ees
\end{enumerate}
\end{proposition} 
\begin{proof}

\begin{enumerate}[(i)]
\item One immediately has:
\bes
\cA_{\cG , \mu} (0) &=& \int [\extd g_\ell ]^3 [\extd \tilde{g}_\ell ]^3 [\extd g_\ell^{ext} ]^3 [\extd \tilde{g}_\ell^{ext} ]^3 \widetilde{\cA}_{\cG , \mu}  ( g_2 , g_3 ; \tilde{g}_2 , \tilde{g}_3 ; \{ g_\ell^{ext} \} ; \{ \tilde{g}_\ell^{ext} \} )\nn \\
&& \times C_{i_1} ( g_1^{ext} , g_2^{ext} , g_3^{ext} ; \tilde{g}_1 , g_2 , g_3 ) \cM_j ( \tilde{g}_1^{\inv} g_1  ) C_{i_2} ( \tilde{g}_1 , \tilde{g}_2 , \tilde{g}_3 , \tilde{g}_1^{ext} , \tilde{g}_2^{ext} , \tilde{g}_3^{ext} ) \nn \\
&=& \int [\extd g_\ell ]^3 [\extd \tilde{g}_\ell ]^3  [\extd g_\ell^{ext} ]^3 [\extd \tilde{g}_\ell^{ext} ]^3 \widetilde{\cA}_{\cG , \mu}  ( g_2 , g_3 ; \tilde{g}_2 , \tilde{g}_3 ; \{ g_\ell^{ext} \} ; \{ \tilde{g}_\ell^{ext} \} ) \nn \\
&& \times C_{i_1} ( g_1^{ext} , g_2^{ext} , g_3^{ext} ; \tilde{g}_1 , g_2 , g_3 ) \cM_j ( g_1 ) C_{i_2} ( \tilde{g}_1 , \tilde{g}_2 , \tilde{g}_3 , \tilde{g}_1^{ext} , \tilde{g}_2^{ext} , \tilde{g}_3^{ext} ) \nn \\
&=& \left( \int \extd g \cM_j (g) \right) \cA_{\cG / \cM_j , \mu} \,,
\ees
where from the first to the second line we made the change of variable $g_1 \to \tilde{g}_1 g_1$.
\item For $\cA_{\cG , \mu}^{(1)} (0)$, a similar change of variables yields:
\bes
\cA_{\cG , \mu}^{(1)} (0) &=& \int [\extd g_\ell ]^2 [\extd \tilde{g}_\ell ]^3 [\extd g_\ell^{ext} ]^3 [\extd \tilde{g}_\ell^{ext} ]^3 \widetilde{\cA}_{\cG , \mu}  ( g_2 , g_3 ; \tilde{g}_2 , \tilde{g}_3 ; \{ g_\ell^{ext} \}; \{ \tilde{g}_\ell^{ext} \} ) \nn \\
&\times& \left( \int \extd g  \, \cM_j (g) \, \cL_{X_g , \tilde{g}_1} \, C_{i_1} ( g_1^{ext} , g_2^{ext} , g_3^{ext} ; \tilde{g}_1 , g_2 , g_3 )  \right) \nn \\
&\times& C_{i_2} ( \tilde{g}_1 , \tilde{g}_2 , \tilde{g}_3 ; \tilde{g}_1^{ext} , \tilde{g}_2^{ext} , \tilde{g}_3^{ext} ) \,. 
\ees
But by invariance of $\cM_j$ under inversion, one also has
\beq
\int \extd g  \, \cM_j (g) \, \cL_{X_g} = - \int \extd g  \, \cM_j (g) \, \cL_{X_g} \; \Rightarrow \; \int \extd g  \, \cM_j (g) \, \cL_{X_g} = 0\,,
\eeq 
hence $\cA_{\cG , \mu}^{(1)} (0) = 0$. 

\item Finally, $\cA_{\cG , \mu}^{(2)} (0)$ can be expressed as:
\bes
\cA_{\cG , \mu}^{(2)} (0) &=& \int [\extd g_\ell ]^2 [\extd \tilde{g}_\ell ]^3 [\extd g_\ell^{ext} ]^3 [\extd \tilde{g}_\ell^{ext} ]^3 \widetilde{\cA}_{\cG , \mu}  ( g_2 , g_3 ; \tilde{g}_2 , \tilde{g}_3 ; \{ g_\ell^{ext} \} ; \{ \tilde{g}_\ell^{ext} \} ) \nn \\
& \times & \left( \int \extd g  \, \cM_j (g) \, (\cL_{X_g , \tilde{g}_1})^{2} \, C_{i_1} ( g_1^{ext} , g_2^{ext} , g_3^{ext} ; \tilde{g}_1 , g_2 , g_3 )  \right) \nn \\
&\times& C_{i_2} ( \tilde{g}_1 , \tilde{g}_2 , \tilde{g}_3 ; \tilde{g}_1^{ext} , \tilde{g}_2^{ext} , \tilde{g}_3^{ext} ) \,. 
\ees
We can decompose the operator $\int \extd g  \, \cM_j (g) \, (\cL_{X_g})^{2}$ into its diagonal and off-diagonal parts with respect to an orthonormal basis $\{ \tau_k \,, k = 1\,, \ldots , 3\}$ in $\su(2)$. The off-diagonal part writes
\beq
\sum_{k \neq l} \int \extd g \, \cM_j (g ) \, X^{k}_g X^{l}_g \, \cL_{\tau_k} \cL_{\tau_l}
\eeq
and can be shown to vanish. Indeed, let us fix $k \neq l$, and $h \in \SU(2)$ such that:
\beq
X^{k}_{h g h^{\inv}} = X^{l}_{g} \; \; ; \qquad X^{l}_{h g h^{\inv}} = - X^{k}_{g} \,.  
\eeq
It follows from the invariance of $\cM_j$ under conjugation that 
\beq
\int \extd g \, \cM_j (g ) \, X^{k}_g X^{l}_g = - \int \extd g \, \cM_j (g ) \, X^{k}_g X^{l}_g
\; \Rightarrow \; \int \extd g \, \cM_j (g ) \, X^{k}_g X^{l}_g = 0\,. 
\eeq
Hence all off-diagonal terms vanish. One is therefore left with the diagonal ones, which contribute in the following way:
\beq
\int \extd g  \, \cM_j (g) \, (\cL_{X_g})^{2} = \sum_{k=1}^{3} \int \extd g  \, \cM_j (g) \, ( X^{k}_g )^{2} (\cL_{\tau_k})^{2} \,.
\eeq
Again, by invariance under conjugation, $\int \extd g  \, \cM_j (g) \, ( X^{k}_g )^{2}$ does not depend on $k$. This implies:
\bes
\int \extd g  \, \cM_j (g) \, (\cL_{X_g})^{2} &=& \int \extd g  \, \cM_j (g) \, ( X^{1}_g )^{2} \sum_{k=1}^{3} (\cL_{\tau_k})^{2} \\
&=& \int \extd g  \, \cM_j (g) \, ( X^{1}_g )^{2} \sum_{k=1}^{3} (\cL_{\tau_k})^{2} \\
&=& \left( \frac{1}{3} \int \extd g  \, \cM_j (g) \, ( X^{k}_g )^{2} \right) \Delta \,. 
\ees
\end{enumerate}
\end{proof}

\subsubsection{Additional external faces and wave-function counter-terms}

Let us first say a word about how the previous results generalize to more external faces, still assuming the absence of wave-function counter-terms. According to Corollary \ref{faces_legs}, the only two possibilities are $F_{ext}(\cM_j) = 2$ or $F_{ext}(\cM_j) = 3$, and in both cases $N \geq 4$. Incidentally, $\omega( \cM_j ) = 0$ or $1$. Moreover, since all the faces have the same color, we always have $N_{ext} (\cM_j) = 2 F_{ext}(\cM_j)$. One defines $\cA_{\cG , \mu} (t)$ by interpolating between the end variables of the external faces, which consist of $F_{ext}(\cM_j)$ pairs of variables, with one variable per propagator in $N_{ext} (\cM_j)$. Assuming their color to be $1$, for instance, the amplitude $\cA_{\cG , \mu} (t)$ can be written as:
\bes
\cA_{\cG , \mu} (t) &=& \int [\extd g_\ell^{k} \extd \tilde{g}_\ell^{k} ] [\extd g_\ell^{ext, k} \extd \tilde{g}_\ell^{ext, k} ]  \widetilde{\cA}_{\cG , \mu}  ( g_2^{k} , g_3^{k} ; \tilde{g}_2^{k} , \tilde{g}_2^{k} ; \{ g_\ell^{ext,k} \} ; \{ \tilde{g}_\ell^{ext,k} \} ) \\
&&\prod_{k = 1}^{F_{ext}(\cM_j)} C_{i_k} ( g_1^{ext, k} , g_2^{ext, k} , g_3^{ext , k} ; g_1^{k} (t) , g_2^{k} , g_3^{k} ) \, \cM_j ( \{ g_1^{k} , \tilde{g}_1^{k} \} ) \nn \\
&& \qquad C_{i_k '} ( \tilde{g}_1^{k} , \tilde{g}_2^{k} , \tilde{g}_3^{k} ; \tilde{g}_1^{ext, k} , \tilde{g}_2^{ext, k} , \tilde{g}_3^{ext, k} ) \,, \nn
\ees
with 
\beq
g_1^{k} (t) = \tilde{g}_1^k e^{t X_{(\tilde{g}_1^{k})^\inv g_1^{k}}} \,, \qquad t \in [0 , 1]\,. 
\eeq
Moreover, we know that under a spanning tree contraction, the external faces of $\cM_j$ get disconnected. This means that the function $\cM_j$ can be factorized as a product
\beq
\cM_j ( \{ g_1^{k} , \tilde{g}_1^{k} \} ) = \prod_{k = 1}^{F_{ext} (\cM_j)} \cM_j^{(k)} ( g_1^{k} , \tilde{g}_1^{k} ) \,,
\eeq
such that each $\cM_j^{(k)}$ verifies all the invariances discussed in the previous paragraph. Thus, the part of the integrand of $\cA_{\cG , \mu} (t)$ relevant to $\cM_j$ is factorized into $k$ terms similar to the integrand appearing in the $F_{ext} = 1$ case. It is then immediate to conclude that all the properties which were proven in the previous paragraph hold in general. Indeed, the Taylor expansions to check are up to order $0$ or $1$ at most. The zeroth order of a product is trivially the product of the zeroth orders. As for the first order, it cancels out since the derivative of each one of the $k$ terms is $0$ at $t=0$. 

\

The effect of wave-function counter-terms is even easier to understand. Indeed, they essentially amount to insertions of Laplace operators. But the heat kernel at time $\alpha$ verifies
\beq
\Delta K_\alpha = \frac{\extd K_\alpha}{\extd \alpha}\,,
\eeq
therefore all the invariances of $K_\alpha$ on which the previous demonstrations rely also apply to $\Delta K_\alpha$. 

\

All in all, the conclusions drawn in the previous paragraph hold for all non-vacuum high divergent subgraphs $\cM_j$.

\subsubsection{Notations and finiteness of the remainders}
\label{sec:remainders}

In the remainder of this chapter, it will be convenient to use the following notations for the local part of the Taylor expansions above:
\beq
\tau_{\cM_j} \cA_{\cG , \mu} = \sum_{k = 0}^{\omega(\cM_j )} \frac{1}{k !} \cA_{\cG , \mu}^{(k)} (0) \,.
\eeq
$\tau_{\cM_j}$ projects the full amplitude $\cA_{\cG , \mu}$ onto effectively local contributions which take into account the relevant contributions of the subgraph $\cM_j \subset \cG$. To confirm that this is indeed the case, one needs to prove that in the remainder 
\beq
R_{\cM_j} \cA_{\cG , \mu} \equiv \int_{0}^{1} \extd t \, \frac{(1 - t)^{\omega(\cM_j )}}{\omega(\cM_j ) !} \cA_{\cG , \mu}^{(\omega(\cM_j ) + 1)} (t) \, ,
\eeq
the (non-local) part associated to $\cM_j$ is power-counting convergent. According to (\ref{sep_scales}), we have:
\bes
R_{\cM_j} \cA_{\cG , \mu} &=& \int_{0}^{1} \extd t \, \frac{(1 - t)^{\omega(\cM_j )}}{\omega(\cM_j ) !} \int [\extd g_\ell ]^3 [\extd \tilde{g}_\ell ]^3 [\extd g_\ell^{ext} ]^3 [\extd \tilde{g}_\ell^{ext} ]^3 \widetilde{\cA}_{\cG , \mu}  ( g_2 , g_3 ; \tilde{g}_2 , \tilde{g}_3 ; \{ g_\ell^{ext} \} ; \{ \tilde{g}_\ell^{ext} \} ) \nn \\
&&\times \, (\cL_{X_{\tilde{g_1}^{\inv} g_1 } , g_1 (t) })^{\omega(\cM_j) + 1} \, C_{i_1} ( g_1^{ext} , g_2^{ext} , g_3^{ext} ; g_1 (t) , g_2 , g_3 ) \nn \\
&& \times \, \cM_j ( g_1 , \tilde{g}_1 ) \,C_{i_2} ( \tilde{g}_1 , \tilde{g}_2 , \tilde{g}_3 ; \tilde{g}_1^{ext} , \tilde{g}_2^{ext} , \tilde{g}_3^{ext} ) \,,
\ees
and therefore:
\bes\label{remainder_scales}
\vert R_{\cM_j} \cA_{\cG , \mu} \vert &\leq& \int_{0}^{1} \extd t \, \frac{\vert 1 - t \vert^{\omega(\cM_j )}}{\omega(\cM_j ) !} \int [\extd g_\ell ]^3 [\extd \tilde{g}_\ell ]^3 [\extd g_\ell^{ext} ]^3 [\extd \tilde{g}_\ell^{ext} ]^3 \nn \\
&&\times \, \widetilde{\cA}_{\cG , \mu}  ( g_2 , g_3 ; \tilde{g}_2 , \tilde{g}_3 ; \{ g_\ell^{ext} \} ; \{ \tilde{g}_\ell^{ext} \} ) \nn \\
&& \times \vert X_{\tilde{g_1}^{\inv} g_1} \vert^{\omega(\cM_j ) + 1} (\cL_{ \tilde{X}_{\tilde{g_1}^{\inv} g_1 } , g_1 (t) })^{\omega(\cM_j) + 1} C_{i_1} ( g_1^{ext} , g_2^{ext} , g_3^{ext} ; g_1 (t) , g_2 , g_3 ) \nn \\
&& \times \cM_j ( g_1 , \tilde{g}_1 ) C_{i_2} ( \tilde{g}_1 , \tilde{g}_2 , \tilde{g}_3 ; \tilde{g}_1^{ext} , \tilde{g}_2^{ext} , \tilde{g}_3^{ext} ) \,,
\ees
where $\tilde{X}_{\tilde{g_1}^{\inv} g_1 }$ is the unit vector of direction $X_{\tilde{g_1}^{\inv} g_1 }$. We can now analyze how the power-counting of expression (\ref{remainder_scales}) differs from that of the amplitude $( \cG , \mu )$. There are two competing effects. The first is a loss of convergence due to the $\omega(\cM_j) + 1$ derivatives acting on $C_{i_1}$. According to (\ref{deriv_bound}), these contributions can be bounded by an additional $M^{( \omega( \cM_{j}) + 1 ) i_1 }$ term. This competes with the second effect, according to which the non-zero contributions of the integrand are concentrated in the region in which $\tilde{g_1}^{\inv} g_1$ is close to the identity. More precisely, the fact that $\cM_j$ contains only scales higher than $j$ imposes that 
\beq
\vert X_{\tilde{g_1}^{\inv} g_1} \vert \leq K M^{- j}
\eeq
where the integrand is relevant. The first line of \eqref{remainder_scales} therefore contributes to the power-counting with a term bounded by $M^{- (\omega( \cM_{j}) + 1 )) j}$. And since by definition $j > i_1$, one concludes that the degree of divergence of the remainder is bounded by:
\beq
\omega ( \cM_j ) + (\omega ( \cM_j ) + 1 ) ( i_1 - j ) \leq - 1 \,.  
\eeq

\section{Perturbative renormalizability}

We now establish a BPHZ theorem for the renormalized series. As in other kinds of field theories, this proof relies on forest formulas, and a careful separation between its high, divergent, and quasi-local parts from additional useless finite contributions.

\

We begin with a (standard) discussion about the compared merits of the renormalized expansion on the one hand, and the effective expansion on the other hand. 

So far we have discussed the renormalization of our model in the spirit of the latter, where each renormalization step (one for each slice) generates effective local couplings at lower scales. It perfectly fits Wilson's conception of renormalization:
in this setting, one starts with a theory with UV cut-off $\Lambda = M^{- 2 \rho}$, and tries to understand the physics in the IR, whose independence from UV physics is ensured by the separation of scales with respect to the cut-off. In order to compute physical processes involving external scales $i_{IR} < \rho$, one can integrate out all the fluctuations in the shell $i_{IR} < i \leq \rho$, resulting in an effective theory at scale $i_{IR}$. 

Because our model is renormalizable, we know that the main contributions in this integration are associated to quasi-local divergent subgraphs, therefore the effective theory can be approximated by a local theory of the same form as the bare one. According to Wilson's renormalization group treatment, in order to better handle the fact that only the high parts of divergent subgraphs contribute to this approximation, one should proceed in individual and iterated steps $i \to i - 1$ instead of integrating out the whole shell $i_{IR} < i \leq \rho$ at once. At each step one can absorb the ultimately divergent contributions (when the cut-off will be subsequently removed) into new effective coupling constants. This procedure, then, naturally generates one effective coupling constant per renormalizable interaction and per scale, as opposed to a single renormalized coupling per interaction in the renormalized expansion. This might look like a severe drawback, but on the other hand a main advantage is that the finiteness of the effective amplitudes becomes clear: the Taylor expansions of the previous section together with the finiteness of the remainders guarantee that all divergences are tamed. In particular, there is no problem of overlapping divergences since high subgraphs at a given scale cannot overlap.   

\

If one wants to be able to work with single renormalized couplings in the Lagrangian, one has to resort to a cruder picture in which the whole renormalization trajectory is approximated by a unique integration step from $\rho$ to $i_{IR}$. The price to pay is that one has no way anymore to isolate  the high (truly divergent) contributions of divergent subgraphs, which will result in additional finite contributions to the renormalized amplitudes. These contributions can build up over scales, explaining the appearance of renormalons, i.e. amplitudes which grow super exponentially, that is to say as a factorial of the number of vertices. This should be contrasted with the effective approach, in which amplitudes grow at most exponentially in the number of vertices. 

While they are not a big issue in perturbative expansions at low orders, renormalons are very problematic in non-perturbative approaches to quantum field theory such as the constructive program \cite{vincent_book}. 
They may be all the more problematic in TGFTs, if one expects continuum space-time physics to show up in a regime dominated by large graphs, and thus to depend on non-perturbative effects\footnote{Here, we only mean non-perturbative in the sense of the perturbative expansion for small coupling constants that we considered in this thesis.}. And this seems in turn unavoidable if one interprets the perturbative expansion we deal with here as an expansion around the 'no space-time' vacuum. 

A second related drawback of the renormalized expansion is the problem of overlapping divergences. Here again, the effective expansion appears to be very helpful. Not only overlapping divergences do not show up in this framework, but this also elucidates their treatment in the renormalized expansion. Indeed, at each step in the trajectory of the renormalization group, divergences are indexed by disconnected subgraphs. When one iterates the process, from high to lower scales, one finds that the divergent subgraphs of a given amplitude $\cA_{\cG , \mu}$ which contribute organize themselves into a forest. This is obvious once we understood that these graphs are high, and therefore correspond to nodes of the Gallavotti-Nicol\`o tree of $(\cG , \mu)$.
In order to pack all these contributions into renormalized couplings for the whole trajectory of the renormalized group, it is therefore necessary to index the counter-terms by all the possible forests of divergent subgraphs (irrespectively of them being high or not), called Zimmermann's forests. Seen from this perspective, it is only when unpacking the renormalized amplitudes by appropriately decomposing them over scale attributions that one makes transparent why and how the Zimmermann's forest formula cures all divergences. Here again, the situation in TGFTs with respect to usual QFTs would suggest to resort to the effective expansion: due to the finer notion of connectedness which indexes the divergent subgraphs (face-connectedness), overlapping divergences are enhanced, the internal structure of the vertices being an additional source of difficulties. This is for instance manifest in super-renormalizable examples of the type studied in the previous chapter, in which overlapping contributions already enter the renormalization of tadpoles. 

\

Despite the two generic drawbacks of the renormalized series, we choose a conservative approach in the following, and decide to outline in some details the proof of finiteness of the usual renormalized amplitudes. We will however start with a sketch of the recursive definition of the effective coupling constants. Since vertex-connectedness lies at the core of the Wilsonian effective expansion, we cannot take full advantage of face-connectedness in this context. This is similar to ordinary quantum field theories, where $1$-particle reducible graphs need to be taken into account in the effective expansion but can be dispensed with in the renormalized expansion. This is the main motivation for resorting to the renormalized expansion, where counter-terms are indexed by forests of divergent subgraphs. We will then decompose the amplitudes over scales and check that all contributions from high divergent subgraphs are correctly cured by the appropriate counter-terms. We will finally perform the sum over scale attributions, showing why the result is finite, and how useless counter-terms can build up to form renormalons.

\subsection{Effective and renormalized expansions}

As briefly explained before, the effective expansion is a reshuffling of the bare theory (with cut-off $\Lambda = M^{-2 \rho}$), in terms of recursively defined effective coupling constants. We therefore start from the connected Schwinger functions decomposed over scale attributions compatible with the cut-off:
\beq\label{bare_s}
\cS_N^\rho = \sum_{\cG , \mu | \mu \leq \rho} \frac{1}{s(\cG)} \left( \prod_{b \in \cB} ( - t_{b}^{\rho} )^{n_b (\cG)} \right) \cA_{\cG , \mu } \,.
\eeq 
In this formula, the sum runs over connected graphs, and $b$ spans all possible interactions, including mass and wave-function counter-terms. Starting from the highest scale $\rho$, we want to construct a set of $\rho + 1$ effective coupling constants per interaction $b$, called $t_{b , i}^{\rho}$ with $0 \leq i \leq \rho$. They will be formal power series in the bare coupling constants $t_{b , \rho}^{\rho} \equiv t_{b}^{\rho}$, such that $t_{b , i}^{\rho}$ is obtained from $t_{b , i + 1}^{\rho}$ by adding to it all the counter-terms associated to high subgraphs at scale $i + 1$. In order to make this statement more precise, it is useful to define $i_b (\cG , \mu)$ as the scale of a vertex $b$ in a graph $(\cG , \mu)$ as:
\beq\label{scale_coupling}
i_b ( \cG , \mu ) \equiv \max \{ i_l (\mu) \vert l \in L_b (\cG) \} \,,
\eeq
where $L_b (\cG)$ is the set of lines of $\cG$ which are hooked to $b$. We aim at a re-writing of \eqref{bare_s} of the form:
\beq\label{s_eff}
\cS_N^\rho = \sum_{\cG , \mu | \mu \leq \rho} \frac{1}{s(\cG)} \left( \prod_{b \in \cB(\cG) } ( - t_{b , i_b ( \cG , \mu )}^{\rho} ) \right) \cA_{\cG , \mu}^{eff} \,,
\eeq
in which the bare coupling constants have been substituted by effective ones at the scale of the bubbles making a graph $(\cG , \mu)$, and the new effective amplitudes are free of divergences. Thanks to the multiscale analysis, we know exactly which face-connected subgraphs are responsible for the divergences of a bare amplitude $\cA_{\cG , \mu}$: they are the high divergent subgraphs, which is a subset of all the quasi-local subgraphs. Unfortunately, they cannot play the leading role in the effective expansion: the divergences in a slice $i+1$ must be packaged into vertex-connected components, and \emph{reabsorbed in effective vertices with external propagators at scales lower or equal to $i$}. This condition on the external scales makes it impossible to act on a face-connected divergent subgraph independently of what it is vertex-connected to. Our language is therefore not adapted to the effective expansion. In order to make this point clearer, \emph{let us assume for the moment that the divergent subgraphs, the GN tree and the $\tau$ contraction operators are defined on the basis of vertex-connectedness}. In this provisional acceptation of the terms, let us moreover call $D_\mu (\cG)$, the \textit{forest of high divergent subgraphs} of $(\cG , \mu)$. The effective amplitudes are then deduced from the bare ones by subtracting the local part of each high divergent subgraph \cite{vincent_book}:
\beq\label{a_eff}
\cA_{\cG , \mu}^{eff} = \prod_{m \in D_\mu (\cG)} ( 1 - \tau_m ) \cA_{\cG , \mu}\,.
\eeq
Finiteness of $\cA_{\cG , \mu}^{eff}$ in the limit of infinite cut-off is then guaranteed. 
In order to make this prescription consistent, we need to reabsorb contributions of the form
\beq
\tau_m \cA_{\cG , \mu}
\eeq
into the effective coupling constants. 
This can be made more precise by defining an inductive version of (\ref{s_eff}):
\beq\label{s_eff_i}
\cS_N^\rho = \sum_{\cG , \mu | \mu \leq \rho} \frac{1}{s(\cG)} \left( \prod_{b \in \cB(\cG)} ( - t_{b , \sup( i , i_b ( \cG , \mu ))}^{\rho} ) \right) \cA_{\cG , \mu}^{eff , i} \,,
\eeq
with
\beq\label{a_eff_i}
\cA_{\cG , \mu}^{eff , i} \equiv \prod_{m \in D_\mu^i (\cG)} ( 1 - \tau_m ) \cA_{\cG , \mu}
\eeq
and
\beq
D_\mu^i (\cG) \equiv \{ m \in D_{\mu} (\cG) \vert i_m > i \} \,.
\eeq
We now proceed to prove (\ref{s_eff_i}), by induction on $i$, which at the same time will provide the recursive relation for the effective coupling constants. For $i = \rho$, (\ref{s_eff_i}) coincides with the bare expansion (\ref{bare_s}), and therefore holds true. Assuming that it holds at rank $i + 1$, let us then see how to prove it at rank $i$. The difference between $\cA_{\cG , \mu}^{eff , i}$ and $\cA_{\cG , \mu}^{eff , i + 1}$ amounts to counter-terms in $D_{\mu}^{i} (\cG) \setminus D_{\mu}^{i + 1} (\cG) = \{ M \in D_{\mu} (\cG) \vert i_M = i + 1 \}$, hence:
\beq
\cA_{\cG , \mu}^{eff , i} - \cA_{\cG , \mu}^{eff , i + 1} = \sum_{S \subset D_{\mu}^{i} (\cG) \setminus D_{\mu}^{i + 1} (\cG) \atop S \neq \emptyset} \prod_{M \in S} (- \tau_M ) \prod_{m \in D_{\mu}^{i + 1} (\cG)} ( 1 - \tau_m ) \cA_{\cG , \mu} .
\eeq
Adding and subtracting this quantity to $\cA_{\cG , \mu}^{eff , i + 1}$ in the equation (\ref{s_eff_i}) at rank $i+1$, one obtains a new equation which now involves $\cA_{\cG , \mu}^{eff , i}$ (thanks to the term added), together with a sum over subsets $S$ (due to the term subtracted). In condensed notations, this can be written as:
\beq\label{s_eff_intermediate}
\cS_N^\rho = \sum_{ (\cG , \mu , S) , \mu \leq \rho \atop S \subset D_{\mu}^{i} (\cG) \setminus D_{\mu}^{i + 1} (\cG) }  \frac{1}{s(\cG)} \left( \prod_{b \in \cB(\cG)} ( - t_{b , \sup( i + 1 , i_b ( \cG , \mu ))}^{\rho} ) \right) \cA_{\cG , \mu , S}^{eff , i} \,,
\eeq
where
\beq
\cA_{\cG , \mu , S}^{eff , i} \equiv - \prod_{M \in S} (- \tau_M ) \prod_{m \in D_{\mu}^{i + 1} (\cG)} ( 1 - \tau_m ) \cA_{\cG , \mu} 
\eeq
when $S \neq \emptyset$ and $\cA_{\cG , \mu , \emptyset}^{eff , i} \equiv \cA_{\cG , \mu}^{eff , i}$. The elements in a set $S$ being vertex-disjoint, we can contract them independently, and absorb the terms associated to $S \neq \emptyset$ into effective coupling constants at scale $i$. 

However, in order to correctly take wave function counter-terms into account, one needs to slightly generalize the notion of contraction previously defined for strictly tensorial interactions. While $\tau_M$ extracts amplitudes of contracted graphs times a pre-factor when $\omega(M) = 0$ or $1$, it is rather a sum of two terms when $\omega(M) = 2$: a zeroth order term proportional to a contracted amplitude, and a second order term proportional to a contracted amplitude supplemented with a Laplacian insertion as in (\ref{wf_ct}). In the latter case, one shall therefore decompose the operators as sums of two operators
\beq
\tau_M = \tau_M^{(0)} + \tau_M^{(2)}
\eeq
corresponding to the two types of counter-terms\footnote{Similarly, one defines $\tau_M \equiv \tau_M^{(0)}$ if $\omega(M) = 0$ or $1$.}. Developing these products, one ends up with a formula akin to (\ref{s_eff_intermediate}), provided that sets $S$ are generalized to
\beq
\hat{S} \equiv \{ (M , k_M) \vert M \in S, \, k_M \in \{ 0 , 2\}, \, k_M \leq \omega(M) \} \,,
\eeq
and that $\tau$ operators are replaced by $\tau_{\hat{M}} \equiv \tau_M^{(k_M)}$ for $\hat{M} = (M , k_M)$.
 Taking the scale attributions into account, we are lead to define the \textit{collapse} $\phi_i$, which sends triplets $( \cG , \mu , \hat{S} )$ with $S \subset D_{\mu}^{i} (\cG) \setminus D_{\mu}^{i + 1} (\cG)$ to its contracted version $(\cG' , \mu' , \emptyset)$. $\cG' \equiv \cG / \hat{S}$ is the graph obtained after the elements of $\hat{S}$ have been contracted, understood in a generalized sense: $\cG / (M , k_M)$ is equivalent to $\cG / M$ when $k_M = 0$ or $1$, and is a graph in which the $2$-point divergent subgraph $M$ has been replace by a Laplace operator if $k_M = 2$. As for $\mu'$, it is simply the restriction of $\mu$ to lines of $\cG'$. The bubbles of $\cG'$ are thought of as new effective interactions, obtained from contractions of vertex-connected graphs 
 . We can therefore factorize the sum in (\ref{s_eff_intermediate}) as:
\beq\label{s_eff_intermediate-2}
\cS_N^\rho = \sum_{\cG' , \mu' } \sum_{ (\cG , \mu , \hat{S}) , \mu \leq \rho \atop \phi_i (\cG , \mu , \hat{S}) = (\cG', \mu', \emptyset) }  \frac{1}{s(\cG)} \left( \prod_{b \in \cB(\cG)} ( - t_{b , \sup( i + 1 , i_b ( \cG , \mu ))}^{\rho} ) \right) \cA_{\cG , \mu , \hat{S}}^{eff , i} \,,
\eeq
with
\beq
\cA_{\cG , \mu , \hat{S}}^{eff , i} \equiv - \prod_{\hat{M} \in \hat{S}} (- \tau_{\hat{M}}) \prod_{m \in D_{\mu}^{i + 1} (\cG)} ( 1 - \tau_m ) \cA_{\cG , \mu} \,
\eeq
when $S \neq \emptyset$. In this last equation, one can act first with $\underset{\hat{M} \in \hat{S}}{\prod} (- \tau_{\hat{M}})$ on $\cA_{\cG , \mu}$. 
After reorganizing all the terms in (\ref{s_eff_intermediate-2}), it is easy to understand how the coupling constants at scale $i$ must be defined. For instance, assuming that $\cG$ has no quadratic divergences to avoid overloaded notations, one notices that (\ref{s_eff_intermediate-2}) reduces to (\ref{s_eff_i}) provided that:
\bes\label{rec}
\frac{1}{s(\cG')} \prod_{b' \in \cB(\cG')} ( - t_{b' , \sup( i , i_{b'} ( \cG' , \mu ))}^{\rho} ) &=& 
\sum_{ (\cG , \mu , S) , \mu \leq \rho \atop \phi_i (\cG , \mu , S) = (\cG', \mu', \emptyset) }  \frac{1}{s(\cG)} 
\prod_{b \in \cB(\cG)} ( - t_{b , \sup( i + 1 , i_b ( \cG , \mu ))}^{\rho} ) \nn \\
&& \prod_{m \in D_\mu^{i + 1} (\cG)} (1 - \tau_m ) \prod_{M \in S } (- \tau_{M}) \cA_{M , \mu} .
\ees
Thanks to usual properties of symmetry factors in quantum field theory, which allow to factorize the symmetry factors of vertex-connected subgraphs, we have 
\beq
s(\cG) = s(\cG') \prod_{M \in D_\mu^{i} (\cG) \setminus D_\mu^{i + 1} (\cG)} s(M)\,.
\eeq
We can therefore readily extract a solution for \eqref{rec}, in the form of a definition of the effective coupling constants at rank $i$:
\bes\label{scale_induction}
- t_{b , i}^{\rho} &=& - t_{b , i + 1}^{\rho} - \sum_{ (\cH , \mu , \{M\}) , \mu \leq \rho \atop  \phi_i (\cH , \mu , \{M\}) = ( b , \mu , \emptyset)} \frac{1}{s(\cH)} \left( \prod_{b' \in \cB(\cH)} ( - t_{b' , i_{b'} ( \cH , \mu )}^{\rho} )  \right) \nn \\
&\times& \; \left( \prod_{m \in D_{\mu} (\cH) \setminus \{M\}} ( 1 - \tau_m ) \right) \prod_{M \in S} (- \tau_{M} ) \, \cA_{M , \mu}\,.
\ees
This concludes the proof of the existence of the effective expansion when vertex-connectedness is used to organize the counter-terms. Had we relied on face-connectedness instead, equation (\ref{rec}) would have had coupling constants at scale $i+1$ also on the left-hand side, which would have made the whole scheme inconsistent with definition (\ref{scale_coupling}). 

\
By construction, $\{ t_{b , \rho}^{\rho} \, , \, b \in \cB \}$ are interpreted as the bare coupling constants. Accordingly, the renormalized constants are to be found at the other end of the scale ladder, namely in the last infrared slice, which corresponds to external legs. This is compatible with a renormalized coupling being defined as the full amputated function corresponding to the type of interaction considered. It can be checked that the latter amounts to set 
\beq
t_{b , ren}^{\rho} \equiv t_{b , -1}^{\rho} \,,
\eeq
and we could look for yet another reshuffling of the Schwinger functions, this time as multi-series in $\{ t_{b , ren}^{\rho} \}$. 

However, we follow a different strategy for the renormalized expansion, and close the vertex-connected parenthesis. Divergent graphs and contraction operators are now again understood in the face-connected sense we advocate in this thesis. 
The natural induction with respect to scales (\ref{scale_induction})
is not available anymore, but can be partially encapsulated into the definition of counter-terms according to an induction with respect to the number of vertices in a diagram. This is nothing but the well-known Bogoliubov induction, which provides the infinite set of counter-terms to be added to the bare Lagrangian. In our case, the induction takes the form:
\beq\label{bogo}
c_\cG = \sum_{\{ g_1, \ldots , g_k \} } \prod_{m \in S} ( - \tau_m ) \cA_{m / \{ g \}} \prod_{i = 1}^{k} c_{g_i} \,, 
\eeq
where $\cG$ is a vertex-connected graph with all its face-connected components $m \in S$ divergent, $c_\cG$ its associated counter-term, and $\{ g_1 , \ldots , g_k \}$ runs over all possible families of disjoint vertex-connected divergent subgraphs of $\cG$, for which counter-terms $\{ c_{g_i} \}$ have been defined at an earlier stage of the induction. Note also that $\cA_{m / \{ g \}}$ is a short-hand notation for the part of the amplitude associated to $m$, once the $g_i$'s it contains have been contracted. Each of these counter-terms will contribute to the renormalization of a coupling constant (or several when quadratically divergent subgraphs are present). A key point to notice is that, because our classification of divergent subgraphs (Table \ref{div}) also applies to vertex-connected components, we are ensured that any vertex-connected union of face-connected divergent subgraphs will have the same boundary as one of the bare interactions. More precisely, one has:
\beq
t_b^\rho = t_{b , ren}^{\rho} + \sum_{n = 1}^{+ \infty} c_n^{b} (t_{b , ren}^{\rho})^n \,,
\eeq
where $c_n^b$ is the sum of all the counter-terms $c_\cG$ at order $n$ of the type $b$ \footnote{The same subtlety as in the previous discussion occurs for quadratically divergent contributions: one has to split the counter-terms $c_\cG$ into mass and wave-function contributions. We kept this step implicit here in order to lighten the notations.}.

It is then a well-known fact that a (formal) perturbative expansion in these new variables generates renormalized amplitudes expressed by Zimmermann's forest formula. The forests appearing in this formula can be called \textit{inclusion forests}, since they are sets of subgraphs $\cF$ with specific inclusion properties: for any $h_1 , h_2 \in \cF$, either $h_1$ and $h_2$ are line-disjoint, or one is included into the other. In this model, the relevant forests are inclusion forests of vertex-connected subgraphs with all their face-connected components divergent. Since each of the graphs in the forests is acted upon by a product of contraction operators $( - \tau_{m} )$, one for each face-connected component, and since in addition face-connectedness is a finer notion than vertex-connectedness, one can actually work with inclusion forests of face-connected subgraphs. Moreover, one needs to strengthen their definition by emphasizing face-disjointness rather than line-disjointness. To avoid any terminology confusion with the usual notion of inclusion forest, we call this new type of forests \textit{strong inclusion forests}.
 
\begin{definition}
Let $\cH \subset \cG$ be a subgraph.
A \textit{strong inclusion forest} $\cF$ of $\cH$ is a set of non-empty and face-connected subgraphs of $\cH$, such that: 
\begin{enumerate}[(i)]
\item for any $h_1 , h_2 \in \cF$, either $h_1$ and $h_2$ are line-disjoint, or one is included into the other; 
\item any line-disjoint $h_1 , \ldots , h_k \in \cF$ are also face-disjoint.
\end{enumerate} 
\end{definition}

A few remarks are in order. First, a strong inclusion forest $\cF$ is always an inclusion forest (condition (i)), hence the nomenclature. 
Second, it is important to understand that the Zimmermann forests relevant to our model are \textit{strong} inclusion forests. To this effect, notice for instance that if $g_1, \ldots ,  g_k \subset \cG$ appear in a same term of the Bogoliubov recursion (\ref{bogo}) for some intermediate subgraph $\cH \subset \cG$, then they form $k$ distinct face-connected components in $\cH$. The existence of such a subgraph is equivalent to the face-disjointness of $g_1 , \ldots ,  g_k$. Third, we point out that the meloforests introduced before are strong inclusion forests. In the wider context of the present chapter, we modify slightly this terminology, and call \textit{meloforest} any strong inclusion forest of \textit{melonic} subgraphs. 
Finally, we simply call \textit{divergent forest} a strong inclusion forest of divergent subgraphs, and note $\cF_{D} (\cG)$ the set of divergent forests of a graph $\cG$ (including the empty forest). In the $\SU(2)$, $d=3$ model, divergent forests are also meloforests, but the converse is not true.

\

In this language, the renormalized amplitudes are related to the bare ones through:
\beq\label{a_ren}
\cA_\cG^{ren} = \left( \sum_{\cF \in \cF_D (\cG)} \prod_{m \in \cF} \left( - \tau_{m} \right) 
\right) \cA_\cG \,.
\eeq 

\

In order to prove the finiteness of the renormalized amplitudes, one should rely on the refined understanding of the divergences provided by the multiscale expansion. To this effect, we will expand equation (\ref{a_ren}) over scales. For fixed scale attribution, contraction operators acting on high divergent subgraphs will provide a convergent power-counting. The sum over scales will finally be achieved thanks to an adapted classification of divergent forests, which is the purpose of the next section.

%

\subsection{Classification of forests}

Before discussing the classification in details, we point out an intriguing property of this model. In light of Proposition \ref{curiosity}, we notice that the melonic subgraphs of a given non-vacuum graph $\cG$ organize themselves into an inclusion forest. It would be therefore tempting to conjecture that they also form a strong inclusion forest (i.e. a meloforest). However, we can actually find examples of overlapping melonic subgraphs, showing that this is incorrect (see Figure \ref{overlap}). Still, and again by Proposition \ref{curiosity}, we notice that the union of two melonic subgraphs cannot be itself melonic, hence cannot be divergent. Therefore, if we restrict our attention to divergent forests, we can actually prove that the previous conjecture hold.

\begin{figure}[h]
\begin{center}
\includegraphics[scale=0.5]{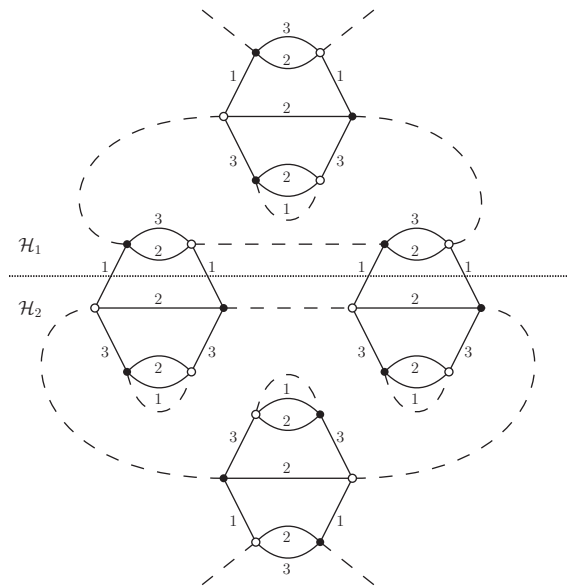}
\caption{Two melonic subgraphs $\cH_1$ and $\cH_2$ ($\cH_2$ being even divergent) which are face-connected in their union.}
\label{overlap}
\end{center}
\end{figure}

\begin{proposition}
Let $\cG$ be a non-vacuum graph. The set of divergent subgraphs of $\cG$ is a strong inclusion forest. We denote it $D(\cG)$.
\end{proposition} 
\begin{proof}
Thanks to proposition \ref{curiosity}, we already know that $D(\cG)$ is an inclusion forest. 
To conclude, we need to show that there exists no subset of line-disjoint subgraphs in $D(\cG)$ which are not also face-disjoint. If this would not be the case, we could certainly find line-disjoint subgraphs $\cH_1 , \ldots , \cH_k \in D(\cG)$ which are face-connected in their union. 
This face-connectedness is necessarily ensured by external faces of the $\cH_1 , \ldots , \cH_k$ which arrange together into internal faces of $\cH_1 \cup \dots \cup \cH_k$. Because their intersection is empty, this can only be achieved if some vertices of $\cH_1 \cup \dots \cup \cH_k$ are shared by several subgraphs $\cH_1 , \ldots , \cH_k$. Let us call $p_{2i}^s$ ($1 \leq i \leq 3$, $2 \leq s \leq 3$) the number of vertices of valency $2i$ which are shared by exactly $s$ subgraphs $\cH_k$. They can be related to the valency of $\cH_1, \ldots , \cH_k$ and $\cH_1 \cup \dots \cup \cH_k$ by the formula:
\beq
N (\cH_1 \cup \dots \cup \cH_k) = \sum_{j = 1}^{k} N (\cH_j) - \sum_{i =1}^3 \sum_{s = 2}^{3} (2 i) (s - 1) p_{2 i}^s \,.
\eeq
This just says that when summing all the individual valencies, one needs to subtract all the contributions of external legs of the connecting vertices, which have been over counted, in order to find the valency of the full subgraph. If a connecting vertex $v$ is connected to exactly $s$ subgraphs $\cH_k$, its external legs have been counted exactly $s - 1$ too many times. Furthermore, the conditions $\omega(\cH_j) \geq 0$ can be summed to yield\footnote{The $\rho$ contributions are all $0$ since $\cH_1, \ldots , \cH_k$ are non-vacuum and melonic}:
\beq
\sum_{j = 1}^k N(\cH_j) \leq 6 k - 4 \sum_{j = 1}^k n_2 (\cH_j) - 2 \sum_{j = 1}^k n_4 (\cH_j) \,. 
\eeq
Remarking that $\sum_{s = 2}^3 s p_{2i}^s \leq \sum_{j = 1}^k n_{2i} (\cH_j)$ for all $i$, we can finally deduce from the two previous inequalities that:
\beq
N(\cH_1 \cup \dots \cup \cH_k) \leq 6 k  - 6 \sum_{i = 1}^{3} \sum_{s = 2}^3 (s -1) p_{2 i}^{s} - 2 \sum_{s = 2}^3 p_{4}^{s} - 4 \sum_{s = 2}^3 p_{2}^{s}\,.
\eeq
We immediately notice that whenever
\beq
\sum_{i = 1}^{3} \sum_{s = 2}^3 (s -1) p_{2 i}^{s} \geq k \,,
\eeq
$\cH_1 \cup \dots \cup \cH_k$ is vacuum, which contradicts the hypothesis that $\cG$ is not. If the previous inequality is not verified, one has instead
\beq\label{tree_ineq}
\sum_{i = 1}^{3} \sum_{s = 2}^3 s p_{2 i}^{s} \leq k + \sum_{i = 1}^{3} \sum_{s = 2}^3 p_{2 i}^s - 1 \,.
\eeq
In order to understand the meaning of this inequality, let us introduced an abstract graph $G$: its nodes are the subgraphs $\cH_1 , \ldots , \cH_k$ and all the vertices shared by more than one subgraph; two nodes are linked by one line in $G$ if and only if one of them is a subgraph, and the other a vertex contained in this subgraph. In equation (\ref{tree_ineq}), on the left-hand side one finds the number of links in $G$, and on the right side its number of nodes minus $1$. Therefore, when the inequality is saturated $G$ is a tree, and when the inequality is strict it is not connected. The latter case is contradictory with our hypotheses. As for when $G$ is a tree, one can find spanning trees $\cT_1 \subset \cH_1 , \ldots , \cT_k \subset \cH_k$ such that there union $\cT \equiv \cT_1 \cup \dots \cup \cT_k$ is a spanning tree of $\cH_1 \cup \dots \cup \cH_k$. But in such a situation, $( \cH_1 \cup \dots \cup \cH_k ) / \cT = (\cH_1 / \cT_1 ) \cup \dots \cup ( \cH_k / \cT_k)$ would be a melopole, contradicting the fact that $\cH_1 \cup \dots \cup \cH_k$ cannot be melonic (see proposition \ref{curiosity}).
\end{proof}

At this stage, we tend to see this result as a curiosity of the specific model we are considering, and only a detailed study of other just-renormalizable models of this type could determine whether it has a wider validity. What is sure is that it is by no means essential to the classification of forests. Still, it allows significant simplifications, which we will take advantage of in the following, in the notations and proofs, since the divergent forests of $\cG$ are exactly the subsets of $D(\cG)$.

\

Recall that $D_\mu (\cG)$ denotes the set of truly divergent subgraphs in $\cG$ for the scale attribution $\mu$. It is a divergent forest, as we already knew from the fact that, modulo $\cG$ itself, it consists exactly in the subgraphs appearing in the GN tree of $(\cG , \mu)$. Furthermore, we now know it to be a subforest of $D(\cG)$. We call its complementary part
\beq
I_\mu (\cG) \equiv D(\cG) \setminus D_\mu (\cG)
\eeq
the \textit{innofensive part} of $D (\cG)$ at scale $\mu$, since it is the set of divergent subgraphs of $\cG$ which do not appear in the GN tree of $(\cG , \mu)$, and therefore do not contribute to divergences at this scale. 

In this model, where the disjoint decomposition $D(\cG) = I_\mu (\cG) \cup D_\mu (\cG)$ involves three sets which are themselves divergent forests, the classification of forest is as trivial as saying that choosing a forest in $D(\cG)$ amounts to choosing a forest in $I_\mu (\cG)$ and a forest in $D_\mu (\cG)$, namely:
\beq
\cF_D (\cG) = \{ \cF_1 \cup \cF_2 \vert \cF_1 \subset I_\mu (\cG) \, , \cF_2 \subset D_\mu (\cG)\}\,.
\eeq
We can use this simple fact in the decomposition of equation (\ref{a_ren}) over scale attributions
\bes
\cA_\cG^{ren} &=& \sum_{\mu} \sum_{\cF \in \cF_D (\cG)} \prod_{m \in \cF} ( - \tau_{m} ) \cA_{\cG, \mu } \\
&=& \sum_{\mu} \sum_{\cF_1 \subset I_\mu (\cG)} \sum_{\cF_2 \subset D_\mu (\cG)} \prod_{m \in \cF_1 \cup \cF_2} ( - \tau_{m} ) \cA_{\cG, \mu }\,.
\ees
We then exchange the first two sums:
\beq
\cA_\cG^{ren} = \sum_{\cF_1 \subset D(\cG)} \, \sum_{\mu \vert \cF_1 \subset I_\mu (\cG)} \prod_{m \in \cF_1} ( - \tau_{m} ) \sum_{\cF_2 \subset D_\mu (\cG)} \prod_{h \in \cF_2} ( - \tau_{m} ) \cA_{\cG, \mu }\,.
\eeq
We can finally reorganize the contraction operators associated to graphs of $D_\mu (\cG)$ to obtain:
\bes
\cA_\cG^{ren} &=& \sum_{\cF \subset D(\cG)} \cA_{\cG , \cF}^{ren}\,,\\
\cA_{\cG , \cF}^{ren} &\equiv& \sum_{\mu | \cF \subset I_\mu (\cG)} \prod_{m \in \cF} (- \tau_m ) \prod_{h \in D_\mu (\cG)} (1 - \tau_h ) \cA_{\cG , \mu}\,.
\ees

This way of splitting the contributions of the different forests according to the scales is in phase with the multi-scale analysis. We shall explain in the next two paragraphs why $\cA_{\cG , \cF}^{ren}$ is convergent. To this effect, we first use the contraction operators indexed by elements of $D_\mu (\cG)$ to show that the renormalized power-counting is improved with respect to the bare one, in such a way that all divergent subgraphs become power-counting convergent. In a second step, we will explain how these decays can actually be used to perform the sum over scale attributions.

\subsection{Convergent power-counting for renormalized amplitudes}\label{su2_cvpc}

We fix a divergent forest $\cF \in D(\cG)$ and a scale attribution $\mu$ such that $\cF \subset I_\mu (\cG)$. We want to find a multi-scale power-counting bound for 
\beq
\prod_{m \in \cF} (- \tau_m ) \prod_{h \in D_\mu (\cG)} (1 - \tau_h ) \cA_{\cG , \mu}\,.
\eeq
Since contraction operators commute, we are free to first act on $\cA_{\cG , \mu}$. In order to properly encode the two possible Taylor orders in $2$-point divergences, we should reintroduce the generalized notations $\hat{m}$ and $\tau_{\hat{m}}$, together with a generalized notion of divergent forest $\hat{\cF}$. Since the argument we are about to make is insensitive to such subtleties, and its clarity would be somewhat affected by the heavy notations, we decide instead to assume that $\cF$ does not contain any quadratically divergent subgraph. It is easily understood that the action of the product of contraction operators disconnects parts of the amplitudes, yielding a product of pieces of the integrand integrated on their internal variables. The exact formula is
\beq
\prod_{m \in \cF} \tau_m \cA_{\cG , \mu} = \cA_{\cG / A_\cF (\cG) } \prod_{m \in \cF} \nu_\mu ( m / A_\cF (m) ) \,,
\eeq
where $A_\cF (m) \equiv \{ g \subset m \vert g \in \cF \}$ is the set of \textit{descendants} of $m$ in $\cF$, and $\nu_\mu ( m / A_\cF (m) )$ is the \textit{amputated amplitude}\footnote{We mean by that that the contributions of external faces are discarded.}
of $m$ contracted by its descendants. The power-counting of each subgraph appearing on the right-hand side of this formula is known, yielding:
\beq
\vert \prod_{m \in \cF} (- \tau_m ) \cA_{\cG , \mu} \vert
\leq K^{L(\cG)} \prod_{m \in \cF \cup \{ \cG \}} \prod_{(i , k)} M^{\omega[ (m / A_\cF (m))_i^{(k)}]} \,. 
\eeq
As expected, we see that the contraction operators associated to inoffensive forests does not improve the power-counting, and are in a sense useless.

On the other hand, we have also seen in section \ref{sec:remainders}, that $(1 - \tau_h )$ operators effectively render subgraphs $h \subset D_\mu (\cG)$ power-counting convergent. We can use this improved power-counting in each $m / A_\cF (m)$ to prove the following proposition:
\begin{proposition}
There exists a constant $K$, such that for any divergent forest $\cF \in D( \cG )$:
\beq\label{improved}
\vert \cA_{\cG , \cF}^{ren} \vert \leq K^{L(\cG)} \sum_{\mu \vert \cF \subset I_\mu (\cG)} \prod_{m \in \cF \cup \{ \cG \}} \prod_{(i , k)} M^{\omega'[ (m / A_\cF (m))_i^{(k)} ]} \, ,
\eeq
where
\beq 
\omega'[ ( m / A_\cF (m) )_i^{(k)} ] = \min \{ -1 , \, \omega[ (m / A_\cF (m))_i^{(k)} ] \}
\eeq
except when $m \in \cF$ and $(m / A_\cF (m))_i^{(k)} = m / A_\cF (m)$, in which case $\omega'[ m / A_\cF (m) ] = 0$.
\end{proposition}
\begin{proof}
If $m$ is compatible with $\cF$ (i.e. $\cF \cup \{ m \}$ is also a strong inclusion forest), we denote by $B_\cF (m)$ the ancestor of $m$ in $\cF \cup \{ m \}$. This notion allows to decompose the product of useful contraction operators as
\beq
\prod_{h \in D_\mu (\cG)} (1 - \tau_h ) = \prod_{m \in \cF \cup \{ \cG\}} \prod_{h \in D_\mu (\cG) \atop B_\cF (h ) = m} (1 - \tau_h )\,.
\eeq
When multiplying this expression by $\underset{m \in \cF}{\prod} (- \tau_m )$, one obtains
\bes
\vert \prod_{m \in \cF} (- \tau_m ) \prod_{h \in D_\mu (\cG)} (1 - \tau_h ) \cA_{\cG , \mu} \vert &=& \left( \prod_{h \in D_\mu (\cG) \atop B_\cF ( h ) = \cG} (1 - \tau_{h / A_\cF (\cG)} ) \vert \cA_{\cG / A_\cF (\cG) , \mu} \vert \right) \\
&& \times \left( \prod_{m \in \cF } \prod_{h \in D_\mu (\cG) \atop B_\cF (h ) = m} (1 - \tau_{h / A_\cF (m)} ) \, \vert \nu_\mu ( m / A_\cF (m) ) \vert \right) \nn
\ees
We recognize in this formula all the useful contractions associated to high divergent subgraphs in each $m / A_\cF (m)$, for which the new degree is at most $-1$, except possibly for the roots $m = m / A_\cM (m)$\footnote{This root can indeed itself be divergent.} when $m \neq \cG$. But because the corresponding amplitudes are amputated, they contribute to the power-counting with a degree $0$. 
\end{proof}

\subsection{Sum over scale attributions}

The improved power-counting (\ref{improved}) allows to decompose renormalized amplitudes into fully convergent\footnote{Recall that a fully convergent graph is a graphs whose face-connected subgraphs all have convergent power-counting.} pieces associated to the contracted subgraphs $m / A_\cF (m)$. We therefore decompose the task of summing over scale attributions into two steps: as in the super-renormalizable example of the previous chapter, we will first recall how this can be performed maintaining a bound in $K^n$ for a fully convergent graph $\cG$; we will then explain how this generalizes to arbitrary renormalized amplitudes, the price to pay being possible factorial growths in $n$ due to contraction operators associated to the inoffensive forests $I_\mu (\cG)$.

\

Let $\cG$ be a fully convergent, vertex-connected, and non-vacuum graph. For any face-connected subgraph $\cH \subset \cG$ such that $\cF(\cH) \neq 0$, we have seen that
\beq
\omega (\cH ) \leq -\frac{N(\cH)}{2}\,.
\eeq
Moreover, $\omega(\cH) = -2$ and $N(\cH) \leq 10$ when $F(\cH) = 0$, therefore one can use a slower decay in $-N (\cH) / 5$ and write
\beq\label{cv_pc}
\cA_{\cG , \mu} \leq K^{L (\cG)} \prod_{(i , k)} M^{- N (\cG_i^{(k)}) / 5}
\eeq
for any scale attribution $\mu$. In order to extract a sufficient decay in $\mu$ from (\ref{cv_pc}), it is crucial to focus on the scales associated to the vertices of $\cG$. Let us therefore introduce $L_b (\cG)$ the set of external lines of a bubble $b \in \cB (\cG)$, and define:
\beq
i_b (\mu) = \sup_{l \in L_b (\cG)} i_l (\mu) \, , \qquad e_b (\mu) = \inf_{l \in L_b (\cG)} i_l (\mu)\,.
\eeq  
The main interest of these two scales lies in the two following facts: a) $b$ touches a high subgraph $\cG_i^{(k)}$ if and only if $i \leq i_b (\mu)$; b) moreover, $b$ is an external vertex of $\cG_i^{(k)}$ if and only if $e_b (\mu) < i \leq i_b (\mu)$. Accordingly, and because $b$ touches at most $6$ high subgraphs, one can distribute a fraction of the decay in the number of lines of high subgraphs to the vertices of $\cG$:
\beq
\prod_{(i , k)} M^{- N (\cG_i^{(k)}) / 5} \leq \prod_{(i , k)} \prod_{b \in \cB(\cG_i^{(k)}) \vert e_b (\mu) < i \leq i_b (\mu)} M^{- 1 / 30} \,.
\eeq
Exchanging the two products yields the interesting bound:
\beq
\cA_{\cG , \mu} \leq K^{L (\cG)} \prod_{b \in \cB (\cG)} \prod_{(i , k) \vert e_b (\mu) < i \leq i_b (\mu)} M^{- \frac{i_b (\mu) - e_b (\mu)}{30}}  \,.
\eeq
Finally, we can distribute the decays among all possible pairs of external legs of each vertex. Since there are at most $6 \times 5 / 2 = 15$ such pairs, we get:
\beq
\cA_{\cG , \mu} \leq K^{L (\cG)} \prod_{b \in \cB (\cG)} \prod_{(l, l') \in L_b (\cG)  \times L_b (\cG) } M^{- \frac{\vert i_l (\mu) - i_l' (\mu) \vert }{450}}  \,.
\eeq
This bound implies the finiteness of $\cA_\cG$. To see this, we can choose a total ordering of the lines $L(\cG) = \{ l_1 , \ldots , l_{L(\cG)} \}$ such that $l_1$ is hooked to an external vertex of $\cG$, and $\{ l_1 , \ldots , l_m \}$  is connected for any $m \leq L(\cG)$. This allows to construct a map $j'$ on the indices $2 \leq j \leq L(\cG)$, such that $1 \leq j'(j) < j$, and\footnote{By convention, one also defines $i_{l_{j'(1)}} = - 1$.}:
\beq
\prod_{b \in \cB (\cG)} \prod_{(l, l') \in L_b (\cG)  \times L_b (\cG) } M^{- \frac{\vert i_l (\mu) - i_l' (\mu) \vert }{450}} \leq \prod_{j = 1}^{L(\cG)} M^{- \vert i_{l_j} (\mu) - i_{l_{j'(j)}} (\mu) \vert / 450 } \,.
\eeq
The sum over $\mu = \{ i_{l_1}, \ldots , i_{l_{L(\cG)}} \}$ of such a product is uniformly bounded by a constant to the power $L(\cG)$, which proves the following theorem:
\begin{theorem}
There exists a constant $K>0$ such that, for any fully convergent, vertex-connected, and non-vacuum graph $\cG$:
\beq
\cA_\cG \leq K^{L(\cG)}\,.
\eeq
\end{theorem}

\

We can apply the same reasoning to the general power-counting (\ref{improved}). Let us fix $\cF$ a divergent forest. The only difference is that graphs $g / A_\cF (g)$ do not have any decay associated to their external legs. One therefore gets one additional scale index to sum over per element of $\cF$. But we can bound them by the maximal scale $i_{max} (\mu)$ in $\mu$ and write:
\bes
\vert \cA_{\cG , \cF}^{ren} \vert &\leq& K^{L(\cG)} \sum_{\mu \vert \cF \subset I_\mu (\cG)} \prod_{m \in \cF \cup \{ \cG \}} \prod_{(i , k)} M^{\omega'[ (m / A_\cF (m))_i^{(k)} ]} \\
&\leq& {K_1}^{L(\cG)} \sum_{i_{max} (\mu) } (i_{max} (\mu))^{\vert \cF \vert} M^{\delta i_{max} (\mu)} \, , 
\ees
where $\delta > 0$ and $K_1 > 0$ are some constants, and $\vert \cF \vert$ is the cardinal of $\cF$. The last sum over $i_{max} (\mu)$ can finally be bounded by $|\cF|! K^{|\cF|}$ for some constant $K > 0$. The final sum on $\cF \subset D(\cG)$ can be absorbed into a redefinition of the constants, since the number of divergent forests is simply bounded by $2^{| D(\cG) |}$. This concludes the proof of the BPHZ theorem.

\begin{theorem}
For any vertex-connected and non-vacuum graph $\cG$, the renormalized amplitude $\cA_\cG^{ren}$ has a finite limit when the cut-off $\Lambda$ is sent to $0$. More precisely, there exists a constant $K>0$ such that the following uniform bound holds:
\beq
\vert \cA_\cG^{ren} \vert \leq K^{L(\cG)} \vert D(\cG) \vert ! 
\eeq
\end{theorem} 
While this theorem proves the renormalizability of the model, it does not preclude the existence of renormalons, since the uniform bound we could find is only factorial. However, we notice that such an unreasonable growth can only exist because of the contraction operators associated to subforest of $I_\mu (\cG)$. On the contrary, if we were to focus on the effective expansion, in which only counter-terms associated to high graphs contribute, one would find instead a uniform bound like the one for fully convergent graphs.






\section{Renormalization group flow}\label{sec:rg_flow}

We conclude this chapter with a preliminary analysis of the renormalization group flow of this TGFT in the deep UV. When it comes to concrete calculations, face-connectedness would bring important practical simplifications. Quite a few graphs which need to be computed in a renormalization scheme based on vertex-connectedness would be absent, for only vertex-connected unions of face-connected divergent graphs would contribute to the flow of the coupling constants in this case. Having localization operators acting on face-connected components would also be advantageous, because their combinatorics is relatively simple.

\
However, the renormalization group flow relies primarily on vertex-connectedness. In order to determine its properties, we therefore need to compute counter-terms associated to vertex-connected divergent graphs, which significantly complicates the task. On the other hand, from the point of view of the renormalization scheme based on face-connectedness developed in this thesis, we might expect the vertex-connected unions of face-connected divergent graphs to generate the most relevant contributions. We therefore outline the general formalism, but only compute the terms associated to this subclass of graphs. The full analysis is in progress

\subsection{Approximation scheme}

Being irrelevant to the question of renormalizability, normalization factors were discarded so far. On the contrary, they are of primary importance in concrete computations of physical coupling constants, where the wave-function renormalization needs to be taken into account. 

\
In order to correctly incorporate these factors into our scheme, we come back to the general Wilsonian perspective outlined in equation (\ref{flow_gene}). Let us call $S_i$ the effective action at scale $i$, with coupling constants $t_{6,2,i}^{phys}$, $t_{6,1,i}^{phys}$ and $t_{4,i}^{phys}$. No mass nor wave-function counter-terms are incorporated in $S_i$, i.e. $CT_{m, i}^{phys} = CT_{\vphi, i}^{phys} = 0$. The covariance at this scale is therefore parametrized by the physical mass $m_{phys,i}$, and we denote it $C^i_{m_{phys,i}}$. In order to determine the effective action at scale $S_{i-1}$, we proceed in two steps. We first define an auxiliary effective action $\widetilde{S}_{i-1}$ by integrating out the slice $i$, with respect to the measure of covariance $C_{i, m_{phys,i}}$:
\beq
\e^{- \widetilde{S}_{i - 1} (\Phi, \overline{\Phi}) } = \int \extd \mu_{C_{i, m_{phys,i}}} (\vphi_{i} , \vphib_{i})\, \e^{- \widetilde{S}_{i}  (\Phi + \vphi_{i}, \overline{\Phi} + \vphib_i )}\,,
\eeq
where $\Phi = \underset{j \leq i - 1}{\sum} \vphi_j$. Following the previous section, the auxiliary effective action can be approximated by:
\beq
\widetilde{S}_{i - 1} \approx \frac{t_{4,i-1}}{2} S_{4} + \frac{t_{6,1, i-1}}{3} S_{6,1} + t_{6,2, i-1} S_{6,2} + CT_{m, i-1} S_{m} + CT_{\vphi, i-1} S_{\vphi} + CT_{0,i-1}\,,
\eeq
where $t_{4,i-1}$, $t_{6,1, i-1}$, $t_{6,2, i-1}$, $CT_{m, i-1}$ and $CT_{\vphi, i-1}$ can be deduced from $t_{6,2,i}^{phys}$, $t_{6,1,i}^{phys}$, and $t_{4,i}^{phys}$ thanks to an induction formula similar to (\ref{scale_induction}), only simpler. Indeed, since $i$ is the highest scale, there is no effective coupling constant at higher scales to be taken into account, and no nested contraction operators to be incorporated, which yields:
\beq\label{scale_induction1}
 t_{b , i - 1} = t_{b , i }^{phys} + \sum_{ (\cH , \mu , \{M\}) , \mu \leq \rho \atop  \phi_i (\cH , \mu , \{M\}) = ( b , \mu , \emptyset)} \frac{1}{s(\cH)} \left( \prod_{b' \in \cB(\cH)} ( - t_{b' , i }^{phys} )  \right) (- \tau_{M} ) \, \cA_{M , \mu}\,.
\eeq 
This is in a sense a Markovian truncation of the general equation (\ref{scale_induction}). $CT_{0,i-1}$ contains the contributions of the vacuum divergent graphs, which we did not analyze. But since it will only add a constant factor $\e^{- CT_{0,i-1}}$ in front of the effective partition function at scale $i-1$, it is irrelevant, and we set it to $0$ from now on. 

We now turn to the second step of the procedure, which consists in reabsorbing the $2$-point counter-terms into the covariance. 
Let us define the operator
\beq
M_{i-1} = - CT_{m , i -1} + CT_{ \vphi , i - 1} \sum_\ell \Delta_\ell \,,
\eeq
corresponding to the kernel of the $2$-point function counter-terms at scale $i$. We also write the covariance at scale $i-1$ before renormalization as
\beq
C^{i-1}_{m_{phys,i}} = P \widetilde{C}^{i-1}_{m_{phys,i}} \,,
\eeq
where $\widetilde{C}^{i-1}_{m_{phys,i}}$ is a covariance without integration on $h$, and $P$ is the group-averaging operator (restoring the integration on $h$). Interestingly, one can prove that 
\beq
\left[ P , \widetilde{C}^{i-1}_{m_{phys,i}} \right] = 0 \;; \qquad \left[ P , M_{i-1} \right] = 0 \,.
\eeq  
Now, at scale $i-1$ the full effective covariance is that of the measure:
\beq
\extd \mu_{C^{i-1}_{m_{phys,i}}} (\Phi , \overline{\Phi} ) \exp\left( \int [\extd g_\ell] [\extd g_\ell'] \Phi(g_1 , g_2 , g_3) \, M_{i-1} (g_\ell ; g_\ell' ) \, \overline{\Phi}(g_1' , g_2' , g_3') \right) \,.
\eeq
Let us call $\overline{C}^{i-1}$ this covariance. It can be computed by summing over connected $2$-point functions in the following way:
\bes
\overline{C}^{i-1} &=& C^{i-1}_{m_{phys,i}} + C^{i-1}_{m_{phys,i}} M_{i-1} C^{i-1}_{m_{phys,i}} + C^{i-1}_{m_{phys,i}} M_{i-1} C^{i-1}_{m_{phys,i}} M_{i-1} C^{i-1}_{m_{phys,i}} + \ldots \nn \\
&=& P \left( \widetilde{C}^{i-1}_{m_{phys,i}} + \widetilde{C}^{i-1}_{m_{phys,i}} M_{i-1} \widetilde{C}^{i-1}_{m_{phys,i}} + \widetilde{C}^{i-1}_{m_{phys,i}} M_{i-1} \widetilde{C}^{i-1}_{m_{phys,i}} M_{i-1} \widetilde{C}^{i-1}_{m_{phys,i}} + \ldots\right) \nn \\
&=& P \frac{\widetilde{C}^{i-1}_{m_{phys,i}}}{1 - \widetilde{C}^{i-1}_{m_{phys,i}} M_{i-1}} \,.
\ees
From the explicit expression of $\widetilde{C}^{i-1}_{m_{phys,i}}$ we can deduce the UV approximation:
\bes
\widetilde{C}^{i-1}_{m_{phys,i}} &=& \int_{M^{- 2 (i-1)}}^{+\infty} \extd \alpha \, \exp\left(- \alpha ( m_{phys,i}^2 - \sum_\ell \Delta_\ell )\right)\nn \\
&=& \frac{ \exp\left( - M^{- 2 (i-1)} ( m_{phys,i}^2 - \sum_\ell \Delta_\ell )\right)}{m_{phys,i}^2 - \sum_\ell \Delta_\ell} \nn \\ 
&\underset{i \to + \infty}{\approx}& \frac{\exp\left(- M^{-2 (i-1)} m_{phys, i}^2\right)}{m_{phys,i}^2 - \sum_\ell \Delta_\ell}
\ees
and hence:
\bes
\overline{C}^{i-1} &\underset{i \to + \infty}{\approx}& \frac{1}{Z_{i-1}} \frac{1}{m^2_{phys, i-1} - \sum_\ell \Delta_\ell} \,,\\
Z_{i-1} &\equiv& \exp\left(- M^{-2 (i-1)} m_{phys, i}^2\right) \left( 1 + CT_{\vphi , i-1} \right)\,, \label{zi} \\
m^2_{phys, i-1} &\equiv& \frac{m_{phys,i}^2 + CT_{m , i-1}}{1 + CT_{\vphi , i-1}} \,. \label{mass_phys}
\ees
In order to determine the physical coupling constants, one normalizes the wave-function parameter, thanks to the field redefinition
\beq
\Phi \to \frac{\Phi}{\sqrt{Z_{i-1}}} \,.
\eeq
The powers of $Z_{i-1}$ subsequently appearing in the interaction part of the action must be reabsorbed into new physical coupling constants:
\beq\label{rescale_couplings}
t_{4 , i - 1}^{phys} \equiv \frac{t_{4 , i - 1}}{{Z_{i-1}}^2} \;; \qquad 
t_{6, 1 , i - 1}^{phys} \equiv \frac{t_{6 , 1,  i-1}}{{Z_{i-1}}^3} \;; \qquad 
t_{6, 2 , i - 1}^{phys} \equiv \frac{t_{6 , 2,  i-1}}{{Z_{i-1}}^3} \,.
\eeq
They parametrize the effective action $S_{i-1}$ at scale $i-1$, together with $CT_{m, i-1}^{phys} = CT_{\vphi, i-1}^{phys} = 0$, and the renormalized covariance is $C^{i-1}_{m_{phys,i-1}}$. This procedure can be reiterated, hence defining flow equations for the physical coupling constants and the mass.  




\subsection{Truncated equations for the counter-terms}





We first provide explicit equations for the auxiliary coupling constants and counter-terms, truncated to their first non-vanishing corrections in the physical parameters. We also restrict our attention to vertex-connected unions of face-connected divergent subgraphs. We assume that the flow equations can be given an analytic meaning when $t_{6,1, i}^{phys}$, $t_{6,2, i}^{phys}$ and $t_{4, i}^{phys}$ are small enough.  
We will use the generic notation $\cO ( t^k )$ for neglected terms of order $k$ in $t_{6,1, i}^{\rho}$, $t_{6,2, i}^{\rho}$ and $t_{4, i}^{phys}$. Of primary importance as regards asymptotic freedom are the $6$-point coupling constants, and the wave-function counter-terms. Indeed, according to Table \ref{div}, the multi-series appearing in the equation for $t_{4,i}$ (\ref{scale_induction1}) do not go beyond order $1$ in $t_{4,i}^{phys}$, which will imply that the UV behavior of $t_{4,i}^{phys}$ can be controlled by those of $t_{6,1,i}^{phys}$.

\
{\bf Remark:} In the following, we will compute the different combinatorial factors by hand, that is by counting the number of graphs contributing to a given term, together with the number of Wick contractions producing each of them. It is also possible to use directly formula (\ref{scale_induction1})\footnote{Recall that $s(\cH)$ can be explicitly computed: it is the number of permutations of the external legs of the labeled graph $\cH$ leaving its colored structure unchanged.}. The latter can instead be used to double-check our computations.

\subsubsection{Mass counter-term}

We first compute $CT_{m, i-1}$ in terms of $CT_{m, i}^{phys}$ plus corrections at order one in the small parameters $t_{6,1, i}^{phys}$, $t_{6,2, i}^{phys}$ and $t_{4, i}^{phys}$. There are four types of graphs contributing at this order: one with a type $(4)$ bubble, one with a $(6,1)$ bubble, and two with a $(6 , 2)$ bubble. They are represented in Figure \ref{mass}. 

\begin{figure}[h]
  \centering
  \subfloat[$G_4^\ell$]{\label{g4l}\includegraphics[scale=0.6]{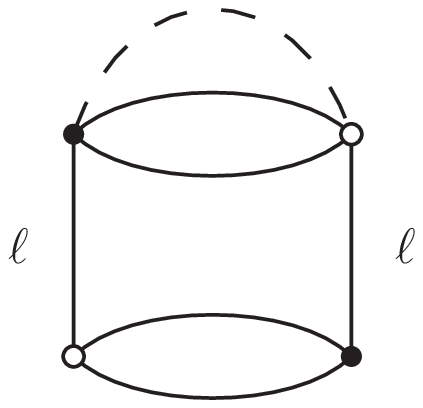}}   	       
  \subfloat[$G_{6,1}^\ell$]{\label{g61l}\includegraphics[scale=0.6]{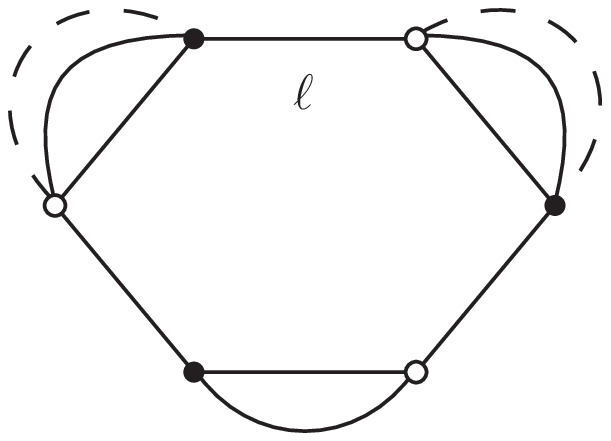}}     
  \subfloat[$G_{6,2}^\ell$]{\label{g62l}\includegraphics[scale=0.6]{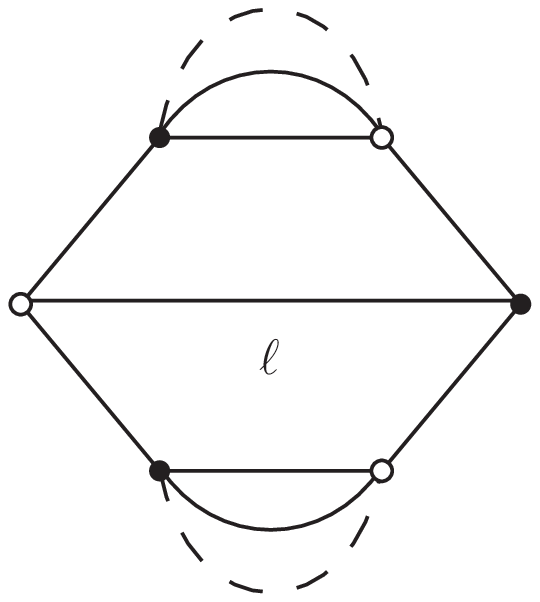}}     
  \subfloat[$G_{6,2}^{\ell \ell'}$]{\label{g62ll}\includegraphics[scale=0.6]{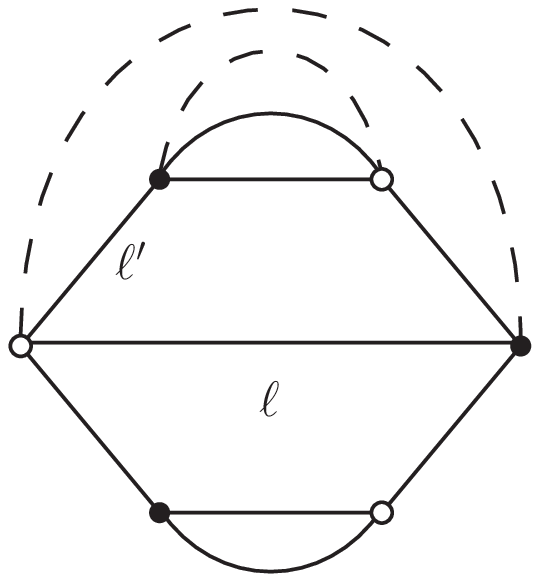}}
  \caption{First order corrections to the mass counter-term with a single-vertex.}\label{mass}
\end{figure}

Let us start with the graph $G_4^\ell$. It has degree $\omega( G_4^\ell ) = 1$, and therefore only its contraction at order $0$ in Taylor expansion contributes. There are two ways of forming its unique line, yielding a combinatorial weight of $2$. $G_4^\ell$ is therefore responsible for a linearly divergent correction:
\bes
- ( - \frac{t_{4 , i}^\rho}{2} ) \times 2 \times  \tau_{h}  \cA_{h} 
&=& t_{4 , i}^\rho \int_{M^{-2i}}^{M^{-2(i-1)}} \extd \alpha \, \e^{- m_{phys,i}^2 \alpha} \int \extd g \, [ K_\alpha (g) ]^2 \nn \\
&=& t_{4 , i}^\rho \int_{M^{-2i}}^{M^{-2(i-1)}} \extd \alpha \, \e^{- m_{phys,i}^2 \alpha} K_{2 \alpha} (\one) \, ,
\ees
where $h$ is the subgraph consisting in the unique line of $G_4^\ell$. 
Such a term arises three times, once for each $\ell$.

\
The graph $G_{6 , 1}^\ell$ has two face-connected components, each with degree $1$. The combinatorial weight is $3$, and the two face-connected components must be contracted independently, yielding a quadratically divergent correction:
\beq
- ( - \frac{t_{6,1 , i}^\rho}{3} ) \times 3 \times \tau_{h_1}  \cA_{h_1} \times \tau_{h_2}  \cA_{h_2} 
=  t_{6, 1 , i}^\rho \left(  \int_{M^{-2i}}^{M^{-2(i-1)}} \extd \alpha \, \e^{- m_{phys,i}^2 \alpha} K_{2 \alpha} (\one) \right)^2
\eeq   
Again, this correction comes in three types, one per color.

\
The situation is very similar for the graph $G_{6,2}^\ell$, which yields a counter-term
\beq
t_{6, 2 , i}^\rho \left(  \int_{M^{-2i}}^{M^{-2(i-1)}} \extd \alpha \, \e^{- m_{phys,i}^2 \alpha} K_{2 \alpha} (\one) \right)^2
\eeq 
appearing three times. 

\
Finally, we focus on $G_{6,2}^{\ell \ell'}$. It is face-connected, has degree two, and a weight $1$. Let us call $\alpha_1$ the Schwinger parameter associated to the elementary melon of $G_{6,2}^{\ell \ell'}$, and $\alpha_2$ that of its second line. The quadratically divergent counter-term associated to $G_{6,2}^{\ell \ell'}$ is then:
\bes
&&t_{6,2, i}^\rho \int_{M^{-2i}}^{M^{-2(i-1)}} \extd \alpha_1 \extd \alpha_2 \, \e^{- m_{phys,i}^2 (\alpha_1 + \alpha_2 )} \int \extd h_1 \extd h_2 \, [K_{\alpha_1} (h_1)]^2 K_{\alpha_1 + \alpha_2} (h_1 h_2) K_{\alpha_2} (h_2) \nn \\
&=& t_{6,2, i}^\rho \int_{M^{-2i}}^{M^{-2(i-1)}} \extd \alpha_1 \extd \alpha_2 \, \e^{- m_{phys,i}^2 (\alpha_1 + \alpha_2 )} \int \extd h  \, [K_{\alpha_1} (h)]^2 K_{\alpha_1 + 2 \alpha_2} (h)
\ees
This correction has to be counted $6$ times, once for each pair $(\ell \ell')$.

\
All in all, we can write:
\beq\label{mass_flow}
CT_{m, i-1} = 3 S_{1,i} \, t_{4, i}^{phys}  + 3 {S_{1,i}}^2 \, ( t_{6,1 , i}^{phys} + t_{6,2 , i}^{phys} ) + 6 S_{2,i} \, t_{6,2 , i}^{phys} + \; \ldots \; + \cO ( t^2 ) \,,
\eeq
where:
\bes
S_{1,i} &\equiv& \int_{M^{-2i}}^{M^{-2(i-1)}} \extd \alpha \, \e^{- m_{phys,i}^2 \alpha} K_{2 \alpha} (\one)\,, \\
S_{2,i} &\equiv& \int_{M^{-2i}}^{M^{-2(i-1)}} \extd \alpha_1 \extd \alpha_2 \, \e^{- m_{phys,i}^2 (\alpha_1 + \alpha_2 )} \int \extd h  \, [K_{\alpha_1} (h)]^2 K_{\alpha_1 + 2 \alpha_2} (h) \,,
\ees
and the dots indicate that we have not taken all the contributions into account.

\subsubsection{Wave-function counter-terms}

At order $1$, and given our truncation, the only type of graphs contributing to the wave-function renormalization is $G_{6,2}^{\ell \ell'}$. $G_{6,1}^{\ell}$ on the other hand is to be left aside for now, since it is made of two face-connected components, each of degree $1$. More precisely, each $G_{6,2}^{\ell \ell'}$ generates a term in $\Delta_{\ell''}$, where $\{ \ell , \ell' , \ell'' \} = \{1 , 2, 3 \}$. This induces a contribution to $CT_{\vphi, i-1}$ of the form
\beq
- t_{6,2, i}^\rho \times \frac{1}{3} \int \extd g \, \cM_i (g) \vert X_g \vert^2
\eeq 
where the function $\cM_i$ is a kernel associated to the external leg of color $\ell''$, and the minus sign comes from the fact that the operator appearing in the action is $(- \Delta_{\ell''})$. The kernel is moreover easily seen to be
\bes
\cM_i (g) &=& \int_{M^{-2i}}^{M^{-2(i-1)}} \extd \alpha_1 \int_{M^{-2 i}}^{M^{-2(i-1)}} \extd \alpha_2 \, \e^{- m_{phys,i}^2 (\alpha_1 + \alpha_2 )} \int \extd h_1 \int \extd h_2 \, [K_{\alpha_1} (h_1)]^2 \nn \\
&& \qquad \times K_{\alpha_1 + \alpha_2} (h_1 h_2) K_{\alpha_2} (h_2) K_{\alpha_2} (h_2 g) 
\ees
Each counter-term $(- \Delta_{\ell''} )$ will receive two such contributions, corresponding to the two possible choices for $(\ell \ell')$, from which we deduce that:
\beq
CT_{\vphi , i -1 } = - \frac{2}{3} \tilde{S}_{2 , i} \, t_{6 , 2 , i}^{phys} + \; \ldots \; + \cO ( t^2 ) \,,
\eeq
with
\bes
\tilde{S}_{2 , i} &\equiv& \int_{M^{-2i}}^{M^{-2(i-1)}} \extd \alpha_1 \int_{M^{-2 i}}^{M^{-2(i-1)}} \extd \alpha_2 \, \e^{- m_{phys,i}^2 (\alpha_1 + \alpha_2 )} \int \extd h_1 \int \extd h_2 \, \extd g \, \vert X_g \vert^2 \nn \\
&& \qquad \times [K_{\alpha_1} (h_1)]^2 K_{\alpha_1 + \alpha_2} (h_1 h_2) K_{\alpha_2} (h_2) K_{\alpha_2} (h_2 g) \,.
\ees

\subsubsection{6-point interactions}

According to Table \ref{div}, $6$-point divergent subgraphs up to second order in physical coupling constants necessarily consist of two $\vphi^6$ vertices. We can also understand that the vertex-connected unions of face-connected divergent subgraphs actually need to be face-connected. Indeed, if there were more than two face-connected components, at least one would have more than eight external legs and would therefore be convergent. More precisely, it turns out that there are only three categories of graphs contributing in our truncation, as shown in Figure \ref{int6}: $H_{6,1}^{\ell \ell'}$ contributes to the renormalization of $(6,1)$ interactions, while $H_{6,2}^{\ell \ell' ; \ell}$ and $H_{6,2}^{\ell \ell' ; \ell'}$ renormalize $(6 , 2)$ interactions. 

\begin{figure}[h]
  \centering
  \subfloat[$H_{6,1}^{\ell \ell'}$]{\label{h61ll}\includegraphics[scale=0.6]{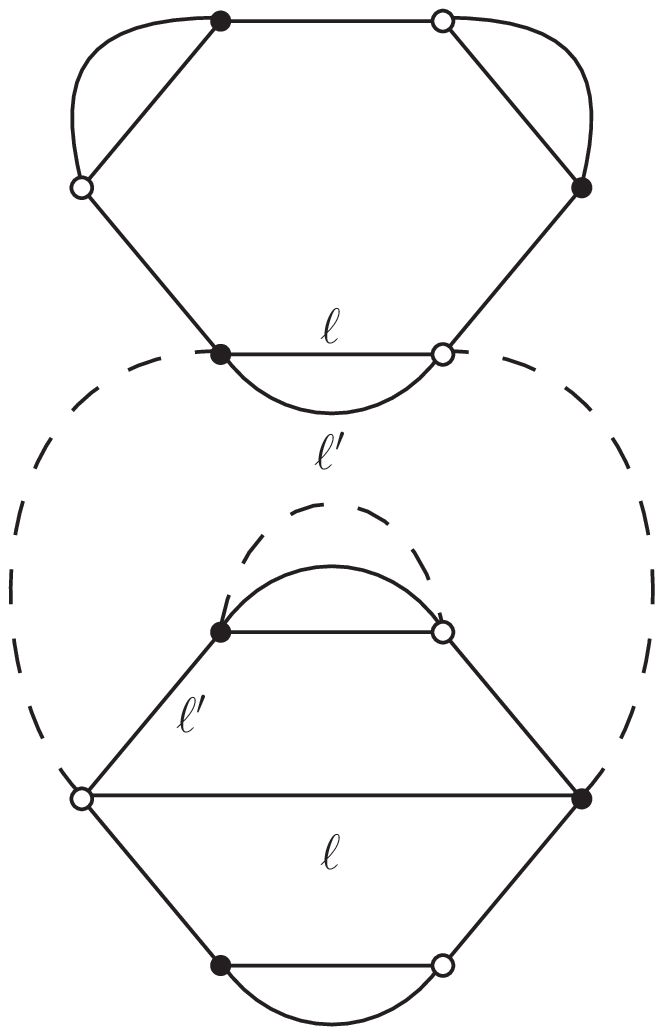}}
  \subfloat[$H_{6,2}^{\ell \ell'; \ell}$]{\label{h62ll}\includegraphics[scale=0.6]{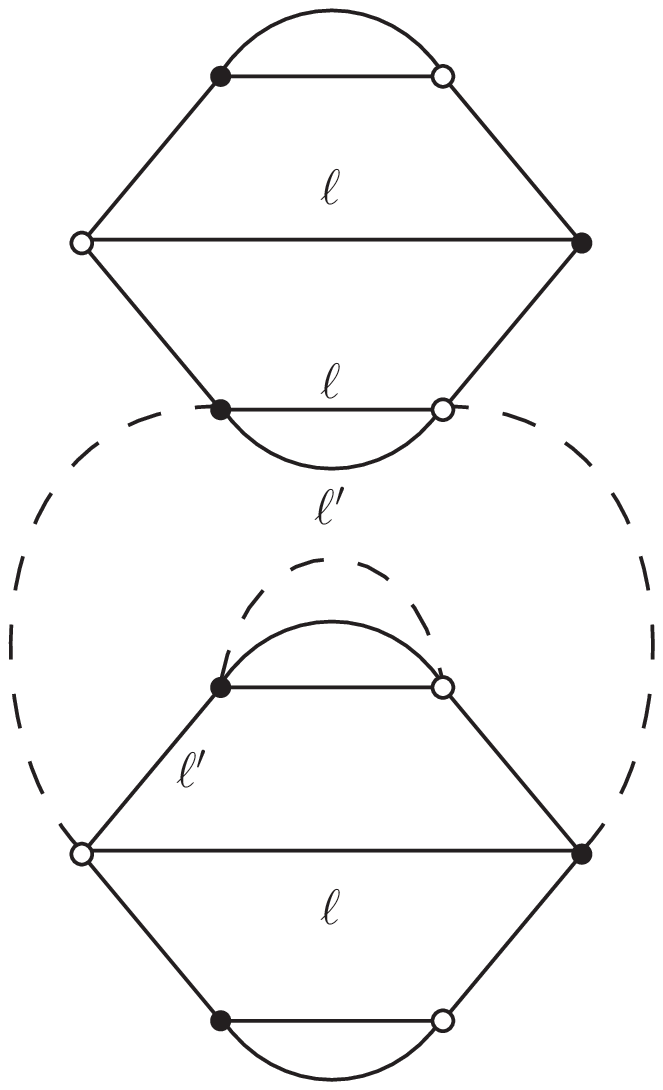}}
  \subfloat[$H_{6,2}^{\ell \ell' ; \ell'}$]{\label{h62lll}\includegraphics[scale=0.6]{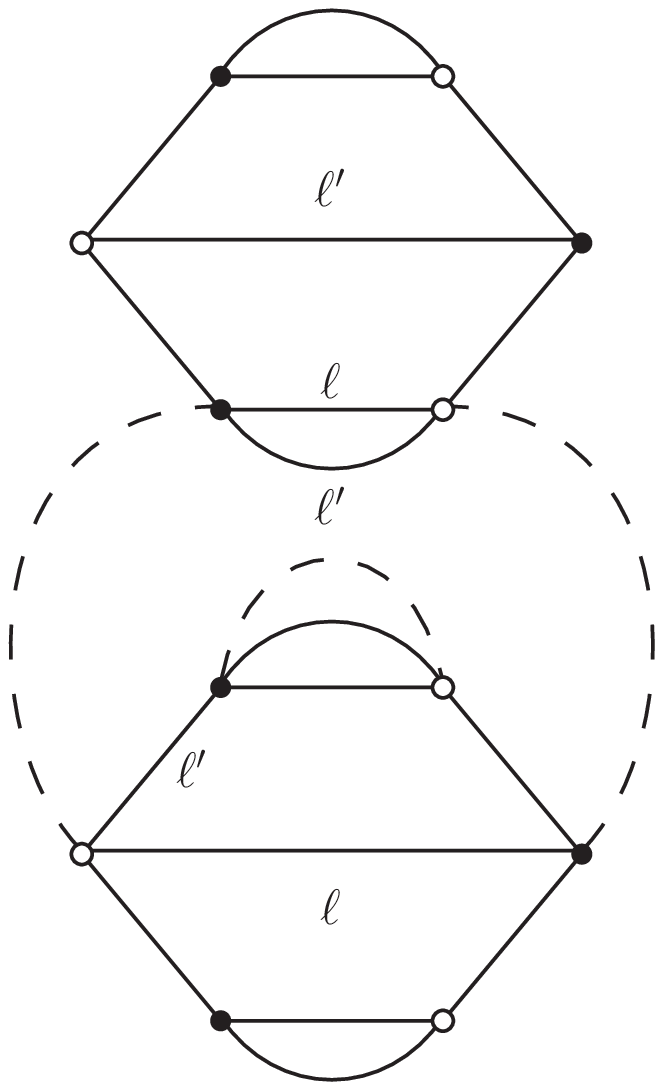}}
  \caption{Second order corrections to the $6$-point interactions.}\label{int6}
\end{figure}

Let us first look at the flow equation for $(6,1)$ interactions. The number of Wick contractions producing the graph $H_{6,1}^{\ell \ell'}$ is $3$, since the only freedom is in the choice of face $(\ell \ell')$ in the upper $(6,1)$ vertex which should be connected to the lower $(6,2)$ vertex. And because $H_{6,1}^{\ell \ell'}$ and $H_{6,1}^{\ell' \ell}$ contribute to the renormalization of a same $(6,1)$ bubble, there is an additional factor $2$ to take into account. The truncated flow therefore takes the form
\beq
\frac{ t_{6 , 1 , i - 1}}{3} = \frac{ t_{6 , 1 , i }^{phys}}{3} - 6 S_{3 , i} \, \frac{ t_{6 , 1 , i }^{phys}}{3} t_{6 , 2 , i }^{phys} + \; \ldots \; +\cO ( t^3 ) \,,
\eeq
or equivalently
\beq
t_{6 , 1 , i - 1} = t_{6 , 1 , i }^{phys} \left( 1 - 6 S_{3 , i} \, t_{6 , 2 , i }^{phys} \right) + \; \ldots \; +\cO ( t^3 ) \,,
\eeq
where $S_{3 , i}$ is define as:
\beq
S_{3,i} \equiv \int_{M^{-2i}}^{M^{-2(i-1)}} \extd \alpha_1 \, \extd \alpha_2 \, \extd \alpha_3 \, \e^{- m_{phys,i}^2 (\alpha_1 + \alpha_2 + \alpha_3 )} \int \extd h  \, [K_{\alpha_1} (h)]^2 K_{\alpha_1 + 2 (\alpha_2 + \alpha_3)} (h)\,. 
\eeq 

We find a similar equation for $(6 , 2)$ interactions, except for combinatorial weights. The number of Wick contractions producing a $H_{6,2}^{\ell \ell' ; \ell}$ is $2$, and only $1$ for $H_{6,2}^{\ell \ell' ; \ell'}$. However the first comes with an additional $1/2!$ contribution (since it consists of two identical vertices). There are moreover two graphs $H_{6,2}^{\ell \ell' ; \ell}$ contributing to a same $(6 , 2)$ interaction (corresponding to two choices for $\ell'$), and likewise two graphs $H_{6,2}^{\ell \ell' ; \ell'}$ (choice of $\ell$). This gives therefore an overall combinatorial factor $2 \times 1/2 \times 2 + 1 \times 2 = 4$, yielding: 
\beq
t_{6 , 2 , i - 1} = t_{6 , 2 , i }^{phys} \left( 1 - 4  S_{3 , i}  \, t_{6 , 2 , i }^{phys} \right) + \; \ldots \; + \cO ( t^3 ) \,.
\eeq

\subsubsection{4-point interactions}

According to the classification of divergent graphs summarized in Table \ref{div}, the multi-series defining $t_{4, i-1}$ stops at order $1$ in $t_{4,i}^{phys}$, and we can therefore write:
\beq
t_{4,i-1} = \left( 1 + f_1 ( t_{6,1,i}^{phys} , t_{6,2,i}^{phys} ) \right) \, t_{4,i}^{phys} + f_2 ( t_{6,1,i}^{phys} , t_{6,2,i}^{phys} )\,, 
\eeq
where $f_1$ and $f_2$ are multi-series in $t_{6,1,i}^{phys}$ and $t_{6,2,i}^{phys}$. 

\
Up to first order in the physical coupling constants, and within the truncation to vertex-connected unions of face-connected divergent subgraphs, $f_1$ and $f_2$ receive contributions from three types of graphs: $I_{6,1}^{\ell}$, $I_{6,2}^{\ell \ell'}$ and $H_4^{\ell \ell'}$, shown in Figure \ref{int4}. The contraction of an $I_{6,1}^{\ell}$ or an $I_{6,2}^{\ell \ell'}$ brings a $S_{1,i}$ coefficient, with combinatorial weights $3$ and $1$ respectively. $I_{6,2}^{\ell \ell'}$ and $I_{6,2}^{\ell ' \ell}$ contribute to the same effective bubble, which brings an additional factor $2$ from this type of graphs. Each $H_4^{\ell \ell'}$ has a combinatorial factor $2$, due to the symmetry of the $4$-valent vertex. Since moreover $H_4^{\ell \ell'}$ and $H_4^{\ell' \ell}$ renormalize the same $4$-valent interaction, we have to take an additional factor $2$ into account. The numerical coefficient resulting from the contraction operation is $S_{3,i}$, just like the $6$-point graph previously computed. In the approximation of $f_1$ and $f_2$ we use, we therefore have:
\beq
\frac{ t_{4, i -1}}{2} \approx \frac{ t_{4, i }^{phys}}{2}  - 2 \times 2 S_{3,i} \, \frac{ t_{6,2, i }^{phys}}{3} \frac{ t_{4, i }^{phys}}{2} +  3 S_{1,i} \, \frac{ t_{6,1, i }^{phys}}{3} + 2 S_{1,i} \, t_{6,2, i }^{phys} + \; \ldots \;,
\eeq  
and hence:
\bes
f_1 ( t_{6,1,i}^{phys} , t_{6,2,i}^{phys} ) &=& - \frac{4}{3} S_{3,i} \, t_{6,2, i }^{phys} + \; \ldots \; + \cO ( t^2 ) \,,\\
f_2 ( t_{6,1,i}^{phys} , t_{6,2,i}^{phys} ) &=& 2 S_{1,i} \left( t_{6,1, i }^{phys} +  2 t_{6,2, i }^{phys} \right) + \; \ldots \; + \cO ( t^2 ) \,. 
\ees


\begin{figure}[h]
  \centering
  \subfloat[$I_{6,1}^{\ell}$]{\label{I61l}\includegraphics[scale=0.6]{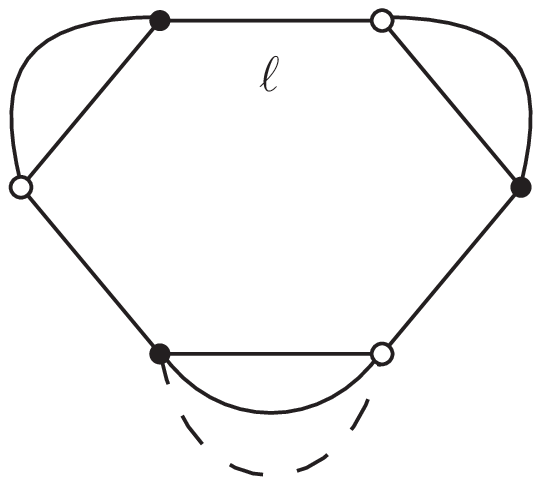}}
  \subfloat[$I_{6,2}^{\ell \ell'}$]{\label{I62ll}\includegraphics[scale=0.6]{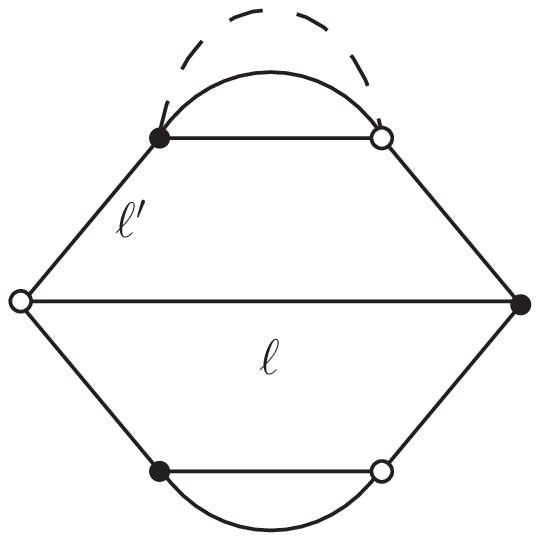}}
  \subfloat[$H_{4}^{\ell \ell'}$]{\label{h4ll}\includegraphics[scale=0.6]{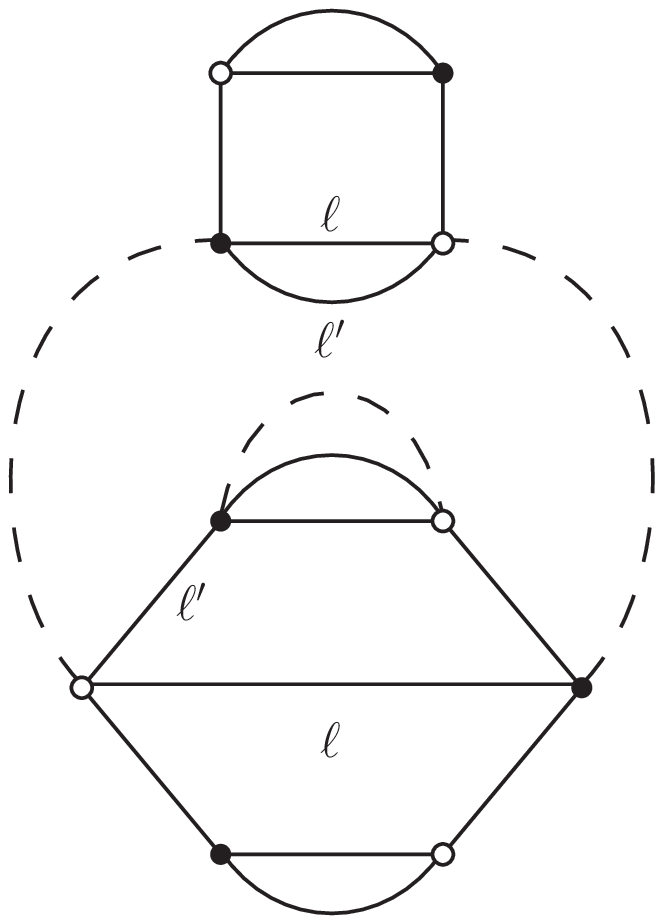}}
  \caption{Divergent graphs contributing to the $4$-point function, up to second order.}\label{int4}
\end{figure}

\subsection{Physical coupling constants: towards asymptotic freedom}





We can now easily deduce flow equations for the physical coupling constants, following (\ref{rescale_couplings}). In order to simplify this analysis, we will from now on make the assumption that the exponential factor in equation (\ref{zi}) can be neglected. Hence we will work under the following hypothesis, whose consistency must be checked at the end of the analysis.

\begin{hypothesis}\label{hyp_mass}
The flow of the mass is such that in the asymptotic UV region:
\beq
\vert \overline{m}_i^2 \vert \equiv \vert \frac{m_{phys,i}^2}{M^{2i}} \vert \ll 1\,.
\eeq
\end{hypothesis}

We can thus use the approximations
\beq
t_{6, 1 , i -1 }^{phys} \approx t_{6, 1 , i -1} ( 1 - 3 CT_{\vphi , i -1}) \;; \qquad t_{6, 2 , i -1 }^{phys} \approx t_{6, 2 , i -1} ( 1 - 3 CT_{\vphi , i -1 })\,,
\eeq
and immediately deduce the flow equations for the $6$-point coupling constants:
\bes
t_{ 6 , 1 , i - 1}^{phys} &=& t_{ 6 , 1 , i}^{phys} \left( 1 + \left[ 2 \tilde{S}_{2 , i}  - 6  S_{3 , i}  \right] t_{ 6 , 2 , i}^{phys} \right) + \; \ldots \; + \cO ( t^3 ) \,, \\
t_{ 6 , 2 , i - 1}^{phys} &=& t_{ 6 , 2 , i}^{phys} \left( 1 + \left[ 2  \tilde{S}_{2 , i} - 4 S_{3 , i}  \right] {t_{ 6 , 2 , i}^{phys}} \right) + \; \ldots \; + \cO ( t^3 ) \,.
\ees 
We define two $\beta$-coefficients, which are the analogues of the $\beta$-functions in our discrete setting:
\bes
\beta_{ 6 , 1 , i} &=& 2 \tilde{S}_{2 , i}  - 6  S_{3 , i} + \; \ldots \;, \\
\beta_{ 6 , 2 , i} &=& 2 \tilde{S}_{2 , i}  - 4  S_{3 , i}  + \; \ldots \;.
\ees
As far as asymptotic freedom is concerned, the crux of the matter is to determine the signs of the $\beta$-coefficients in the deep UV: if they converge to a finite and strictly positive value when $i \to + \infty$, then the $6$-point coupling constants run to $0$ in the UV direction. These coefficient can be evaluated through a Laplace approximation \cite{beta_su2}, and in our approximation it turns out that $\tilde{S}_{2 , i} - 3 S_{3 , i}$ is strictly positive. Hence, if vertex-connected unions of face-connected divergent subgraphs do determine the nature of the flow, the two $6$-point coupling constants necessarily converge to $0$ in the UV. This encouraging indication needs to be confirmed by a complete analysis, which will be reported on in a future publication \cite{beta_su2}.

\

We now turn to the $4$-valent interactions. We have:
\beq
t_{4,i-1}^{phys} \approx t_{4,i-1} \left( 1 - 2 CT_{\vphi, i-1} \right) \,, 
\eeq
from which we deduce
\bes
t_{4,i-1}^{phys} &=& \left( 1 + F_1 ( t_{6,1,i}^{phys} , t_{6,2,i}^{phys} ) \right) \, t_{4,i}^{phys} + F_2 ( t_{6,1,i}^{phys} , t_{6,2,i}^{phys} )\,, \\
F_1 ( t_{6,1,i }^{phys} , t_{6,2,i}^{phys} ) &=&  \frac{4}{3} \left( \tilde{S}_{2 , i} - S_{3,i} \right) \, t_{6,2, i }^{phys} + \; \ldots \; + \cO ( t^2 ) \,,\\
F_2 ( t_{6,1,i }^{phys} , t_{6,2,i}^{phys} ) &=& 2 S_{1,i} \left( t_{6,1, i }^{phys} +  2 t_{6,2, i }^{phys} \right) + \; \ldots \; +  \cO ( t^2 ) \,. 
\ees
According to Table \ref{div}, $F_1$ can only contain $\omega = 0$ terms, while $F_2$ also has linearly divergent counter-terms. In order to correctly deal with these divergences, we can therefore define a 'dimensionless' coupling constant
\beq
\overline{t}_{4, i} \equiv \frac{t_{4, i}^{phys}}{M^{i}}\,.
\eeq 
The flow equation then writes:
\beq
M^{-1} \overline{t}_{4,i-1} = \overline{t}_{4,i} \left( 1 + \frac{4}{3} \left( \tilde{S}_{2 , i} - S_{3,i} \right) \, t_{6,2, i }^{phys} + \cO ( t^2 ) \right) + 2 \frac{S_{1,i}}{M^i} \left( t_{6,1, i }^{phys} +  2 t_{6,2, i }^{phys} \right) + \; \ldots \; + \cO ( t^2 ) \,, 
\eeq
where $\cO (t^2)$ is to be understood in the sense of $t_{6,1,i}^{phys}$ and $t_{6,2,i}^{phys}$ only. 
When the latter are both negligible, one simply has
\beq
\overline{t}_{4,i} = M^{-1} \overline{t}_{4,i-1} \; \Rightarrow \; \overline{t}_{4,i} \sim K \, M^{-i}\,. 
\eeq
Hence, provided that the $t_{6,1,i}^{phys}$ and $t_{6,2,i}^{phys}$ decay to $0$ in the UV, so do $t_{4,i}^{phys}$.

\subsection{Mass and consistency of the assumptions}




We now have to check that, within the asymptotic freedom conjecture, the hypothesis \ref{hyp_mass} is self-consistent. Equations (\ref{mass_phys}) and (\ref{mass_flow}) allow to compute the physical mass flow up to second order in the coupling constants:
\bes
m_{phys,i-1}^2 &=& m_{phys,i}^2 \left( 1 + \frac{2}{3} \tilde{S}_{2,i} \, t_{6,2,i}^{phys} \right) + 3 {S_{1,i}} \, t_{4,i}^{phys} + 3 {S_{1,i}}^2 \, t_{6,1,i}^{phys} \\
&& + \; \left( 6 S_{2,i} + 3 {S_{1,i}}^2 \right) t_{6,2,i}^{phys} + \; \ldots \; + \cO( t^2 )\,.
\ees

\
The assumptions we made concerns the rescaled mass $\overline{m}_i$, so one should divide this equation by $M^{2 i}$. In the process, we can discard the terms with vanishing coefficients in the UV and write: 
\beq
M^{-2}\, \overline{m}_{i-1}^2 = \overline{m}_{i}^2 \left( 1 + \frac{2}{3} \tilde{S}_{2,i} \, t_{6,2,i}^{phys} \right) + 3 \frac{{S_{1,i}}^2}{M^{2 i}} \, t_{6,1,i}^{phys}  + \frac{6 S_{2,i} + 3 {S_{1,i}}^2}{M^{2i}} t_{6,2,i}^{phys} + \; \ldots \; + \cO( t^2 )\,.
\eeq
We can use the same idea as for the $4$-point interactions and conclude that, due to the $M^{-2}$ factor, $\overline{m}_{i-1}^2$ must go to $0$ in the UV if the $6$-point coupling constants do so. Therefore our approximations are consistent with asymptotic freedom.

\

This partial analysis leads us to the following conjecture, currently under investigation:
\begin{conjecture}
The $\SU(2)$ rank-$3$ TGFT studied in this chapter is \text{asymptotically free}.
\end{conjecture}

\chapter{Conclusions and perspectives}

We are finally reaching the conclusion of this manuscript. The results we have been able to gather concern two closely related topics: the $1/N$ expansion in colored GFTs, and renormalization theory for TGFTs. We briefly recall what has been achieved in these two respects below, and discuss some possible avenues to explore in the future. 

\section{The $1/N$ expansion in colored GFTs}

As recalled in Chapter \ref{color_tensor}, research in GFTs has been recently boosted by the introduction of colored models, which drastically simplify the combinatorial structure of the $2$-complexes generated in perturbative expansion. These modified models significantly gained support after the discovery of the $1/N$ expansion, but are also strongly motivated by independent topological considerations, and to a lesser extent (i.e. in three dimensions only) by discrete gravity symmetries. Particular attention has been given to even simpler theories, colored tensor models, which can be considered as the back bone of more involved GFTs. As far as the emergence of continuum gravity is concerned, determining which ingredients of these refined models can imprint the large $N$ behavior (if any) is of primary importance. Tensor models provide us with universal tools to explore further the consequences of GFTs, but can also be considered as purely combinatorial challengers. Repeating the successes of the latter with the simplest GFTs implementing additional discrete geometric data is a primary goal in this respect: not only as a way to develop the GFT formalism itself and bridge the gap with 4d quantum gravity models; but also to compare the merits of the two approaches.  

\subsection{Achievements}

In this thesis, we focused on the simplest GFTs whose amplitudes can be given a lattice gauge theory interpretation, namely the Boulatov and Ooguri models. The existence of a $1/N$ expansion for their colored versions has been established alongside that of colored tensor models by new and powerful methods. In Chapter \ref{largeN}, which is based on papers in collaboration with Daniele Oriti \cite{vertex, edge}, these results were revisited and strengthened, thanks to representations better adapted to the symmetries of discrete $BF$ theory.

\
Taking advantage of the metric formulation of GFTs, in non-commutative variables, we were able to analyze in details the vertex representation of the colored Boulatov model, and the edge representation of the colored Ooguri model. The idea, initially put forward in \cite{ef_poincare} and \cite{diffeos}, is to take the generators of translation symmetries as elementary variables. These translation symmetries are associated to the vertices of the simplicial complexes in 3d, and to their edges in 4d. They are deformed non-commutative symmetries, therefore they require particular care: for instance, the change of variables we presented relies on a non-commutative group Fourier transform. 

\
We then put these two reformulations to good use, and derived scaling bounds in the cut-off parameter. This cut-off, implemented with a heat kernel, is what plays the role of large $N$ parameter, hence the relevance of these scaling bounds to the $1/N$ expansion. 

The first bounds we focused on, the bubble bounds, have no counterpart in the colored tensor literature. It is only thanks to the vertex formulation in $d=3$ and the edge formulation in $d=4$ that we could access the topological information encoded in the $d$-bubbles. The $d$-bubbles are particularly interesting, due to the occurrence of topological singularities, even in colored GFTs. In this respect, we could prove that all the contributions to the free energy are topological manifolds in the first few orders of the $1/N$ expansion. More precisely, for each color $\ell$, we could define a positive quantity which so to speak captures the "degree of manifoldness" of a given configuration: it is $0$ when all the bubbles of color $\ell$ are spherical, and gets bigger and bigger the more non-spherical bubbles there are, or the more topologically (and combinatorially in 4d) non-trivial they are. This quantity is what governs the decay in $N$ of the bubble bounds. Interestingly, it is given by the sum of genera of the bubbles in 3d, and the sum of their degrees (taken from the tensor models) in 4d, hence supporting the idea that the degree is the (non-topological) generalization of genus one should use in GFTs. 

The other bounds we could derive are expressed in terms of the jackets, which are particular discretized surfaces embedded in the simplicial complexes. They notably enter the definition of the degree of a colored graph: it is simply the sum of the genera of all the jackets. They are therefore at the basis of the original derivation of the $1/N$ expansion. The topological information about the jackets could also be recovered with our methods, though in a less natural fashion than for bubbles. What is very interesting however, is that our jacket bounds improve the original ones: the decay we could obtain is governed by the maximum of all the genera of the jackets, rather than their average. This subtlety does not play any role at leading order, since again degree $0$ graphs (hence melonic \cite{critical}) dominate the large $N$ behavior of the colored Boulatov and Ooguri models, but is suggestive of important deviations with respect to tensor models at subleading orders. 

\subsection{Discussion and outlook}





The first piece of information we can take out from this study is that, not very surprisingly, recasting GFT models in forms better adapted to their symmetries gives easier access to their properties. In particular, the bubble and improved jacket bounds could not be anticipated with the usual formulation of the Boulatov-Ooguri models. They could not even be deduced a posteriori from the power-counting of \cite{vm1, vm2, vm3} which, though exact, does not easily give access to the topological information we are primarily interested in. Second, our original expectation was that it would be easier to pin-point differences between colored tensor models and more complicated GFTs by focusing on the properties only possessed by the latter. And we see indeed that, thanks to a change of variable which is only possible in the presence of a gauge invariance condition, bounds which are not satisfied by the amplitudes of the tensor models could be derived. This analysis should in our view be pursued further. 

At present, only the melonic sector of colored tensor models has been analyzed in detail, revealing the existence of a critical point at which a branched-polymer continuum phase is reached. We know that, likewise, the leading order sector of the Boulatov-Ooguri models is populated by the melonic graphs, but we should not forget that they are weighted differently. The key question to elucidate is whether these additional weights are able to alter the critical regime obtained in the purely combinatorial case. A first step in this direction has been taken in \cite{boulatov_phase}, in collaboration with Aristide Baratin, Daniele Oriti, James Ryan and Matteo Smerlak. An exact expression of the melonic amplitudes in the large $N$ limit could be derived, including finite pre-factors. These are evaluated thanks to a generalization of the tree-matrix theorem, and reduce to a counting problem for $2$-trees, which are to $2$-complexes what trees are to graphs. These formula, associated with simple bounds, allow to prove that the free energy of these models has a finite radius of analyticity, hence confirming the existence of a melonic phase transition. A detailed account of its properties remains challenging, but will hopefully be facilitated by the exact formula proven in \cite{boulatov_phase}, giving a new opportunity to find key differences with tensor models.

But even if it turned out that the Boulatov-Ooguri models also describe a crumpled phase after the melonic phase transition, deviations from the tensor case could occur at higher orders. In particular, the recent double-scaling results for colored tensor models \cite{wjd_double, Dartois, schaeffer} open the way to a similar achievement in GFTs. And according to our refined jacket bounds, a double scaling for the Boulatov-Ooguri might possibly retain different classes of graphs than in the tensor case. 

\

The scaling bounds presented in this thesis could and should be improved in several respects. First, we only considered vacuum graphs, and the question arises as to how the presence of boundaries would affect the $1/N$ expansion. In the long run, the goal is to generalize the $1/N$ expansion to 4d quantum gravity models. The scaling behavior of such theories is only crudely understood, therefore adopting a similar strategy as the one used in this thesis is an intriguing possibility. In particular, we might want to investigate how the edge formulation of the Ooguri model would be affected by the imposition of the simplicity constraints. Heuristically, their role is to make the $B$ field of $BF$ theory geometrical, that is to turn it into a wedge product of discrete triads. It is therefore tempting to conjecture that upon imposition of the simplicity constraints, the Ooguri $\so(4)$ variables associated to the edges would be turned into geometrical $\mathbb{R}^4$ vectors. If so, one should then determine whether a part of the broken symmetry under (now geometric) edge translations survives. A possibility would be the presence of a vertex translation invariance, expressed as simultaneous translations of some of the edges. If so, we could look for a further change of variables leading to a vertex formulation of 4d quantum gravity models. It remains to be seen whether this very optimistic scenario can be realized, but if so we could greatly benefit from the experience gained in this thesis to develop a $1/N$ expansion for these physically relevant theories. 

\

Finally, a third idea one might want to explore further is the use of geometric symmetries, such as the vertex translations of the Boulatov model, as defining features of GFTs. Just like the Poincaré symmetry allows to classify the possible interactions in local relativistic field theories, we could use the invariance under vertex translations as an avenue towards a construction of the Boulatov model or a generalization thereof from "first principles". A key obstacle however lies in the deformed character of these symmetries, which suggest the implementation of a non-trivial braiding. Despite several attempts, by several people, such a braiding could not be found to date, precluding progress in this direction. With insights from the second part of this thesis, we would suggest to reconsider the same strategy in the TGFT context, where a general locality principle is also available. What is the fate of the translation symmetries in tensor invariant versions of the Boulatov and Ooguri models? Can they be given a geometric interpretation in 3d? 

\section{Renormalization of TGFTs}

Let us now turn to renormalization. Again, the main innovation of this thesis has been to incorporate the gauge invariance condition of spin foam models into the already existing renormalization scheme for tensor fields, and is thought of as a first step towards 4d quantum gravity. This part of the thesis is the outcome of a collaboration with Daniele Oriti and Vincent Rivasseau.

\subsection{Achievements}

In Chapter \ref{renormalization}, we introduced the particular class of tensorial group field theories we focused on, which are characterized by: a) an infinite set of interactions, labeled by colored bubbles, based on the tensorial locality principle; b) non-trivial propagators implementing a gauge invariance condition on the fields, supplemented with a Laplace operator which softly breaks the tensorial invariance of the interaction. The first ingredient is directly imported from the uncolored tensor models, and can also arise from the integration of fields in specific colored group field theories based on simplicial interactions. The gauge invariance condition turns the Feynman amplitudes into lattice gauge theories and is one of the two main ingredients of group field theories for gravity (the other being, in 4d, the so-called simplicity constraints). The Laplace operator launches the renormalization group flow, and is also partially motivated by the analysis of the radiative corrections of simpler GFTs, with ultra-local propagators. The rank $d$ of the tensors, as well as the dimension $D$ of the compact group indexing the tensors, were in a first stage kept arbitrary. A detailed analysis of the power-counting of such models allowed to derive stringent restrictions on $d$ and $D$ in order to achieve renormalizability. In particular, it was shown that only five combinations of such parameters can potentially support (interacting) just-renormalizable models. Among these, only $(d , D) = (3 , 3)$ can be directly related to a space-time theory, namely topological BF theory or 3d quantum gravity, with $G$ the symmetry group for Lorentzian or Riemannian spaces of dimension $d$. In particular, the case $(4,6)$, that would include a TGFT version of the Ooguri model, is found to be non-renormalizable. The divergences of the renormalizable examples are mostly due to melonic subgraphs, in a suitably generalized sense adapted to TGFTs.   

In Chapter \ref{chap:u1}, we explored first examples in details, namely Abelian $\U(1)$ models in $d = 4$. Even if they have no geometric interpretation, they gave us the opportunity to understand better the general formalism. In particular, we could define a generalization of the Wick ordering procedure, called melordering, which fully renormalizes such super-renormalizable theories. It correctly takes the nested structure of the divergent melonic tadpoles (the melopoles) into account, and contracts them to their local part. It should especially be noted that, contrary to usual local field theories, tadpoles are not automatically local, and not automatically quasi-local either. This is one instance of the difficulties introduced by the new notion of quasi-locality, forced upon us by tensor invariance, which we called traciality. With this tool at hand, we could rigorously prove the finiteness of the renormalized amplitudes.  

In the last chapter, we went on to study in details the only just-renormalizable model with a geometric interpretation, in the Riemannian case $G = \SU(2)$. In order to classify the divergences, proven to be all melonic, we generalized the multiscale techniques to the non-Abelian case. The tensorial interactions were shown to be renormalizable up to order $6$, and to generate up to quadratically divergent subgraphs. The same multiscale techniques could then be used to reabsorb divergences into tensorial effective coupling constants, as well as wave-function counter-terms, thus defining renormalized amplitudes. Computed as sums over particular types of Zimmermann forests, they could finally be proven finite at all orders of perturbation, which is the main result of Chapter \ref{chap:su2}. Additionally, divergent forests were found to be unexpectedly rigid in their structure, which helped simplifying some aspects of the proof of renormalizability. We finally briefly sketched the Wilsonian effective expansion, taking the wave-function renormalization into account. A complete computation of the flow equations, truncated to their first non-vanishing contributions, should allow to determine whether this theory is asymptotically free, as is expected from other TGFT models. The partial study we presented supports this conjecture.

\subsection{Discussion and outlook}

The present study provides a few lessons which in our opinion will have to be kept in mind in the construction and renormalization analysis of more elaborate models, in particular models for 4d quantum gravity. First, concerning TGFTs \textit{per se}, the message we would like to convey is that, in order to efficiently index the divergences, the most appropriate notion of connectedness is face-connectedness rather than vertex-connectedness. This is particularly true in models implementing the gauge invariance condition, in which the amplitudes are functions of holonomies around faces. The natural coarse-graining procedure in this situation is indeed to erase "high energy" faces rather than internal lines, and this can be consistently implemented in what we called tracial subgraphs. Face-connectedness also crucially enters the power-counting theorem of such models, through the rank of the incidence matrix between faces and lines of a given graph. While a proof of renormalizability can certainly be achieved with a notion of vertex-connected high divergent subgraphs only, the face-connected high divergent subgraphs we relied on in this thesis capture the fine structure of the divergences, and henceforth avoid many redundancies in the renormalization. This is exemplified by the fact that divergent subgraphs in the sense of face-connectedness do not overlap (in non-vacuum graphs), but rather organize themselves into a strong inclusion forest. Had we worked in the coarser vertex-connectedness picture, overlapping divergences would have been generic, and redundant counter-terms would have been introduced. An intriguing question to ask, in this respect, is whether face-connectedness might prove more fundamental in simpler TGFTs as well, for example in the original model \cite{tensor_4d}, where no connection degrees of freedom are introduced.  

\

We also pointed out clear limitations of face-connectedness, notably so from the effective Wilsonian point of view. This prevented us from settling down the question of asymptotic freedom in the $\SU(2)$ model of Chapter \ref{chap:su2}, which will require more work. In order to make such studies more tractable, it seems to us that the analogy between the face-connectedness condition in TGFT and $1$-particle irreducibility in ordinary quantum field theory should be investigated further.  

\

Let us now turn to the hard question of the renormalizability of quantum gravity models in four dimensions. We do not have any definitive statement to make on this issue, since, as we explained earlier, the analysis presented here does not immediately generalizes to models involving simplicity constraints. The latter GFTs, beside the fact that they are not in the class of models considered in this thesis, are also not based on group manifolds as such, but rather submanifolds of the Lorentz group. Still, it seems to us that the three dimensional $\SU(2)$ model studied in this thesis suggests to reconsider and improve the current spin foam models for quantum gravity in two essential ways, before attempting any complete study of renormalizability. The first concerns the much debated nature of scales in such models, and the definition of non-trivial propagators which decay in the UV. In particular, we think that the results of the present thesis suggest  that, in general, semi-classical reasoning interpreting the large-$j$ limit of spin foam models (the Feynman amplitudes of GFTs) as the IR general relativistic limit should be taken with care. Indeed, perturbative divergences being associated to the same large-$j$ sector (or equivalently small Schwinger parameter $\alpha$), it actually plays the role of the UV in our TGFT setting. In the Wilsonian point of view, it is therefore only for boundary states with scales much lower than the cut-off that the theory retains some predictive power. This points in the direction of large boundary geometries having to be constructed as collections of many small cells rather than a few big ones. And in practice, this means that one will have to address the question of approximate effective schemes, in order to control such regimes with large numbers of particles. An intriguing possibility would be the occurrence of one or several phase transitions along the renormalization flow. This scenario might already start to be tested in the three dimensional case, the first step being the computation of $\beta$-functions. In any case, we need to understand how to choose non-trivial kernels for propagators in four dimensional models. While there are some hints \cite{ValentinJoseph} that the Laplace-Beltrami operator is naturally generated by the quantum dynamics, when no simplicity constraints are imposed, whether the same is true in the presence of simplicity constraints is unclear at the present. We believe that this is the first open question to address as far as the renormalizability of TGFTs for four dimensional quantum gravity is concerned. The second point, which deserves similar attention, is the possible interplay between tensorial invariance and simplicity constraints, as it is not immediately clear whether the geometric meaning and motivations for such constraints, as well as the details of their implementation, straightforwardly generalize to bubble interactions. If these two important questions can be elucidated, one might try to apply the techniques used in this thesis to determine whether four dimensional TGFT models for quantum gravity with such simplicity constraints are renormalizable or not, and up to which order of interactions. 


\appendix
\chapter{Technical appendix}
\section{Heat Kernel}\label{app_heat}

Consider the $S_3$ representation of $SU(2)$, the identity $\one$ being at the north pole,
$H_0 = S_2$ being the equator and $- \one$ being the south pole. The north and south 
open hemispheres are noted respectively as $H_N$ and $H_S$. 

The heat kernel between two points $g$ and $g'$ is:
\beq\label{heat_k}
K_\alpha(g,g') = \sum_{j \in {\mathbb N}/2}(2j + 1)e^{-j(j+1)\alpha} \frac{\sin((2j + 1) \psi(g'g^{-1})}{\sin \psi(g'g^{-1})} \,,
\eeq
where $\psi(g) \in [0, \pi]$ is the class angle of $g \in SU(2)$, which is 0 at $\one$
and $\pi$ at $-\one$.
It is also the sum over Brownian paths in $SU(2)$ from $g$ to $g'$
$$
K_\alpha(g,g') = \int \extd P_{\alpha} (g,g')[\omega]
$$
where $\extd P_{\alpha} (g,g')[\omega]$ is the Wiener measure over Brownian paths $\omega$ going from
$g$ to $g'$ in time $\alpha$.

The northern heat kernel with Dirichlet boundary conditions, called
$K^{N,D}_\alpha(g,g')$  is the same integral, but in which the 
Brownian paths are constrained to lie entirely in $H_N$, except
possibly their end points $g$ and $g'$, which  are allowed to belong to the {\emph{closed}}
hemisphere $\bar H_N$. Obviously:
\beq\label{dirichlet}
K^{N,D}_\alpha(g,g')  \le K_\alpha(g,g')
\eeq
since there are less paths in the left hand side than in the right hand side.

From the Markovian character of the heat kernel 
$K_\alpha$ we have a convolution equation for 
with $g \in H_S$, in terms of the first hitting point $g'$
where the path visits the equatorial boundary: 

\beq\label{conv}
K_\alpha  (I, g)  = \int_0^{\alpha} \extd \alpha' \int_{g' \in H_0}  \extd g' K^{N,D}_{\alpha'} (I, g') K_{\alpha- \alpha'} (g', g) .
\eeq

\section{Proof of heat kernel bounds}

In order to prove lemma \ref{heat}, we first re-express the heat kernel on $\SU(2)$ in terms of the third Jacobi $\theta$-function
\beq
\theta_3 (z , t) = 1 + 2 \sum_{n = 1}^{+ \infty} \e^{{\rm i} \pi n^2 t} \cos (2 \pi n z)\,,
\eeq
defined for any $(z , t) \in \mathbb{C} \times \mathbb{R}$. We note $\theta_3'$ its derivative with respect to $z$:
\beq
\theta_3' (z , t) = - 4 \pi \sum_{n = 1}^{+ \infty} n \e^{{\rm i} \pi n^2 t} \sin (2 \pi n z)\,.
\eeq 
From equation (\ref{heat_k}), we deduce that:
\beq
K_{\alpha} (g) = \frac{- \e^{\alpha / 4}}{4 \pi \sin \psi(g)} \theta_3' \left( \frac{\psi(g)}{2 \pi} , \frac{i \alpha}{4 \pi} \right)\,. 
\eeq
The main interest of this expression is that $\theta_3$ transforms nicely under the modular group, and in particular\footnote{This is a consequence of the Poisson summation formula, so one might as well directly use this theorem instead of introducing $\theta$.}:
\beq
\theta_3 ( \frac{z}{t} , \frac{- 1}{t}) = \sqrt{- {\rm i} t } \e^{\frac{{\rm i} \pi z^2 }{t}} \theta_3 ( z , t)\,.
\eeq
Differentiation with respect to $z$ yields:
\beq
\theta_3' (z , t) = \frac{\e^{- \frac{{\rm i} \pi z^2 }{t}}}{t \sqrt{- {\rm i} t }} \left( \theta_3' \left( \frac{z}{t} , \frac{- 1}{t} \right) - 2 \pi {\rm i} z \theta_3 \left( \frac{z}{t} , \frac{- 1}{t} \right) \right)\,,
\eeq
and allows to express the heat kernel as
\beq
K_{\alpha} (g) = \frac{\e^{- \frac{ {\psi(g)}^2 }{\alpha}}}{ \alpha^{3/2} } 
\times \frac{ (4 \pi)^{1/2} {\rm i} \e^{\alpha / 4}}{ \sin \psi(g)} \left( \theta_3' \left( \frac{ - 2 {\rm i} \psi(g) }{\alpha} , \frac{4 {\rm i} \pi}{\alpha} \right) - {\rm i} \psi(g) \theta_3 \left(  \frac{ - 2 {\rm i} \psi(g) }{\alpha} , \frac{4 {\rm i} \pi}{\alpha} \right) \right) \,.
\eeq

Using the explicit expressions of $\theta_3$ and $\theta_3'$, we finally obtain:

\beq\label{nice_form}
K_{\alpha} (g) = K_{\alpha}^{0} (g) \frac{ \sqrt{4 \pi} \e^{\alpha / 4} \psi(g)}{\sin{\psi(g)}} F_\alpha (\psi(g))\,,
\eeq
where
\bes
K_{\alpha}^{0} (g) &\equiv& \frac{\e^{- \frac{ {\psi(g)}^2 }{\alpha}}}{ \alpha^{3/2} } \,, \\
F_{\alpha} (g) &\equiv&  1 + \sum_{n = 1}^{+ \infty} \e^{- 4 \pi^2 n^2 / \alpha} 
\left( 2 \cosh( \frac{4 \pi n \psi(g)}{\alpha} ) - \frac{4 \pi n}{\psi(g)} \sinh( \frac{4 \pi n \psi(g)}{\alpha}
) \right) \,.
\ees

\

This formula is suitable for investigating the behavior of $K_\alpha$ away from $- \one$. In particular, simple integral bounds on $F_\alpha$ allow to prove that:
\beq
K_\alpha (g) \underset{\alpha \to 0}{\sim} \frac{\e^{- \frac{ {\psi(g)}^2 }{\alpha}}}{ \alpha^{3/2} } 
\frac{ \sqrt{4 \pi} \psi(g)}{\sin{\psi(g)}} 
\eeq
uniformly on any compact $H$ such that $- \one \notin H$. We shall therefore first study the behavior of $K_\alpha$ and its derivatives on the fixed compact $H_{\frac{3 \pi}{4}} = \{ g \in \SU(2) \vert \psi(g) \leq \frac{3 \pi}{4}\}$. Relying on convolution properties of the heat kernel, we will then extend these results to all of $\SU(2)$. 

\

\noindent {\bf Bounds on $H_{\frac{3 \pi}{4}}$}

\
$K_{\alpha}^{0}$ is easy to analyze, as it is nothing but the flat version of $K_{\alpha}$. The function $\psi \mapsto \frac{\psi(g)}{\sin \psi(g)}$ is analytic on $H_{\frac{3 \pi}{4}}$, therefore its contributions to $K_\alpha$ and its derivatives will be uniformly bounded. The non trivial point of the proof consists in proving that $F_\alpha$ and all its derivatives are also uniformly bounded, by constants independent of $\alpha \in ] 0 , 1]$. By expanding the hyperbolic functions, we can first write:
\bes
F_\alpha (\psi) &=& 1 + \sum_{n = 1}^{+ \infty} \e^{- 4 \pi^2 n^2 / \alpha}
	\sum_{p = 0}^{\infty} a_p (n , \alpha) \psi^{2 p} ,  \\
a_p (n , \alpha) &\equiv& \frac{2}{(2 p) !} \left( \frac{4 \pi n}{ \alpha } \right)^{2 p} \left[ 1 - \frac{16 \pi^2 n^2}{(2 p + 1) \alpha} \right]	.
\ees
We can fix $0 < \epsilon < 1$, and find a constant $K_\epsilon$ such that:
\beq
\e^{- 4 \pi^2 n^2 / \alpha} |a_p (n , \alpha ) | \leq K_\epsilon \e^{- 4 \pi^2 n^2 (1 - \epsilon) / \alpha} \frac{2}{(2 p) !} \left( \frac{4 \pi n}{ \alpha } \right)^{2 p} \,.
\eeq
This implies the following bounds, for any $k \in \mathbb{N}$:
\beq
\vert F_\alpha^{(k)} (\psi) \vert \leq \frac{\partial^{k}}{\partial \psi^k} \left( 1 + 2 K_\epsilon \sum_{n = 1}^{+ \infty} \e^{- 4 \pi^2 n^2 (1 - \epsilon) / \alpha} \cosh ( \frac{4 \pi n \psi}{\alpha} ) \right) \,.
\eeq
When $n \geq 1$, we can use the fact that
\beq
\frac{\partial^{k}}{\partial \psi^k} \cosh ( \frac{4 \pi n \psi}{\alpha} ) \leq  ( \frac{4 \pi n}{\alpha} )^{k} \cosh ( \frac{4 \pi n \psi}{\alpha} ) \,,
\eeq
and the exponential decay in $n^2 / \alpha$ to deduce bounds without derivatives. All in all, we see that for any $\epsilon$, we can find constants $K_\epsilon^{(k)}$ such that:
\bes
\vert F_\alpha (\psi) \vert &\leq& 1 + K_\epsilon^{(0)} \sum_{n = 1}^{+ \infty} \e^{- 4 \pi^2 n^2 (1 - \epsilon) / \alpha} \cosh ( \frac{4 \pi n \psi}{\alpha} ) \,, \\
\vert F_\alpha^{(k)} (\psi) \vert &\leq& K_\epsilon^{(k)} \alpha^{\frac{- k}{2}} \sum_{n = 1}^{+ \infty} \e^{- 4 \pi^2 n^2 (1 - \epsilon) / \alpha} \cosh ( \frac{4 \pi n \psi}{\alpha} ) \,.
\ees
 Following \cite{tensor_4d}, let us assume that $\epsilon \leq \frac{5 \pi}{8}$, in order to ensure that the function $$
x \mapsto \e^{- 4 (1 - \epsilon) \pi^2 x^2 / \alpha} \cosh ( \frac{4 \pi x \psi} {\alpha} ) 
$$
decreases on $[1, + \infty [$ for any $\psi \in [ 0 , \frac{3 \pi}{4} ]$. This provides us with the following integral bound:
\beq
\sum_{n = 1}^{+ \infty} \e^{- 4 (1 - \epsilon) \pi^2 n^2 / \alpha} \cosh ( \frac{4 \pi n \psi}{\alpha} ) \leq \int_{1}^{ + \infty} \e^{- 4 (1 - \epsilon) \pi^2 x^2 / \alpha} \cosh ( \frac{4 \pi x \psi} {\alpha} ) \extd x \,.
\eeq
Putting the latter in Gaussian form yields an expression in terms of the error function $\rm{erfc}(x) \equiv \int_{x}^{+ \infty} \e^{- t^2} \extd t \leq \e^{- x^2}$:
\bes
\int_{1}^{ + \infty} \e^{- 4 (1 - \epsilon) \pi^2 x^2 / \alpha} \cosh ( \frac{4 \pi x \psi} {\alpha} ) \extd x 
&=& \frac{ \e^{\frac{\psi^2}{(1 - \epsilon) \alpha}}\sqrt{ \pi \alpha}}{8 \pi \sqrt{1 - \epsilon}} \left[ \rm{erfc}\left( 2 \pi \sqrt{\frac{1 - \epsilon}{\alpha}}\pi + \frac{\psi}{\sqrt{(1 - \epsilon) \alpha}}\right) \right. \nn \\
&& \qquad \left. + \; \rm{erfc}\left( 2 \pi \sqrt{\frac{1 - \epsilon}{\alpha}}\pi - \frac{\psi}{\sqrt{(1 - \epsilon) \alpha}}\right) \right] \\
&\leq& \frac{ \e^{\frac{\psi^2}{(1 - \epsilon) \alpha}}\sqrt{ \pi \alpha}}{8 \pi \sqrt{1 - \epsilon}} \left[ \e^{- (2 \pi \sqrt{\frac{1 - \epsilon}{\alpha}}\pi + \frac{\psi}{\sqrt{(1 - \epsilon) \alpha}})^{2}} \right. \nn \\
&& \qquad \left. + \; \e^{- (2 \pi \sqrt{\frac{1 - \epsilon}{\alpha}}\pi - \frac{\psi}{\sqrt{(1 - \epsilon) \alpha}})^{2}} \right] \\
&\leq& \frac{\sqrt{ \pi \alpha}}{8 \pi \sqrt{1 - \epsilon}} \e^{- 4 \pi^2 \frac{1 - \epsilon}{\alpha}}
 \left[ \e^{\frac{4 \pi \psi}{\alpha}} + \e^{- \frac{4 \pi \psi}{\alpha}} \right]\,.
\ees
The last expression is bounded by a constant independent of $\alpha$ and $\psi \in [0 , \frac{3 \pi}{4}]$ provided that $\epsilon \leq \frac{1}{4}$, which is an admissible choice. This concludes the proof of the existence of constants $K^{(k)}$ such that:
\bes\label{bound_F}
\vert F_\alpha (\psi) \vert &\leq& 1 + K^{(0)} \sqrt{\alpha}  \,, \\
\vert F_\alpha^{(k)} (\psi) \vert &\leq& K^{(k)} \alpha^{\frac{1 - k}{2}} \,,
\ees
on $H_{\frac{3 \pi}{4}}$. Using equation (\ref{nice_form}), it is then easy to prove that when $\alpha \in ] 0 , 1 ]$, $K_\alpha$ verifies the same type of bounds as $K_\alpha^0$ on $H_{\frac{3 \pi}{4}}$, therefore concluding the proof of lemma \ref{heat} on this subset.

%

\

\noindent {\bf Extension to $\SU(2)$}

\

Suppose that $g \in \SU(2) \setminus H_{3 \pi /4}$. We can use formula (\ref{conv}) and (\ref{dirichlet}) to write:
\beq
K_\alpha (g) \leq \int_0^{\alpha} \extd \alpha' \int_{g' \in H_0}  \extd g' K_{\alpha'} (g') K_{\alpha- \alpha'} (g'^{\inv} g) \,.
\eeq
This upper bound involves only heat kernels evaluated in $H_{3 \pi /4}$. Moreover, $K_{\alpha'} (g')$ does not depend on the particular value of $g' \in H_0$, the squared distance to $\one$ of the latter and of $g'^{\inv} g$ being bounded from below by a constant $c > 0$. 

From the discussion above, we know that there exists constants $\delta_1$ and $K_1$ such that:
\beq
K_\alpha (g) \leq K_1 \int_{g' \in H_0} \extd g' \int_0^{\alpha} \extd \alpha'  \frac{ \e^{- \delta_1 \vert g' \vert^2 / \alpha'}}{\alpha'^{3/2}} \frac{\e^{- \delta_1 \vert g'^{\inv} g \vert^2 / (\alpha- \alpha')}}{(\alpha- \alpha')^{3/2}} \,.
\eeq
To take care of the singularities in $\alpha' = 0$ and $\alpha' = \alpha$, we can decompose the integral over $\alpha'$ into two components: from $0$ to $\alpha/2$, and from $\alpha/2$ to $\alpha$. Each of these integrals can then be bounded independently, for instance:
\bes
\int_0^{\alpha / 2} \extd \alpha'  \frac{ \e^{- \delta_1 \vert g' \vert^2 / \alpha'}}{\alpha'^{3/2}} \frac{\e^{- \delta_1 \vert g'^{\inv} g \vert^2 / (\alpha- \alpha')}}{(\alpha- \alpha')^{3/2}}
&\leq& \int_0^{\alpha / 2} \extd \alpha' \frac{ \e^{- \delta_1 c / \alpha'}}{\alpha'^{3/2}} \frac{\e^{- \delta_1 c / \alpha}}{(\alpha / 2)^{3/2}} \\
&=& K_2 \frac{\e^{- \delta_1 c / \alpha}}{\alpha^{3/2}} \,.
\ees
We can bound the second integral in the same way, and therefore conclude that:
\beq\label{ext_bound}
K_{\alpha}(g) \leq K \frac{\e^{- \delta_2 / \alpha}}{\alpha^{3/2}} \leq K \frac{\e^{- \delta \vert g \vert^{2} / \alpha}}{\alpha^{3/2}}
\eeq
for some constants $K > 0$ and $\delta > 0$.

\

We can proceed in a similar way for the derivatives of $K_\alpha$. We fix $k \geq 1$, and a normalized Lie algebra element $X$. From (\ref{conv}), we deduce 
\bes
\vert (\cL_X)^n K_\alpha  (g) \vert  &=& \vert \int_0^{\alpha} \extd \alpha' \int_{H_0}  \extd g' K^{N,D}_{\alpha'} (g') (\cL_X)^k K_{\alpha- \alpha'} (g'^{\inv} g) \vert \\
&\leq& \int_0^{\alpha} \extd \alpha' \int_{H_0}  \extd g' K_{\alpha'} (g') \vert (\cL_X)^k K_{\alpha- \alpha'} (g'^{\inv} g) \vert \\
&\leq&  K_1 \int_{g' \in H_0} \extd g' \int_0^{\alpha} \extd \alpha'  \frac{ \e^{- \delta_1 \vert g' \vert^2 / \alpha'}}{\alpha'^{3/2}} \frac{\e^{- \delta_1 \vert g'^{\inv} g \vert^2 / ( \alpha- \alpha')}}{(\alpha- \alpha')^{(3 + k)/2}}\,,
\ees
for some constants $K_1$ and $\delta_1$. 
The same method as before allows to show that 
\beq
\vert (\cL_X)^k K_\alpha  (g) \vert \leq  K \frac{\e^{- \delta \vert g \vert / \sqrt{\alpha}}}{\alpha^{(3+ k)/2}}\,,
\eeq 
for some constants $K > 0$ and $\delta > 0$.







\bibliographystyle{hunsrt}
\bibliography{biblio}

\begin{thebibliography}{100}

\bibitem{einstein_grossmann1913}
A~Einstein and M~Grossmann.
\newblock Entwurf einer verallgemeinerten relativitätstheorie und eine theorie
  der gravitation.
\newblock {\em Zeitschrift für Mathematik und Physik}, 62:225--261, 1913.

\bibitem{einstein1915}
A~Einstein.
\newblock Entwurf einer verallgemeinerten relativitätstheorie und eine theorie
  der gravitation.
\newblock {\em Preussische Akademie der Wissenschaften, Sitzungsberichte},
  pages 844--847, 1915.

\bibitem{bohr1983}
Niels Bohr.
\newblock Discussion with {E}instein on epistemological problems in atomic
  physics.
\newblock In J.~A. Wheeler and W.~H. Zurek, editors, {\em Quantum Theory and
  Measurement}, chapter I.1. Princeton University Press, Princeton, NJ, 1983.

\bibitem{unruh1984}
WG~Unruh.
\newblock Steps towards a quantum theory of gravity.
\newblock {\em Quantum Theory of Gravity}, 1:234, 1984.

\bibitem{ashtekar_book}
Abhay Ashtekar and Ranjeet~S Tate.
\newblock {\em Lectures on non-perturbative canonical gravity}, volume~6.
\newblock World Scientific Publishing Company Incorporated, 1991.

\bibitem{rovelli_book}
Carlo Rovelli.
\newblock {\em Quantum gravity}.
\newblock Cambridge University Press, 2004.

\bibitem{thiemann_book}
Thomas Thiemann.
\newblock {\em Modern canonical quantum general relativity}.
\newblock Cambridge University Press, 2007.

\bibitem{perez_review2004}
Alejandro Perez.
\newblock {Introduction to loop quantum gravity and spin foams}.
\newblock 2004, gr-qc/0409061.

\bibitem{perez_review2012}
Alejandro Perez.
\newblock {The Spin Foam Approach to Quantum Gravity}.
\newblock {\em Living Rev.Rel.}, 16:3, 2013, 1205.2019.

\bibitem{goroff_sagnotti}
Marc~H Goroff and Augusto Sagnotti.
\newblock The ultraviolet behavior of einstein gravity.
\newblock {\em Nuclear Physics B}, 266(3):709--736, 1986.

\bibitem{smolin_background}
Lee Smolin.
\newblock {The Case for background independence}.
\newblock 2005, hep-th/0507235.

\bibitem{bojowald2001}
Martin Bojowald.
\newblock Absence of a singularity in loop quantum cosmology.
\newblock {\em Physical Review Letters}, 86(23):5227--5230, 2001.

\bibitem{ashtekar2006}
Abhay Ashtekar, Tomasz Pawlowski, and Parampreet Singh.
\newblock {Quantum nature of the big bang}.
\newblock {\em Phys.Rev.Lett.}, 96:141301, 2006, gr-qc/0602086.

\bibitem{bojowald2012}
Martin Bojowald and George~M. Paily.
\newblock {Deformed General Relativity and Effective Actions from Loop Quantum
  Gravity}.
\newblock {\em Phys.Rev.}, D86:104018, 2012, 1112.1899.

\bibitem{burgess2003}
C.P. Burgess.
\newblock {Quantum gravity in everyday life: General relativity as an effective
  field theory}.
\newblock {\em Living Rev.Rel.}, 7:5, 2004, gr-qc/0311082.

\bibitem{donoghue2012}
John~F. Donoghue.
\newblock {The effective field theory treatment of quantum gravity}.
\newblock {\em AIP Conf.Proc.}, 1483:73--94, 2012, 1209.3511.

\bibitem{reuter2012}
Martin Reuter and Frank Saueressig.
\newblock {Quantum Einstein Gravity}.
\newblock {\em New J.Phys.}, 14:055022, 2012, 1202.2274.

\bibitem{dario2013}
Dario Benedetti.
\newblock {On the number of relevant operators in asymptotically safe gravity}.
\newblock 2013, 1301.4422.

\bibitem{nicolai2005}
Hermann Nicolai, Kasper Peeters, and Marija Zamaklar.
\newblock {Loop quantum gravity: An Outside view}.
\newblock {\em Class.Quant.Grav.}, 22:R193, 2005, hep-th/0501114.

\bibitem{dewitt2011pursuit}
C{\'e}cile DeWitt-Morette.
\newblock {\em The pursuit of quantum gravity}.
\newblock Springerverlag Berlin Heidelberg, 2011.

\bibitem{dirac_qm}
PAM Dirac.
\newblock Lectures on quantum mechanics.
\newblock 1964.

\bibitem{dewitt_can}
Bryce~S DeWitt.
\newblock Quantum theory of gravity. i. the canonical theory.
\newblock {\em Physical Review}, 160(5):1113, 1967.

\bibitem{dewitt_cov}
Bryce~S DeWitt.
\newblock Quantum theory of gravity. ii. the manifestly covariant theory.
\newblock {\em Physical Review}, 162(5):1195, 1967.

\bibitem{ashtekar1986}
Abhay Ashtekar.
\newblock New variables for classical and quantum gravity.
\newblock {\em Physical Review Letters}, 57(18):2244--2247, 1986.

\bibitem{barbero1994}
J.~Fernando Barbero~G.
\newblock {Real Ashtekar variables for Lorentzian signature space times}.
\newblock {\em Phys.Rev.}, D51:5507--5510, 1995, gr-qc/9410014.

\bibitem{jacobson1995}
Ted Jacobson.
\newblock {Thermodynamics of space-time: The Einstein equation of state}.
\newblock {\em Phys.Rev.Lett.}, 75:1260--1263, 1995, gr-qc/9504004.

\bibitem{lorenzo2009}
L.~Sindoni.
\newblock {Emergent Gravity: The Analogue Models Perspective}.
\newblock 2009.

\bibitem{barcelo2005}
Carlos Barcelo, Stefano Liberati, and Matt Visser.
\newblock {Analogue gravity}.
\newblock {\em Living Rev.Rel.}, 8:12, 2005, gr-qc/0505065.

\bibitem{rovelli_smolin1994}
Carlo Rovelli and Lee Smolin.
\newblock {Discreteness of area and volume in quantum gravity}.
\newblock {\em Nucl.Phys.}, B442:593--622, 1995, gr-qc/9411005.

\bibitem{pranzetti_sigma}
Jacobo Diaz-Polo and Daniele Pranzetti.
\newblock {Isolated Horizons and Black Hole Entropy In Loop Quantum Gravity}.
\newblock {\em SIGMA}, 8:048, 2012, 1112.0291.

\bibitem{weinberg1}
Steven Weinberg.
\newblock {\em The Quantum Theory Of Fields: Foundations}, volume~1.
\newblock Cambridge university press, 1996.

\bibitem{wilson_nobel}
Kenneth~G Wilson.
\newblock Nobel lecture.
\newblock {\em Nobelprize.org}, 1965.

\bibitem{vincent_book}
Vincent Rivasseau.
\newblock {\em From perturbative to constructive renormalization}.
\newblock Princeton University Press New Jersey, 1991.

\bibitem{salmhofer}
Manfred Salmhofer.
\newblock {\em Renormalization}.
\newblock Springer Verlag, 1999.

\bibitem{straumann}
Norbert Straumann.
\newblock {\em General Relativity with Applications to Astrophysics}.
\newblock Springer-Verlag, 2004.

\bibitem{vincent_tt1}
Vincent Rivasseau.
\newblock {Quantum Gravity and Renormalization: The Tensor Track}.
\newblock {\em AIP Conf.Proc.}, 1444:18--29, 2011, 1112.5104.

\bibitem{vincent_tt2}
Vincent Rivasseau.
\newblock {The Tensor Track: an Update}.
\newblock 2012, 1209.5284.

\bibitem{edr}
Etera~R. Livine, Daniele Oriti, and James~P. Ryan.
\newblock {Effective Hamiltonian Constraint from Group Field Theory}.
\newblock {\em Class.Quant.Grav.}, 28:245010, 2011, 1104.5509.

\bibitem{gfc}
Steffen Gielen, Daniele Oriti, and Lorenzo Sindoni.
\newblock {Cosmology from Group Field Theory Formalism for Quantum Gravity}.
\newblock 2013, 1303.3576.

\bibitem{bianca_cyl}
Bianca Dittrich.
\newblock {From the discrete to the continuous: Towards a cylindrically
  consistent dynamics}.
\newblock {\em New J.Phys.}, 14:123004, 2012, 1205.6127.

\bibitem{bianca_review}
Bianca Dittrich.
\newblock {How to construct diffeomorphism symmetry on the lattice}.
\newblock {\em PoS}, QGQGS2011:012, 2011, 1201.3840.

\bibitem{bahr2012}
Benjamin Bahr, Bianca Dittrich, Frank Hellmann, and Wojciech Kaminski.
\newblock {Holonomy Spin Foam Models: Definition and Coarse Graining}.
\newblock {\em Phys.Rev.}, D87:044048, 2013, 1208.3388.

\bibitem{boulatov}
D.V. Boulatov.
\newblock {A Model of three-dimensional lattice gravity}.
\newblock {\em Mod.Phys.Lett.}, A7:1629--1646, 1992, hep-th/9202074.

\bibitem{vertex}
Sylvain Carrozza and Daniele Oriti.
\newblock {Bounding bubbles: the vertex representation of 3d Group Field Theory
  and the suppression of pseudo-manifolds}.
\newblock {\em Phys.Rev.}, D85:044004, 2012, 1104.5158.

\bibitem{edge}
Sylvain Carrozza and Daniele Oriti.
\newblock {Bubbles and jackets: new scaling bounds in topological group field
  theories}.
\newblock {\em JHEP}, 1206:092, 2012, 1203.5082.

\bibitem{u1}
Sylvain Carrozza, Daniele Oriti, and Vincent Rivasseau.
\newblock {Renormalization of Tensorial Group Field Theories: Abelian U(1)
  Models in Four Dimensions}.
\newblock 2012, 1207.6734.

\bibitem{su2}
Sylvain Carrozza, Daniele Oriti, and Vincent Rivasseau.
\newblock {Renormalization of an SU(2) Tensorial Group Field Theory in Three
  Dimensions}.
\newblock {\em Commun.Math.Phys.}, to appear.

\bibitem{beta_su2}
Sylvain Carrozza.
\newblock {In preparation}.

\bibitem{boulatov_phase}
Aristide Baratin, Sylvain Carrozza, Daniele Oriti, James~P. Ryan, and Matteo
  Smerlak.
\newblock {Melonic phase transition in group field theory}.
\newblock 2013, 1307.5026.

\bibitem{Laddha:2011mk}
Alok Laddha and Madhavan Varadarajan.
\newblock {The Diffeomorphism Constraint Operator in Loop Quantum Gravity}.
\newblock {\em Class.Quant.Grav.}, 28:195010, 2011, 1105.0636.

\bibitem{casey_weak}
Casey Tomlin and Madhavan Varadarajan.
\newblock {Towards an Anomaly-Free Quantum Dynamics for a Weak Coupling Limit
  of Euclidean Gravity}.
\newblock {\em Phys.Rev.}, D87:044039, 2013, 1210.6869.

\bibitem{adm}
R~Arnowitt, S~Deser, and CW~Misner.
\newblock Canonical variables for general relativity.
\newblock {\em Physical Review}, 117(6):1595, 1960.

\bibitem{Geiller:2012dd}
Marc Geiller and Karim Noui.
\newblock {A note on the Holst action, the time gauge, and the Barbero-Immirzi
  parameter}.
\newblock 2012, 1212.5064.

\bibitem{henneaux_teitelboim}
Marc Henneaux and Claudio Teitelboim.
\newblock {\em Quantization of gauge systems}.
\newblock Princeton university press, 1992.

\bibitem{lost}
Jerzy Lewandowski, Andrzej Okolow, Hanno Sahlmann, and Thomas Thiemann.
\newblock {Uniqueness of diffeomorphism invariant states on holonomy-flux
  algebras}.
\newblock {\em Commun.Math.Phys.}, 267:703--733, 2006, gr-qc/0504147.

\bibitem{Ashtekar:1994wa}
Abhay Ashtekar and Jerzy Lewandowski.
\newblock {Differential geometry on the space of connections via graphs and
  projective limits}.
\newblock {\em J.Geom.Phys.}, 17:191--230, 1995, hep-th/9412073.

\bibitem{Zapata:2004xq}
Jose~A. Zapata.
\newblock {Loop quantization from a lattice gauge theory perspective}.
\newblock {\em Class.Quant.Grav.}, 21:L115--L122, 2004, gr-qc/0401109.

\bibitem{Ashtekar:1994mh}
Abhay Ashtekar and Jerzy Lewandowski.
\newblock {Projective techniques and functional integration for gauge
  theories}.
\newblock {\em J.Math.Phys.}, 36:2170--2191, 1995, gr-qc/9411046.

\bibitem{Reisenberger:1996pu}
Michael~P Reisenberger and Carlo Rovelli.
\newblock {'Sum over surfaces' form of loop quantum gravity}.
\newblock {\em Phys.Rev.}, D56:3490--3508, 1997, gr-qc/9612035.

\bibitem{Thiemann:1996ay}
T.~Thiemann.
\newblock {Anomaly - free formulation of nonperturbative, four-dimensional
  Lorentzian quantum gravity}.
\newblock {\em Phys.Lett.}, B380:257--264, 1996, gr-qc/9606088.

\bibitem{Bianchi:2012nk}
Eugenio Bianchi and Frank Hellmann.
\newblock {The Construction of Spin Foam Vertex Amplitudes}.
\newblock {\em SIGMA}, 9:008, 2013, 1207.4596.

\bibitem{eprl}
Jonathan Engle, Etera Livine, Roberto Pereira, and Carlo Rovelli.
\newblock {LQG vertex with finite Immirzi parameter}.
\newblock {\em Nucl.Phys.}, B799:136--149, 2008, 0711.0146.

\bibitem{Alexandrov:2011ab}
Sergei Alexandrov, Marc Geiller, and Karim Noui.
\newblock {Spin Foams and Canonical Quantization}.
\newblock {\em SIGMA}, 8:055, 2012, 1112.1961.

\bibitem{Baez:1997zt}
John~C. Baez.
\newblock {Spin foam models}.
\newblock {\em Class.Quant.Grav.}, 15:1827--1858, 1998, gr-qc/9709052.

\bibitem{bc}
John~W. Barrett and Louis Crane.
\newblock {A Lorentzian signature model for quantum general relativity}.
\newblock {\em Class.Quant.Grav.}, 17:3101--3118, 2000, gr-qc/9904025.

\bibitem{fk}
Laurent Freidel and Kirill Krasnov.
\newblock {A New Spin Foam Model for 4d Gravity}.
\newblock {\em Class.Quant.Grav.}, 25:125018, 2008, 0708.1595.

\bibitem{Dupuis:2011fz}
Maite Dupuis and Etera~R. Livine.
\newblock {Holomorphic Simplicity Constraints for 4d Spinfoam Models}.
\newblock {\em Class.Quant.Grav.}, 28:215022, 2011, 1104.3683.

\bibitem{Dupuis:2011wy}
Maite Dupuis, Laurent Freidel, Etera~R. Livine, and Simone Speziale.
\newblock {Holomorphic Lorentzian Simplicity Constraints}.
\newblock {\em J.Math.Phys.}, 53:032502, 2012, 1107.5274.

\bibitem{bo_bc}
Aristide Baratin and Daniele Oriti.
\newblock {Quantum simplicial geometry in the group field theory formalism:
  reconsidering the Barrett-Crane model}.
\newblock {\em New J.Phys.}, 13:125011, 2011, 1108.1178.

\bibitem{bo_holst}
Aristide Baratin and Daniele Oriti.
\newblock {Group field theory and simplicial gravity path integrals: A model
  for Holst-Plebanski gravity}.
\newblock {\em Phys.Rev.}, D85:044003, 2012, 1111.5842.

\bibitem{Baratin:2010wi}
Aristide Baratin and Daniele Oriti.
\newblock {Group field theory with non-commutative metric variables}.
\newblock {\em Phys.Rev.Lett.}, 105:221302, 2010, 1002.4723.

\bibitem{ponzano_regge}
Giorgio Ponzano and Tullio Regge.
\newblock Semiclassical limit of racah coefficients.
\newblock Technical report, Princeton Univ., NJ, 1969.

\bibitem{turaev_viro}
V.G. Turaev and O.Y. Viro.
\newblock {State sum invariants of 3 manifolds and quantum 6j symbols}.
\newblock {\em Topology}, 31:865--902, 1992.

\bibitem{Barrett:1995mg}
John~W. Barrett.
\newblock {Quantum gravity as topological quantum field theory}.
\newblock {\em J.Math.Phys.}, 36:6161--6179, 1995, gr-qc/9506070.

\bibitem{Bahr:2011uj}
Benjamin Bahr, Bianca Dittrich, and Sebastian Steinhaus.
\newblock {Perfect discretization of reparametrization invariant path
  integrals}.
\newblock {\em Phys.Rev.}, D83:105026, 2011, 1101.4775.

\bibitem{Dittrich:2011zh}
Bianca Dittrich, Frank~C. Eckert, and Mercedes Martin-Benito.
\newblock {Coarse graining methods for spin net and spin foam models}.
\newblock {\em New J.Phys.}, 14:035008, 2012, 1109.4927.

\bibitem{Dittrich:2011vz}
Bianca Dittrich and Sebastian Steinhaus.
\newblock {Path integral measure and triangulation independence in discrete
  gravity}.
\newblock {\em Phys.Rev.}, D85:044032, 2012, 1110.6866.

\bibitem{dPFKR}
Roberto De~Pietri, Laurent Freidel, Kirill Krasnov, and Carlo Rovelli.
\newblock {Barrett-Crane model from a Boulatov-Ooguri field theory over a
  homogeneous space}.
\newblock {\em Nucl.Phys.}, B574:785--806, 2000, hep-th/9907154.

\bibitem{GFT_rovelli_reisenberg}
Michael~P. Reisenberger and Carlo Rovelli.
\newblock {Space-time as a Feynman diagram: The Connection formulation}.
\newblock {\em Class.Quant.Grav.}, 18:121--140, 2001, gr-qc/0002095.

\bibitem{Perez:2003}
Alejandro Perez.
\newblock {Spin foam models for quantum gravity}.
\newblock {\em Class.Quant.Grav.}, 20:R43, 2003, gr-qc/0301113.

\bibitem{freidel_gft}
Laurent Freidel.
\newblock {Group field theory: An Overview}.
\newblock {\em Int.J.Theor.Phys.}, 44:1769--1783, 2005, hep-th/0505016.

\bibitem{daniele_rev2006}
Daniele Oriti.
\newblock {The Group field theory approach to quantum gravity}.
\newblock 2006, gr-qc/0607032.

\bibitem{Krajewski:2010yq}
Thomas Krajewski, Jacques Magnen, Vincent Rivasseau, Adrian Tanasa, and
  Patrizia Vitale.
\newblock {Quantum Corrections in the Group Field Theory Formulation of the
  EPRL/FK Models}.
\newblock {\em Phys.Rev.}, D82:124069, 2010, 1007.3150.

\bibitem{ds_3}
Steffen Gielen and Daniele Oriti.
\newblock {Discrete and continuum third quantization of Gravity}.
\newblock 2011, 1102.2226.

\bibitem{lqg_propa1}
Emanuele Alesci and Carlo Rovelli.
\newblock {The Complete LQG propagator. I. Difficulties with the Barrett-Crane
  vertex}.
\newblock {\em Phys.Rev.}, D76:104012, 2007, 0708.0883.

\bibitem{lqg_propa2}
Emanuele Alesci and Carlo Rovelli.
\newblock {The Complete LQG propagator. II. Asymptotic behavior of the vertex}.
\newblock {\em Phys.Rev.}, D77:044024, 2008, 0711.1284.

\bibitem{lqg_propa3}
Eugenio Bianchi, Elena Magliaro, and Claudio Perini.
\newblock {LQG propagator from the new spin foams}.
\newblock {\em Nucl.Phys.}, B822:245--269, 2009, 0905.4082.

\bibitem{david1985planar}
Fran{\c{c}}ois David.
\newblock Planar diagrams, two-dimensional lattice gravity and surface models.
\newblock {\em Nuclear Physics B}, 257:45--58, 1985.

\bibitem{Ginsparg:1991bi}
Paul~H. Ginsparg.
\newblock {Matrix models of 2-d gravity}.
\newblock 1991, hep-th/9112013.

\bibitem{DiFrancesco:1993nw}
P.~Di~Francesco, Paul~H. Ginsparg, and Jean Zinn-Justin.
\newblock {2-D Gravity and random matrices}.
\newblock {\em Phys.Rept.}, 254:1--133, 1995, hep-th/9306153.

\bibitem{Ginsparg:1993is}
Paul~H. Ginsparg and Gregory~W. Moore.
\newblock {Lectures on 2-D gravity and 2-D string theory}.
\newblock 1993, hep-th/9304011.

\bibitem{staudacher_moore_seiberg}
Gregory~W. Moore, Nathan Seiberg, and Matthias Staudacher.
\newblock {From loops to states in 2-D quantum gravity}.
\newblock {\em Nucl.Phys.}, B362:665--709, 1991.

\bibitem{legall}
J.-F. Le~Gall and G.~Miermont.
\newblock {Scaling limits of random trees and planar maps}.
\newblock 1101.4856.

\bibitem{Ambjorn_tensors}
Jan Ambjorn, Bergfinnur Durhuus, and Thordur Jonsson.
\newblock {Three-dimensional simplicial quantum gravity and generalized matrix
  models}.
\newblock {\em Mod.Phys.Lett.}, A6:1133--1146, 1991.

\bibitem{gross}
Mark Gross.
\newblock {Tensor models and simplicial quantum gravity in > 2-D}.
\newblock {\em Nucl.Phys.Proc.Suppl.}, 25A:144--149, 1992.

\bibitem{Sasakura:1990fs}
Naoki Sasakura.
\newblock {Tensor model for gravity and orientability of manifold}.
\newblock {\em Mod.Phys.Lett.}, A6:2613--2624, 1991.

\bibitem{Gurau:2010nd}
Razvan Gurau.
\newblock {Lost in Translation: Topological Singularities in Group Field
  Theory}.
\newblock {\em Class.Quant.Grav.}, 27:235023, 2010, 1006.0714.

\bibitem{Ambjorn:2012jv}
J.~Ambjorn, A.~Goerlich, J.~Jurkiewicz, and R.~Loll.
\newblock {Nonperturbative Quantum Gravity}.
\newblock {\em Phys.Rept.}, 519:127--210, 2012, 1203.3591.

\bibitem{Barbieri1997}
A.~Barbieri.
\newblock {Quantum tetrahedra and simplicial spin networks}.
\newblock {\em Nucl.Phys.}, B518:714--728, 1998, gr-qc/9707010.

\bibitem{BaezBarrett_tetra}
John~C. Baez and John~W. Barrett.
\newblock {The Quantum tetrahedron in three-dimensions and four-dimensions}.
\newblock {\em Adv.Theor.Math.Phys.}, 3:815--850, 1999, gr-qc/9903060.

\bibitem{daniele_rev2011}
Daniele Oriti.
\newblock {The microscopic dynamics of quantum space as a group field theory}.
\newblock pages 257--320, 2011, 1110.5606.

\bibitem{lr_loll}
Renate Loll.
\newblock {Discrete approaches to quantum gravity in four-dimensions}.
\newblock {\em Living Rev.Rel.}, 1:13, 1998, gr-qc/9805049.

\bibitem{PR3}
Laurent Freidel and Etera~R. Livine.
\newblock {Ponzano-Regge model revisited III: Feynman diagrams and effective
  field theory}.
\newblock {\em Class.Quant.Grav.}, 23:2021--2062, 2006, hep-th/0502106.

\bibitem{majidfreidel}
Laurent Freidel and Shahn Majid.
\newblock {Noncommutative harmonic analysis, sampling theory and the Duflo map
  in 2+1 quantum gravity}.
\newblock {\em Class.Quant.Grav.}, 25:045006, 2008, hep-th/0601004.

\bibitem{karim}
E.~Joung, J.~Mourad, and K.~Noui.
\newblock {Three Dimensional Quantum Geometry and Deformed Poincare Symmetry}.
\newblock {\em J.Math.Phys.}, 50:052503, 2009, 0806.4121.

\bibitem{cdr_duflo}
Carlos Guedes, Daniele Oriti, and Matti Raasakka.
\newblock {Quantization maps, algebra representation and non-commutative
  Fourier transform for Lie groups}.
\newblock 2013, 1301.7750.

\bibitem{daniele_hydro}
Daniele Oriti.
\newblock {Group field theory as the microscopic description of the quantum
  spacetime fluid: A New perspective on the continuum in quantum gravity}.
\newblock {\em PoS}, QG-PH:030, 2007, 0710.3276.

\bibitem{stachel1986einstein}
John Stachel.
\newblock Einstein and the quantum: fifty years of struggle.
\newblock {\em From Quarks to Quasars: Philosophical Problems of Modern
  Physics}, pages 349--81, 1986.

\bibitem{Gurau:2009tw}
Razvan Gurau.
\newblock {Colored Group Field Theory}.
\newblock {\em Commun.Math.Phys.}, 304:69--93, 2011, 0907.2582.

\bibitem{jimmy}
James~P. Ryan.
\newblock {Tensor models and embedded Riemann surfaces}.
\newblock {\em Phys.Rev.}, D85:024010, 2012, 1104.5471.

\bibitem{RazvanN}
Razvan Gurau.
\newblock {The 1/N expansion of colored tensor models}.
\newblock {\em Annales Henri Poincare}, 12:829--847, 2011, 1011.2726.

\bibitem{RazvanVincentN}
Razvan Gurau and Vincent Rivasseau.
\newblock {The 1/N expansion of colored tensor models in arbitrary dimension}.
\newblock {\em Europhys.Lett.}, 95:50004, 2011, 1101.4182.

\bibitem{Gurau:2011xq}
Razvan Gurau.
\newblock {The complete 1/N expansion of colored tensor models in arbitrary
  dimension}.
\newblock {\em Annales Henri Poincare}, 13:399--423, 2012, 1102.5759.

\bibitem{diffeos}
Aristide Baratin, Florian Girelli, and Daniele Oriti.
\newblock {Diffeomorphisms in group field theories}.
\newblock {\em Phys.Rev.}, D83:104051, 2011, 1101.0590.

\bibitem{Gurau:2011xp}
Razvan Gurau and James~P. Ryan.
\newblock {Colored Tensor Models - a review}.
\newblock {\em SIGMA}, 8:020, 2012, 1109.4812.

\bibitem{critical}
Valentin Bonzom, Razvan Gurau, Aldo Riello, and Vincent Rivasseau.
\newblock {Critical behavior of colored tensor models in the large N limit}.
\newblock {\em Nucl.Phys.}, B853:174--195, 2011, 1105.3122.

\bibitem{dr_dually}
Dario Benedetti and Razvan Gurau.
\newblock {Phase Transition in Dually Weighted Colored Tensor Models}.
\newblock {\em Nucl.Phys.}, B855:420--437, 2012, 1108.5389.

\bibitem{v_revisiting}
Valentin Bonzom.
\newblock {Revisiting random tensor models at large N via the Schwinger-Dyson
  equations}.
\newblock {\em JHEP}, 1303:160, 2013, 1208.6216.

\bibitem{v_new}
Valentin Bonzom.
\newblock {New 1/N expansions in random tensor models}.
\newblock {\em JHEP}, 1306:062, 2013, 1211.1657.

\bibitem{jr_branched}
Razvan Gurau and James~P. Ryan.
\newblock {Melons are branched polymers}.
\newblock 2013, 1302.4386.

\bibitem{wjd_double}
Wojciech Kaminski, Daniele Oriti, and James~P. Ryan.
\newblock {Towards a double-scaling limit for tensor models: probing
  sub-dominant orders}.
\newblock 2013, 1304.6934.

\bibitem{vrv_ising}
Valentin Bonzom, Razvan Gurau, and Vincent Rivasseau.
\newblock {The Ising Model on Random Lattices in Arbitrary Dimensions}.
\newblock {\em Phys.Lett.}, B711:88--96, 2012, 1108.6269.

\bibitem{v_dimers}
Valentin Bonzom.
\newblock {Multicritical tensor models and hard dimers on spherical random
  lattices}.
\newblock {\em Phys.Lett.}, A377:501--506, 2013, 1201.1931.

\bibitem{vh_dimers}
Valentin Bonzom and Harold Erbin.
\newblock {Coupling of hard dimers to dynamical lattices via random tensors}.
\newblock {\em J.Stat.Mech.}, 1209:P09009, 2012, 1204.3798.

\bibitem{pspin}
Valentin Bonzom, Razvan Gurau, and Matteo Smerlak.
\newblock Universality in p-spin glasses with correlated disorder.
\newblock {\em Journal of Statistical Mechanics: Theory and Experiment},
  2013(02):L02003, 2013.

\bibitem{vf_packed}
Valentin Bonzom and Frédéric Combes.
\newblock {Fully packed loops on random surfaces and the 1/N expansion of
  tensor models}.
\newblock 2013, 1304.4152.

\bibitem{universality}
Razvan Gurau.
\newblock {Universality for Random Tensors}.
\newblock 2011, 1111.0519.

\bibitem{r_vir}
Razvan Gurau.
\newblock {A generalization of the Virasoro algebra to arbitrary dimensions}.
\newblock {\em Nucl.Phys.}, B852:592--614, 2011, 1105.6072.

\bibitem{uncoloring}
Valentin Bonzom, Razvan Gurau, and Vincent Rivasseau.
\newblock {Random tensor models in the large N limit: Uncoloring the colored
  tensor models}.
\newblock {\em Phys.Rev.}, D85:084037, 2012, 1202.3637.

\bibitem{tensor_4d}
Joseph Ben~Geloun and Vincent Rivasseau.
\newblock {A Renormalizable 4-Dimensional Tensor Field Theory}.
\newblock {\em Commun.Math.Phys.}, 318:69--109, 2013, 1111.4997.

\bibitem{FerriGagliardi}
Massimo Ferri and Carlo Gagliardi.
\newblock Crystallisation moves.
\newblock {\em Pacific J. Math}, 100(1):85--103, 1982.

\bibitem{Vince_gene}
Andrew Vince.
\newblock n-graphs.
\newblock {\em Discrete Mathematics}, 72(1):367--380, 1988.

\bibitem{lin}
Joseph Ben~Geloun, Thomas Krajewski, Jacques Magnen, and Vincent Rivasseau.
\newblock {Linearized Group Field Theory and Power Counting Theorems}.
\newblock {\em Class.Quant.Grav.}, 27:155012, 2010, 1002.3592.

\bibitem{r_sdgene}
Razvan Gurau.
\newblock {The Schwinger Dyson equations and the algebra of constraints of
  random tensor models at all orders}.
\newblock {\em Nucl.Phys.}, B865:133--147, 2012, 1203.4965.

\bibitem{r_constructive}
Razvan Gurau.
\newblock {The 1/N Expansion of Tensor Models Beyond Perturbation Theory}.
\newblock 2013, 1304.2666.

\bibitem{rz_loop1}
Vincent Rivasseau and Zhituo Wang.
\newblock {Constructive Renormalization for $\Phi^{4}_2$ Theory with Loop
  Vertex Expansion}.
\newblock {\em J.Math.Phys.}, 53:042302, 2012, 1104.3443.

\bibitem{rz_loop2}
Vincent Rivasseau and Zhituo Wang.
\newblock {How to Resum Feynman Graphs}.
\newblock 2013, 1304.5913.

\bibitem{josephsamary}
Joseph Ben~Geloun and Dine~Ousmane Samary.
\newblock {3D Tensor Field Theory: Renormalization and One-loop
  $\beta$-functions}.
\newblock 2012, 1201.0176.

\bibitem{joseph_etera}
Joseph Ben~Geloun and Etera~R. Livine.
\newblock {Some classes of renormalizable tensor models}.
\newblock 2012, 1207.0416.

\bibitem{fabien_dine}
Dine~Ousmane Samary and Fabien Vignes-Tourneret.
\newblock {Just Renormalizable TGFT's on $U(1)^{d}$ with Gauge Invariance}.
\newblock 2012, 1211.2618.

\bibitem{joseph_d2}
Joseph~Ben Geloun.
\newblock {Renormalizable Models in Rank $d \geq 2$ Tensorial Group Field
  Theory}.
\newblock 2013, 1306.1201.

\bibitem{cboulatov}
Joseph Ben~Geloun, Jacques Magnen, and Vincent Rivasseau.
\newblock {Bosonic Colored Group Field Theory}.
\newblock {\em Eur.Phys.J.}, C70:1119--1130, 2010, 0911.1719.

\bibitem{laurentdiffeo}
Laurent Freidel and David Louapre.
\newblock {Diffeomorphisms and spin foam models}.
\newblock {\em Nucl.Phys.}, B662:279--298, 2003, gr-qc/0212001.

\bibitem{biancadiffeos}
Bianca Dittrich.
\newblock {Diffeomorphism symmetry in quantum gravity models}.
\newblock 2008, 0810.3594.

\bibitem{biancabenny}
Benjamin Bahr and Bianca Dittrich.
\newblock {(Broken) Gauge Symmetries and Constraints in Regge Calculus}.
\newblock {\em Class.Quant.Grav.}, 26:225011, 2009, 0905.1670.

\bibitem{BarrettCraneWdW}
John~W. Barrett and Louis Crane.
\newblock {An Algebraic interpretation of the Wheeler-DeWitt equation}.
\newblock {\em Class.Quant.Grav.}, 14:2113--2121, 1997, gr-qc/9609030.

\bibitem{eterasimonevalentin}
Valentin Bonzom, Etera~R. Livine, and Simone Speziale.
\newblock {Recurrence relations for spin foam vertices}.
\newblock {\em Class.Quant.Grav.}, 27:125002, 2010, 0911.2204.

\bibitem{Vince_2d}
Andrew Vince.
\newblock The classification of closed surfaces using colored graphs.
\newblock {\em Graphs and Combinatorics}, 9(1):75--84, 1993.

\bibitem{lrd}
Laurent Freidel, Razvan Gurau, and Daniele Oriti.
\newblock {Group field theory renormalization - the 3d case: Power counting of
  divergences}.
\newblock {\em Phys.Rev.}, D80:044007, 2009, 0905.3772.

\bibitem{ValentinJoseph}
Joseph Ben~Geloun and Valentin Bonzom.
\newblock {Radiative corrections in the Boulatov-Ooguri tensor model: The
  2-point function}.
\newblock {\em Int.J.Theor.Phys.}, 50:2819--2841, 2011, 1101.4294.

\bibitem{scaling3d}
Jacques Magnen, Karim Noui, Vincent Rivasseau, and Matteo Smerlak.
\newblock {Scaling behaviour of three-dimensional group field theory}.
\newblock {\em Class.Quant.Grav.}, 26:185012, 2009, 0906.5477.

\bibitem{francesco}
Francesco Caravelli.
\newblock {A Simple Proof of Orientability in Colored Group Field Theory}.
\newblock {\em SpringerPlus}, 1:6, 2012, 1012.4087.

\bibitem{Ooguri}
Hirosi Ooguri.
\newblock {Topological lattice models in four-dimensions}.
\newblock {\em Mod.Phys.Lett.}, A7:2799--2810, 1992, hep-th/9205090.

\bibitem{vincent_2002}
Vincent Rivasseau.
\newblock An introduction to renormalization.
\newblock In {\em Poincar{\'e} Seminar 2002}, pages 139--177. Springer, 2003.

\bibitem{Simon}
Barry Simon.
\newblock {\em The $P(\Phi)_2$ Euclidean (quantum) field theory}.
\newblock Princeton University Press (Princeton, NJ), 1974.

\bibitem{josephaf}
Joseph Ben~Geloun.
\newblock {Two and four-loop $\beta$-functions of rank 4 renormalizable tensor
  field theories}.
\newblock {\em Class.Quant.Grav.}, 29:235011, 2012, 1205.5513.

\bibitem{dineaf}
Dine~Ousmane Samary.
\newblock {Beta functions of $U(1)^d$ gauge invariant just renormalizable
  tensor models}.
\newblock 2013, 1303.7256.

\bibitem{aldo}
Aldo Riello.
\newblock {Self-Energy of the Lorentzian EPRL-FK Spin Foam Model of Quantum
  Gravity}.
\newblock 2013, 1302.1781.

\bibitem{vincent_renGFT}
Vincent Rivasseau.
\newblock {Towards Renormalizing Group Field Theory}.
\newblock {\em PoS}, CNCFG2010:004, 2010, 1103.1900.

\bibitem{vm1}
Valentin Bonzom and Matteo Smerlak.
\newblock {Bubble divergences from cellular cohomology}.
\newblock {\em Lett.Math.Phys.}, 93:295--305, 2010, 1004.5196.

\bibitem{vm2}
Valentin Bonzom and Matteo Smerlak.
\newblock {Bubble divergences from twisted cohomology}.
\newblock {\em Commun.Math.Phys.}, 312:399--426, 2012, 1008.1476.

\bibitem{vm3}
Valentin Bonzom and Matteo Smerlak.
\newblock {Bubble divergences: sorting out topology from cell structure}.
\newblock {\em Annales Henri Poincare}, 13:185--208, 2012, 1103.3961.

\bibitem{pr1}
Laurent Freidel and David Louapre.
\newblock {Ponzano-Regge model revisited I: Gauge fixing, observables and
  interacting spinning particles}.
\newblock {\em Class.Quant.Grav.}, 21:5685--5726, 2004, hep-th/0401076.

\bibitem{addendum}
Joseph~Ben Geloun and Vincent Rivasseau.
\newblock {Addendum to 'A Renormalizable 4-Dimensional Tensor Field Theory'}.
\newblock 2012, 1209.4606.

\bibitem{ef_poincare}
Florian Girelli and Etera~R. Livine.
\newblock {A Deformed Poincare Invariance for Group Field Theories}.
\newblock {\em Class.Quant.Grav.}, 27:245018, 2010, 1001.2919.

\bibitem{Dartois}
Stephane Dartois, Razvan Gurau, and Vincent Rivasseau.
\newblock {Double Scaling in Tensor Models with a Quartic Interaction}.
\newblock 2013, 1307.5281.

\bibitem{schaeffer}
Gurau Razvan and Schaeffer Gilles.
\newblock {Regular colored graphs of positive degree}.
\newblock 2013, 1307.5279.

\end{thebibliography}
\addcontentsline{toc}{chapter}{Bibliography}

\end{document}